\documentclass[acmsmall]{acmart}

\AtBeginDocument{%
  }

\setcopyright{cc}
\setcctype{by}
\acmDOI{10.1145/3720509}
\acmYear{2025}
\acmJournal{PACMPL}
\acmVolume{9}
\acmNumber{OOPSLA1}
\acmArticle{144}
\acmMonth{4}
\received{2024-10-16}
\received[accepted]{2025-02-18}

\usepackage{float}
\usepackage{multirow}
\usepackage{mathpartir}
\usepackage{fancyvrb}
\usepackage{xcolor}
\usepackage{subcaption}
\usepackage{caption}
\usepackage{listings}
\usepackage{tikz}
\usepackage{syntax}
\usepackage{matlab-prettifier}
\usepackage{adjustbox}
\usepackage{xspace}
\usepackage{enumitem}
\usepackage{comment}
\usepackage{amsmath, bm}
\usepackage{wrapfig}
\usetikzlibrary{positioning}
\usepackage{stmaryrd}
\usetikzlibrary{calc}
\usetikzlibrary{fit}
\usepackage{relsize}
\usepackage{xfrac}
\usepackage{mathtools}

\newtheorem{definition}{Definition}[section]  % Separate numbering for definitions
\newtheorem{lemma}{Lemma}[section]            % Separate numbering for lemmas

\definecolor{WowColor}{rgb}{.75,0,.75}
\definecolor{SubtleColor}{rgb}{0,0,.50}
\newcounter{margincounter}

\newcommand{\cf}{\textsc{ProveSound}\xspace}
\newcommand{\oldtool}{\textsc{ConstraintFlow}\xspace}

\newcommand{\var}{x}
\newcommand{\constant}{c}
\newcommand{\expr}{e}
\newcommand{\sym}{\cfkeywords{\texttt{sym}}\xspace}
\newcommand{\curr}{\cfkeywords{\texttt{curr}}\xspace}
\newcommand{\prev}{\cfkeywords{\texttt{prev}}\xspace}

\newcommand{\shape}{\cfkeywords{\texttt{shape}}\xspace}
\newcommand{\transformer}{\cfkeywords{\texttt{Transformer}}\xspace}
\newcommand{\summ}{\cfkeywords{\texttt{sum}}\xspace}
\newcommand{\avg}{\cfkeywords{\texttt{avg}}\xspace}
\newcommand{\len}{\cfkeywords{\texttt{len}}\xspace}
\newcommand{\argmax}{\cfkeywords{compare}\xspace}
\newcommand{\maxx}{\cfkeywords{\texttt{max}}\xspace}
\newcommand{\minn}{\cfkeywords{\texttt{min}}\xspace}
\newcommand{\map}{\cfkeywords{\texttt{map}}\xspace}
\newcommand{\mapl}{\cfkeywords{\texttt{mapList}}\xspace}
\newcommand{\foo}{\texttt{applyFunc}\xspace}
\newcommand{\dott}{\cfkeywords{\texttt{dot}}\xspace}
\newcommand{\concat}{\cfkeywords{concat}\xspace}
\newcommand{\traverse}{\cfkeywords{\texttt{traverse}}\xspace}
\newcommand{\solver}{\cfkeywords{\texttt{solver}}\xspace}
\newcommand{\lp}{\cfkeywords{\texttt{solver}}\xspace}
\newcommand{\maximize}{\cfkeywords{\texttt{maximize}}\xspace}
\newcommand{\minimize}{\cfkeywords{\texttt{minimize}}\xspace}
\newcommand{\forward}{\cfkeywords{\texttt{forward}}\xspace}
\newcommand{\backward}{\cfkeywords{\texttt{backward}}\xspace}
\newcommand{\flow}{\cfkeywords{\texttt{Flow}}\xspace}
\newcommand{\affine}{\cfkeywords{\texttt{Affine}}\xspace}
\newcommand{\relu}{\cfkeywords{\texttt{ReLU}}\xspace}
\newcommand{\abs}{\cfkeywords{\texttt{Abs}}\xspace}
\newcommand{\hsigmoid}{\cfkeywords{\texttt{HardSigmoid}}\xspace}
\newcommand{\htanh}{\cfkeywords{\texttt{HardTanh}}\xspace}
\newcommand{\hswish}{\cfkeywords{\texttt{HardSwish}}\xspace}
\newcommand{\leakyrelu}{\cfkeywords{\texttt{ReLU6}}\xspace}
\newcommand{\maxpool}{\cfkeywords{\texttt{MaxPool}}\xspace}
\newcommand{\minpool}{\cfkeywords{\texttt{MinPool}}\xspace}
\newcommand{\avgpool}{\cfkeywords{\texttt{AvgPool}}\xspace}
\newcommand{\dotprod}{\cfkeywords{\texttt{DotProduct}}\xspace}

\newcommand{\revaffine}{\cfkeywords{\texttt{rev\_Affine}}\xspace}
\newcommand{\revrelu}{\cfkeywords{\texttt{rev\_ReLU}}\xspace}

\newcommand{\weight}{\cfkeywords{\texttt{weight}}\xspace}
\newcommand{\bias}{\cfkeywords{\texttt{bias}}\xspace}
\newcommand{\layer}{\cfkeywords{\texttt{layer}}\xspace}
\newcommand{\equations}{\cfkeywords{\texttt{equations}}\xspace}
\newcommand{\deeppolyH}{BALANCE Cert\xspace}
\newcommand{\deeppolyNew}{REUSE Cert\xspace}
\newcommand{\sympoly}{SymPoly\xspace}

\newcommand{\abshape}{s}
\newcommand{\types}{t}

\newcommand{\expand}{\mathsf{expanded}}
\newcommand{\ex}{\mathsf{expand}}
\newcommand{\exn}{\mathsf{expandN}}
\newcommand{\hh}{\mathcal{D}}
\newcommand{\cc}{\mathcal{C}}
\newcommand{\prop}{\mathcal{P}}
\newcommand{\fstore}{F}
\newcommand{\tstore}{\Theta}
\newcommand{\sstore}{\sigma}
\newcommand{\trav}{\cfkeywords{\texttt{traverse}}\xspace}
\newcommand{\addg}{\mathsf{add}}
\newcommand{\cinvariant}{\mathsf{Inv}}
\newcommand{\cinduction}{\mathsf{Ind}}
\newcommand{\unsat}{\mathsf{unsat}}
\newcommand{\sat}{\mathsf{sat}}
\newcommand{\g}{\Gamma, \tau_s}
\newcommand{\store}{\rho}

\newcommand{\cfkeywords}[1]{\texttt{\footnotesize\bfseries\textcolor{keywords}{#1}}}
\newcommand{\cftypewords}[1]{\texttt{\footnotesize\bfseries\textcolor{typewords}{#1}}}

\newcommand{\defshape}{\cfkeywords{{\texttt{Def shape as}}}\xspace}
\newcommand{\func}{\cfkeywords{\texttt{Func}}\xspace}
\newcommand{\mathif}{\cfkeywords{\texttt{if}}\xspace}
\newcommand{\embed}{\texttt{<>}\xspace}

\newcommand{\polyexp}{\cftypewords{\texttt{PolyExp}}\xspace}

\newcommand{\symexp}{\cftypewords{\texttt{SymExp}}\xspace}
\newcommand{\ct}{\cftypewords{\texttt{Ct}}\xspace}
\newcommand{\float}{\cftypewords{\texttt{Real}}\xspace}
\newcommand{\intt}{\cftypewords{\texttt{Int}}\xspace}
\newcommand{\bool}{\cftypewords{\texttt{Bool}}\xspace}
\newcommand{\typeneuron}{\cftypewords{\texttt{Neuron}}\xspace}
\newcommand{\typesym}{\cftypewords{\texttt{Sym}}\xspace}
\newcommand{\typenoise}{\cftypewords{\texttt{Sym}}\xspace}

\newcommand{\dir}{\delta}
\newcommand{\andx}{\mathsf{and}}
\newcommand{\orx}{\mathsf{or}}
\newcommand{\notx}{\mathsf{not}}
\newcommand{\binop}{\oplus}
\newcommand{\add}{\mathsf{sum}}
\newcommand{\lpop}{\mathsf{op}}
\newcommand{\noise}{{\epsilon}}
\newcommand{\val}{\nu}
\newcommand{\pval}{\val_p}
\newcommand{\zval}{\val_s}
\newcommand{\bval}{\val_b}
\newcommand{\lval}{\val_l}
\newcommand{\cval}{\val_c}
\newcommand{\sval}{\mu}
\newcommand{\bsval}{\mu_b}
\newcommand{\lsval}{\mu_l}

\newcommand{\vset}{\mathsf{V}}

\newcommand{\context}{\fstore, \store, \hh_C}
\newcommand{\scontext}{\fstore, \sstore, \hh_S, \cc}
\newcommand{\symc}{\fstore, \sstore, \hh_S, \cc}

\newcommand{\env}{\fstore, \tstore, \hh_C}

\newcommand{\true}{\mathsf{true}}
\newcommand{\false}{\mathsf{false}}
\newcommand{\filter}{\mathsf{Ft}}
\newcommand{\priority}{\mathsf{P}}
\newcommand{\neigh}{\mathsf{N}}
\newcommand{\ver}{v}
\newcommand{\vertices}{\mathsf{neurons}}
\newcommand{\rrr}{\mathsf{real}}

\newcommand{\type}{\mathsf{type}}
\newcommand{\inputt}{\mathsf{input}}
\newcommand{\outputt}{\mathsf{output}}
\newcommand{\height}{\mathsf{height}}

\newcommand{\reduced}{\mathcal{R}}
\newcommand{\nn}{\mathcal{N}}
\newcommand{\ee}{\mathcal{E}}
\newcommand{\rr}{\mathcal{R}}
\newcommand{\ttt}{\mathcal{T}}
\newcommand{\sss}{\mathcal{S}}
\newcommand{\ooo}{\mathcal{OP}}

\newcommand{\m}{\mathcal{M}}
\newcommand{\cs}{\mathsf{Constants}}
\newcommand{\ns}{\mathsf{Neurons}}
\newcommand{\es}{\mathsf{SymbolicVars}}
\newcommand{\cts}{\mathsf{Constraints}}
\newcommand{\ps}{\mathsf{PolyExps}}
\newcommand{\zs}{\mathsf{SymExps}}
\newcommand{\dom}{\mathsf{dom}}
\newcommand{\range}{\mathsf{range}}
\newcommand{\vars}{\mathsf{vars}}

\usepackage{graphicx}
\usepackage{caption}
\usepackage[skip=0.5ex]{subcaption}

\makeatletter
\newsavebox{\@brx}
\newcommand{\llangle}[1][]{\savebox{\@brx}{\(\m@th{#1\langle}\)}%
  \mathopen{\copy\@brx\kern-0.5\wd\@brx\usebox{\@brx}}}
\newcommand{\rrangle}[1][]{\savebox{\@brx}{\(\m@th{#1\rangle}\)}%
  \mathclose{\copy\@brx\kern-0.5\wd\@brx\usebox{\@brx}}}
\makeatother
\renewcommand{\ll}{\llangle}
\renewcommand{\gg}{\rrangle}

\newcounter{number}
\newcommand{\mycounter}{...(\thenumber) \stepcounter{number}}
\definecolor{diagramcolor}{rgb}{0.17,0.37,0.69}
\definecolor{keywords}{rgb}{0.05,0.05,0.9}
\definecolor{typewords}{rgb}{0,0.5,0}
\definecolor{greencomments}{rgb}{0,0.5,0}
\definecolor{turqusnumbers}{rgb}{0.17,0.57,0.69}
\definecolor{redstrings}{rgb}{0.5,0,0}

\definecolor{codegreen}{rgb}{0,0.6,0}
\definecolor{codegray}{rgb}{0.5,0.5,0.5}
\definecolor{codepurple}{rgb}{0.58,0,0.82}
\definecolor{backcolour}{RGB}{250, 250, 250}

\lstdefinelanguage{ConstraintFlow}
    {morekeywords={def, shape, as, curr, prev, prev0, prev1, mapList, transformer, ReLU, Affine, HardSwish, Maxpool, DotProduct, rev_ReLU, rev_Affine, rev_Maxpool, rev_Max, rev_Min, rev_Add, rev_Mult, func, map, true, false, traverse, dot, flow, forward, backward, sum, layer, sym, compare, avg, len, max, min, and, in, solver, currList, equations, minimize, maximize, mult, add, sigmoid, tanh, lp, Abs, eps},
    morekeywords = [4]{Bool, Int, Real, PolyExp, SymExp, Neuron, Noise, Ct},
    keywordstyle = \bfseries\color{keywords},
    keywordstyle = [4]{\bfseries\color{typewords}},
    sensitive=false, 
    morecomment=[l][\color{greencomments}]{///},
    morecomment=[l][\color{greencomments}]{//},
    morecomment=[s][\color{greencomments}]{{(*}{*)}},
    morestring=[b]",
    stringstyle=\color{redstrings}
    }

\lstnewenvironment{cflisting}
    {\lstset{language=ConstraintFlow,
                basicstyle=\ttfamily,
                breaklines=true,
                columns=fullflexible
    }}
    {}

\makeatletter
\lst@AddToHook{OnEmptyLine}{\addtocounter{lstnumber}{-1}}% Remove line number increment from listings
\makeatother

\begin{document}

\title{Automated Verification of Soundness of DNN Certifiers}

\author{Avaljot Singh}
\email{avaljot2@illinois.edu}
\orcid{0009-0006-4167-8709}
\affiliation{%
  \institution{University of Illinois Urbana-Champaign}
  \country{USA}
}

\author{Yasmin Chandini Sarita}
\affiliation{%
  \institution{University of Illinois Urbana-Champaign}
  \country{USA}
}
\email{ysarita2@illinois.edu}

\author{Charith Mendis}
\affiliation{%
  \institution{University of Illinois Urbana-Champaign}
  \country{USA}
}
\email{charithm@illinois.edu}

\author{Gagandeep Singh}
\affiliation{%
  \institution{University of Illinois Urbana-Champaign}
  \country{USA}
}
\email{ggnds@illinois.edu}

\begin{abstract}
The uninterpretability of Deep Neural Networks (DNNs) hinders their use in safety-critical applications. Abstract Interpretation-based DNN certifiers provide promising avenues for building trust in DNNs. Unsoundness in the mathematical logic of these certifiers can lead to incorrect results. However, current approaches to ensure their soundness rely on manual, expert-driven proofs that are tedious to develop, limiting the speed of developing new certifiers. Automating the verification process is challenging due to the complexity of verifying certifiers for arbitrary DNN architectures and handling diverse abstract analyses.

We introduce \cf, a novel verification procedure that automates the soundness verification of DNN certifiers for arbitrary DNN architectures. Our core contribution is the novel concept of a \textit{symbolic DNN}, using which, \cf reduces the soundness property, a universal quantification over arbitrary DNNs, to a tractable symbolic representation, enabling verification with standard SMT solvers. By formalizing the syntax and operational semantics of \oldtool, a DSL for specifying certifiers, \cf efficiently verifies both existing and new certifiers, handling arbitrary DNN architectures. 

Our code is available at \href{https://github.com/uiuc-focal-lab/constraintflow.git}{https://github.com/uiuc-focal-lab/constraintflow.git}
\end{abstract}




\maketitle

\section{Introduction}
While DNNs can achieve impressive performance, there is a growing need for their safety and robustness in safety-critical domains like autonomous driving~\cite{driving}, healthcare~\cite{AMATO201347}, etc., due to their susceptibility to environmental and adversarial noise~\cite{madry2018towards, xu2023robust}. Formal certification of DNNs can be used to assess their performance on a large, potentially infinite set of inputs, thereby providing guarantees on DNN behavior. Abstract Interpretation-based DNN certifiers are used widely for formally certifying DNNs, balancing cost and precision tradeoffs~\cite{deepz, refinezono, AI2, optAndAbs, deeppoly, zhang2018crown, alphacrown, star, cnncert, dutta, ehlers2017formal, scalablever, fastcrown, gpupoly, semidefinite, convexrelaxation, krelu, tjandraatmadja, vincent19, imagestars, wang2018neurify, wang2018, wang2021beta, weng18a, WongK18, wu2020, xiang2017, Zelazny2022OnOB, syrenn, redprod, incremental1, incremental2, relational1,relational2, banerjee2024interpreting}. 

Abstract Interpretation-based DNN certifiers must satisfy the \textit{over-approximation-based soundness} property to ensure correctness. Currently, when a new DNN certifier is proposed, its soundness is proved manually using arduous pen-and-paper proofs. These proofs show that the outputs computed by abstract transformers over-approximate the outputs of the DNN on concrete inputs. Developing these proofs demands an expert-level understanding of abstract interpretation and substantial experience in proving mathematical lemmas and theorems. Consequently, the development of DNN certifiers is often confined to a small group of experts. Automating the verification of DNN certifiers would significantly reduce these barriers, enabling more widespread development of reliable certifiers. However, this automation presents several challenges, which we outline below.

\textbf{Challenge 1: Imperative Programming.}
While one approach to verifying the mathematical soundness of DNN certifiers could be to use program verifiers such as Dafny~\cite{dafny}, they are unsuitable because the commonly-used libraries implementing the DNN certifiers, such as auto\_LiRPA~\cite{Lirpa:20}, ELINA~\cite{fastpoly}, and ERAN~\cite{deeppoly}, are extensive code-bases in general-purpose programming languages, employing complex imperative programming paradigms, such as pointer arithmetic. Verifying the soundness of these libraries would require isolating the mathematical logic from their implementation and modeling the algorithm's behavior on an arbitrary DNN. 

\textbf{Challenge 2: Universal Quantification. }
Since a DNN is an input to a DNN certifier, the over-approximation-based soundness of the certifier is a universally quantified assertion over all possible DNNs, which significantly complicates its verification. To illustrate this, consider verifying the certifier for a fixed DNN, where the architecture is known. In this case, the soundness can be verified by representing all neurons and edges in the DNN (represented as a Directed Acyclic Graph) using symbolic variables and then executing the certifier symbolically. The difficulty arises when the input DNN is arbitrary and so, cannot be directly represented symbolically. A DNN might be a simple fully-connected network with ReLU activations, or a more complex architecture such as ResNet, with arbitrary residual connections and activations. These DNNs have drastically different architectures, and the DNN certifier may have different execution traces for them. So, verifying the soundness of the certifier for one architecture does not guarantee soundness for arbitrary DNNs.

\textbf{Challenge 3: Complex DNN Certifiers. }
Popular DNN certifiers like~\cite{deeppoly, zhang2018crown, forwardbackward} associate polyhedral bounds with each neuron, which makes it difficult to naively model the certifier behavior using symbolic execution. For example, a polyhedral lower bound for a neuron $n$ might be expressed as $n \geq 5n_1 + n_2$, where the neurons $n_1, n_2$ are neurons located anywhere in the DNN, independent of the DNN architecture. This adds a structure over the neurons (beyond the DNN architecture) that is unknown before executing the certifier. Further, $n, n_1, n_2, \cdots$ are symbolic variables even during a concrete execution of the certifier. So, modeling the certifier behavior using symbolic execution entails modeling the symbolic variables (neurons) as SMT symbolic variables. The correctness of this modeling is unclear and is not explored in existing work~\cite{rosette, symexec}. 

\textbf{Challenge 4: Huge Query Size. }
One approach would be to represent a DNN as a complete DAG where each neuron is a vertex, but this results in massive graphs (i.e. $10^4$ neurons in a modest-size DNN will have around \((10^4)^2\) edges), with a weight of zero in the DAG representing the absence of an edge in the DNN. However, a complete DAG would lead to a huge query, which would overwhelm current SMT solvers, making them either fail or take an impractically long time. So, naively modeling arbitrary DNNs as a complete DAG is impractical for realistic-size DNNs.

To the best of our knowledge, no existing technique can automatically verify the soundness of abstract interpretation-based DNN certifiers while accommodating a diverse range of certifiers, ensuring soundness for arbitrary DNNs, and maintaining efficiency and scalability.   

\textbf{This work. }
We design a novel automated bounded verification procedure---\cf---which can verify the soundness of DNN certifiers for arbitrary DNNs. \cf is based on the novel concept of a \textit{symbolic DNN}---an abstract neural network that represents all subgraphs of any arbitrary DNN on which a DNN certifier can be applied (\S~\ref{sec:verification}). By leveraging symbolic DNNs, we transform the universally quantified soundness conditions into a tractable symbolic representation, verifying which is sufficient to prove the certifier's soundness on arbitrary DNNs. We offload the verification of this tractable symbolic representation to off-the-shelf SMT solvers. Recently, a preliminary design of a Domain Specific Language (DSL)---\oldtool---was proposed for specifying the core mathematical logic of abstract interpretation-based DNN certifiers decoupling it from any implementation details~\cite{constraintflow}. However, its syntax and semantics are not formalized. So, we design a BNF grammar, type-system, and operational semantics for \oldtool, which enables \cf to verify the soundness of certifier specifications within \oldtool. 

\textbf{Main contributions. }
\begin{itemize}[noitemsep, nolistsep]
    \item We develop a type-system for ensuring well-typed programs in \oldtool and also provide operational semantics.  We also develop symbolic semantics for \oldtool and a novel concept of a symbolic DNN to devise a verification procedure---\cf---to automatically find bugs or verify the soundness of the specified DNN certifiers. 
    \item We establish formal guarantees and provide proofs that include type-soundness, and the soundness of the automated verification procedure, \cf, w.r.t. the operational semantics of \oldtool.
    \item We provide an extensive evaluation to demonstrate that \cf enables proving the correctness or detecting bugs in existing and new abstract transformers for contemporary DNN certifiers and new DNN certifiers with new abstract domains. Using \cf, for the first time, we can automatically verify the soundness of DNN certifiers for DNNs with an arbitrary number of layers, each with millions of learned parameters. 
\end{itemize}

\section{Background}
\label{sec:background}
In this section, we provide the necessary background needed for abstract interpretation-based DNN certifiers. While the concepts introduced are relevant to a broad range of certifiers, we describe the widely used DeepPoly certifier~\cite{deeppoly} and use it as our running example throughout the paper.

\subsection{Abstract Interpretation-Based DNN Certifiers }
We use a definition of DNNs similar to the one used in ~\cite{constraintflow}. A DNN is represented as a Directed Acyclic Graph (DAG) with neurons as the vertices and edges corresponding to the non-zero weights in the DNN architecture. The value of each neuron is determined by a DNN operation $f$, which receives as input a set of neurons, referred to as the \textit{previous} neurons $p$. DNN operations can be categorized into two categories: (i) primitive operations and (ii) composite operations. Primitive operations include the addition and multiplication of two neurons as well as non-linear activations like ReLU, sigmoid, etc. Composite operations are operations that can be expressed as combinations of primitive functions. Examples include affine transformation of neurons (fully connected layers or convolution layers) or activations like maxpool ,etc.

For a given DNN operation \(f\), the input consists of \(m\) neurons, where \(m\) denotes the arity of $f$ (e.g., \(f_{add}: \mathbb{R} \times \mathbb{R} \rightarrow \mathbb{R}\) has \(m = 2\)). Let \(\mathbf{x}\) represent an \(m\)-dimensional input to a layer, with each dimension corresponding to a neuron. DNN certifiers take a potentially infinite set of inputs, represented as \(c = \{\mathbf{x_i}\}\) and $c \in \mathcal{C}$, where \(\mathcal{C}\) is the concrete domain. Concrete elements \(c_1, c_2 \in \mathcal{C}\) are ordered by subset inclusion \(\subseteq\). Certification involves defining an abstract domain \(\mathcal{A}\) and abstract transformers \(f^\sharp\) for each \(f\). The DNN certifiers map concrete inputs to abstract elements via an abstraction function \(\alpha\) and propagate these through the network using abstract transformers. Abstract elements \(a \in \mathcal{A}\) can be mapped back to concrete values using a concretization function \(\gamma\).

\begin{definition}
\label{def:soundness}
    An abstract transformer $f^\sharp$ is sound w.r.t. the DNN operation $f$ if 
    $\ \forall a \in \mathcal{A} \cdot \forall c \in \mathcal{C} \cdot c \subseteq \gamma(a) \implies f(c) \subseteq \gamma(f^\sharp(a))$, where the semantics of $f$ are lifted to the natural set semantics.
\end{definition}

\subsection{DeepPoly DNN Certifier }
\label{background:deepoly}
We focus on abstract domains that associate fields with each neuron $n$ to impose constraints on their values. These fields form an \textit{abstract shape} $\abshape$ with corresponding constraints denoted as $\prop(\abshape, n)$. Popular abstract interpretation-based certifiers, including DeepPoly, use such domains. In the DeepPoly abstract domain, an abstract element \(a \in \mathcal{A}\) is represented as a conjunction of constraints over the neurons' abstract shapes, i.e., $a = (\abshape_1, \dots, \abshape_N)$, where $N$ is the total number of neurons. For each neuron \(n\), its abstract shape is $\abshape_n = \langle l_n, u_n, L_n, U_n \rangle$, where $l_n, u_n \in \mathbb{R} \cup \{-\infty, \infty\}$, and $L_n, U_n$ are affine expressions of neurons in the DNN. The associated over-approximation-based constraints are $\prop(\abshape, n) \triangleq (l_n \leq n \leq u_n) \wedge (L_n \leq n \leq U_n)$. Thus, the concretization function $\gamma(a) = \{(n_1, \dots, n_m) \in \mathbb{R}^m \ | \ \forall i \in [m], (l_{n_i} \leq n_i \leq u_{n_i}) \wedge (L_{n_i} \leq n_i \leq U_{n_i})\}$

An abstract transformer updates the abstract shape of the output neuron based on the concrete operation \(f\) while leaving the others unchanged. For the Affine operation, the updated abstract shape is $\abshape'_n = \langle l'_n, u'_n, L'_n, U'_n \rangle$, where $L'_n = U'_n = b + \sum^l_{i=1} w_i n_i$, where the bias ($b$) and the weights ($w_i$) are the DNN's learned parameters. To compute the lower concrete bound ($l'_n$), DeepPoly performs a backsubstitution step which starts with the lower polyhedral expression, $e = L'_n$.  At each step, $e = c_0' + \sum^l_{i=1} c'_i n_i$, each $n_i$ in $e$ is replaced with its own lower or upper polyhedral bound depending on the sign of the coefficient $c'_i$, i.e., $e \leftarrow c'_0 + \sum^l_{i=1} (c'_i \geq 0 \ ? \ c'_iL_{n_i} : c'_iU_{n_i})$. This step is repeated until all the neurons in $e$ are in the input layer, after which the constituent neurons are replaced with their respective lower or upper concrete bounds, i.e., if $e = c''_0 + \sum^l_{i=1} c''_i n_i$, then $l'_n = c''_0 + \sum^l_{i=1} (c''_i \geq 0 \ ? \ c''_il_{n_i} : c''_iu_{n_i})$. The upper concrete bound $u'_n$ is also computed similarly. 

\section{Overview}
\label{sec:overview}
We first provide an overview of DNN certifier specification in \oldtool using the DeepPoly specification from~\cite{constraintflow} as a running example, followed by the novel type-system and semantics for \oldtool. Finally, we show the soundness verification of the certifier specification.

\subsection{\oldtool}
% \begin{wrapfigure}{r}{0.3\textwidth}
%   \begin{minipage}{0.3\textwidth}
%         \centering
%         \input{sections/subtyping}
%     \caption{Subtyping Lattice}
%     \label{fig:subtyping}
%     \end{minipage}
% \end{wrapfigure}
% We define a type system and semantics for \oldtool that facilitates symbolic reasoning and high-level abstractions for formal verification. We represent the syntax in \cfkeywords{\texttt{blue}} and types in \cftypewords{\texttt{red}}.

\oldtool introduces datatypes specific to DNN certifiers including \typeneuron, \polyexp, and \ct. Neurons are represented as \typeneuron. The type \polyexp represents affine expressions over neurons and \ct represents symbolic constraints. Since some DNN certifiers use symbolic variables to specify constraints over the neuron values ~\cite{deepz, star, refinezono}, we introduce the $\sym$ construct to declare a symbolic variable of the type \typesym. We also introduce \symexp to capture symbolic expressions over these symbolic variables. 
% 
% We define a subtyping relation \(\sqsubset\) for the basic types in \oldtool, organized as a lattice (Fig.~\ref{fig:subtyping}). An expression is type-checked to ensure it has a type other than \(\top\) or \(\bot\) (\S~\ref{sec:typechecking}).
% 
By treating polyhedral and symbolic expressions as first-class members, we can define the operational semantics of constructs that can directly operate on these new types. These include (i) binary arithmetic operations like `+', (ii) \map, which applies a function to each constituent neuron or symbolic variable in a polyhedral or symbolic expression, and (iii) \traverse, which repeatedly applies \map to a polyhedral expression until a termination condition is met. The formal semantics (discussed in detail in \S~\ref{sec:operationalsemantics}) enable automated reasoning and verification. 

In \oldtool, a DNN certifier is specified through three main steps: (i) specifying the abstract shape for each neuron along with its soundness constraints, (ii) defining the abstract transformers for each DNN operation, and (iii) determining how constraints propagate through the network. We illustrate the different steps of specifying a DNN certifier in \oldtool using the DeepPoly specification in Fig.~\ref{fig:intro}. 

\begin{figure}
    \begin{lstlisting}
Def shape as (Real l, Real u, PolyExp L, PolyExp U) {(curr[l] <= curr) and (curr[u] >= curr) and (curr[L] <= curr) and (curr[U] >= curr)};

Func priority(Neuron n) = n[layer];
Func concretize_lower(Neuron n, Real c) = (c >= 0) ? (c * n[l]) : (c * n[u]);
Func concretize_upper(Neuron n, Real c) = (c >= 0) ? (c * n[u]) : (c * n[l]);
Func replace_lower(Neuron n, Real c) = (c >= 0) ? (c * n[L]) : (c * n[U]);
Func replace_upper(Neuron n, Real c) = (c >= 0) ? (c * n[U]) : (c * n[L]);
Func backsubs_lower(PolyExp e, Neuron n) = (e.traverse(backward,priority,false,replace_lower)  {e <= n}).map(concretize_lower);
Func backsubs_upper(PolyExp e, Neuron n) = (e.traverse(backward,priority,false,replace_upper)  {e >= n}).map(concretize_upper);

Transformer DeepPoly{
  Affine -> (backsubs_lower(prev.dot(curr[w]) + curr[b], curr),
             backsubs_upper(prev.dot(curr[w]) + curr[b], curr), 
             prev.dot(curr[w]) + curr[b], 
             prev.dot(curr[w]) + curr[b]);
  Relu -> prev[l] > 0 ? 
            (prev[l], prev[u], prev, prev) : 
            (prev[u] < 0 ? 
                (0, 0, 0, 0) : 
                (0, prev[u], 0, ((prev[u] / (prev[u] - prev[l])) * prev) - ((prev[u] * prev[l]) / (prev[u] - prev[l]))));
}

Flow(forward, -priority, false, DeepPoly);\end{lstlisting}
    \caption{DeepPoly specification in \oldtool}
    \label{fig:intro} 
\end{figure}

\subsubsection{Abstract Domain. }
The specification of a DNN certifier starts by defining the abstract domain used by the certifier (Line 1 of Fig.~\ref{fig:intro}). In \oldtool, this is done by defining the abstract shape ($\abshape$) associated with each neuron and the constraints defining the over-approximation-based soundness condition ($\prop$). These are specified for the \curr neuron, which serves as a syntactic placeholder for all neurons in the DNN. For example, the DeepPoly abstract shape and its constraints can be defined in \oldtool as illustrated in Fig.~\ref{fig:intro}, where \(l, u, L, U\) are user-defined members of the abstract shape, accessed via square bracket notation (\curr[·]). The DeepPoly soundness condition is encoded as: $(l\leq n) \wedge (u\geq n) \wedge (L\leq n) \wedge (U\geq n$).

We formalize the syntax for \oldtool (\S~\ref{sec:syntax}), allowing the users to define arbitrary abstract shapes. For instance, abstract domains can combine polyhedral and novel symbolic expressions. Symbolic variables (\(\noise\)) are subject to default constraints, \( -1 \leq \noise_i \leq 1 \), defining multi-dimensional polyhedra. The constraint \(\curr \ \embed \ \curr[Z]\) indicates that \curr is embedded in the polyhedron defined by \(\curr[Z]\), meaning there exists an assignment to the symbolic variables in \(\curr[Z]\) such that \(\curr = \curr[Z]\):
\begin{lstlisting}[numbers=none]
Def shape as (Real l, Real u, PolyExp L, PolyExp U, SymExp Z) {curr[l] <= curr, curr[u] >= curr, curr[L] <= curr, curr[U] >= curr, curr <> curr[Z]};  
\end{lstlisting}

\subsubsection{Abstract Transformers. }
\label{sec:overviewtransformers}
After defining the abstract domain, the second step is to specify the abstract transformers for different DNN operations. In Fig.~\ref{fig:intro}, lines 2-8 show the user-defined functions used within the transformer definitions in lines 9-19 within the \transformer construct. The implicit inputs to the \transformer construct are \curr, representing the current neuron, and \prev, representing the previous neurons. \prev is a list for DNN operations with multiple inputs, like \affine, and a single neuron in case of operations with a single input, like \relu. The transformer for each DNN operation specifies the computations for updating the four fields of the abstract shape: \(l\), \(u\), \(L\), and \(U\). The transformers for \affine and \relu operations are shown in Fig.~\ref{fig:intro} in lines 10 and 14 respectively. Using the semantics of the \oldtool constructs, we show how the DeepPoly specification in Fig.~\ref{fig:intro} simulates the mathematical logic of DeepPoly (explained in \S~\ref{sec:background}). The \oldtool semantics also allow us to explore variants of DeepPoly.  

In the DeepPoly \affine transformer, the polyhedral bounds ($L$ and $U$) are given by \prev.\dott(\curr[w]) + \curr[b]. There are many ways to compute the concrete lower $l$ and upper bounds $u$. Consider \texttt{\footnotesize{concretize\_lower}} and \texttt{\footnotesize{replace\_lower}} functions from Fig.~\ref{fig:intro} that respectively replace a neuron with its lower or upper concrete and polyhedral bounds based on its coefficient. We can compute the lower concrete bound for \curr, by applying the \texttt{\footnotesize{concretize\_lower}} to all the neurons in the lower polyhedral expression, i.e., (\prev.\dott(\curr[w]) + \curr[b]).\map(\texttt{\footnotesize{concretize\_lower}}). We can compute a more precise polyhedral lower bound by first applying \texttt{\footnotesize{replace\_lower}} to each constituent neuron, i.e., (\prev.\dott(\curr[w]) + \curr[b]).\map(\texttt{\footnotesize{replace\_lower}}). We can repeat this several times, following which, we can apply \texttt{\footnotesize{concretize\_lower}} to concretize the bound. In the standard implementation, the number of applications of \texttt{\footnotesize{replace\_lower}} is unknown because it is applied until the polyhedral bound only contains neurons from the input layer of the DNN. Although this is precise, it might be costly to perform this computation until the input layer is reached. So, custom stopping criteria can be decided, balancing the tradeoff between precision and cost. Note that the order in which the neurons are substituted with their bounds also impacts the output's precision. 

To specify arbitrary graph traversals succinctly, we provide the \traverse construct, which decouples the stopping criterion from the neuron traversal order. \traverse operates on polyhedral expressions and takes as input the direction of traversal and three functions---a user-defined stopping function, a priority function over neurons specifying the order of traversal and a neuron replacement function. In each step, \traverse applies the priority function to each constituent neuron in the polyhedral expression. Then, it applies the neuron replacement function to each constituent neuron with the highest priority among the neurons on which the stopping condition evaluates to $\false$. The outputs are then summed up to generate a new polyhedral expression. This process continues until the stopping condition is $\true$ on all the constituent neurons or all the neurons are in the input or output layer depending on the traversal order. We can use \traverse to specify the backsubstitution step and hence the DeepPoly \affine transformer as shown in Fig. ~\ref{fig:intro}.

\subsubsection{Flow of Constraints. }
\label{sec:flow}
Existing DNN certifiers propagate constraints from the input to the output layer or in reverse~\cite{forwardbackward, 10.1145/3428253, inbook}. Further, the order in which abstract shapes of neurons are computed impacts analysis precision. In \oldtool, the specification of the order of application is decoupled from the actual transformer specification, so the soundness verification of the transformer remains independent of the traversal order. We formalize this syntax and semantics to provide adjustable knobs to define custom flow orders, using a direction, priority function, and a stopping condition. The user specifies these arguments and the transformer using the \flow construct, as demonstrated in Fig.~\ref{fig:intro}, Line 20, for the DeepPoly certifier. This code assigns higher priority to lower-layer neurons, resulting in a BFS traversal. The stopping function is set to $\cfkeywords{false}$, stopping only when reaching the output layer. We verify the soundness of all specified transformers in the \transformer construct. Based on the DNN operation, \flow applies the corresponding transformer, ensuring a composition of only sound transformers. 

\subsection{\cf: Automated Bounded Verification of the DNN Certifier }
To establish the soundness of a certifier, it is necessary to verify the soundness of each abstract transformer $f^\sharp$ w.r.t. its concrete counterpart $f$, i.e.,
\begin{equation}
\label{formula}
\forall a \in \mathcal{A} \cdot \forall c \in \mathcal{C} \cdot c \subseteq \gamma(a) \implies f(c) \subseteq \gamma(f^\sharp(a))
\end{equation}

Equation~\ref{formula} is universally quantified over both the abstract element \(a\) and the concrete element \(c\). The abstract element, a tuple of abstract shapes, over-approximates the values of neurons in the DNN, while the concrete element represents specific valuations for the neurons. Since the DNN architecture—its topology, number of neurons, and consequently the number of abstract shapes—can vary, the universal quantification in equation~\ref{formula} presents a challenge for verification. 

So, we introduce the concept of a \textit{Symbolic DNN} to represent an arbitrary DNN and the corresponding abstract shapes symbolically. The symbolic DNN is an abstract neural network representing all subgraphs of any arbitrary DNN on which the specified transformer can be applied. It consists of symbolic values representing only the necessary neurons for executing the transformer specification. So, verifying the soundness of the specified transformer on a finite symbolic DNN is sufficient to prove its soundness on an arbitrarily large DNN with any topology. 

The symbolic DNN is initialized only with \curr and \prev, along with their abstract shapes so the specified abstract transformer can be symbolically executed. However, in some cases, the symbolic execution of a transformer requires more neurons to be initialized in the symbolic DNN. We do so by a \textit{Symbolic DNN Expansion}, where we statically analyze the transformer and only introduce neurons and their abstract shapes necessary for the symbolic execution. We explain these steps using an example in \S~\ref{sec:overviewgraphcreation}, \S~\ref{sec:overviewgraphexpand}. After the creation and expansion steps, we have a symbolic representation of the DNN and corresponding abstract shapes sufficient for symbolic execution to generate the final verification query which can be off-loaded to an off-the-shelf SMT solver (\S~\ref{sec:overviewexecution}). 

To better illustrate these steps, we introduce a new DeepPoly transformer for \relu which has a better runtime than the original transformer but is slightly less precise. We then show the above-mentioned steps for the verification of the new transformer. As introduced in \S~\ref{sec:background}, the DeepPoly abstract shape consists of 4 fields---$l, u, L, U$, where $l, u$ are the concrete bounds and $L, U$ are the polyhedral bounds of the neuron. Consider the DeepPoly \relu transformer. It takes in as input the abstract shape of the \prev neuron and computes the new abstract shape for \curr neuron. It has 3 cases based on the values $\prev[l], \prev[u]$ of the input abstract shape - (i) $\prev[l]\geq 0$, (ii) $\prev[u] \leq 0$, and (iii) $\prev[l] < 0 < \prev[u]$. We focus only on the first case for illustration. In this case, the concrete bounds are set to the input concrete bounds, i.e., $\curr[l] \gets \prev[l]$ and $\curr[u] \gets \prev[u]$. Both the lower and upper polyhedral bounds are set to $\prev$, i.e., $\curr[L] \gets \prev$ and $\curr[U] \gets \prev$. In the new transformer for \relu, instead of setting the polyhedral bounds of \curr in terms of the neurons of the previous layer, i.e., \prev, we set them using the lower and upper polyhedral bounds of \prev, which are $\prev[L]$ and $\prev[U]$ respectively. In \oldtool, these polyhedral bounds can be computed using $\map(\texttt{\footnotesize{replace\_lower}})$ and $\map(\texttt{\footnotesize{replace\_upper}})$ respectively. The user-defined functions $\texttt{\footnotesize{replace\_lower}}$ and $\texttt{\footnotesize{replace\_upper}}$ replace a neuron with its lower or upper polyhedral bounds based on its coefficient. The \map construct applies a function to all neurons in a polyhedral expression. So, the expression for the upper polyhedral bound (and similarly for lower) can thus be written as $\expr \equiv \prev[U].\map(\texttt{\footnotesize{replace\_upper}})$. 

\begin{figure}
\resizebox{\textwidth}{!}{
\begin{tikzpicture}
    \begin{scope}[local bounding box=f1box]
        \node (curr) at (0,0) {$\curr$};
        \node[below=0cm of curr.center, rotate=270, right=0.01cm] (node2) {$\mapsto$};
        \node[below=0.1cm of node2.center] (node3) {$\sval_c$};
        \node[left=1.0cm of curr] (prev) {$\prev$};
        \node[below=0cm of prev.center, rotate=270, right=0.01cm] (node4) {$\mapsto$};
        \node[below=0.1cm of node4.center] (node5) {$\sval_p$};
        \node[right=0.4cm of curr, yshift=0.9cm] (l) {$l \mapsto \sval_c^l$};
        \node[right=0.4cm of curr, yshift=0.3cm] (u) {$u \mapsto \sval_c^u$};
        \node[right=0.4cm of curr, yshift=-0.3cm] (L) {$L \mapsto \sval_c^L$};
        \node[right=0.4cm of curr, yshift=-0.9cm] (U) {$U \mapsto \sval_c^U$};
        \node[left=0.4cm of prev, yshift=0.9cm] (pl) {$\sval_p^l \mapsfrom l$};
        \node[left=0.4cm of prev, yshift=0.3cm] (pu) {$\sval_p^u \mapsfrom u$};
        \node[left=0.4cm of prev, yshift=-0.3cm] (pL) {$\sval_p^L \mapsfrom L$};
        \node[left=0.4cm of prev, yshift=-0.9cm] (pU) {$\sval_p^U \mapsfrom U$};

        \node[below] at (current bounding box.south) {};
        \node[below] at (current bounding box.south) {};
        \node[below] at (current bounding box.south) {$\prop(\curr) = (\sval_c^l \leq \sval_c \leq \sval_c^u) \wedge (\sval_c^L \leq \sval_c \leq \sval_c^U)$};
        \node[below] at (current bounding box.south) {$\prop(\prev) =  (\sval_p^l \leq \sval_p \leq \sval_p^u) \wedge (\sval_p^L \leq \sval_p \leq \sval_p^U) $};
        \node[below] at (current bounding box.south) {$\cc_\eta =  (\sval_p \leq 0 \implies \sval_c = 0) \wedge (\sval_p > 0 \implies \sval_c = \sval_p)$};
        \node[below] at (current bounding box.south) {$\cc = \prop(\curr) \wedge \prop(\prev) \wedge \cc_\eta$};
        \node[below, yshift=-0.3cm] at (f1box.south) {(a)};
        \draw[dotted] (curr.east) -- (l.west);
        \draw[dotted] (curr.east) -- (u.west);
        \draw[dotted] (curr.east) -- (L.west);
        \draw[dotted] (curr.east) -- (U.west);
    
        \draw[dotted] (prev.west) -- (pl.east);
        \draw[dotted] (prev.west) -- (pu.east);
        \draw[dotted] (prev.west) -- (pL.east);
        \draw[dotted] (prev.west) -- (pU.east);
        
        \draw[->] (prev.east) -- (curr.west) node[midway, above] {\relu};
        \node[draw, above=of prev] {$\hh_S$};
    \end{scope}

    \begin{scope}[shift={(11.0cm,1.0cm)},local bounding box=f2box]
        \node (curr) at (0,0) {$\curr$};
        \node[below=0cm of curr.center, rotate=270, right=0.01cm] (node2) {$\mapsto$};
        \node[below=0.1cm of node2.center] (node3) {$\sval_c$};
        \node[left=1.0cm of curr] (prev) {$\prev$};
        \node[below=0cm of prev.center, rotate=270, right=0.01cm] (node4) {$\mapsto$};
        \node[below=0.1cm of node4.center] (node5) {$\sval_p$};
        \node[right=0.4cm of curr, yshift=0.9cm] (l) {$l \mapsto \sval_c^l$};
        \node[right=0.4cm of curr, yshift=0.3cm] (u) {$u \mapsto \sval_c^u$};
        \node[right=0.4cm of curr, yshift=-0.3cm] (L) {$L \mapsto \sval_c^L$};
        \node[right=0.4cm of curr, yshift=-0.9cm] (U) {$U \mapsto \sval_c^U$};
    
        \node[left=0.4cm of prev, yshift=0.9cm] (pl) {$\sval_p^l \mapsfrom l$};
        \node[left=0.4cm of prev, yshift=0.3cm] (pu) {$\sval_p^u \mapsfrom u$};
        \node[left=0.4cm of prev, yshift=-0.3cm] (pL) {$\sval_p^L \mapsfrom L$};
        \node[left=0.4cm of prev, yshift=-0.9cm] (pU) {$\sval_r^1 + \sval_r^2*\textcolor{red}{\bm{\sval_{n_1}}} +  \sval_r^3*\textcolor{red}{\bm{\sval_{n_2}}} \mapsfrom U$};
    
        \node[draw, above=of pU, yshift=-0.2cm] {$\hh'_S$};
        \begin{scope}[shift={(-2cm,-2.5cm)},local bounding box=f3box]
            \node (n2) at (0,0) {$n_2$};
            \node[below=0cm of n2.center, rotate=270, right=0.01cm] (n2map) {$\mapsto$};
            \node[below=0.1cm of n2map.center] (n2var) {$\textcolor{red}{\bm{\sval_{n_2}}}$};
            \node[left=2.5cm of n2] (n1) {$n_1$};
            \node[below=0cm of n1.center, rotate=270, right=0.01cm] (n1map) {$\mapsto$};
            \node[below=0.1cm of n1map.center] (n1var) {$\textcolor{red}{\bm{\sval_{n_1}}}$};

            \node[left=0.4cm of n1, yshift=0.9cm] (n1l) {$\sval_{n_1}^l \mapsfrom l$};
            \node[left=0.4cm of n1, yshift=0.3cm] (n1u) {$\sval_{n_1}^u \mapsfrom u$};
            \node[left=0.4cm of n1, yshift=-0.3cm] (n1L) {$\sval_{n_1}^L \mapsfrom L$};
            \node[left=0.4cm of n1, yshift=-0.9cm] (n1U) {$\sval_{n_1}^U \mapsfrom U$};

            \node[right=0.4cm of n2, yshift=0.9cm] (n2l) {$l \mapsto \sval_{n_2}^l$};
            \node[right=0.4cm of n2, yshift=0.3cm] (n2u) {$u \mapsto \sval_{n_2}^u$};
            \node[right=0.4cm of n2, yshift=-0.3cm] (n2L) {$L \mapsto \sval_{n_2}^L$};
            \node[right=0.4cm of n2, yshift=-0.9cm] (n2U) {$U \mapsto \sval_{n_2}^U$};
    
            \draw[dotted] (n2.east) -- (n2l.west);
            \draw[dotted] (n2.east) -- (n2u.west);
            \draw[dotted] (n2.east) -- (n2L.west);
            \draw[dotted] (n2.east) -- (n2U.west);
        
            \draw[dotted] (n1.west) -- (n1l.east);
            \draw[dotted] (n1.west) -- (n1u.east);
            \draw[dotted] (n1.west) -- (n1L.east);
            \draw[dotted] (n1.west) -- (n1U.east);

            \draw[red, dotted, <->] ($(pU.west) - (-1.75,0.3)$) -- ($(n1.north)$);
            \draw[red, dotted, <->] ($(pU.west) - (-3.5,0.2)$) -- ($(n2.north west)$);
        \end{scope}
    
        \node[below] at (f2box.south) {};
        \node[below] at (f2box.south) {};
        \node[below] at (f2box.south) {$\prop(\prev) = (\sval_p^l \leq \sval_p \leq \sval_p^u) \wedge (\sval_p^L \leq \sval_p \leq (\sval_r^1 + \sval_r^2*\sval_{n_1} +  \sval_r^3*\sval_{n_2}))$};
        \node[below] at (f2box.south) {$\cc = \prop(\curr) \wedge \prop(\prev) \wedge \cc_\eta \wedge \prop(n_1) \wedge \prop(n_2)$};
        \node[below, yshift=-0.23cm] at (f2box.south) {(b)};
        \draw[dotted] (curr.east) -- (l.west);
        \draw[dotted] (curr.east) -- (u.west);
        \draw[dotted] (curr.east) -- (L.west);
        \draw[dotted] (curr.east) -- (U.west);
    
        \draw[dotted] (prev.west) -- (pl.east);
        \draw[dotted] (prev.west) -- (pu.east);
        \draw[dotted] (prev.west) -- (pL.east);
        \draw[dotted] (prev.west) -- (pU.east);
        
        \draw[->] (prev.east) -- (curr.west) node[midway, above] {\relu};
    \end{scope}
    
    \node at ($(f1box.east)!0.45!(f2box.west)$) {\mbox{\Huge $ \leadsto$ }};
\end{tikzpicture}
}
% \vspace{-3mm}
\caption{Symbolic DNN creation and expansion for DeepPoly. $\prop(n) \equiv (l \leq n \leq u) \wedge (L \leq n \leq U)$
}
% \vspace{-3mm}
    \label{fig:symbolicdnn}
\end{figure}
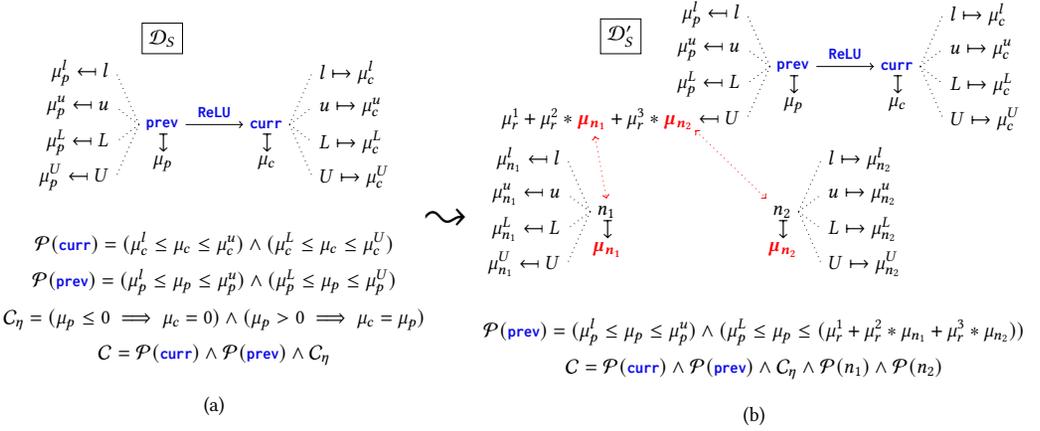

\subsubsection{Symbolic DNN Creation}
\label{sec:overviewgraphcreation}
For each DNN operation $\eta$ (e.g., \relu in this case), given the abstract transformer, we create a symbolic DNN (Fig.~\ref{fig:symbolicdnn}a) with neurons representing \prev and \curr that are respectively the input and output of $\eta$. These neurons are associated with symbolic variables $\mu_p$ and $\mu_c$ representing their valuations respectively. The edges are only between \curr and \prev neurons representing the $\relu$ operation. Here, \prev represents only a single neuron. However, for DNN operations like \affine, the symbolic DNN is initialized with $\prev_1, \cdots \prev_k$ where $k$ is a sufficiently large parameter. We do not make any assumptions about the DNN's architecture, resulting in the absence of any extra neurons or edges between $\prev_i$ and $\prev_j$ and thus, no additional constraints over symbolic variables. Fig.~\ref{fig:symbolicdnn}a shows the symbolic DNN for the \relu transformer for the DeepPoly certifier. The soundness property $\prop$ for this certifier is that for each neuron $n$, $(l \leq n \leq u) \wedge (L \leq n \leq U)$. Each shape member and metadata associated with these neurons is also initialized with fresh symbolic variables. For instance, $\mu_p^l$, $\mu_p^u$ represent the lower and upper concrete bounds respectively, and $\mu_p^L$, $\mu_p^U$ are the lower and upper polyhedral bounds of \prev. The symbolic DNN is associated with constraints representing the edge relations between the neurons and the soundness property assumptions before applying the transformer. In Fig.~\ref{fig:symbolicdnn}a, these constraints are presented as $\cc = \prop(\curr) \wedge \prop(\prev) \wedge \cc_\eta$, where $\prop(\curr)$ and $\prop(\prev)$ represent the soundness property over \curr and \prev respectively. $\cc_\eta$ represents the semantics of the $\relu$ operation, i.e., $\curr = 0$ when $\prev < 0$, and $\curr = \prev$ otherwise. The formal definition and details of a symbolic DNN can be found in \S~\ref{sec:graphcreation}.

\subsubsection{Symbolic DNN Expansion}
\label{sec:overviewgraphexpand}
Initially, polyhedral bounds such as \(\prev[L]\) and \(\prev[U]\) are represented as single symbolic variables. However, for operations like \map, the polyhedral values need to be expanded into expressions of the form \(x_0 + x_1\cdot n_1 + x_2 \cdot n_2 \dots\), where \(x_i\) are coefficients and \(n_i\) are neurons. This is necessary for the semantics of \map, as functions like \texttt{\footnotesize{replace\_upper}} are applied to each constituent neuron and coefficient within the polyhedral expression. For example, consider \(\expr \equiv \prev[U].\map(\texttt{\footnotesize{replace\_upper}})\). Initially, \(\prev[U]\) is a single symbolic variable \(\sval_p^U\) (Fig.~\ref{fig:symbolicdnn}a), but to symbolically evaluate \(\expr\), the expression must be expanded into its constituent terms, e.g., \(\sval_r^1 + \sval_r^2 \cdot \sval_{n_1} + \sval_r^3 \cdot \sval_{n_2}\), where \(\sval_r^1, \sval_r^2,\) and \(\sval_r^3\) are symbolic coefficients, and \(\sval_{n_1}, \sval_{n_2}\) represent new neurons. In this case, the expansion introduces two neurons, but in general, the number of neurons \(n_{sym}\) is a sufficiently large parameter. No architectural assumptions are made about the new neurons, but they must be added to the symbolic DNN along with their metadata, and the soundness property \(\prop\) must be assumed for them. Fig.~\ref{fig:symbolicdnn}b shows the updated symbolic DNN after one expansion step. Similarly, before executing the expression for the polyhedral lower bound \(\expr \equiv \prev[L].\map(\texttt{\footnotesize{replace\_lower}})\), \(\sval_p^L\) must also be expanded. This expansion is performed through static analysis of the transformer. Once the symbolic DNN is expanded, the associated constraints \(\cc\) are updated to reflect the new neurons and the expanded values. Detailed steps for Symbolic DNN Expansion are in \S~\ref{sec:graphexpansion}.

\subsubsection{Generating the Verification Query}
\label{sec:overviewexecution}
\begin{figure}
% \fbox{
    \begin{minipage}{0.29\textwidth}
        \centering
        \begin{tikzpicture}[scale=0.85]
        \draw[dashed] (-0,0.2) ellipse (0.8cm and 1.8cm);
        \fill (0, 1) circle (2pt) node[left] {\scriptsize{$\gamma(f^\sharp(a))$}};
        \fill (0, 0.5) circle (2pt) node[left] {\scriptsize$f(c)$};
        \fill (0.3, -0.5) circle (2pt) node[above right] {\scriptsize$\gamma(a)$};
        \fill (0, -1.0) circle (2pt) node[left] {\scriptsize$c$};
        \draw[->, color=diagramcolor, thick] (-0, -0.9) to node[left] {\scriptsize$2$} (-0, 0.4);
        \begin{scope}[shift={(2.5cm,0cm)},local bounding box=f1box]
            \draw[dashed] (0,0.2) ellipse (0.8cm and 1.8cm);
            \fill (0, 1) circle (2pt) node[right] {\scriptsize$f^\sharp(a)$};
            \fill (0, -0.5) circle (2pt) node[right] {\scriptsize$a$};
            \draw[->, color=diagramcolor, thick] (0, -0.4) to node[right] {\scriptsize$3$} (0, 0.9);
        \end{scope}
        \draw[->, color=diagramcolor, thick] (2.4, -0.5) to node[below] {\scriptsize$1$} (0.4, -0.5);
        \draw[->, color=diagramcolor, thick] (2.4, 1) to node[above] {\scriptsize$4$} (0.1, 1);
    \end{tikzpicture}
    \caption{Soundness of $f^\sharp$ w.r.t. $f$}
    \label{fig:diagram}
    \vspace{2em}
    \end{minipage}\hfill
    % \vrule
    \begin{minipage}{0.65\textwidth}
    % \resizebox{\textwidth}{!}{
        {\footnotesize 
        \begin{align*}
    &c \subseteq \gamma(a) \stackrel{?}{\implies} f(c) \subseteq \gamma(f^\sharp(a)) \\
    &\equiv c \subseteq \gamma(a) \stackrel{?}{\implies} \Big( (\cdot\cdot, p, c, \cdot\cdot) \in f(c)   \implies   (\cdot \cdot, p, c, \cdot\cdot) \in \gamma(f^\sharp(a)) \Big) \\
    &\equiv \Big( c \subseteq \gamma(a) \wedge  (\cdot\cdot, p, c, \cdot\cdot) \in f(c) \Big)  \stackrel{?}{\implies} \Big(   (\cdot \cdot, p, c, \cdot\cdot) \in \gamma(f^\sharp(a)) \Big) \\
    &\equiv  \ \Big( \prop(\abshape_{c}, c) \wedge \prop(\abshape_{p},p) \wedge c=f(p) \Big) \stackrel{?}{\implies} \Big( a'=f^\sharp(a) \implies \big(\prop(\abshape'_{c}, c)\big) \Big) \\
    &\equiv  \ \Big( \varphi_0 \wedge \varphi_1 \wedge \varphi_2 \Big) \stackrel{?}{\implies} \Big( \varphi_3 \implies \varphi_4 \Big) \\
    &\equiv  \ \Big( \varphi_0 \wedge \varphi_1 \wedge \varphi_2 \wedge \varphi_3 \Big) \stackrel{?}{\implies} \varphi_4 
        \end{align*}
        }
    \caption{SMT query for Soundness of $f^\sharp$ w.r.t. $f$}
            \label{eq:soundred}
    \vspace{2em}
    \end{minipage}
    \begin{minipage}{\textwidth}
    \captionof{table}{Generating SMT query for verifying one case of the \relu transformer for DeepPoly certifier.}
\label{table:diagram}
        \resizebox{\textwidth}{!}{
\scriptsize
\begin{tabular}{p{0.4\textwidth}|p{0.5\textwidth}}
\toprule
\textbf{Steps in Fig.~\ref{fig:diagram}} & \textbf{DeepPoly Translation for \relu Operation} \\
\midrule
Let $a = (\cdots, \abshape_{n_1}, \abshape_{n_2}, \abshape_{p}, \abshape_{c}, \cdots)$ & Declare fresh symbolic variables for all neurons, metadata, and shape fields in the expanded symbolic DNN \\
\midrule
(\textcolor{diagramcolor}{1}) Let $(\cdots, n_1, n_2, p, c, \cdots) = \gamma(a)$, $c \subseteq \gamma(a)$ & $\varphi_1 \equiv \prop(\abshape_{n_1},n_1) \wedge \prop(\abshape_{n_2},n_2) \wedge \prop(\abshape_{p},p) \wedge \prop(\abshape_{c},c)$  \\
\midrule
(\textcolor{diagramcolor}{2}) Apply $f$ to $c$ & $\varphi_2 \equiv c = f_r(p)$ \\
\midrule
(\textcolor{diagramcolor}{3}) Let $a' = f^\sharp(a)$ & Declare new symbolic variables for output: \\ & $\varphi_3 \equiv (a' == (\cdots, \abshape_{n_1}, \abshape_{n_2}, \abshape_{p}, \abshape'_{c}, \cdots))$ \\
\midrule
(\textcolor{diagramcolor}{4}) Apply $\gamma$ to $a'$ & $\varphi_4 \equiv \prop(\abshape'_{c},c)$ \\
\bottomrule
\end{tabular}
}
    \end{minipage}
    % \caption{Caption for the entire figure}
    % \label{fig:entire_figure}
% }
        % \vspace{-1.5em}
\end{figure}

Once the symbolic DNN is expanded, we can translate the soundness check of a DNN certifier (Formula~\ref{formula}) into a closed-form SMT query. In the case of \relu, the symbolic DNN corresponds to an abstract element \(a\), a tuple of abstract shapes \(a = (\cdots, \abshape_{n_1}, \abshape_{n_2}, \abshape_p, \abshape_c, \cdots)\), where \(\abshape_{n_1}, \abshape_{n_2}, \abshape_p,\) and \(\abshape_c\) represent the abstract shapes of \(n_1\), \(n_2\), \prev, and \curr, respectively. As shown in Fig.~\ref{fig:diagram}, the verification process consists of two steps (\textcolor{diagramcolor}{1}, \textcolor{diagramcolor}{2}) to compute \(f(c)\), and two steps (\textcolor{diagramcolor}{3}, \textcolor{diagramcolor}{4}) to compute \(\gamma(f^\sharp(a))\), starting from \(a\). Table~\ref{table:diagram} outlines the computations for each step, with an example for the first case of the DeepPoly \relu transformer (\(\varphi_0 \equiv \prev[l] \geq 0\)).

\begin{enumerate}
    \item[\textcolor{diagramcolor}{1}] \(c \subseteq \gamma(a)\), representing the set of neuron value tuples satisfying \(\prop\). This is denoted by \(\varphi_1\).
    \item[\textcolor{diagramcolor}{2}] Applying \(f\) to \prev to compute \curr. Any \(v \in f(c)\), with \(v = (\cdots, p, c, \cdots)\), must satisfy \(\varphi_2 \equiv c=f(p)\), where, in the case of \relu, \(f_r\) is defined as $f_r(p) = \max(p, 0)$.
    \item[\textcolor{diagramcolor}{3}] Applying \(f^\sharp\) to \(a\), updating only the abstract shape of \curr: \(a' = (\cdots, \abshape_{n_1}, \abshape_{n_2}, \abshape_p, \abshape'_c, \cdots)\). The new shape fields \(l\), \(u\), \(L\), and \(U\) are computed symbolically. For example, \(\curr[U]\) is set to \(\prev[U].\map(\texttt{\footnotesize{replace\_upper}})\).We start this computation by computing $\prev[U]$ as $\sval_r^1 + \sval_r^2* \sval_{n_1} + \sval_r^3* \sval_{n_2}$. Then we apply \texttt{\footnotesize{replace\_upper}} to each constituent summands to compute the final value as $\sval_r^1 + If(\sval_r^2 \geq 0, \sval_r^2*\sval_{n_1}^U, \sval_r^2*\sval_{n_1}^L) + If(\sval_r^3 \geq 0, \sval_r^3*\sval_{n_2}^U, \sval_r^3*\sval_{n_2}^L)$. Here, $If(c, l, r)$ is a Z3 construct. Similarly, the lower polyhedral bound is also computed.
    \item[\textcolor{diagramcolor}{4}] Applying \(\gamma\) to \(a'\) results in \(\varphi_4 \equiv \prop(\abshape'_c, c)\).
\end{enumerate}

\noindent
The verification reduces to checking if \((\varphi_0 \wedge  \varphi_1 \wedge \varphi_2 \wedge \varphi_3 ) \implies \varphi_4\), as illustrated in Fig.~\ref{eq:soundred}. More details on the symbolic semantics and the steps to generate the final query can be found in \S~\ref{sec:symexecution},~\ref{sec:smtverification}.

\subsubsection{Soundness and Completeness of \cf}
The target of the verification procedure is to ensure that if using the operational semantics of \oldtool, the abstract transformer is applied to any concrete DNN along with its abstract element that satisfies the specified property, the updated abstract element still maintains the over-approximation-based soundness. For this, \cf creates a symbolic DNN and executes the specified transformer using symbolic semantics to generate an SMT query. We prove that verifying the transformer using symbolic semantics over a symbolic DNN ensures the verification using operational semantics over any concrete DNN. We explain this in detail in \S~\ref{sec:verificationsoundness}.

\paragraph{Soundness}
We introduce the notion of a symbolic DNN over-approximating a concrete DNN and symbolic semantics over-approximating the operational semantics.  As a result, we use a bisimulation argument to prove that if the transformer is verified for a symbolic DNN, then it is also verified for all concrete DNNs that the symbolic DNN over-approximates. 

\paragraph{Completeness}
Symbolic execution is not complete for \traverse because it involves loops with input-dependent termination conditions. So, to verify programs using \traverse, we check the correctness and subsequently use the \textit{inductive invariant} provided by the programmer. We also provide a construct \solver in \oldtool that can be used for calls to external solvers. For example, finding the minimum value of an expression $e_1$ under some constraints $e_2$ can be encoded as $\solver(\minimize, e_1, e_2)$. Since we do not have access to the solver, instead of symbolically executing it, we use \textit{function contracts} to represent the output, i.e., a fresh variable $x$ is declared that represents the output. Under the conditions $\expr_2$, the output $x$ must be less than $\expr_1$, i.e., $e_2 \implies x \leq e_1$. Due to the invariants and contracts not being the strongest, the verification is not complete. However, it is complete for programs that do not use these constructs. 
\section{Formalising \oldtool}
\label{sec:constraintflow}
We formally develop the syntax, type-system, and operational semantics of \oldtool.

\subsection{Syntax}
\label{sec:syntax}
\subsubsection{Statements}
In \oldtool, a program $\Pi$ starts with the shape declaration ($d$) and is followed by a sequence of statements ($s$), i.e., $\Pi$ ::= $d \ ; \ s$. As shown in Fig.~\ref{fig:syntax}, statements include function definitions ($f$) - specified using \func construct, transformer definitions ($\theta$) - specified using \transformer construct, the flow of constraints - specified using \flow construct, and sequence of statements separated by $;$. The output of a function is an expression $\expr$, while the output of a transformer ($\theta_r$) is either a tuple of expressions $t \equiv (\expr_1, \cdots)$, where $\expr_i$ represents the output of each member of the abstract shape, or $(\expr \ ? \ \theta_{r_1} : \theta_{r_2})$, where $ \_ ? \_ : \_$ is the ternary operator.  

\setlength{\grammarindent}{10em}
\begin{figure}
    \centering
    \begin{grammar}
        <Expression> $\expr$ ::= $\constant$ | $\var$ | $\sym$ | $\expr_1 \ \oplus \ \expr_2$  | $\expr [ x ]$ | $f_c(\expr_1, \cdots)$ | $\var.\traverse(\dir, f_{c_1}, f_{c_2}, f_{c_3})\{\expr\}$ | $\expr.\map(f_c)$  | $\solver(\minimize, e_1, e_2)$ | $\cdots$
        
        <Shape-decl> $d$ ::= $\defshape \ (\types_1 \ x_1, \types_2 \ x_2, \cdots) \{\expr\}$  
        
        <Function-def> $f$ ::= $\func \ \var(\types _1 \ \var_1, \types_2 \ \var_2, \cdots) = e$
        
        <DNN-operation> $\eta$ ::= $\affine$ | $\relu$ | $\maxpool$ | $\dotprod$ | $\cfkeywords{Sigmoid}$ | $\cfkeywords{Tanh}$ | $\cdots$
        
        <Transformer-decl> $\theta_d$ ::= $\transformer \ \var$
        
        <Transformer-ret> $\theta_r$ ::= $(\expr_1, \expr_2, \cdots)$ | $(\expr \ ? \ \theta_{r_1} : \theta_{r_2})$
        
        <Transformer> $\theta$ ::= $\theta_d \ \{\eta_1 \rightarrow \theta_{r_1}; \eta_2 \rightarrow \theta_{r_2}; \cdots\}$
        
        <Statement> $s$ ::= $\flow(\dir, f_{c_1}, f_{c_2}, \theta_c)$ | $f$ | $\theta$ | $s_1 \ ; \ s_2$
        
        <Program> $\Pi$ ::= $d \ ; \ s$
    \end{grammar}
    \caption{A part of the BNF grammar for \oldtool. The complete grammar can be found in Appendix A}
    \label{fig:syntax}
\end{figure}

\subsubsection{Expressions}
As shown in Fig.~\ref{fig:syntax}, apart from constants ($\constant$) and variables ($\var$), \sym is also an expression, which can be used to declare a new symbolic variable $\noise$. For every symbolic variable, we implicitly add the constraint most commonly used in DNN certifiers, i.e., $-1 \leq \noise \leq 1$.  We allow the standard binary operators, list operators, function calls, etc. Some operators like `+' are overloaded to also apply to polyhedral and symbolic expressions. Each neuron is associated with its abstract shape and metadata, which can be accessed by square bracket notation, for instance - $\curr[l]$. The \map construct takes in a function name and an expression of type $\polyexp$ (or $\symexp$). The function is applied to all the constituent neurons (or symbolic variables) and adds the results to give a new polyhedral (or symbolic) expression. \traverse is applied to a variable ($\var$) representing a polyhedral expression, and takes in the direction of traversal ($\dir$), a priority function ($f_{c_1}$), a stopping function ($f_{c_2}$), a replacement function $(f_{c_3})$, and a user-defined invariant ($\expr$), needed for verification. We also provide the \solver construct in \cf, which allows calls to external solvers. For example, minimizing an expression \(e_1\) under constraints \(e_2\) can be expressed as \(\solver(\minimize, e_1, e_2)\).

\subsubsection{Specifying Constraints}
To verify a DNN certifier, one must provide the soundness property ($\prop$) along the abstract shape. Also, for \traverse, the programmer must provide an invariant. To define constraints in \oldtool, the operators $==, \le, \ge$ are overloaded and can be used to compare polyhedral expressions as well as \oldtool symbolic expressions. For example, the constraint $n_1 + n_2 \leq n_3$ means that for all possible values of $n_1, n_2,$ and $n_3$ during concrete execution, the constraint must be true. Further, the construct $\embed$ can be used to define constraints such as $e_1 \ \embed \ e_2$, where $e_1$ is a polyhedral expression, and $e_2$ is a symbolic expression. Mathematically, the constraint $n_1 + n_2 \ \embed \ \sym_1 + 2 \ \sym_2$ means $\forall n_1, n_2 \cdot \exists \ \sym_1, \sym_2 \in [-1, 1], s.t., n_1 + n_2 = \sym_1 + 2 \ \sym_2$. In \oldtool, the constraints are expressions of type $\ct$. The binary operators like $\wedge, \vee$ are also overloaded. For example, if $e_1$ and $e_2$ are of the type $\ct$, then $e_1 \wedge e_2$ is a constraint of type $\ct$.

\subsection{Type Checking}
\label{sec:typechecking}
\begin{figure}
    \begin{minipage}{0.3\textwidth}
        \centering
        \resizebox{\textwidth}{!}{
\begin{tikzpicture}

% Define the lattice elements
\node (1) at (0,0) {$\top$};
\node (2) at (-2.25,-1.125) {$\ct$};
\node (3) at (0,-1.125) {$\symexp$};
\node (4) at (2.25,-1.125) {$\polyexp$};
\node (5) at (1.125,-2.25) {$\float$};
\node (6) at (-2.25,-3.125) {$\bool$};
\node (7) at (-1.125,-3.125) {$\typenoise$};
\node (8) at (1.125,-3.125) {$\intt$};
\node (9) at (2.25,-3.125) {$\typeneuron$};
\node (10) at (0,-4.5) {$\bot$};

% Draw the lattice relations
\draw[->] (10) -- (6);
\draw[->] (10) -- (7);
\draw[->] (10) -- (8);
\draw[->] (10) -- (9);
\draw[->] (2) -- (1);
\draw[->] (3) -- (1);
\draw[->] (4) -- (1);
\draw[->] (5) -- (3);
\draw[->] (5) -- (4);
\draw[->] (6) -- (2);
\draw[->] (7) -- (3);
\draw[->] (8) -- (5);
\draw[->] (9) -- (4);

\end{tikzpicture}
}
        \caption{Subtyping Lattice}
        \label{fig:subtyping}
    \end{minipage}
    \hspace{2mm}
    % \vrule
    \hspace{1mm}
    \begin{minipage}{0.66\textwidth}
        \fbox{$\footnotesize \cdot \ \vdash \Pi \ : \ \Gamma, \tau_s \qquad \cdot \ \vdash d \ : \ \tau_s \qquad \Gamma, \tau_s \vdash s \ : \ \Gamma'$}
        \centering
        \resizebox{\textwidth}{!}{
        $
            \begin{array}{c}
            \inferrule*[lab={\textsc{T-program}}]
            {
            \cdot \vdash d : \tau_s \quad
            \cdot, \tau_s \vdash s : \Gamma 
            }
            {
            \cdot \ \vdash d \ ; \ s \ : \ \Gamma, \tau_s
            }
            \hspace{0.5cm}
            \inferrule*[lab={\textsc{T-shape}}]
            {
            \tau_s = [x_1 \mapsto \types_1, \cdots x_n \mapsto \types_n]  \\\\
            \forall i \in [n], \bot \sqsubset \types_i \sqsubset \top
            }
            {
            \cdot \ \vdash \defshape \ (\types_1 \ x_1, \cdots, \types_n \ x_n) : \tau_s
            }
            \\\\
            \inferrule*[lab={\textsc{T-transformer}}]
            {
            \var \not\in \Gamma \quad 
            \Gamma' = \Gamma[\curr \mapsto \typeneuron)][\prev \mapsto \overline{\typeneuron}] \\\\
            \forall i \in [m], \Gamma', \tau_s \vdash \theta_{r_i} : (\types^1_i, \cdots, \types^n_i) \\\\
            \forall j \in [n], \types^j = \sqcup_{i \in [m]}(\types^j_i) \quad  
            \forall j \in [n], \types^j \sqsubseteq \tau^j_s
            }
            {
            \Gamma, \tau_s \vdash \transformer \ \var = \{\eta_1 : \theta_{r_1}, \cdots\} : \qquad \qquad \\\\
            \qquad \qquad \Gamma[x \mapsto (\typeneuron \times \overline{\typeneuron}) \rightarrow (\types^1, \cdots)]
            } 
            \end{array}
        $
        }
    \caption{Type-checking Rules ($\ttt$)}
    \label{fig:typecheckingmain}
    \end{minipage}
\end{figure}

We define a subtyping relation \(\sqsubset\) for the basic types in \oldtool, organized as a lattice (Fig.~\ref{fig:subtyping}). An expression is type-checked to ensure that it has a type other than \(\top\) or \(\bot\). Type-checking involves recording the types of the members of the abstract shape in a record \(\tau_s\) (referred to as \textsc{T-shape} in Fig.~\ref{fig:typecheckingmain}). A static environment \(\Gamma\) maps program identifiers to their respective types, and the tuple \((\Gamma, \tau_s)\) forms the typing context in \oldtool (\textsc{T-program}). We utilize standard function types of the form \(\types_1 \times \cdots \times \types_n \rightarrow \types\), where \(\types_i\) are the argument types and \(\types\) is the return type. The \transformer construct encapsulates the abstract transformers associated with each DNN operation. In rule \textsc{T-transformer}, the output of an abstract transformer \(\theta_{r}\) is a tuple of expressions that undergo recursive type-checking to ensure consistency with \(\tau_s\). The implicit inputs to \transformer are \curr and \prev. For \(n\) members in the user-defined abstract shape and \(m\) DNN operations, the corresponding abstract transformers yield tuples of types \((\types^1_i, \cdots, \types^n_i)\). For each abstract shape element, we define the type \(\types^j = \sqcup_{i \in [m]} \types^j_i\). The transformer type checks if \(j \in [n]\) and \(\types^j_i \sqsubseteq \tau^j_s\), where \(\tau^j_s\) is the type of the \(j\)-th shape member. The type of \curr is \typeneuron, while the type of \prev depends on the DNN operation; for simplicity, we assume \prev is of type \(\overline{\typeneuron}\). If all abstract transformers in the \transformer construct pass type-checking, a new binding is created in \(\Gamma\) mapping the transformer name to the type \(\typeneuron \times \overline{\typeneuron} \rightarrow (\types^1, \cdots, \types^m)\). The detailed description of type-checking in \oldtool can be found in Appendix B.

\subsection{Operational Semantics}
\label{sec:operationalsemantics}
The input concrete DNN is represented as a record \(\hh_C\) that maps the metadata and abstract shape members of all neurons to their respective values. While executing statements in \oldtool, two stores are maintained: (i) \(\fstore\), which maps function names to their arguments and return expressions, and (ii) \(\tstore\), which maps transformer names to their definitions. The general form for the operational semantics of statements in \oldtool is given by: \(\langle s, \fstore, \tstore, \hh_C \rangle \Downarrow \fstore', \tstore', \hh'_C.\) Function definitions add entries to \(\fstore\), while transformer definitions add entries to \(\tstore\). The \flow construct applies transformer \(\theta_c\) to the neurons in the DNN \(\hh_C\), modifying it to \(\hh'_C\).

Each expression in \oldtool evaluates to a value (\(\val\)), with the formal definition of values provided in Appendix C. A record \(\store\) maps variables in \oldtool to concrete values. The general form for the operational semantics of expressions in \oldtool is: \(\langle \expr, \fstore, \store, \hh_C \rangle \Downarrow \val\), with most operations, including unary and binary, following their natural operational semantics.

\begin{figure}
    \begin{minipage}{\textwidth}
        \fbox{$\footnotesize \langle \Pi, \hh_C \rangle \Downarrow \hh_C' \qquad \langle s, \env \rangle \Downarrow \fstore', \tstore', \hh_C' \qquad \langle \expr, \context \rangle \Downarrow \val$}\\
        \centering
        $
        \begin{array}{c}
            \inferrule*[lab = \textsc{OP-map}]
            {
            \langle \expr, \context \rangle \Downarrow \constant_0 + \sum_{i=0}^{i=l} \constant_i \cdot \ver_i \\\\
            \forall i \in [l], \ \langle f_c(\ver_i, \constant_i), \context \rangle  \Downarrow \val_i 
            }
            {
            \langle \expr.\map(f_c), \context \rangle \Downarrow \constant_0 + \sum_{i=0}^{i=l} \val_i
            }
            \hspace{1cm}
            \inferrule*[lab = \textsc{OP-traverse-2}]
            {
            \vset' = \priority(\vset, f_{c_1}, \context) \quad 
            \val = \constant + \val_{\vset'} + \val_{\overline{\vset'}} \\\\
            \langle \val_{\vset'}.\map(f_{c_3}), \context \rangle \Downarrow \val' \\\\
            \val'' = \constant + \val' + \val_{\overline{\vset'}} \\\\
            \vset'' = \filter((\vset \setminus \vset') \cup \neigh(\vset', \dir), f_{c_2}, \context) \\\\
            \langle \val''.\traverse, \context, \vset'' \rangle \Downarrow \val'''
            }
            {
            \langle \val.\traverse(\dir, f_{c_1}, f_{c_2}, f_{c_3}), \context, \vset \rangle \Downarrow \val'''
            }
            \end{array}
        $
    \caption{Big-step Operational Semantics ($\ooo$) of \oldtool}
    \label{fig:operationalsemantics}
    \end{minipage}
\end{figure}

The operational semantics of \map (\textsc{OP-map} in Fig.~\ref{fig:operationalsemantics}) begins by recursively evaluating the input expression \(\expr\), yielding a polyhedral or symbolic expression denoted as \(\bval\). The input function \(f\) is then applied to each component of \(\bval\), resulting in individual outputs \(\val_i\) that are summed to produce the final output. For \traverse, the input expression \(\expr\) is first evaluated to yield a polyhedral value \(\val\). Then, an active vertex set \(\vset\) is established by retrieving constituent neurons from \(\val\)  and filtering out neurons that satisfy the stopping condition \(f_{c_2}\), i.e., \(\vset \gets \filter(\vertices(\val), f_{c_2}, \context)\). This set initializes \(\vset\) and is iterated upon until it is empty. In each iteration, shown in \textsc{OP-traverse-2} (Fig.~\ref{fig:operationalsemantics}), the priority function \(f_{c_1}\) is applied to each neuron in \(\vset\), selecting the highest-priority neurons: \(\vset' \gets \priority(\vset, f_{c_1}, \context)\). The value \(\val\) can be decomposed into three parts: a constant \(\constant\), the value associated with neurons in \(\vset'\), and the value for neurons not in \(\vset'\): \(\val = \constant + \val_{\vset'} + \val_{\overline{\vset'}}\). The replacement function \(f_{c_3}\) is applied only to \(\val_{\vset'}\), retaining the coefficients of the other neurons, resulting in a new polyhedral value: \(\val'' = \constant + \val' + \val_{\overline{\vset'}}\). The active set is updated by removing neurons from \(\vset'\) and adding their neighbors, filtered again to satisfy the stopping condition: \(\vset'' = \filter((\vset \setminus \vset') \cup \neigh(\vset', \dir), f_{c_2}, \context)\). This process continues until the final value is computed. More detailed operational semantics for \traverse and other constructs can be found in Appendix D.

\subsection{Type Soundness}
We demonstrate that if a program type-checks according to the rules of \oldtool, then applying the program according to operational semantics produces an updated abstract element for the input neural network (Theorem~\ref{thm:welltyped}). Lemmas~\ref{lemma:optype-checking} and \ref{lemma:statementtype-checking} establish that if an expression or statement type-checks, it will evaluate according to operational semantics, with the output type consistent with the type computed during type-checking.
Detailed proofs are in Appendix E.

\begin{lemma}
\label{lemma:optype-checking}
Given $(\Gamma, \tau_s)$ and $(\fstore, \store, \hh_C)$ with finite $\hh_C$ such that $(\fstore, \store, \hh_C)$ is consistent with $(\Gamma, \tau_s)$, if $\ \Gamma, \tau_s \vdash \expr : \types$ and $\bot \sqsubset \types \sqsubset \top$, then $\langle \expr, \fstore, \store, \hh_C \rangle \Downarrow \val$ and $\vdash \val : \types'$ s.t. $\types' \sqsubseteq \types$.
\end{lemma}

\begin{lemma}
\label{lemma:statementtype-checking}
Given $(\Gamma, \tau_s)$ and $(\fstore, \store, \hh_C)$ with finite $\hh_C$ such that $(\fstore, \store, \hh_C)$ is consistent with $(\Gamma, \tau_s)$, if $\Gamma, \tau_s \vdash s : \Gamma'$, then $\langle s, \fstore, \store, \hh_C \rangle \Downarrow \fstore', \store', \hh'_C$ s.t. $(\fstore', \store', \hh'_C)$ is consistent with $(\Gamma', \tau_s)$.
\end{lemma}

\begin{theorem}
\label{thm:welltyped}
    A well-typed program in \oldtool successfully terminates according to the operational semantics, i.e., $\ttt \models \ooo$. Formally, if $\cdot \vdash \Pi : \Gamma, \tau_s$ then $\langle\Pi, \hh_C\rangle \Downarrow \hh'_C$
\end{theorem}

\begin{proof}[Proof sketch]
Theorem~\ref{thm:welltyped} follows directly from Lemmas~\ref{lemma:optype-checking} and \ref{lemma:statementtype-checking}. The lemmas are proved by induction on the structures of \(\expr\) and \(s\). For Lemma~\ref{lemma:optype-checking}, the case where \(\expr \equiv \var \cdot \traverse(\dir, f_1, f_2, f_3)\{\_\}\) is particularly intricate as it involves traversing the DNN. We demonstrate this by constructing a bit vector \(B\) representing the neurons in the DNN, ordered topologically (as a DAG), where 1 indicates the presence in the active set and 0 indicates absence. We show that the value of \(B\) is bounded and decreases by at least 1 in each iteration.
\end{proof}

\section{\cf---Bounded Automatic Verification}
\label{sec:verification}
We present bounded automated verification for the soundness verification of every abstract transformer specified for a DNN certifier. Bounds are assumed on the maximum number of neurons in the previous layer ($n_{prev}$), and the maximum number of \cf symbolic variables used by the certifier ($n_{sym}$). We reduce this verification task to a first-order logic query which can be handled with an off-the-shelf SMT solver. In this section, the terms \textit{symbolic variables} and \textit{constraints} refer to SMT symbolic variables and constraints over them, not the \cf symbolic variables $\noise$ or constraints unless stated otherwise. When executing the certifier using operational semantics, the input is a concrete DNN. So, the soundness of the certifier must be verified for all possible inputs, i.e., all possible DNNs. Our key insight is a \textit{Symbolic DNN} that can represent arbitrary concrete DNNs within the above-stated bounds. In a nutshell, given a \cf program, we perform the following steps: (i) create a symbolic DNN (\S~\ref{sec:graphcreation}), (ii) expand the symbolic DNN to be able to execute the program (\S~\ref{sec:graphexpansion}), (iii) execute the program on the symbolic DNN using symbolic semantics (\S~\ref{sec:symexecution}), (iv) generate the verification query and verify the query using an off-the-shelf SMT solver (\S~\ref{sec:smtverification}). We prove the \textit{soundness} of the symbolic semantics w.r.t. the operational semantics (\S~\ref{sec:verificationsoundness}). So, verifying the soundness of a certifier for a symbolic DNN ensures the soundness of any concrete DNN within the bounds.

\subsection{Symbolic DNN Creation}
\label{sec:graphcreation}
We introduce the concept of a \textit{Symbolic DNN} to represent an arbitrary DNN and the corresponding abstract shapes symbolically. It represents all subgraphs of any arbitrary neural network on which the specified transformer can be applied. So, it consists of symbolic values representing neurons necessary for executing the transformer. 
\begin{definition}
    A symbolic DNN is a graph $\langle V, E, \hh_S, \cc \rangle$, where $V$ is the set of neurons and $E$ is the set of edges representing the DNN operations (e.g., \affine, \relu). Each node is associated with an abstract shape and metadata. $\hh_S$ is a record that maps each neuron, its shape members, and metadata to symbolic variables and $\cc$ represents constraints over the symbolic variables.
\end{definition}
As explained in \S~\ref{sec:overviewgraphcreation}, \ref{sec:overviewgraphexpand}, for each DNN operation $\eta$ (e.g., \relu), we initialize a symbolic DNN with neurons representing \prev and \curr that are respectively the input and output of $\eta$. The edges are only between \curr and \prev neurons and represent the operation $\eta$. $\cc$ encodes $\eta$ and the assumption of the user-specified property $\prop$ over all of the neurons in the symbolic DNN. Each shape member and metadata associated with these neurons is set to symbolic variables in $\hh_S$. In subsequent sections, we omit $V$ and $E$ and refer to $\hh_S, \cc$ as a symbolic DNN. 
\begin{figure}
\centering 
\resizebox{\textwidth}{!}{
    $
    \begin{array}{c}
        \inferrule*[lab = \textsc{E-shape-b}]
        {
        \langle \expr, \fstore, \sstore, \hh_s, \cc \rangle \downarrow n, \_ \\\\
        \exn(n, \var, \tau_s, \hh_S, \cc, \prop) = \hh'_S, \cc'
        }
        {
        \ex(\expr[x], \tau_s, \symc, \prop) = \hh'_S, \cc' 
        }
        \hspace{1cm}
        \inferrule*[lab = \textsc{G-map}]
        {
        \tau_s, \scontext, \prop \models \expr \leadsto \hh_{S}', \cc' \\\\
        \ex(\expr, \tau_s, \fstore, \sstore, \hh_{S}', \cc', \prop) = \hh_{S_0}, \cc_0 \\\\
        \langle \expr, \fstore, \sstore, \hh_{S_0}, \cc_{0} \rangle \downarrow \sval_{b_0} + \sum^j_{i=1} n_i * \sval_{b_i} \\\\
        \forall i \in [j] \quad \tau_s, \fstore, \sstore_i, \hh_{S_{i-1}}, \cc_{i-1} \models f_c(n_i, \sval_{b_i}) \leadsto \hh_{S_i}, \cc_i
        }
        {
        \tau_S, \symc, \prop \models \expr\cdot\map(f_c) \leadsto \hh_{S_j}, \cc_{S_j} 
        }
        \\\\
        \inferrule*[lab = \textsc{Expand-poly-r}]
        {
        \tau_s(\var) = \polyexp \qquad
        \hh_S[n[\var]] = \sval_{b_r} \qquad
        \nn = [n'_1, \cdots n'_j] \\\\
        \hh_{S_0} = \hh_S \qquad
        \forall i \in [j], \hh_{S_i}, \cc_i = \addg(n'_i, \tau_s, \hh_{S_{i-1}}, \prop, \cc_{i-1}) \\\\
        \bsval = \sval_{r_0} + \sum^j_{i=1} \sval_{r_i} * n'_i \qquad
        \hh_S' = \hh_{S_j}[n[\var] \mapsto \bsval]
        }
        {
        \exn(n, \var, \tau_s, \hh_S, \cc_0, \prop) = \hh'_S, \cc_j
        }
        \\\\
        \inferrule*[lab = \textsc{G-traverse}]
        {
        \nn = [n_1, \cdots n_j] \\\\
        \tau_s, \scontext, \prop \models \expr \leadsto \hh_{S_0}, \cc_0 \qquad  
        \forall i \in [j], \hh_{S_i}, \cc_i = \addg(n_i, \tau_s, \hh_{S_{i-1}}, \prop, \cc_{i-1}) \\\\
        \bsval = \sval_{b_0} + \sum^j_{i=1} \sval_{b_i} * n_i \quad \sval_{b}, \sval_{b_0}, \sval_{b_i}\text{ are fresh symbolic variables}\\\\
        \hh'_{S_0} = \hh_{S_j} \qquad 
        \cc'_0 = \cc_j \qquad 
        \forall i \in [j], \tau_s, \fstore, \sstore, \hh'_{S_{i-1}}, \cc'_{i-1}, \prop \models f_{c_2}(n_i, \sval_{b_i}) \leadsto \hh'_{S_i}, \cc'_i \\\\
        \hh''_{S_0} = \hh'_{S_j} \qquad 
        \cc''_0 = \cc'_j \qquad
        \forall i \in [j], \tau_s, \fstore, \sstore, \hh''_{S_{i-1}}, \cc''_{i-1}, \prop \models f_{c_3}(n_i, \sval_{b_i}) \leadsto \hh''_{S_i}, \cc''_i
        }
        {
        \tau_s, \scontext, \prop \models \var\cdot\trav(\dir, f_{c_1}, f_{c_2}, f_{c_3})\{\expr\} \leadsto \hh''_{S_j}, \cc''_j
        }
    \end{array}
    $
    }
    \caption{Symbolic DNN Expansion}
    \vspace{-2mm}
    \label{fig:graphexpansion}
\end{figure}
Next, to enable symbolic execution of the specified transformer, we may need to expand the symbolic DNN. For example, in the case of the expression $\expr \equiv \prev[U].\map(\texttt{\footnotesize{foo}})$, where \texttt{\footnotesize{foo}} is a user-defined function, $\prev[U]$ must be expanded before we can apply $\texttt{\footnotesize{foo}}$. The symbolic DNN expansion step is written in the form $\tau_s, \symc, \prop \models \expr \leadsto \hh_S', \cc'$ (\S~\ref{sec:graphexpansion}). After the symbolic DNN expansion step of an expression $\expr$, it can be symbolically executed using the symbolic semantics. The symbolic semantics are defined in the form $\langle \expr, \fstore, \sstore, \hh_S, \cc\rangle \downarrow \sval, \cc'$ (\S~\ref{sec:symexecution}).   

\subsection{Symbolic DNN Expansion}
\label{sec:graphexpansion}
The expansion step is done by statically analyzing the transformer specification and expanding the symbolic DNN accordingly. A subset of the rules for symbolic DNN expansion is shown in Fig.~\ref{fig:graphexpansion}. The complete set of rules can be found in Appendix G. This step analyzes the expression $\expr$ for the presence of one of three constructs - \map, function call, or $\traverse$. The rules for \map and \traverse are shown in rules \textsc{G-map} and \textsc{G-traverse} (in Fig.~\ref{fig:graphexpansion}). In the rule \textsc{G-map}, the graph expansion is recursively applied to the input expression $\expr$ in the first line. Then, since it is a \map construct, it must be ensured that the output of $\expr$ is in expanded form. This is done in the second line. The third line asserts that the output from the symbolic execution of $\expr$ is already in the expanded form $\sval_{b_0} + \sum^j_{i=1} n_i * \sval_{b_i}$. Since the \map construct applies the function call to all the individual summands of the output, the DNN expansion step is applied to each function call before symbolically executing it. This is shown in the fourth line of \textsc{G-map} rule. 

Now, we explain the $\ex(\expr, \tau_s, \fstore, \sstore, \hh_S, \cc, \prop)$ rules used to ensure that the output of symbolically executing $\expr$ is in expanded form. Here, $\ex$ takes in an expression, $\expr$, $\tau_S$, $\fstore$, $\sstore$, $\hh_S$, $\cc$, and the abstract shape constraint definition $\prop$. The output of $\ex$ is $\hh'_S$, which can contain new shape members and expanded versions of existing shape members, and $\cc'$, which is extended to include the soundness property assumptions on any new neurons added to the symbolic DNN or the constraint $-1 \leq \noise \leq 1$ for any new \cf symbolic variables. In Fig.~\ref{fig:graphexpansion}, we show one of the base cases of this step, \textsc{Expand-poly-r}, where we expand the accessed polyhedral shape member of the input neuron. In the first line, we symbolically execute $\expr$ to get the neuron $n$. Then, if $x$ is of the type $\polyexp$ or $\symexp$, we add new symbolic variables to the symbolic DNN accordingly. 

Another interesting case for graph expansion is the expressions $x.\trav(d, f_{c_1}, f_{c_2}, f_{c_3})$ shown in the rule \textsc{G-traverse}, where we recursively call the graph expansion for the invariant $\expr$ in line 1. Since we cannot symbolically execute the \traverse construct due to it being a loop with an undetermined number of iterations at the analysis time, we declare new neurons to represent the output. In line 2, these new neurons and their corresponding metadata are added to the symbolic DNN. So, the output of symbolically executing \traverse is represented as $\bsval = \sval_{b_0} + \sum^j_{i=1} \sval_{b_i} * n_i$ in line 3. When generating the query, we also need to assume that the stopping condition ($f_{c_2}$) is true on all summands of the final output, and also the function $f_{c_3}$ is applied to all the summands. So, in lines 4-5, we recursively apply the symbolic DNN expansion on all the summands using $f_{c_2}$ and $f_{c_3}$. 

\subsection{Symbolic Semantics}
\label{sec:symexecution}
\begin{figure}
\centering
    $
    \begin{array}{c}
        \inferrule*[lab = \textsc{Sym-ternary}]
        {
        \langle \expr_1, \symc \rangle \downarrow \sval_{1}, \cc_1 \quad
        \langle \expr_2, \fstore, \sstore, \hh_S, \cc_1 \rangle \downarrow \sval_2, \cc_2 \quad 
        \langle \expr_3, \fstore, \sstore, \hh_S, \cc_2 \rangle \downarrow \sval_3, \cc_3 
        } 
        {\langle (\expr_1 ? \expr_2 : \expr_3), \symc \rangle \downarrow If(\sval_{1}, \sval_2, \sval_3), \cc_3}
        \\\\
        \inferrule*[lab = \textsc{Check-induction}]
        {
        \nn = [n'_1, \cdots, n'_j] \qquad
        \bsval = \sval^\rrr_{0} + \sum^j_{i=1} \sval^\rrr_{i} * n'_i \qquad
        \sstore' = \sstore[\var \mapsto \bsval] \\\\ 
        \langle \expr, \fstore, \sstore', \hh_S, \cc \rangle \downarrow \bsval', \cc_0 \qquad
        \forall i \in [j], 
        \langle f_{c_2}(n_i, \sval_{r_i}), \fstore, \sstore', \hh_S, \cc_{i-1} \rangle \downarrow \sval'_i, \cc_i \\\\ 
        \cc'_0 = \cc_j\qquad
        \forall i \in [j], \langle f_{c_3}(n_i, \sval_{r_i}), \fstore, \sstore', \hh_{S}, \cc'_{i-1} \rangle \downarrow \sval''_i, \cc'_i \\\\
        \sval'' = \sval_{r_0} + \sum^j_{i=1} If(\sval'_i, \sval''_i, \sval_{r_i}*n_i) \qquad 
        \sstore'' = \sstore[\var \mapsto \sval''] \qquad 
        \langle \expr, \fstore, \sstore'', \hh_S, \cc'_j \rangle \downarrow \sval''', \cc'' 
        }
        {
        \cinduction(\var \cdot \traverse(\dir, f_{c_1}, f_{c_2}, f_{c_3})\{\expr\}, \fstore, \sstore, \hh_S, \cc) = \unsat(\neg(\cc_0 \wedge \bsval' \implies \cc''_j \wedge \sval'''))
        } 
        \\\\
        \inferrule*[lab = \textsc{Check-invariant}]
        {
        \langle \expr, \symc \rangle \downarrow \sval, \cc' \quad
        \bsval = \unsat(\neg(\cc' \implies \sval))
        \\\\
        \bsval' = \cinduction(\var \cdot \traverse(\dir, f_{c_1}, f_{c_2}, f_{c_3})\{\expr\}, \fstore, \sstore, \hh_S, \cc) 
        }
        {
        \cinvariant(\var \cdot \traverse(\dir, f_{c_1}, f_{c_2}, f_{c_3})\{\expr\}, \fstore, \sstore, \hh_S, \cc) = \bsval \wedge \bsval', \cc'
        }
        \\\\
        \inferrule*[lab = \textsc{Sym-traverse}]
        {
        \cinvariant(\var\cdot\traverse(\dir, f_{c_1}, f_{c_2}, f_{c_3})\{\expr\}, \symc) = \true, \cc' \\\\
        \bsval = \sval_{0} + \sum^j_{i=1} \sval_{i} * \sval'_i \qquad 
        \sstore' = \sstore[x \mapsto \bsval] \qquad
        \langle \expr, \fstore, \sstore', \hh_S, \cc' \rangle \downarrow \sval, \cc''  
        }
        {
        \langle \var\cdot\traverse(\dir, f_{c_1}, f_{c_2}, f_{c_3})\{\expr\}, \symc \rangle \downarrow \bsval, \sval \wedge \cc''
        }
    \end{array}
    $
    \caption{Symbolic Semantics ($\sss$) for \cf expressions: $\langle \expr, \symc \rangle \downarrow \sval, \cc'$}
    \label{fig:symsemantics}
\end{figure}
Like operational semantics, symbolic semantics ($\sss$) use $\fstore$ which maps function names to their formal arguments and return expressions. However, instead of the concrete store $\store$ used in operational semantics, it uses a symbolic store, $\sstore$, which maps the identifiers to their symbolic values $\sval$ in expanded form. Symbolic semantics output a symbolic value $\sval$, and also add additional constraints to $\cc$, i.e., $\langle \expr, \fstore, \sstore, \hh_S, \cc \rangle \downarrow \sval, \cc'$. Constants, variables, and the introduction of new \cf symbolic variables using the $\noise$ construct are the base cases of the symbolic semantics of \cf. Unary, binary, and ternary operations are straightforward recursive cases. We show \textsc{Sym-ternary} in Fig.~\ref{fig:symsemantics}, where the three expressions $\expr_1, \expr_2,$ and $\expr_3$ are recursively executed to output $\sval_1, \sval_2, \sval_3$, respectively. The output value of the ternary operation is thus returned as $\text{If}(\sval_1, \sval_2, \sval_3)$, where $\text{If}$ is a Z3 construct. Also, the constraints are accumulated in the recursive calls. The symbolic semantics for \map construct are similar to the operational semantics and are therefore omitted here. We now discuss the semantics for the more challenging \traverse construct. Detailed semantics for other constructs are available in Appendices F and G.

Due to the lack of DNN architecture information, full symbolic execution of the loop specified by the \traverse construct is not feasible. So, we validate the user-provided invariant's soundness and subsequently use it for the symbolic semantics of \traverse. In the rule \textsc{Sym-traverse} in Fig.~\ref{fig:symsemantics}, $\expr$ is the user-defined invariant for the traversal, $\bsval$ is the output symbolic polyhedral expression, and $\sval$ is the result of applying the invariant $\expr$ to $\bsval$. We check the soundness of this invariant in two steps (\textsc{Check-invariant} in Fig.~\ref{fig:symsemantics}). First, we verify that the invariant is satisfied at the initial state. Here, $\sval$ represents the evaluated invariant expression $\expr$ applied to the input state of \traverse. $\unsat(\neg(\cc' \implies \sval))$ implies that $\sval$ is true under the conditions, $\cc'$, which are valid before executing $\traverse$. Second, we verify that the invariant is inductive (\textsc{Check-induction}). In $\cinduction$, $\unsat(\neg(\cc_0 \wedge \bsval' \implies \cc''_j \wedge \sval'''))$ means that under the assumption that the invariant holds before an iteration of \traverse, the invariant must hold after the iteration of \traverse. If the invariant is validated, we create a symbolic value of the form $\sval_{0} + \sum^j_{i=1} \sval_{i} *\sval'_{i}$ to represent the output of $\var.\traverse(d, f_{c_1}, f_{c_2}, f_{c_3})\{\expr\}$ and assume, in $\cc$, that the invariant holds on this output. 

\subsection{Queries for Verification}
\label{sec:smtverification}
Initially, it is assumed that the property $\prop$ holds for all the neurons in the symbolic DNN. To compute the new abstract shape, the user-specified abstract transformer is executed using the symbolic semantics as described in \S~\ref{sec:symexecution}. This results in the new abstract shape (for \curr) - a tuple of symbolic values $(\sval_1, \cdots, \sval_n)$ and a condition, $\cc'$ that encodes constraints over $\sval_i$. To verify the soundness of the abstract transformer, we need to check that if the property $\prop$ holds for all the neurons in the symbolic DNN ($\forall n \in \hh_S, \prop(\alpha_n, n)$), then it also holds for the new symbolic abstract shape values, $\prop(\alpha'_{\curr}, \curr)$, where $\alpha'_{\curr} = (\sval_1, \cdots, \sval_n)$. We split the query into two parts: (i) antecedent $p$---encoding the initial constraints on the symbolic DNN, the computations of the new abstract shape for $\curr$, represented by $\rr$, the semantic relationship $\eta$ between \curr and \prev, and any path conditions relevant to the specific computations we are verifying, $\cc'$. $p \triangleq (\forall n \in \hh_S, \prop(\alpha_n, n)) \wedge \curr=\eta(\prev) \wedge \rr \wedge \cc'$ (ii) consequent $q$---encoding the property $\prop$ applied to the new abstract shape of $\curr$. $q \triangleq \prop(\alpha'_{\curr},\curr)$. So, the final query is $\textsf{checkValid}(p \implies q)$.  

\begin{figure}
    % \resizebox{\textwidth}{!}{
    \begin{tikzpicture}
        \node[circle, draw, minimum size=1cm] (node) at (0,0) {$n$};
        \node[right=0.5cm of node, yshift=0.9cm] (l) {$l  = -1$};
        \node[right=0.5cm of node, yshift=0.3cm] (u) {$u  = 3$};
        \node[right=0.5cm of node, yshift=-0.9cm] (L) {$L = 4 + 5n_1 + 6n_2$};
        \node[right=0.5cm of node, yshift=-0.3cm] (U) {$U = 3 - 2n_1$};
        \draw[dotted] (node.east) -- (l.west);
        \draw[dotted] (node.east) -- (u.west);
        \draw[dotted] (node.east) -- (L.west);
        \draw[dotted] (node.east) -- (U.west);

        \draw[dotted, ->] (-150:1.5) -- (node.west);
        \draw[dotted, <-] (-170:1.5) -- (node.west);
        \draw[dotted, ->] (-190:1.5) -- (node.west);
        \draw[dotted, <-] (-210:1.5) -- (node.west);

        \node[draw, above=of current bounding box, yshift=-0.7cm, xshift=-2cm] (name) {$\hh_C$};
        \node[draw, right=of name, xshift=6cm] {$\hh_S$};

        \begin{scope}[shift={(10cm,0cm)},local bounding box=f1box]
            \node (node) at (0,0) {$\prev$};
            
            \node[below=0cm of node.center, rotate=270, right=0.01cm] (node2) {$\mapsto$};
            \node[below=0.1cm of node2.center] (node3) {$\sval_p$};
            
            \node[left=0.5cm of node, yshift=0.9cm] (l) {$l  \mapsto \sval_p^l$};
            \node[left=0.5cm of node, yshift=0.3cm] (u) {$u  \mapsto \sval_p^u$};
            \node[left=0.5cm of node, yshift=-0.3cm] (L) {$L \mapsto \sval_p^L$};
            \node[left=0.5cm of node, yshift=-0.9cm] (U) {$U \mapsto \sval_{r}^1 + \sval_{r}^2*\sval_{n_1} + \sval_{r}^3*\sval_{n_2}$};

            \draw[dotted] (node.west) -- (l.east);
            \draw[dotted] (node.west) -- (u.east);
            \draw[dotted] (node.west) -- (L.east);
            \draw[dotted] (node.west) -- (U.east);

            \draw[dotted, ->] (node) -- (0:1.5);
            
        \end{scope}
    \end{tikzpicture}
    % }
    % \vspace{-2mm}
    \caption{Parts of Concrete DNN $\hh_C$ and Symbolic DNN $\hh_S$}
% \vspace{-4mm}
    \label{fig:overapproxneuron}
\end{figure}
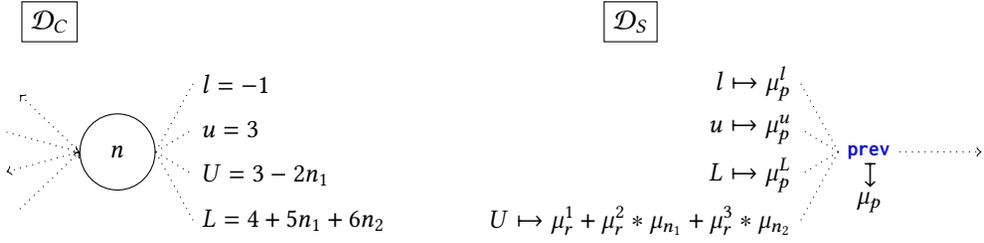
    
\subsection{Correctness of Verification Procedure}
\label{sec:verificationsoundness}
We define a notion of \textit{over-approximation} of a concrete DNN by a symbolic DNN, a concrete value by a symbolic value, etc. So, any property proved by our verification algorithm for a symbolic DNN also holds for any concrete DNN that is over-approximated by the symbolic DNN. This notion lets us establish the correctness of the \cf verification procedure.

\subsubsection{Over-Approximation}
Fig.~\ref{fig:overapproxneuron} shows parts of a concrete DNN $\hh_C$ and a symbolic DNN $\hh_S$ from Fig.~\ref{fig:symbolicdnn}b. The neuron \prev in $\hh_S$ over-approximates the neuron $n$ in the concrete DNN $\hh_C$ if $\varphi$ is satisfiable, where $\varphi \equiv (\sval_p^l = -1) \wedge (\sval_p^u = 3) \wedge (\sval_p^L = 4 + 5n_1 + 6n_2) \wedge (\sval_{r}^1 + \sval_{r}^2*\sval_{n_1} + \sval_{r}^3*\sval_{n_2} = 3 - 2n_1)$. Further, if $\sval_p, \sval_{n_1}, \sval_{n_2}$ represent $n, n_1, n_2$ respectively, they must also be equal, i.e., $\varphi_1 \equiv \varphi \wedge (\sval_p = n) \wedge (\sval_{n_1} = n_1) \wedge (\sval_{n_2} = n_2)$ must be satisfiable. Note that the neurons in $\hh_C$ are not assigned any values and are therefore symbolic themselves. So, $\varphi_1$ must be satisfiable for all possible values of $n, n_1, n_2$ in $\hh_C$. Further, the symbolic DNN has another component $\cc$ which imposes constraints on $\sval_i$. So, the formula must be satisfiable under the constraints $\cc$, i.e., $\varphi_2$ must be true.
\begin{equation}
\label{eq:eq1}
    \varphi_2 = \forall \{n, n_1,n_2\} \cdot \exists \{\sval_p, \sval_{n_1}, \sval_{n_2}, \sval_p^l, \cdots\} \cdot \big(\varphi \wedge (\sval_p = n) \wedge (\sval_{n_1} = n_1) \wedge (\sval_{n_2} = n_2) \wedge \cc \big)
\end{equation}
In the symbolic DNN, $\cc$ contains (i) the constraints encoded by the property $\prop$ assumed on all the neurons in the symbolic DNN, and (ii) the edge relationship between \curr and \prev. 
\begin{definition}
    A symbolic DNN $\hh_S, \cc$ over-approximates a concrete DNN $\hh_C$ if $\forall Y \cdot \exists W \cdot (\cc \bigwedge_{t \in \dom(\hh_S)}\hh_S(t) = \hh_C(t))$, where $Y$ is the set of neurons and \cf symbolic variables in $\hh_C$ and $W$ is the set of all SMT symbolic variables in $\hh_S$. 
\end{definition}

Further, in Equation~\ref{eq:eq1}, all the variables inside universal quantifier ($n, n_1, n_2$) are set equal to variables in the existential quantifier $\sval_p, \sval_{n_1}, \sval_{n_2}$. So, the equation can be rewritten by simply replacing the variables within the universal quantifier with corresponding variables in the existential quantifier, and removing the corresponding equality constraints, i.e., $\varphi_3 = \varphi_2$, where $\varphi_3 = \forall \{\sval_p, \sval_{n_1}, \sval_{n_2}\} \cdot \exists \{\sval_p^l, \sval_p^u, \sval_p^L, \sval_{r}^1, \sval_r^2, \sval_r^3\} \cdot (\varphi \wedge \cc)$.  

In our example in Fig.~\ref{fig:overapproxneuron}, $Y = \{\sval_p, \sval_{n_1}, \sval_{n_2}$\}, and $W$ is the set of all the other symbolic variables used in $\hh_S$. So, a symbolic DNN $\hh_S, \cc$ over-approximates a concrete DNN $\hh_C$ if $\forall Y \cdot \exists W \cdot (\cc \wedge \bigwedge_{t \in \dom(\hh_S)}(\hh_S(t) = \hh_C(t)))$. There are two types of symbolic variables in $W$---ones that represent constants during concrete execution and ones that represent polyhedral or symbolic expressions. So, we partition $W$ into two sets, $X$ and $Z$, where $X$ contains the symbolic variables representing constants, while $Z$ contains the other symbolic variables. So, we can then re-write $\varphi_3$ as $\varphi_4 = \forall Y \cdot \exists X \cdot \exists Z \cdot (\cc \wedge \bigwedge_{t \in \dom(\hh_S)}(\hh_S(t) = \hh_C(t)))$. Note that in the example above, $X = \{\sval_p^l, \sval_p^u, \sval_r^1, \sval_r^2, \sval_r^3\}, Z = \{\sval_p^L\}$. From Equation~\ref{eq:eq1}, since $\sval_p^l, \sval_p^u, \sval_r^1, \sval_r^2, \sval_r^3$ are independent of $n, n_1, n_2$, we bring the set $X$ out of the $\forall$ quantifier. Generalizing this notion, we use the definition \textsc{Over-approx-DNN} in Fig.~\ref{fig:defs}. The over-approximation of a concrete DNN $\hh_C$ by a symbolic DNN $\hh_S, \cc$ is represented as $\hh_C \prec_\cc \hh_S$. The definitions for a symbolic value over-approximating a concrete value and a symbolic store over-approximating a concrete store can be found in (Appendix H). 

Using these definitions of over-approximation, we prove two important properties. First, if a symbolic DNN over-approximates the concrete DNN, then expanding the symbolic DNN maintains the over-approximation. Second, we show that given a \cf expression that type-checks, if one starts with a symbolic DNN $\hh_S, \cc$ and a concrete DNN $\hh_C$ such that $\hh_C \prec_\cc \hh_S$, then the output of applying symbolic semantics on $\hh_S, \cc$ over-approximated the output of applying operational semantics on $\hh_C$. We prove this using bisimulation (rule \textsc{Bisimulation} in Fig.~\ref{fig:defs}), where we simultaneously apply the operational semantics to the concrete DNN and the symbolic semantics to the symbolic DNN $\hh_S, \cc$. The complete details can be found in Appendix I.

\begin{figure}
    \small
    $
    \begin{array}{c}
        \inferrule*[lab={\textsc{Over-approx-DNN}}]
        {
        \dom(\hh_S) \subseteq \dom(\hh_C) \qquad 
        X = \cs(\hh_S, \cc) \qquad  
        Y = \ns(\hh_S, \cc) \cup \es(\hh_S, \cc) \\\\
        Z = \ps(\hh_S, \cc) \cup \zs(\hh_S, \cc) \cup \cts(\hh_S, \cc) \\\\
        \exists X \cdot \forall Y \cdot \exists Z \cdot \Big(\cc \wedge \bigwedge_{t \in \dom(\hh_S)}\hh_S(t) = \hh_C(t) \Big)
        }
        {
        \hh_C \prec_{\cc} \hh_S
        } 
        \\\\
        \inferrule*[lab={\textsc{Bisumation}}]
        {
        \langle \expr, \fstore, \store, \hh_C \rangle \Downarrow \val \qquad 
        \langle \expr, \fstore, \sstore, \hh_S, \cc \rangle \downarrow \sval, \cc' \\\\
        \exists! X \cdot \forall Y \cdot\exists!Z \cdot \Big( CS(\sval,\val) \wedge \cc' \wedge \m  \wedge \bigwedge_{t \in \dom(\hh_S)} \hh_S(t) = \hh_C(t) \Big)
        }
        {
        \ll \expr, \fstore, \store, \sstore, \hh_C, \hh_S, \cc \gg \updownarrow \val, \sval, \cc', \m
        } 
    \end{array}
    $
    \caption{Definitions for Over-approximation and Bisimulation}
    \label{fig:defs}
\end{figure}

\subsubsection{Soundness and Completeness}
We show that if \cf concludes that the abstract transformers specified in the program are verified to maintain the user-defined property $\prop$, then executing the program on any concrete DNN also maintains the property $\prop$. We prove this by initially creating a symbolic DNN with only the neurons representing \curr and \prev and edges representing their corresponding DNN operation ($\eta$ - for example \relu). This over-approximates any part of an arbitrary concrete DNN (within the bounds of verification) which is the output of $\eta$. Next, the over-approximation is maintained during symbolic DNN expansion and executing symbolic semantics. Finally, the query is generated over symbolic values that overapproximate the corresponding concrete values. So, if the SMT solver concludes that the property $\prop$ is maintained over the symbolic DNN, then we can conclude that $\prop$ will also be maintained over all over-approximated concrete DNNs. Further, since the symbolic semantics are not exact only for \traverse and \solver constructs, \cf is complete, excluding these constructs.
\begin{theorem}[Soundness]
\label{soundnesstheorem}
    For a well-typed program $\Pi$, if \cf verification procedure proves it maintains the property $\prop$, then upon executing $\Pi$ on all concrete DNNs within the bounds of verification, the property $\prop$ will be maintained at all neurons in the DNN. 
\end{theorem}

\begin{theorem}[Completeness]
\label{completenesstheorem}
    If executing a well-typed program $\Pi$ that does not use \traverse and \solver constructs on all concrete DNNs within the bounds of verification maintains the property $\prop$ for all neurons in the DNN, then it can be proved by the \cf verification procedure.
\end{theorem}

\section{Evaluation}
\label{sec:evaluation}
We demonstrate that designing the formal semantics for \oldtool and the verification procedure \cf enables users to design and verify new DNN certifiers. The new designs include---(i) variations to the existing certifiers, (ii) supporting new DNN operations within the existing abstract domains, and (iii) completely new abstract domains and transformers. In practice, the implementations of existing DNN certifiers~\cite{interval, deeppoly, zhang2018crown, deepz, refinezono, forwardbackward} employ various techniques to adjust the scalability vs precision tradeoff. Incorporating such modifications to the original algorithms unintentionally alters their mathematical logic. However, the original pen-and-paper proofs do not ensure the correctness of the certifiers with these modifications. In \S~\ref{sec:casestudydeeppoly}, we demonstrate that these modified certifiers can be verified using \cf by specifying them in \oldtool. In \S~\ref{sec:newoperations}, we extend DNN certifiers to support new operations such as \abs, \hsigmoid, etc. by designing abstract transformers, which has not been addressed by any existing work~\cite{deeppoly}. We also show the verification of their soundness using \cf. In \S~\ref{sec:newcertifiers}, we design new abstract domains and their corresponding transformers in \oldtool and verify their soundness using \cf. 

Finally, in \S~\ref{sec:existing}, we show that \oldtool can specify and verify the above-mentioned diverse existing DNN certifiers, covering various abstract domains, transformers, and flow directions. We evaluated a diverse set of state-of-the-art DNN certifiers, including IBP~\cite{interval}, DeepPoly~\cite{deeppoly}, CROWN~\cite{zhang2018crown}, DeepZ~\cite{deepz}, RefineZono~\cite{refinezono}, Vegas~\cite{forwardbackward}, and Hybrid Zonotope~\cite{hybzono}. For all our experiments, we demonstrate that our verification procedure, \cf, can automatically prove the soundness of the certifiers specified in \oldtool or detect unsoundness. The benchmarks for testing the unsoundness detection using \cf were created by introducing random bugs programmatically in the DNN certifiers, following a methodology established in prior research~\cite{bugs}. The details are provided in Appendix K.1. 

\subsubsection*{DNN Operations}
We focus on the widely used DNN operations, including primitive operations like \relu, \cfkeywords{Max}, \cfkeywords{Min}, \cfkeywords{Add}, \cfkeywords{Mult}, etc., and composite operations like \affine, \maxpool, etc. The primitive operations are the ones that take a small, fixed number of inputs, like the addition or multiplication of 2 neurons. Since these can be composed to define composite operations, such as Attention layers, the corresponding abstract transformers can also be composed accordingly. Although verifying transformers for primitive operations directly implies the soundness of arbitrary compositions, in some cases, transformers can be more precise if specified directly for composite operations. In such cases, we show the specification and verification for composite operations. 

We focus on the abstract transformers where the verification problem is known to be decidable. Although it is possible to express transformers for activation functions like Sigmoid and Tanh in \oldtool (Appendix K.2), their verification may become undecidable~\cite{DBLP:journals/corr/abs-2305-06064}. In the future, \cf verification can be extended to handle these transformers using $\delta$-complete decision procedures~\cite{Gao2012CompleteDP}. Currently, our verification queries fall under SMT of Nonlinear Real Arithmetic (NRA), decidable with a doubly exponential runtime in the worst case~\cite{Khanh2012SMTFP}. 

\subsubsection*{Verification Bounds}
For verification of composite operations - \affine and \maxpool, the parameters, $n_{prev}$ (maximum number of neurons in a layer) and $n_{sym}$ (maximum length of a polyhedral or symbolic expression) are used during the graph expansion step and impact the verification times. For our experiments, we set $n_{sym} = n_{prev}$. Note that $n_{prev}$ is an upper bound for the maximum number of neurons in a single layer, without restricting the total neuron count in the DNN. Therefore, the DNN can have an arbitrary number of layers, each with at most $n_{prev}$ neurons, thereby, allowing for an arbitrary total number of neurons in the DNN. We set these parameters based on the sizes of layers within DNNs that existing certifiers currently handle~\cite{gpupoly, deeppoly, zhang2018crown, cloud}. For \maxpool, \minpool, and \avgpool, existing certifiers handle at most $10$ neurons at a time, so we set $n_{prev} = n_{sym} = 10$. The \affine layer includes DNN operations like convolution layers and fully-connected layers. In Table~\ref{table:complexcertifiers}, we present the computation times for \affine with $n_{prev}=n_{sym}=2048$. In Fig.~\ref{fig:affine}, we show how the verification time scales with parameter values ($n_{prev}=n_{sym}$), ranging from $32$ to $8192$, for \affine transformers. Note that $n_{prev} = 8192$ corresponds to over 64 million parameters per layer. Existing DNN certifiers~\cite{gpupoly, deeppoly, zhang2018crown, cloud} usually do not operate on larger sizes than this, but the verification time for larger sizes can be extrapolated from the graph for higher values.

\subsubsection*{Experimental setup}
We implemented the automated verification procedure in Python and used Z3 SMT solver ~\cite{z3} to verify the generated queries. All our experiments were run on a 2.50 GHz 16 core 11th Gen Intel i9-11900H CPU with a main memory of 64 GB.

\subsection{Verifying Modified DNN Certifiers}
\label{sec:casestudydeeppoly}
Implementations of DNN certifiers often include modifications to balance the scalability vs. precision tradeoff. It is crucial to ensure the soundness of the modified certifiers. Verifying them using pen-and-paper proofs can be complicated. In contrast, \oldtool and \cf provide a way to specify and verify these certifiers respectively. For illustration purposes, we focus on the DeepPoly abstract domain and key DNN operations—\affine, \maxpool, and \relu. However, the concepts introduced can be applied to other certifiers and DNN operations. We present two case studies: \deeppolyH and \deeppolyNew, and show the evaluation results in  Table~\ref{table:primitivecertifiersnew}.

\subsubsection{\deeppolyH (Balanced Efficiency and Precision Certifier)}
We use the same abstract shape as the DeepPoly certifier and design transformers that balance precision and efficiency.\\
\textbf{Affine.} The most expensive part of the DeepPoly certifier is the backsubstitution step in the \affine transformer. To improve efficiency, albeit with reduced precision, \deeppolyH employs a custom stopping function within the \traverse construct to stop the backsubstitution at an intermediate layer, specifically, two layers back rather than always proceeding to the input layer. \\
\textbf{ReLU.} In the case of unstable neurons, there are two commonly used lower polyhedral bounds - 0 and \prev. In \deeppolyH, a heuristic determines which polyhedral lower bound to store based on $\prev[l]$ and $\prev[u]$. \\
\textbf{MaxPool.} 
For \maxpool, we use the new abstract transformer designed in ~\cite{constraintflow}, which is more precise than DeepPoly. We compute a list of neurons whose concrete lower bound is greater than or equal to the concrete upper bounds of all other neurons in $\prev$. If this list is non-empty, we set the polyhedral lower and upper bounds to the average of the neurons in this list. Otherwise, we use the same polyhedral bounds used in DeepPoly. The complete code can be found in Appendix K.4.

\subsubsection{\deeppolyNew (Reused Bounds for Enhanced Efficiency)}
In an existing implementation of DeepPoly~\cite{eran}, the certifier stores previously computed polyhedral bounds from earlier layers to reuse them instead of recalculating them for current layer bounds. This approach prioritizes efficiency while accepting a slight trade-off in precision. In \oldtool, this can be easily specified by additionally storing the cached polyhedral bounds as separate members of the abstract shape $L_c, U_c$. For the \affine abstract transformer, the users can first use the new polyhedral bounds. If the results are not sufficiently precise (based on a heuristic), then the computation falls back to the original computation using the \traverse construct. This transformer significantly boosts efficiency by leveraging cached values of previous \affine layer backsubstitutions rather than computing them anew at each layer. The transformers for \relu and \maxpool can be similarly defined for \deeppolyNew. The complete code can be found in Appendix K.4.
\begin{table}
    \centering
    \captionsetup{justification=centering}
    \caption{Query generation time (G), verification time (V) for correct implementation, and bug-finding time for randomly introduced bugs (B) in seconds for new DNN certifiers (\S~\ref{sec:casestudydeeppoly}, \S~\ref{sec:newoperations}).}
    \begin{subtable}[t]{\textwidth}
        \centering
        \captionsetup{justification=centering}
        \caption{New Transformers introduced in \S~\ref{sec:casestudydeeppoly}}
        \label{table:primitivecertifiersnew}
        \resizebox{0.75\textwidth}{!}{
        \begin{tabular}{@{}l | c  c c  | c c c | c c c@{}}
        \toprule
        Certifiers & \multicolumn{3}{c}{\affine} & \multicolumn{3}{c}{\maxpool} & \multicolumn{3}{c}{\relu}\\
        \text{ } & G & V  & B & G & V  & B & G & V  & B \\ 
        \midrule
        \deeppolyH & 0.230 & 1.921 & 0.318 & 0.172 & 0.844 & 0.069 & 0.252 & 1.397 & 0.099 \\ 
        \deeppolyNew & 0.263 & 2.843 & 0.667 & 0.176 & 1.029 & 0.073 & 0.242 & 2.895 & 0.359\\ 
        \bottomrule
        \end{tabular}
        }
        
    \end{subtable}
    
    \begin{subtable}[t]{\textwidth}
        \centering
        \captionsetup{justification=centering}
        \caption{New DNN operations introduced in \S~\ref{sec:newoperations}}
        \label{table:newoperations}
        \resizebox{\textwidth}{!}{
        \begin{tabular}{@{}l | c  c c | c c c | c  c c | c c c | c c c@{}}
        \toprule
        Certifiers & \multicolumn{3}{c}{\leakyrelu} & \multicolumn{3}{c}{\abs} & \multicolumn{3}{c}{\hsigmoid} & \multicolumn{3}{c}{\htanh} & \multicolumn{3}{c}{\hswish}\\
        \text{ } & G & V  & B & G & V  & B & G & V  & B & G & V  & B & G & V  & B\\ 
        \midrule
        DeepPoly/CROWN & 0.299 & 2.454 & 0.543 & 0.199 & 5.252 & 0.069 & 0.319 & 2.238 & 0.147 & 0.304 & 3.016 & 0.354 & 0.277 & 2.963 & 0.383 \\ 
        Vegas(Backward)& 0.216 & 1.264 & 0.145 & 0.078 & 0.237 & 0.102 & 0.206 & 0.900 & 0.076 & 0.166 & 1.154 & 0.095 & 0.186 & 0.812 & 0.065  \\ 
        DeepZ & 0.150 & 1.25 & 0.363 & 0.116 & 0.462 & 0.369 & 0.172 & 1.634 & 0.550 & 0.148  & 2.677 &  0.526 & 0.290  & 3.457 &  0.886 \\ 
        RefineZono & 0.233 & 2.084 & 0.347 &  0.165 & 0.870 & 0.128 & 0.259 & 2.847 & 0.150 & 0.178 & 2.444 & 0.657 & 0.542 & 2.42  & 0.564 \\ 
        IBP & 0.102 & 0.237 & 0.289 & 0.147 & 0.455 & 0.059 & 0.098 & 0.228 & 0.071  & 0.123 & 0.269 & 0.065 & 0.205 & 0.653 & 0.218 \\ 
        Hybrid Zonotope &  0.109 & 0.388 & 0.456 & 0.125 & 0.930 & 0.121 & 0.118 & 0.369 & 0.403 & 0.175 & 0.405 & 0.197 & 0.238 & 2.256 & 0.065 \\ 
        \deeppolyH & 0.230 & 1.921 & 0.318 & 0.172 & 0.844 & 0.069 & 0.252 & 1.397 & 0.099 & 0.229 & 2.433 & 0.083 & 0.198 & 2.070 & 0.462 \\ 
        \deeppolyNew & 0.263 & 2.843 & 0.667 & 0.176 & 1.029 & 0.073 & 0.242 & 2.895 & 0.359 & 0.227 & 4.354 & 0.446 & 0.234  & 3.733 & 0.121 \\ 
        \bottomrule
        \end{tabular}
        }
        
    \end{subtable}
    
    % \vspace{-5mm}
\end{table}

\subsection{Abstract Transformers for New DNN Operations}
\label{sec:newoperations}
As deep learning frameworks continually introduce new activations, the need for designing \textit{sound} abstract transformers becomes increasingly critical. We demonstrate the effectiveness of \oldtool syntax and formal semantics and \cf verification procedure in this context by specifying and verifying abstract transformers for novel DNN operations not currently supported by existing DNN certifiers. These new operations include \leakyrelu, \abs, \hsigmoid, \htanh, and \hswish. Detailed transformers for each operation can be found in Appendix K. Evaluation results across different DNN certifiers are presented in Table~\ref{table:newoperations}, demonstrating that most transformers for these operations can be verified (or disproved) within 1 second. For illustration, 
we show the DeepPoly transformer for \hswish ($\hswish(x) = x \cdot \min(1, \min(0, \frac{x+3}{6}))$). 

\begin{lstlisting}
Func slope(Real x1, Real x2) = ((x1 * (x1 + 3)) - (x2 * (x2 + 3))) / (6 * (x1-x2));
Func intercept(Real x1, Real x2) = x1 * ((x1 + 3) / 6) - (slope(x1, x2) * x1);
Func f1(Real x) = x < 3 ? x * ((x + 3) / 6) : x;
Func f2(Real x) = x * ((x + 3) / 6);
Func f3(Neuron n) = max(f2(n[l]), f2(n[u]));
Transformer DeepPoly{
    HardSwish -> 
    (prev[l] < -3) ? 
        (prev[u] < -3 ? 
            (0, 0, 0, 0) : 
            (prev[u] < 0 ? 
                (-3/8, 0, -3/8, 0) : 
                (-3/8, f1(prev[u]), -3/8, f1(prev[u]) * (prev - prev[l])))) : 
            ((prev[l] < 3) ? ((prev[u] < 3) ? 
                (-3/8, f3(prev), -3/8, prev*slope(prev[u], prev[l]) + intercept(prev[u],prev[l])):
                (-3/8, prev[u], -3/8, prev[u] * ((prev + 3) / (prev[u] + 3)))) :
        (prev[l], prev[u], prev, prev)); 
}
\end{lstlisting}

\subsection{Designing New DNN Certifiers with New Abstract Domains}
\label{sec:newcertifiers}
We show that \cf allows verifying the soundness of new DNN certifiers based on completely new abstract domains and transformers. Specifying them in \oldtool is only possible due to the novel formalism including type system and semantics introduced in this work. 

\subsubsection*{\sympoly DNN Certifier}
Several state-of-the-art DNN certifiers, including DeepPoly, CROWN, etc., approximate the value of each neuron in the DNN by imposing polyhedral constraints over each of them. However, in the case of piecewise-linear activation functions, these certifiers rely on heuristics to choose appropriate polyhedral bounds from more than one possible choice. For instance, in the case of an unstable \relu neuron, there are infinite possibilities for a potential lower polyhedral bound. We argue that in general, the lower polyhedral bound can be of the form $c \cdot \prev$ where c is any real coefficient s.t. $0 \leq c \leq 1$. The two most commonly used lower bounds - $\prev$ and $0$ are only two extreme cases of the general lower bound. Using the \oldtool syntax and semantics, the users can directly specify the general transformer, i.e., $\curr[L] \gets \frac{1+\sym}{2} * \prev$. \cf can be used to prove the soundness of this lower bound. In this way, \cf allows a user to verify the soundness of a space of abstract transformers, which can be leveraged to automatically synthesize the optimal transformer using a cost function encoding the precision of the transformer based on the DNN certification problem. Further, since each invocation of the $\sym$ construct outputs a new symbolic value, different values of the symbolic coefficient can be chosen for different neurons in the DNN. A slightly different version is explored in the DNN certifier $\alpha-$CROWN~\cite{alphacrown}, where $\alpha$ is a concrete but learnable coefficient, learned using gradient descent.   

Based on this idea, the DNN certifier \sympoly can be found in Fig.~\ref{fig:alpha-crown}. The abstract domain consists of two concrete bounds $\texttt{\footnotesize{l, u}}$ and two polyhedral bounds with symbolic coefficients $\texttt{\footnotesize{symL, symU}}$. The abstract transformer for \relu is specified in 3 cases - (i) $\curr[u] < 0$, (ii) $\curr[l]>0$, and (iii) $\curr[l] \leq 0 \leq \curr[u]$. In the more challenging third case, the lower polyhedral bound is set to $\frac{1+\sym}{2} * \prev$. 
The abstract transformers can be similarly designed for activations such as $\htanh$, $\hsigmoid$, $\hswish$, $\abs$, etc. These can be found in Appendix K.4. Notably, we can verify the soundness of these transformers in runtimes similar to the DeepPoly certifier.
\begin{figure}
     \centering
     \begin{subfigure}[b]{\textwidth}
         \begin{lstlisting}
Def shape as (Real l, Real u, PolyExp symL, PolyExp symU) {(curr[l]<=curr) and (curr[u]>=curr) and (curr[symL]<=curr) and (curr[symU]>=curr)};
Transformer SymPoly{
  Relu -> prev[l] > 0 ? (prev[l], prev[u], prev, prev) : 
            (prev[u] < 0 ? (0, 0, 0, 0) : 
                (0, prev[u], ((1+sym)/2) * prev, ((prev[u] / (prev[u] - prev[l])) * prev) - ((prev[u] * prev[l]) / (prev[u] - prev[l]))));
}\end{lstlisting}
    \caption{SymPoly}
    \label{fig:alpha-crown}
     \end{subfigure}
     \hfill
     \begin{subfigure}[b]{\textwidth}
         \begin{lstlisting}
Def shape as (Real l, Real u, PolyExp L, PolyExp U, SymExp Z) {curr[l]<=curr and curr[u]>=curr and curr[L]<=curr and curr[U]>=curr and curr <> curr[Z]};
Func min_symexp(Sym e, Real c) = c > 0 ? -c : c;
Func lower_sym(Neuron List prev, Neuron curr) = (prev[Z] * curr[w] + curr[b]).map(min_symexp);
Func lower_poly(Neuron List prev, Neuron curr) = backsubs_lower(prev * curr[w] + curr[b]);
Transformer PolyZ{
    Affine ->  (max(lower_sym(prev, curr), lower_poly(prev, curr)), 
                min(upper_sym(prev, curr), upper_poly(prev, curr)),
                prev * curr[w] + curr[b], prev * curr[w] + curr[b], prev[Z] * curr[w] + curr[b]);
}\end{lstlisting}
        \caption{PolyZ}
        \label{fig:PolyZ}
     \end{subfigure}
        \caption{Code Sketches for new DNN certifiers. The complete codes can be found in Appendix K.4}
        % \vspace{-5mm}
\end{figure}

\subsubsection*{PolyZ DNN Certifier}
\label{sec:poly}
We show another new abstract domain - PolyZ - a reduced product of the popular DeepZ and DeepPoly domains using polyhedral and symbolic constraints. The abstract shape consists of 5 members - two concrete interval bounds, $l$ and $u$ of the type $\float$, two polyhedral bounds $L$ and $U$ of the type $\polyexp$, and a symbolic expression $Z$ of the type $\symexp$. The shape constraints state that the neuron's value satisfies the bounds $l$, $u$, $L$, and $U$ and $\curr \ \embed \ Z$. We also define the abstract transformers for this new domain. The \affine transformer is shown in Fig.~\ref{fig:PolyZ} and the complete specification is in Appendix K.4. PolyZ is more precise than both DeepPoly and DeepZ, and we can verify its soundness using the \cf verification procedure.

\subsection{State-of-the-Art DNN Certifiers}
\label{sec:existing}
The existing DNN certifiers evaluated in this section include IBP~\cite{interval} (Interval Bound Propagation), DeepPoly~\cite{deeppoly} (or CROWN~\cite{zhang2018crown}), DeepZ~\cite{deepz}, RefineZono~\cite{refinezono}, Vegas~\cite{forwardbackward},  and Hybrid Zonotope~\cite{hybzono}. The abstract shapes of DeepPoly, CROWN, and Vegas include polyhedral expressions represented by the $\polyexp$ datatype and use the \traverse construct to compute the concrete bounds. DeepZ, RefineZono, and Hybrid Zonotope use symbolic expressions represented by $\symexp$ in their abstract shapes. RefineZono uses $\ct$ to encode constraints over the possible values of the neurons. RefineZono and Vegas use the \solver construct to compute the concrete bounds. The users can define functions using the \func construct, promoting code reusability and facilitating a modular design. The \oldtool codes for these DNN certifiers are presented in Appendix K. 
    \begin{figure}
        \centering
        \begin{lstlisting}
Func lower(Neuron n1, Neuron n2) = min([n1[l]*n2[l], n1[l]*n2[u], n1[u]*n2[l], n1[u]*n2[u]]);
Func upper(Neuron n1, Neuron n2) = max([n1[l]*n2[l], n1[l]*n2[u], n1[u]*n2[l], n1[u]*n2[u]]);
Transformer DeepPoly{
    Max -> (prev0[l] >= prev1[u]) ? (prev0[l], prev0[u], prev0, prev0) : ((prev1[l] >= prev0[u]) ? 
                    (prev1[l], prev1[u], prev1, prev1) : 
                    (max(prev0[l], prev1[l]), max(prev0[u], prev1[u]), max(prev0[l], prev1[l]), 
                        max(prev0[u], prev1[u])));
    Mult -> (lower(prev0, prev1), upper(prev0, prev1), lower(prev0, prev1), upper(prev0, prev1));
}\end{lstlisting}
    \caption{Max and Mult transformers for DeepPoly Certifier}
    \label{fig:attention}
    \end{figure}
    \begin{figure}
        \begin{minipage}{0.6\textwidth}
            % \raggedright
    \begin{lstlisting}
Def shape as (Real l, Real u, PolyExp L, PolyExp U) {...};

Transformer DeepPoly_forward{ReLU -> ... ;}
Transformer DeepPoly_backward{rev_ReLU -> ... ;}

Flow(forward, ..., ... , DeepPoly_forward);
Flow(backward, ..., ..., DeepPoly_backward);    \end{lstlisting}
        \caption{Code Sketch for Vegas Certifier}
        \label{fig:fbincomplete}
        \end{minipage}
        \hspace{0.5cm}
        \begin{minipage}{0.35\textwidth}
            \includegraphics[width=\linewidth]{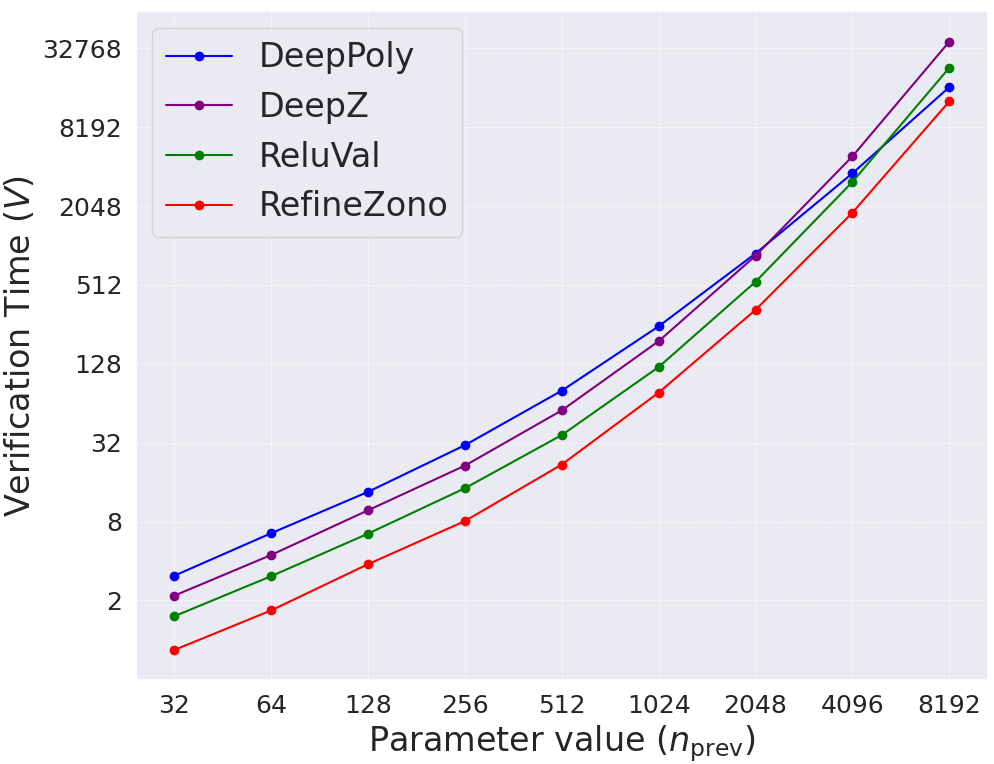}
            \captionsetup{justification=centering}
            \captionof{figure}{Verification time (in s) for \affine transformers.}
            \label{fig:affine}
        \end{minipage}
    \end{figure}
    \begin{table}
    \centering
    \captionsetup{justification=centering}
    \caption{Query generation time (G), verification time (V) for correct implementation, and bug-finding time for randomly introduced bugs (B) in seconds for transformers of existing DNN certifiers (\S~\ref{sec:existing}).}
    \begin{subtable}[t]{\textwidth}
        \centering
        \captionsetup{justification=centering}
        \caption{Primitive operations}
        \label{table:primitivecertifiers}
        \resizebox{\textwidth}{!}{
        \begin{tabular}{@{}l | c  c c | c c c | c  c c | c c c | c c c@{}}
        \toprule
        Certifiers & \multicolumn{3}{c}{\relu} & \multicolumn{3}{c}{\cfkeywords{Max}} & \multicolumn{3}{c}{\cfkeywords{Min}} & \multicolumn{3}{c}{\cfkeywords{Add}} & \multicolumn{3}{c}{\cfkeywords{Mult}}\\
        \text{ } & G & V  & B & G & V  & B & G & V  & B & G & V  & B & G & V  & B\\ 
        \midrule
        DeepPoly/CROWN & 0.196 & 1.526 & 0.066 & 0.095 & 2.618 & 0.074 & 0.128 & 2.829 & 0.601 & 0.0812 & 0.136 & 0.205 & 0.209 & 2.104 & 0.129 \\ 
        Vegas(Backward)& 0.142 & 0.584 & 0.221 & 0.047 & 0.139 & 0.084 & 0.052 & 0.115 & 0.087  & 0.056 & 0.097 & 0.153 & 0.388 & 0.486 & 0.110   \\ 
        DeepZ & 0.0832 & 0.534 & 0.336 & 0.115 & 0.703 & 0.145 & 0.119 & 0.691 & 0.215 & 0.0815  & 0.091 & 0.256 &  0.234 & 0.498 & 0.427 \\ 
        RefineZono & 0.158 & 0.980 & 0.071 &  0.199 & 1.235 & 0.262 & 0.213 & 1.263 & 0.331 & 0.089 & 0.117 & 0.242 & 0.404  & 17.197 & 0.468 \\ 
        IBP & 0.112 & 0.493 & 0.364 & 0.132 & 0.508 & 0.081 & 0.136 & 0.545 & 0.333 & 0.0716 & 0.060 & 0.158 & 0.217 & 1.160 & 0.259 \\ 
        Hybrid Zonotope & 0.260 & 1.003 & 0.341 & 0.132 & 0.775 & 0.292 & 0.132 & 0.724 & 0.626 & 0.086 & 0.286 & 0.204 & 0.209 & 0.520 & 1.397 \\ 
        \bottomrule
        \end{tabular}
        }
        
    \end{subtable}
    
    \begin{subtable}[t]{\textwidth}
        \centering
        \captionsetup{justification=centering}
        \caption{Composite operations}
        \label{table:complexcertifiers}
        \resizebox{\textwidth}{!}{
        \begin{tabular}{@{}l | c  c c | c c c | c  c c | c c c@{}}
        \toprule
        \multirow{2}{*}{Certifiers} & \multicolumn{3}{c}{\affine} & \multicolumn{3}{c}{\maxpool} & \multicolumn{3}{c}{\minpool} & \multicolumn{3}{c}{\avgpool} \\ 
        & G & V  & B & G & V  & B & G & V  & B & G & V  & B   \\ 
        \midrule
        DeepPoly / CROWN &  5.496 & 889.607 & 9.825 & 14.744 & 196.651 & 1396.132 & 13.917 & 194.871 & 1419.119 & 0.137 & 0.363 &  0.131\\ 
        Vegas (Backward) & 2.436 & 25.447 & 25.898 & - & - & - & - & - & - & - & - & -\\ 
        DeepZ & 4.569 & 854.548 & 833.314  & 54.217 & 364.859 &  1780.938 &  52.140 & 292.806 & 1366.977 & 0.0818 & 0.265 &  0.763 \\ 
        RefineZono & 5.436 & 329.994 & 152.825  & 54.788 & 376.177 & 1451.729 &  56.427 & 308.570 & 1799.091 & 0.095 & 0.306 & 0.301\\ 
        IBP & 2.997 & 540.865 & 183.707 & 0.089 & 4.077 & 0.253 & 0.090 & 4.114 & 4.605 & 0.067 & 0.0117 & 0.921 \\ 
        Hybrid Zonotope & - & - & - & 1.816 & 10.610 & 2.892 & 1.503 & 10.598 & 3.395 & 0.318 & 11.499 & 2.568 \\ 
        \bottomrule
        \end{tabular}
        }
        
    \end{subtable}
    
    % \vspace{-5mm}
\end{table}
% \end{minipage}

Notably, with the formal syntax and the operational semantics, \oldtool can handle various flow directions effectively. For instance, the Vegas certifier~\cite{forwardbackward}, which employs both forward and backward flows, is easily expressed in \oldtool. We provide the code for the Vegas certifier in Fig.~\ref{fig:fbincomplete}. The abstract shape and the transformer for the forward direction are the same as the DeepPoly analysis, while the transformer for the backward analysis replaces operations like $\relu$ with $\revrelu$. We can also verify its soundness using \cf (Tables~\ref{table:primitivecertifiers},~\ref{table:complexcertifiers}).

For primitive operations like \cfkeywords{Max}, \cfkeywords{Mult}, etc., there are two implicit inputs to the transformer definitions, namely the input neurons - \cfkeywords{prev0} and \cfkeywords{prev1}. DeepPoly transformers for \cfkeywords{Max} and \cfkeywords{Mult} are shown in Fig.~\ref{fig:attention}. 
The primitive operations - \relu, \cfkeywords{Max}, \cfkeywords{Min}, \cfkeywords{Add}, \cfkeywords{Mult} shown in Table~\ref{table:primitivecertifiers} can be verified in fractions of a second. In Table~\ref{table:complexcertifiers}, we show the evaluation results for the composite operations. For \maxpool and \minpool, the DeepZ and RefineZono transformers are harder to verify because their queries are doubly quantified due to the  \embed operator in their specifications. IBP is the easiest to verify because the limited abstract shape does not allow it to be as precise as other transformers for \maxpool and \minpool. Also, for Vegas, the backward transformers for \maxpool, \minpool, and \avgpool are not available in existing works. Similarly, for the Hybrid Zonotope, the transformer for \affine is defined in terms of transformers for primitive operations. So, we skip these in Table~\ref{table:complexcertifiers}. For \affine, DeepPoly is the hardest because it uses the \traverse construct, which requires additional queries to check the validity of the invariant. Vegas takes the least time because of a relatively simpler verification query. Note that the verification times are not correlated to the runtimes of certifiers on concrete DNNs. In Appendix K.3, we provide the \oldtool code for several of these certifiers. The complexity inherent in these certifiers and their implementations suggests that verifying them solely through pen-and-paper proofs or automated theorem provers is impractical. 

\section{Related Work}
\label{sec:related}
\paragraph{DNN Certification. }
The recent advancements in DNN certification techniques~\cite{albarghouthi} have led to the organization of competitions to showcase DNN certification capabilities~\cite{vnncomp}, the creation of benchmark datasets~\cite{vnnlib}, the introduction of a DSL for specifying certification properties~\cite{reliableneuralspecification, dnnv}, and the development of a library for DNN certifiers~\cite{li2023sok, socrates}. However, these platforms lack formal soundness guarantees and do not offer a systematic approach to designing new certifiers.

\paragraph{DSL for Abstract Interpretation. }
Although~\cite{constraintflow} proposed a preliminary design for \oldtool using a few examples, the absence of formal semantics hinders its use for designing and verifying new DNN certifiers. We equip \oldtool with a BNF grammar, type-system, operational semantics, and symbolic semantics that enable users to specify existing DNN certifiers, design new ones, and verify their soundness using \cf. 

Similarly, ~\cite{tsl} designs TSL---a DSL for abstract interpreters for conventional programs. TSL allows users to specify the concrete semantics and the abstract domain and automatically produces an abstract interpreter based on these specifications. However, it does not provide any specialized datatypes needed to specify DNN certifiers easily. It also does not guarantee the soundness of the abstract interpreter. In contrast, \cf can verify the soundness of the certifier specification.   

\paragraph{Symbolic Execution. }
Similar to \cf DNN expansion step, ~\cite{lazy, pathfinder} employ lazy initialization for symbolic execution of complex data structures like lists, trees, etc. The object fields are initialized with symbolic values only when accessed by the program. Unlike these works, which possess prior knowledge of the exact structure of the objects, DNN certifiers deal with arbitrary DAGs representing DNNs. The graph nodes (neurons) are intricate data structures with unknown graph topology. We believe that we are the first to create a symbolic DNN with sufficient generality to represent arbitrary graph topologies to verify the soundness of DNN certifiers.

\paragraph{Correctness of Symbolic Execution. }
Some existing works prove the correctness of the symbolic execution w.r.t. the language semantics~\cite{10.1145/3547628}. However, these methods do not establish correctness in cases where symbolic execution also represents symbolic variables used in concrete executions. On the other hand, we provide elaborate proofs establishing the correctness of \cf where we encode the symbolic variables within the program as SMT symbolic variables. 
\section{Discussion and Future Work}
\label{sec:discussion}
We develop \cf, a novel bounded automated verification procedure to automatically verify the overapproximation-based soundness of abstract interpretation-based DNN certifiers. We also develop a formal syntax, type-system, operational semantics, and symbolic semantics for \oldtool. For the first time, we can verify the soundness of DNN certifiers for arbitrary (but bounded) DAG topologies. Given the growing concerns about AI safety, we believe that \cf, coupled with \oldtool, allows the development of new DNN certifiers without proving their soundness manually. This work allows the following future directions: 
\paragraph{Multi-neuron specifications} - \cf can be extended to verify multi-neuron abstract shapes \cite{prima} by allowing their specification in \oldtool.
\paragraph{Sequence of Operations} - \cf can also be extended to automatically verify a sequence of DNN operations, like $\affine+\relu$. To do so, while generating the final query, we would execute the concrete semantics of the composition of \affine and \relu. 
\paragraph{Automating Abstract Interpretation} - \cf and \oldtool facilitate the automated generation of abstract transformers~\cite{automatingabs, synthabs, posthat} by offering all the basic components - (i) a DSL for defining the search space of candidate transformers, (ii) the semantics of the DSL, and (iii) a procedure for verifying the soundness of each candidate. This can be explored in future research.
\paragraph{Verification Property} - Currently, the verification property is the over-approximation-based soundness of a DNN certifier. Nevertheless, given that all the necessary formalism for verification has been established, the property can be extended to encompass more intricate aspects, such as encoding precision of a DNN certifier w.r.t. a baseline.
\clearpage
\section{Data-Availability Statement}
The artifact\cite{artifact} consists of \cf implementation and the \oldtool specifications of the DNN certifiers presented in Section~\ref{sec:evaluation} and Appendix K. The code, accompanied by the instructions to run it, can be found  \href{https://zenodo.org/records/14597703}{here}. 
\section*{Acknowledgments}
We thank the anonymous reviewers for their insightful comments. This work was supported in part by NSF Grants No. CCF-2238079, CCF-2316233, CNS-2148583.
\bibliographystyle{ACM-Reference-Format}
\bibliography{main}

\appendix
\clearpage
\section{\oldtool DSL}
\label{appendix:grammar}
\begin{grammar}
<Types> $\types$ ::= $\intt$ | $\float$ | $\bool$ | $\typeneuron$ | $\typenoise$ | $\polyexp$ | $\symexp$ | $\ct$ | $\overline{\types}$

<Binary-op> $\oplus$ ::= $+$ | $-$ | $*$ | $/$ | $\andx$ | $\orx$ | $\geq$ | $\leq$ | $==$ | $\embed$

<Unary-op> $\sim$ ::= $\notx$ 

<Neuron-metadata> $m$ ::= $\weight$ | $\bias$ | $\layer$ | $\equations$

<List-operations> $F_l$ ::= $\maxx$ | $\minn$

<Solver-op> $\lpop$ ::= $\maximize$ | $\minimize$

<Function-call> $f_c$ ::= $\var$

<Transformer-call> $\theta_c$ ::= $\var$

<Direction> $\dir$ ::= $\backward$ | $\forward$

<Expression> $\expr$ ::= $\constant$ | $\var$ | $\sym$ 
\alt $\expr_1 \ \oplus \ \expr_2$ | $\sim \expr$
\alt $\expr_1 \ ? \ \expr_2 \ : \ \expr_3$
\alt $\expr [ m ]$ | $\expr [ x ]$ 
\alt $\var.\traverse(\dir, f_{c_1}, f_{c_2}, f_{c_3})\{\expr\}$ | $\solver(\lpop, \expr_1, \expr_2)$
\alt $F_l(\expr)$ | $F_l(\expr_1, \expr_2)$ | $\argmax(\expr, f_c)$ 
\alt $\summ(\expr)$ | $\avg(\expr)$ | $\len(\expr)$ | $\expr_1.\dott(\expr_2)$ | $\expr_1.\concat(\expr_2)$
\alt $\expr.\map(f_c)$ | $\expr.\mapl(f_c)$
\alt $f_c(\expr_1, \expr_2, \cdots)$

% <Tuple> $u$ ::= $(\expr_1, \expr_2, \cdots)$

<Shape-decl> $d$ ::= $\defshape \ (\types_1 \ x_1, \types_2 \ x_2, \cdots) \{\expr\}$  

<Function-decl> $f_d$ ::= $\func \ \var(\types _1 \ \var_1, \types_2 \ \var_2, \cdots)$

<Function-definition> $f$ ::= $f_d = e$

<DNN-operation> $\eta$ ::= $\affine$ | $\relu$ | $\maxpool$ | $\dotprod$ | $\cfkeywords{Sigmoid}$ | $\cfkeywords{Tanh}$ | $\cdots$

<Transformer-decl> $\theta_d$ ::= $\transformer \ \var$

<Transformer-ret> $\theta_r$ ::= $(\expr_1, \expr_2, \cdots)$ | $(\expr \ ? \ \theta_{r_1} : \theta_{r_2})$

<Transformer> $\theta$ ::= $\theta_d  \{\eta_1 \rightarrow \theta_{r_1}; \eta_2 \rightarrow \theta_{r_2}; \cdots\}$

<Statement> $s$ ::= $\flow(\dir, f_{c_1}, f_{c_2}, \theta_c)$ | $f$ | $\theta$ | $s_1 \ ; \ s_2$

<Program> $\Pi$ ::= $d \ ; \ s$

\end{grammar}

Other than standard types such as $\intt, \float, \bool$, we provide datatypes that are specific to DNN certification including $\typeneuron, \typenoise, \polyexp, \symexp, \ct$. Further, the expression of type $\overline{\types}$ are lists with all the elements of type $\types$.  

We have the standard unary, binary, and ternary operators for expressions. The list operations include computing the sum of a list, finding the maximum or minimum element of a list, the average of a list, the length of a list, the dot product of two lists, and concatenating two lists. Further, the construct \argmax can be used to define the maximum of a list of neurons which takes in a user-defined function that compares two neurons. Other operations are \map, \mapl, \traverse, \solver. function call. These are explained in the main text of the paper. 

The metadata of a neuron is configurable. For now, we provide weight, bias, layer number, and equations as the metadata. This can be extended to include other things that are a part of the DNN architecture. For example, we can add boolean metadata such as the ones indicating whether a neuron is a part of the input layer or a part of the output layer. The metadata $\equations$ is used to refer to a list of equations relating the neurons from the current layer to the neurons of the next layer. When traversing the DNN in the backward direction, Each neuron in \prev corresponds to one equation in $\curr[\equations]$. The metadata weight and bias are referred to as w and b in the main text.

% For backward abstract transformer uses the abstract shapes of neurons in the next layer to refine the neuron's abstract shape members. Here, instead of the operations being $\relu$, $\affine$, etc., they are $\revrelu$, $\revaffine$, etc. In addition to \curr and \prev, the abstract transformer has another formal argument \currlist, which is a list of neurons in the same layer as \curr. Also, in this code, the metadata $\equations$ is used to refer to a list of equations relating the neurons from the current layer to the neurons of the next layer. Each neuron in \prev corresponds to one equation in $\curr[\equations]$. For the $\revaffine$ operation, the code uses the \solver construct to find the minimum and maximum value of the neuron given the bounds of the neurons in the next layer.
\clearpage
\section{Type checking}
\label{appendix:typechecking}
\subsection{Type checking Rules for Expressions}
The general form for type-checking for expressions is $\g \vdash \expr \ : \ \types$. $\Gamma$ is a static record that stores the types of identifiers defined in the code. These identifiers include function names, transformer names, and the formal arguments of functions and transformers. $\tau_s$ stores the types of the shape members defined by the user. The binary operators like $+, and, \leq$ are overloaded to be also used for the new datatypes defined in \oldtool. As in the rule \textsc{T-comparison-1}, the comparison of integers or real numbers outputs a boolean, however, in rule \textsc{T-comparison-2}, the comparison of two $\polyexp$ expressions outputs a constraint of the type $\ct$. 
\label{appendix:typecheckingexpressions}
\begin{figure}[H]
\raggedright
%     \begin{subfigure}{\textwidth}
         \fbox{$\g \vdash \expr \ : \ \types$}
         % \resizebox{\textwidth}{!}{
  \[
  \begin{array}{c}
    \inferrule*[lab={\textsc{T-var}}]
    {
    }
    {
    \g \vdash x : \Gamma(x)
    }
    \hspace{1cm}
    \inferrule*[lab={\textsc{T-noise}}]
    {
    }
    {
    \g \vdash \noise : \typenoise
    } 
    \hspace{1cm}
    \inferrule*[lab={\textsc{T-unary}}]
    {
    \g \vdash \expr : \bool
    }
    {
    \g \vdash \sim\expr : \bool
    }  
    \\\\
    % \hspace{1cm}
    \inferrule*[lab={\textsc{T-comparison-1}}]
    {
    \g \vdash \expr_1 : \types_1 \quad 
    \g \vdash \expr_2 : \types_2 \\\\ 
    \types_1, \types_2 \in \{\intt, \float \} \quad 
    \oplus \in \{\geq, \leq, == \}
    }
    {
    \g \vdash \expr_1 \oplus \expr_2 : \bool  
    }
    \hspace{1cm}
    \inferrule*[lab={\textsc{T-comparison-2}}]
    {
    \g \vdash \expr_1 : \types_1 \quad 
    \g \vdash \expr_2 : \types_2 \\\\ 
    \types_1 \sqcup \types_2 = \types \quad 
    \float \sqsubset \types \sqsubset \ct \\\\
    % \types_1 \sqcup \types_2 \in \{\typeneuron, \typenoise, \polyexp, \symexp\} \\\\ 
    \oplus \in \{\geq, \leq, == \}
    }
    {
    \g \vdash \expr_1 \oplus \expr_2 : \ct   
    }
    \\\\
    % \hspace{1cm}
    \inferrule*[lab={\textsc{T-comparison-3}}]
    {
    \g \vdash \expr_1 : \types_1 \quad 
    \g \vdash \expr_2 : \types_2 \\\\ 
    \types_1 \sqsubseteq \polyexp \\\\ 
    \types_2 \sqsubseteq \symexp
    }
    {
    \g \vdash \expr_1 \ \embed \ \expr_2 : \ct   
    }
    % \\\\
    \hspace{1cm}
    \inferrule*[lab={\textsc{T-binary-bool}}]
    {
    \g \vdash \expr_1 : \types_1 \\\\ 
    \g \vdash \expr_2 : \types_2 \\\\ 
    \oplus \in \{\andx, \orx \} \quad 
    \types_1, \types_2 \in \{\bool, \ct\}
    }
    {
    \g \vdash \expr_1 \oplus \expr_2 : \types_1 \sqcup \types_2  
    }
    \\\\
    % \hspace{1cm}
    \inferrule*[lab={\textsc{T-binary-arith-1}}]
    {
    \g \vdash \expr_1 : \types_1 \quad 
    \g \vdash \expr_2 : \types_2 \\\\ 
    \types_1 \sqcup \types_2 \in \{\polyexp, \symexp\} \\\\ 
    \oplus \in \{+, -\}
    }
    {
    \g \vdash \expr_1 \oplus \expr_2 : \types_1 \sqcup \types_2  
    }
    \hspace{1cm}
    % \\\\
    \inferrule*[lab={\textsc{T-binary-arith-2}}]
    {
    \g \vdash \expr_1 : \types_1 \quad 
    \g \vdash \expr_2 : \types_2 \\\\ 
    \types_2 \sqsubseteq \types_1 \quad 
    \types_2 \sqsubseteq \float \\\\ 
    \oplus \in \{*, /\}
    }
    {
    \g \vdash \expr_1 \oplus \expr_2 : \types_1
    }
    % \hspace{1cm}
    \\\\
    \inferrule*[lab={\textsc{T-binary-arith-3}}]
    {
    \g \vdash \expr_1 : \types_1 \quad 
    \g \vdash \expr_2 : \types_2 \\\\ 
    \types_1 \sqsubseteq \typenoise \quad 
    \types_2 \in \{\typeneuron, \polyexp\} 
    % \oplus \in \{*, /\}
    }
    {
    \g \vdash \expr_1 * \expr_2 : \polyexp
    }
    \hspace{1cm}
    % \\\\
    \inferrule*[lab={\textsc{T-binary-arith-4}}]
    {
    \g \vdash \expr_1 : \types_1 \quad 
    \g \vdash \expr_2 : \types_2 \\\\ 
    \types_1 \in \{\typeneuron, \polyexp\} \quad 
    \types_2 \sqsubseteq \typenoise \\\\ 
    \oplus \in \{*, /\}
    }
    {
    \g \vdash \expr_1 \oplus \expr_2 : \polyexp
    }
    \\\\
    % \hspace{1cm}
    \inferrule*[lab={\textsc{T-binary-mult}}]
    {
    \g \vdash \expr_1 : \types_1 \quad 
    \g \vdash \expr_2 : \types_2 \\\\ 
    \types_1 \sqcup \types_2 \in \{\polyexp, \symexp\} \\\\
    \types_1 \sqsubset \types_2 \ \orx \ \types_2 \sqsubset \types_1
    % \types_1 \in \{\intt, \float\},  \types_2 \in \{\typeneuron, \polyexp, \symexp\} \text{ or }\\\\
    % \types_2 \in \{\intt, \float\},  \types_1 \in \{\typeneuron, \polyexp, \symexp\}
    % \\\\
    % \types_1, \types_2 \in \{\intt, \float, \polyexp, \symexp \} \\\\ 
    % \types_1 < \types_2 \text{ or }\types_2 < \types_1 \quad
    % \joinable(L_1, L_2)
    }
    {
    \g \vdash \expr_1 * \expr_2 : \types_1 \sqcup \types_2
    }
    % \inferrule*[lab={\footnotesize{T-BINARY-BOOL}}]
    % {
    % \g \vdash \expr_1 : \bool, L_1 \\\\ 
    % \g \vdash \expr_2 : \bool, L_2 \\\\ 
    % \oplus \in \{\andx, \orx \} \quad
    % \joinable(L_1, L_2)
    % }
    % {
    % \g \vdash \expr_1 \oplus \expr_2 : \bool, L_1 \join L_2  
    % }
    \hspace{1cm}
    % \\\\
    \inferrule*[lab={\textsc{T-ternary}}]
    {
    \g \vdash \expr_1 : \bool \\\\ 
    \g \vdash \expr_2 : \types_1  \\\\
    \g \vdash \expr_3 : \types_2 
    }
    { 
    \g \vdash (\expr_1 \ ? \  \expr_2 : \expr_3) : \types_1 \sqcup \types_2  
    }
    % \\\\
    % \hspace{1cm}
    % \inferrule*[lab={\footnotesize{T-METADATA-W}}]
    % {
    % \g \vdash \expr : \typeneuron, L \\\\ 
    % \affine \preceq L
    % }
    % {
    % \g \vdash \expr[\weight] : \overline{\float}, \affine
    % }
    \hspace{1cm}
    \inferrule*[lab={\textsc{T-metadata-b}}]
    {
    \g \vdash \expr : \typeneuron 
    % \affine \preceq L
    }
    {
    \g \vdash \expr[\bias] : \float
    }
    % \\\\
    % 
    % \\\\
    % \inferrule*[lab={\footnotesize{T-SUM-1}}]
    % {
    % \g \vdash \expr : \overline{\types}, L \\\\
    % \types \in \{\intt, \float\}
    % }
    % {
    % \g \vdash \add(\expr) : \types, L
    % }
    % \hspace{1cm}
    % % \\\\
    % \inferrule*[lab={\footnotesize{T-SUM-2}}]
    % {
    % \g \vdash \expr : \overline{\typeneuron}, L
    % }
    % {
    % \g \vdash \add(\expr) : \polyexp, L
    % }
    % \hspace{1cm}
    % \inferrule*[lab={\footnotesize{T-DOT}}]
    % {
    % \g \vdash \expr_1 : \overline{\typeneuron}, L_1 \\\\ 
    % \g \vdash \expr_2 : \tau, L_2 \quad 
    % \tau \leq \overline{\float}\\\\
    % \joinable(L_1, L_2)
    % }
    % {
    % \g \vdash \expr_1.\dotx(\expr_2) : \overline{\typeneuron}, L_1 \join L_2
    % }
    % \\\\
    \end{array}
    \]
    % }
\end{figure}
  %    \end{subfigure}
  %    \ContinuedFloat
  %    \begin{subfigure}{\textwidth}
\begin{figure}[H]
\raggedright
%     \resizebox{\textwidth}{!}{
    \[
  \begin{array}{c}
    % \hspace{1cm}
    % \inferrule*[lab={\footnotesize{T-METADATA-L}}]
    % {
    % \g \vdash \expr : \typeneuron, L
    % }
    % {
    % \g \vdash \expr[\layer] : \intt, L
    % }
    % \hspace{1cm}
    \inferrule*[lab={\textsc{T-shape}}]
    {
    \g \vdash \expr : \typeneuron \\\\ 
    \tau_s(x) = \types 
    }
    {
    \g \vdash \expr[x] : \types 
    }
    \hspace{1cm}
    % \inferrule*[lab={\footnotesize{T-METADATA-W2}}]
    % {
    % \g \vdash \expr : \overline{\typeneuron}, L \\\\ 
    % \affine \preceq L
    % }
    % {
    % \g \vdash \expr[\weight] : \overline{\overline{\float}}, \affine
    % }
    % \\\\
    \inferrule*[lab={\textsc{T-metadata-b2}}]
    {
    \g \vdash \expr : \overline{\typeneuron}
    }
    {
    \g \vdash \expr[\bias] : \overline{\float}
    }
    % \hspace{1cm}
    % \inferrule*[lab={\footnotesize{T-METADATA-L2}}]
    % {
    % \g \vdash \expr : \overline{\typeneuron}, L
    % }
    % {
    % \g \vdash \expr[\layer] : \overline{\intt}, L
    % }
    \hspace{1cm}
    \inferrule*[lab={\textsc{T-metadata-shape2}}]
    {
    \g \vdash \expr : \overline{\typeneuron} \\\\ 
    \tau_s(x) = \types 
    }
    {
    \g \vdash \expr[x] : \overline{\types}
    }
    \\\\
    \inferrule*[lab={\textsc{T-min-max}}]
    {
    F_l \in \{\minn, \maxx\} \\\\
    \g \vdash \expr : \overline{\types} \quad 
    \types \sqsubseteq \float 
    }
    {
    \g \vdash F_l(\expr) : \types
    }
    \hspace{1cm}
    \inferrule*[lab={\textsc{T-compare}}]
    {
    % F_l \in \{\argmin, \argmax\} \\\\
    \g \vdash \expr : \overline{\types} \\\\ 
    \Gamma(f_c) = \types \times \types \rightarrow \bool
    }
    {
    \g \vdash \argmax(\expr, f_c) : \overline{\types}
    }
    \hspace{1cm}
    % \\\\
    \inferrule*[lab={\textsc{T-function-call}}]
    {
    \Gamma(f_c) = (\Pi_i^n \types_i)\rightarrow \types \\\\
    \forall i, \g \vdash \expr_i : \types_i
    }
    {
    \g \vdash f_c(\expr_1, \cdots, \expr_n) : \types
    }
    \\\\
    \inferrule*[lab={\textsc{T-map-poly}}]
    {
    \g \vdash \expr : \polyexp \\\\
    \g \vdash f_c : \typeneuron \times \float \rightarrow \types \\\\
    \types \sqsubseteq \polyexp 
    % \\\\
    % \types = \typeneuron \implies \types' = \polyexp \quad \types \not= \typeneuron \implies \types' = \types
    }
    {
    \g \vdash \expr\cdot\map(f_c) : \polyexp 
    }
    \hspace{1cm}
    \inferrule*[lab={\textsc{T-map-sym}}]
    {
    \g \vdash \expr : \symexp \\\\
    \g \vdash f_c : \typenoise \times \float \rightarrow \types \\\\
    \types \sqsubseteq \symexp 
    % \\\\
    % \types = \typeneuron \implies \types' = \polyexp \quad \types \not= \typeneuron \implies \types' = \types
    }
    {
    \g \vdash \expr\cdot\map(f_c) : \polyexp 
    }
    \\\\
    \inferrule*[lab={\textsc{T-map-list}}]
    {
    \g \vdash \expr : \overline{\types} \\\\
    \g \vdash f_c : \types \rightarrow \types' \\\\
    \types \sqsubseteq \polyexp 
    \\\\
    \types \in \{\typeneuron, \polyexp\} \implies \types' = \polyexp \\\\ 
    \types \in \{\typenoise, \symexp\} \implies \types' = \symexp \\\\ 
    \text{otherwise }\types' = \types
    }
    {
    \g \vdash \expr\cdot\map(f_c) : \types' 
    }
    \hspace{1cm}
    \inferrule*[lab={\textsc{T-solver}}]
    {
    \g \vdash \expr_1 : \polyexp \\\\ 
    \g \vdash \expr_2 : \ct
    }
    {
    \g \vdash \solver(\lpop, \expr_1, \expr_2) : \float
    }
    \\\\
    \inferrule*[lab={\textsc{T-traverse}}]
    {
    \g \vdash \expr_1 : \polyexp \quad 
    \g \vdash \expr_2 : \ct \\\\ 
    \g \vdash f_{c_1} : \typeneuron \rightarrow \types' \quad 
    \g \vdash f_{c_2} : \typeneuron \rightarrow \bool \\\\ 
    \g \vdash f_{c_3} : \typeneuron \times \float \rightarrow \types'' \\\\
    \bot \sqsubset \types' \sqsubseteq \float \quad
    \bot \sqsubset \types'' \sqsubseteq \polyexp
    % \types \in \{\typeneuron, \polyexp\} \implies \types'' = \polyexp \\\\
    % \types \not\in \{\typeneuron, \polyexp\} \implies \types'' = \types 
    }
    {
    \g \vdash \expr_1.\traverse(\dir, f_{c_1}, f_{c_2}, f_{c_3})\{\expr_2\} : \polyexp
    }
    % \hspace{1cm}
    % 
    \end{array}
    \]
    % }
%      \end{subfigure}
%     \caption{Caption}
%     \label{fig:enter-label}
\end{figure}
% Following the subtyping lattice for the types in \cf. To complete the lattice, we introduce the types $\top, \bot$. An expression type-checks if the type of the expression is not $\top$ or $\bot$.
% \begin{enumerate}
%     \item $\intt \sqsubset \float$
%     \item $\bool \sqsubset \ct$
%     \item $\typenoise \sqsubset \symexp$
%     \item $\float \sqsubset \symexp$
%     \item $\typeneuron \sqsubset \polyexp$
%     \item $\float \sqsubset \polyexp$
% \end{enumerate}
% \begin{figure}[H]
%     \centering
%     \input{sections/subtyping}
%     \caption{Subtyping relations}
%     % \label{fig:subtyping}
% \end{figure}

\subsection{Type checking Rules for Statements}
The rules for type-checking statements update either $\tau_s$ or $\Gamma$. For function definition statements where the arguments are of the type $\types_1, \cdots, \types_n$ and the return expression is of type $\types$, the type of the function is $(\Pi_i^n\types_i) \rightarrow \types$. $\Gamma$ is updated to map the function name to this type. For transformer definition statements, if they take in the inputs \curr and \prev, which are optional, the inputs are of type $\typeneuron$ and $\overline{\typeneuron}$. The type of the output tuples returned by transformers are the join of the output tuples returned by the return expression of each operation. The type of each transformer is also stored in $\Gamma$. The type of the shape declaration statement is $[\types_1, \cdots, \types_n]$, where $\types_i$ is the type of the $i$th declared shape member.\\

% \section{Type checking of Statements in \cf}
% \begin{figure}[H]
% \raggedright
  \fbox{$\Gamma, \tau_s \vdash s \ : \ \Gamma'$}
  % \fbox{$\Gamma \vdash \ct \ : \ \types_c$}
  % \resizebox{\textwidth}{!}{
  % $
  \[
  \begin{array}{c}
    \inferrule*[lab={\footnotesize{\textsc{T-func}}}]
    {
    \Gamma' = \Gamma[x_1\mapsto \types_1, \cdots x_n \mapsto \types_n] \\\\
    \var \not\in \Gamma \quad
    \Gamma', \tau_s \vdash \expr : \types \quad 
    \bot \sqsubset \types \sqsubset \top
    }
    {
    \Gamma, \tau_s \vdash \func \ x(\types_1 \ x_1, \cdots \types_n \ x_n) = \expr : \Gamma[x \mapsto (\Pi_i^n\types_i) \rightarrow \types]
    }
    \hspace{1cm}
    \inferrule*[lab = \textsc{T-transformer-call}]
    {\Gamma(\theta_c) = \tau}
    {\g \vdash \theta_c : \tau}
    \\\\
    \inferrule*[lab={\textsc{T-transformer-ret-1}}]
    {
    \forall i \in [n], \g \vdash \expr_i : \types_i \quad 
    \types_i \sqsubseteq \tau_s(i)
    }
    {
    \g \vdash (\expr_1, \cdots, \expr_n) : (\types_1, \cdots, \types_n)
    }
    \hspace{1cm}
    \inferrule*[lab={\textsc{T-transformer-ret-2}}]
    {
    \g \vdash \expr : \bool \\\\
    \g \vdash \theta_{r_1} : (\types'_1, \cdots ,\types'_n) \\\\
    \g \vdash \theta_{r_2} : (\types''_1, \cdots ,\types''_n) \\\\
    \forall i \in [n], \types_i = \types'_i \sqcup \types''_i \sqsubset \top
    }
    {
    \g \vdash \mathif(\expr, \theta_{r_1}, \theta_{r_2}) : (\types_1, \cdots ,\types_n)
    }
    \\\\
    \inferrule*[lab={\textsc{T-transformer}}]
    {
    \Gamma' = \Gamma[\curr \mapsto \typeneuron)][\prev \mapsto \overline{\typeneuron}] \\\\
    \forall i \in [m], \Gamma', \tau_s \vdash \theta_{r_i} : (\types^1_i, \cdots, \types^n_i)\\\\
    \var \not\in \Gamma \quad 
    \forall j \in [n], \types^j = \sqcup_{i \in [m]}(\types^j_i) \quad 
    \types^j \sqsubseteq \tau^j_s
    }
    {
    \Gamma, \tau_s \vdash \transformer \ \var(\curr, \prev) = \{\eta_1 : \theta_{r_1}, \cdots\} : \Gamma[x \mapsto (\typeneuron \times \overline{\typeneuron}) \rightarrow (\types^1, \cdots)]
    } 
    \\\\
    \inferrule*[lab={\textsc{T-flow}}]
    {
    % \g \vdash \expr : \typepoly \\\\ 
    \g \vdash f_{c_1} : \typeneuron \rightarrow \types' \\\\ 
    \g \vdash f_{c_2} : \typeneuron \rightarrow \bool \\\\ 
    \theta_c \in \Gamma \quad
    % \g \vdash \theta_c : \tau \\\\
    \types' \sqsubseteq \float
    % \types = \typeneuron \implies \types' = \typepoly \quad \types \not= \typeneuron \implies \types' = \types
    }
    {
    \g \vdash \flow(\dir, f_{c_1}, f_{c_2}, \theta_{c}) : \Gamma
    }
    \hspace{1cm}
    \inferrule*[lab={\textsc{T-seq}}]
    {
    \g \vdash s_1 : \Gamma' \\\\
    \Gamma', \tau_s \vdash s_2 : \Gamma''
    }
    {
    \g \vdash s_1 \ ; \ s_2 : \Gamma''
    } 
\end{array}
  \]
  % }
  
% \caption{Typing judgment for statements}
% \label{fig:typestatement}
% \end{figure}

% \begin{figure}[H]
% \raggedright
  \fbox{$\cdot \ \vdash d \ : \ \tau_s$}
  % \fbox{$\Gamma \vdash \ct \ : \ \types_c$}
  % \resizebox{\textwidth}{!}{
  % $
  \[
  \begin{array}{c}
    \inferrule*[lab={\textsc{T-shape}}]
    {
    \tau_s = [x_1 \mapsto \types_1, \cdots x_n \mapsto \types_n]  \\\\
    \forall i \in [n], \bot \sqsubset \types_i \sqsubset \top
    % \tau_{s_2} = [1 \mapsto \tau_1, \cdots n \mapsto \tau_n] 
    }
    {
    \cdot \ \vdash \defshape \ (\types_1 \ x_1, \cdots, \types_n \ x_n) : \tau_s
    }
\end{array}
  % $
  % }
  \]
% \caption{Typing judgment for shape declaration}
% \label{fig:typekappa}
% \end{figure}

% \begin{figure}[H]
% \raggedright
%\resizebox{0.5\textwidth}{!}{
  \fbox{$\cdot \ \vdash \Pi \ : \ \Gamma, \tau_s$}
  % \fbox{$\Gamma \vdash \ct \ : \ \types_c$}
  % \resizebox{\textwidth}{!}{
  % $
  \[
  \begin{array}{c}
    \inferrule*[lab={\textsc{T-program}}]
    {
    \cdot \vdash d : \tau_s \quad
    \cdot, \tau_s \vdash s : \Gamma 
    }
    {
    \cdot \ \vdash d \ ; \ s \ : \ \Gamma, \tau_s
    }
\end{array}
  % $
% }
  \]
% \caption{Typing judgment for Program}
% \label{fig:typeprogram}
% \end{figure}

% \begin{figure}[H]
%     \centering
%     \input{sections/subtyping}
%     \caption{Subtyping relations}
%     % \label{fig:subtyping}
% \end{figure}

\clearpage
\section{Concrete Values in \oldtool}
\label{appendix:value}
\subsection{Definition}
\begin{grammar}
<constant> $\constant$ ::= $1$ | $2$ | $3$ | $\cdots$ | $\true$ | $\false$

<PolyVal> $\pval$ ::= $\constant$ | $\constant * n$ | $\pval' + \pval''$

<SymVal> $\zval$ ::= $\constant$ | $\constant * \noise$ | $\zval' + \zval''$

<CtVal> $\cval$ ::= $\pval == \pval$ | $\pval\  \embed\  \zval$

<Base-val> $\bval$ ::= $\constant$ | $\pval$ | $\zval$

<List-val> $\lval$ ::= $[\val_{b_1}, \val_{b_2}, \cdots]$

<Val> $\val$ ::= $\bval$ | $\lval$
\end{grammar}

These are the possible values a concrete expression can evaluate using the operational semantics. The value can be a constant, integer, real number, or boolean. The value can also be a neuron or symbolic variable, or an affine combination of either, which would result in a polyhedral expression or symbolic expression. Since we allow solvers in the operational semantics, constraints, stating that two polyhedral values are equal or a polyhedral value is embedded in a symbolic expression, are also values. Lists of values are also valid values.

\subsection{Operations on Concrete Values}
% \begin{figure}[H]
% \raggedright
    \fbox{$\maxx(\lval)$}
    % \resizebox{\textwidth}{!}{
    % $
    \[
    \begin{array}{c}
        \inferrule*[lab = \textsc{V-max-emp}]
        {
        \lval = []
        }
        {
        \maxx(\lval) = 0 
        } 
        \hspace{1cm}
        \inferrule*[lab = \textsc{V-max-non-emp}]
        {
        \lval = [\bval]
        }
        {
        \maxx(\lval) = \bval
        } 
        \hspace{1cm}
        \inferrule*[lab = \textsc{V-max-non-emp-r}]
        {
        \lval = \bval''::\lval' \quad 
        \bval' = \maxx(\lval') \\\\ 
        \bval'' \geq \bval' \implies \bval = \bval'' \\\\
        \bval'' < \bval' \implies \bval = \bval' 
        }
        {
        \maxx(\lval) = \bval
        } 
    \end{array}
    \]
    These are the rules for computing the maximum value of a list. The lists on which this operation is supported can only be of types $\overline{\intt}$ and $\overline{\float}$ so the operations $\ge$ and $<$ are defined. These rules recursively compute the maximum element of the list by comparing each element to the maximum of the remaining elements in the list.\\
%     $
%     }
% \end{figure}
% \noindent
% \begin{figure}[H]
% \raggedright

    \fbox{$\minn(\lval)$}
    % \resizebox{\textwidth}{!}{
    % $
    \[
    \begin{array}{c}
        \inferrule*[lab = \textsc{V-min-emp}]
        {
        \lval = []
        }
        {
        \minn(\lval) = 0 
        } 
        \hspace{1cm}
        \inferrule*[lab = \textsc{V-min-non-emp}]
        {
        \lval = [\bval]
        }
        {
        \minn(\lval) = \bval
        } 
        \hspace{1cm}
        \inferrule*[lab = \textsc{V-min-non-emp-r}]
        {
        \lval = \bval''::\lval' \quad 
        \bval' = \minn(\lval') \\\\ 
        \bval'' \leq \bval' \implies \bval = \bval'' \\\\
        \bval'' > \bval' \implies \bval = \bval' 
        }
        {
        \minn(\lval) = \bval
        } 
    \end{array}
    \]
%     $
%     }
% \end{figure}
These rules are similar to the above rules, except they compute the minimum element of a list, instead of the maximum one.

\noindent
% \begin{figure}[H]
% \raggedright
    \fbox{$\argmax(\lval, f_c, \context)$}
    % \resizebox{\textwidth}{!}{
    % $
    \[
    \begin{array}{c}
        \inferrule*[lab = \textsc{Compare}]
        {
        \lval = [\val_{b_1}, \cdots \val_{b_n}] \quad 
        \fstore[f_c] = (\var_1, \var_2), \expr \\\\ 
        \forall i \in [n], \store_i = \store[\var_1 \mapsto \bval, \var_2 \mapsto \val_{b_i}] \\\\ 
        \langle \expr, \fstore, \store_i, \hh_C \rangle \Downarrow \val'_{b_i} \quad
        \bval' = \bigwedge^{n}_{i=1}\val'_{b_i} \\\\
        \bval' = \true \implies \lval'' = \bval::\lval' \quad 
        \bval' = \false \implies \lval'' = \lval'
        }
        {
        \argmax(\bval, \lval, \lval', f_c, \context) = \lval''
        } 
        \hspace{0.5cm}
        \inferrule*[lab = \textsc{Compare-non-emp}]
        {
        \val_{l_1} = \bval::\lval \\\\
        \
        \val_{l_4} = \argmax(\bval, \val_{l_2}, \val_{l_3}, \context) \\\\
        \val_{l_5} = \argmax(\lval, \val_{l_2}, \val_{l_4}, f_c, \context)
        }
        {
        \argmax(\val_{l_1}, \val_{l_2}, \val_{l_3}, f_c, \context) = \val_{l_5}
        } 
        \\\\
        \inferrule*[lab = \textsc{Compare-emp}]
        {
        \val_{l_1} = []
        }
        {
        \argmax(\val_{l_1}, \val_{l_2}, \val_{l_3}, f_c, \context) = \val_{l_3}
        % \argmaxr(\lval', \lval'', \bval, f_c, \fstore, \rho) = \bval
        } 
        \hspace{1cm}
        % 
        % \\\\
        \inferrule*[lab = \textsc{V-compare-emp}]
        {
        \lval = []
        }
        {
        \argmax(\lval, f_c, \context) = \lval 
        } 
        % \hspace{1cm}
        \\\\
        \end{array}
    \]
%     $
%     }
% \end{figure}
% \begin{figure}
% \raggedright
% \resizebox{\textwidth}{!}{
%     $
\[
    \begin{array}{c}
        \inferrule*[lab = \textsc{V-compare-non-emp}]
        {
        \lval = [\bval]
        }
        {
        \argmax(\lval, f_c, \context) = \lval
        }
        \hspace{1cm}
        \inferrule*[lab = \textsc{V-compare-non-emp-r}]
        {
        \lval = [\val_{b_1}, \cdots \val_{b_n}] \\\\ 
        \lval' = \argmax(\lval, \lval, [], f_c, \context)
        }
        {
        \argmax(\lval, f_c, \context) = \lval'
        } 
    \end{array}
    \] \\
%     $
%     }
% \end{figure}

These rules define the $\argmax$ operation. This operation takes as input a list and a function. The function has to take in two arguments, which will be members of the list, and return true if the first argument is greater than the second argument. The $\argmax$ operation returns a list of all of the elements in the input list that are greater than all other elements. Depending on the function provided, this list could be empty, have one element, or more than one element. The rules \textsc{Compare}, \textsc{Compare-non-emp}, and \textsc{Compare-emp} compute the list of maximum elements while keeping track of an accumulated list of maximums. The rules \textsc{V-compare-emp}, and \textsc{V-compare-non-emp} state that the maximum of a list with 0 or 1 element is that same list. 

\subsection{Value typing}
% \begin{figure}[H]
% \raggedright
% \resizebox{\textwidth}{!}{
%     $
\[
  \begin{array}{c}
    \inferrule*[lab={\textsc{V-int}}]
    {
    }
    {
    \vdash \constant_i : \intt
    }
    \hspace{1cm}
    \inferrule*[lab={\textsc{V-real}}]
    {
    }
    {
    \vdash \constant_r : \float
    }
    \hspace{1cm}
    \inferrule*[lab={\textsc{V-true}}]
    {
    }
    {
    \vdash \true : \bool
    }
    \hspace{1cm}
    \inferrule*[lab={\textsc{V-false}}]
    {
    }
    {
    \vdash \false : \bool
    }
    \\\\
    \inferrule*[lab={\textsc{V-sym}}]
    {
    }
    {
    \vdash \noise : \typenoise
    }
    \hspace{1cm}
    \inferrule*[lab={\textsc{V-neuron}}]
    {
    }
    {
    \vdash n : \typeneuron
    } 
    \hspace{1cm}
    \inferrule*[lab={\textsc{V-polyexp}}]
    {
    }
    {
    \vdash \constant + \sum \constant*n : \polyexp
    }
    \hspace{1cm}
    \inferrule*[lab={\textsc{V-symexp}}]
    {
    }
    {
    \vdash \constant + \sum \constant*\noise : \symexp
    }
    \\\\
    \inferrule*[lab={\textsc{V-ct-1}}]
    {
    \vdash \val_1 : \types_1 \quad
    \vdash \val_2 : \types_2 \\\\
    \types_1 \sqcup \types_2 \in \{\polyexp, \typeneuron\} \\\\
    \oplus \in \{\geq, \leq, ==\}
    }
    {
    \vdash \val_1 \oplus \val_2 : \ct
    }
    \hspace{1cm}
    \inferrule*[lab={\textsc{V-ct-2}}]
    {
    \vdash \val_1 : \polyexp \quad
    \vdash \val_2 : \symexp 
    }
    {
    \vdash \val_1 \embed \val_2 : \ct
    }
    \hspace{1cm}
    \inferrule*[lab={\textsc{V-list}}]
    {
    \vdash \val_i : \types_i \quad  
    \types = \bigsqcup^n_{i=1} \types_i
    }
    {
    \vdash [\val_1, \cdots \val_n] : \overline{\types}
    }
  \end{array}
  \]
  % $
  % }
% \caption{Value typing Judgment}
% \label{fig:value-typing}
% \end{figure}

These are the types of possible values in concrete execution. Integers, real numbers, and boolean values have the standard types. Neurons, polyhedral expressions, symbolic variables, and symbolic expressions have the types $\typeneuron$, $\polyexp$, $\typenoise$ and $\symexp$ respectively. Constraints also use the specialized type, $\ct$. A list of values, which have the type $\types$, has the type $\overline{\types}$
\clearpage
\section{Operational Semantics of \oldtool}
\label{appendix:operationalsemantics}
\subsection{Operational semantics for expressions}
The general form for operational semantics for expressions is $\langle \expr, \fstore, \store, \hh_C \rangle \Downarrow \val $. $\fstore$ is a static record that stores the function names, their formal arguments, and the return expression from the function definitions. $\store$ maps variables to values. $\hh_C$ represents the concrete DNN. It maps each neuron's shape members and metadata to the corresponding concrete value, which comes from the input to the DNN certifier algorithm. The rules \textsc{Op-shape} and \textsc{Op-metadata} evaluate the expression before the shape member or metadata access, and then for the neuron, or each neuron in a list, output the value of the specified shape member or metadata from $\hh_C$. The rule \textsc{Op-compare} refers to \textsc{Compare}, \textsc{Compare-non-emp}, and \textsc{Compare-emp} defined above. The \textsc{Traverse} rule maintains an active set, that starts from the neurons in its input polyhedral expression. Then, it filters out the elements of the active set for which the stopping condition, $f_{c_2}$, evaluates to true, using the rule \textsc{Filter}. After this, it applies the replacement function, $f_{c_3}$, to the neurons, and their corresponding coefficient, in the active set with the highest priority (found using the rule \textsc{Priority}). For each neuron to which the replacement function is applied, this neuron is removed from the active set and replaced by its neighbors (found using the rule \textsc{Neighbour}). The neighbors are defined as the nodes with outgoing edges to a neuron when traversing in the backward direction and the neurons with incoming edges from a neuron when traversing in the forward direction. All of these steps are repeated while the active set is not empty.\\

 % \begin{figure}[H]
    \fbox{$\langle \expr, \context \rangle \Downarrow \val$}\\
    \[
    \begin{array}{c}
        \inferrule*[lab = \textsc{Op-const}]
        { }
        {\langle \constant, \context \rangle \Downarrow \constant}
        \hspace{1cm}
        \inferrule*[lab = \textsc{Op-var}]
        { }
        {\langle \var, \context \rangle \Downarrow \store(\var)}
        \hspace{1cm}
        \inferrule*[lab = \textsc{Op-sym}]
        { }
        {\langle \noise, \context \rangle \Downarrow \noise_{new}}
        % \hspace{1cm}
        \\\\
        \inferrule*[lab = \textsc{Op-unary}]
        {\langle \expr, \context \rangle \Downarrow \val} 
        {\langle \sim\expr, \context \rangle \Downarrow \ \sim\val}
        % \\\\
        \hspace{1cm}
        \inferrule*[lab = \textsc{Op-binary}]
        {
        \langle \expr_1, \context \rangle \Downarrow \val_1 \\\\
        \langle \expr_2, \context \rangle \Downarrow \val_2
        } 
        {\langle \expr_1 \oplus \expr_2, \context \rangle \Downarrow \val_1 \binop \val_2}
        \hspace{1cm}
        \inferrule*[lab = \textsc{Op-ternary}]
        {
        \langle \expr_1, \context \rangle \Downarrow \val_{b_1} \\\\
        \langle \expr_2, \context \rangle \Downarrow \val_2 \\\\
        \langle \expr_3, \context \rangle \Downarrow \val_3 \\\\
        \val_{b_1} = \true \implies \val = \val_2 \\\\
        \val_{b_1} = \false \implies \val = \val_3 
        } 
        {\langle (\expr_1 ? \expr_2 : \expr_3), \context \rangle \Downarrow \val}
        % \hspace{1cm}
        \\\\
        \inferrule*[lab = \textsc{Op-metadata}]
        {
        \langle \expr, \context \rangle \Downarrow n \\\\
        \val = \hh_C[n[m]] 
        } 
        {\langle \expr[m], \context \rangle \Downarrow \val}
        % \\\\
        \hspace{1cm}
        \inferrule*[lab = \textsc{Op-shape}]
        {
        \langle \expr. \context \rangle \Downarrow n \\\\
        \val = \hh_C[n[x]] 
        } 
        {\langle \expr[x], \context \rangle \Downarrow \val}
        % \hspace{1cm}
        \\\\
        \inferrule*[lab = \textsc{Op-max}]
        {
        \langle \expr, \context \rangle \Downarrow \val'
        } 
        {\langle \maxx(\expr), \context \rangle \Downarrow \max(\val')}
        \hspace{1cm}
        \inferrule*[lab = \textsc{Op-compare}]
        {
        \langle \expr, \context \rangle \Downarrow \val' \\\\
        \val = \argmax(\val', f_c, \context)
        } 
        {\langle \argmax(\expr, f_c), \context \rangle \Downarrow \val}
        \\\\
        \\\\
        \end{array}
    \]
    \[
    \begin{array}{c}
        \inferrule*[lab = \textsc{Op-func-call}]
        {
        \forall i \in [n], \ \langle \expr, \context \rangle \Downarrow \val_i \\\\
        \fstore(f_c) = (\var_1, \cdots, \var_n), \expr \\\\
        \store' = \store[\var_1 \mapsto \val_1, \cdots \var_n \mapsto \val_n] \\\\ 
        \langle \expr, \fstore, \store', \hh_C \rangle \Downarrow \val 
        % \fstore, \store', \ver, \dir \models \expr \Downarrow \val  
        }
        {
        \langle f_c(\expr_1, \cdots, \expr_n), \context \rangle  \Downarrow \val
        }
        \hspace{1cm}
        \inferrule*[lab = \textsc{Op-map}]
        {
        \langle \expr, \context \rangle \Downarrow \bval' \quad 
        \bval' = \constant_0 + \sum_{i=0}^{i=l} \constant_i \cdot \ver_i \\\\
        \forall i \in [l], \ \langle f_c(\ver_i, \constant_i), \context \rangle  \Downarrow \val_i \quad 
        \bval = \constant_0 + \sum_{i=0}^{i=l} \val_i
        }
        {
        \langle \expr.\map(f_c), \context \rangle \Downarrow \bval
        }
        \\\\
        \inferrule*[lab = \textsc{Filter}]
        {
        \vset' = \{\ver' \ | \ (\ver' \in \vset) \wedge (\langle f_c(\ver'), \context \rangle \Downarrow \false)\}
        }
        {
        \filter(\vset, f_c, \context) = \vset'
        }
        % \hspace{1cm}
        \\\\
        \inferrule*[lab = \textsc{Priority}]
        {
        \forall \ver_i \in \vset, \langle f_c(\ver_i), \context \rangle \Downarrow \val_i \\\\
        m = \max_i(\val_i)\\\\
        \vset' = \{\ver' \ | \ (\ver' \in \vset) \wedge (\langle f_c(\ver'), \context \rangle \Downarrow m\}
        }
        {
        \priority(\vset, f_c, \context) = \vset'
        }
        \\\\
        \inferrule*[lab = \textsc{Neighbour}]
        {
        \dir = \forward \implies \vset' = \{\ver' \ | \ (\ver, \ver' \in E)\} \\\\
        \dir = \backward \implies \vset' = \{\ver' \ | \ (\ver', \ver \in E)\}
        }
        {
        \neigh(\vset, \dir) = \vset'
        }
        % \hspace{1cm}
        \\\\
        \inferrule*[lab = \textsc{Vertices}]
        {
        \val = \constant_0 + \constant_1 \cdot \ver_1 + \cdots + \constant_l \cdot \ver_l
        }
        {
        \vertices(\val) = \{\ver_1, \cdots ,\ver_l\}
        }
        \hspace{1cm}
        \inferrule*[lab = \textsc{Vertices-2}]
        {
        \val = \constant_0 + \constant_1 \cdot \ver_1 + \cdots + \constant_l \cdot \ver_l
        }
        {
        \val_\vset = \sum_{\ver_j \in \vset} \constant_j \cdot \ver_j
        }
        % \hspace{1cm}
        \\\\
        \inferrule*[lab = \textsc{Op-traverse-1}]
        { }
        {
        \langle \val.\traverse(\dir', f_{c_1}, f_{c_2}, f_{c_3}), \context, \{\} \rangle \Downarrow \val
        }
        % \hspace{1cm}
        \\\\
        \inferrule*[lab = \textsc{Op-traverse-2}]
        {
        \vset' = \priority(\vset, f_{c_1}, \context) \quad 
        \val = \constant + \val_{\vset'} + \val_{\overline{\vset'}} \\\\
        \langle \val_{\vset'}.\map(f_{c_3}), \context \rangle \Downarrow \val' \\\\
        \val'' = \constant + \val' + \val_{\overline{\vset'}} \\\\
        \vset'' = \filter((\vset \setminus \vset') \cup \neigh(\vset', \dir), f_{c_2}, \context) \\\\
        % \vset''' = \filter(\vset'', f_{c_2}, \context) \\\\
        \langle \val''.\traverse, \context, \vset'' \rangle \Downarrow \val'''
        }
        {
        \langle \val.\traverse(\dir, f_{c_1}, f_{c_2}, f_{c_3}), \context, \vset \rangle \Downarrow \val'''
        }
        \\\\
        \inferrule*[lab = \textsc{Op-traverse}]
        {
        \langle \expr, \context \rangle \Downarrow \val \quad 
        \vset = \filter(\vertices(\val), f_{c_2}, \context) \\\\
        \langle \val.\traverse(\dir, f_{c_1}, f_{c_2}, f_{c_3}), \context, \vset \rangle \Downarrow \val'
        }
        {
        \langle \expr_1.\traverse(\dir, f_{c_1}, f_{c_2}, f_{c_3})\{\_\}, \context \rangle \Downarrow \val'
        }
    \end{array}
    \]
% \end{figure}

\subsection{Operational semantics for \oldtool statements}
In operational semantics for statements, function definition statements modify $\fstore$ by adding the function name, arguments, and return expression. Transformer definitions similarly modify $\tstore$ by adding the return tuples for each operation to $\tstore$, which is mapped to by the transformer name. The rules for the flow statement specify that if the constraints flow in the forward direction, the initial active set of neurons is the ones in the input layer. Then, this set is filtered to only retain the neurons for which the stopping condition is evaluated to $\false$. After this, the transformer specified in the flow statement is applied to each neuron in the active set. Lastly, each neuron in the active set is replaced with its neighbors. These steps are repeated until the active set of neurons in empty.
\begin{figure}[H]
\raggedright
    \fbox{$\langle s, \env \rangle \Downarrow \fstore', \tstore', \hh_C'$}
    \[
    \begin{array}{c}
        \inferrule*[lab = \textsc{Op-func-def}]
        { }
        {
        \langle \func \ \var(\types _1 \ \var_1, \types_2 \ \var_2, \cdots) = \expr, \env \rangle \Downarrow \fstore[\var \mapsto ((\var_1, \var_2, \cdots), \expr)], \tstore, \hh_C
        }
        \\\\
        \inferrule*[lab = \textsc{Op-transformer-def}]
        { }
        {
        \langle \transformer \ \var = \{\eta_1 \rightarrow \theta_{r_1}; \cdots\}, \env \rangle \Downarrow \fstore, \tstore[\var \mapsto (\{\eta_1 \rightarrow \theta_{r_1}; \cdots\})], \hh_C
        }
        \\\\
        \inferrule*[lab = \textsc{Op-transformer-ret}]
        {
        \forall i \in [n], \langle \expr_i, \context \rangle \Downarrow \val_i
        }
        {
        \langle (\expr_1, \cdots, \expr_n), \context \rangle \Downarrow (\val_1, \cdots, \val_n)
        }
        \hspace{1cm}
        \inferrule*[lab = \textsc{Op-transformer-ret-if}]
        {
        \langle \expr, \context \rangle \Downarrow \val \\\\
        \langle \theta_{r_1}, \context \rangle \Downarrow t_1 \quad
        \langle \theta_{r_2}, \context \rangle \Downarrow t_2 \\\\
        \val = \true \implies t = t_1 \quad
        \val = \false \implies t = t_2
        }
        {
        \langle \mathif(\expr, \theta_{r_1}, \theta_{r_2}), \context \rangle \Downarrow t
        }
        \\\\
        \inferrule*[lab = \textsc{Transformer-call}]
        {
        \tstore(\theta_c) = \{\eta_1 \rightarrow \theta_{r_1}, \cdots, \eta_m \rightarrow \theta_{r_m}\} \\\\
        \forall \ver_i \in \vset, \ \store_i = [\curr \mapsto \ver_i][\prev \mapsto \neigh(\ver_i, \dir)] \\\\
        \ver_i.\type = \eta_j \quad
        \langle \theta_{r_j}, \fstore, \store_i, \hh_C \rangle \Downarrow (\val^1_i, \cdots, \val^n_i) \\\\
        \hh'_C = \hh_C[\ver_i.\shape \mapsto (\val^1_i, \cdots, \val^n_i)]
        }
        {
        \langle \theta_c, \env, \vset, \dir \rangle \Downarrow \hh_C'
        }
        \\\\
        \inferrule*[lab = \textsc{Flow-emp}]
        { }
        {
        \langle \flow(\dir, f_{c_1}, f_{c_2}, \theta_c), \env, \{\} \rangle \Downarrow \hh_C
        }
        \hspace{1cm}
        \inferrule*[lab = \textsc{Flow-non-emp}]
        {
        \vset' = \priority(\vset, f_{c_1}, \context) \\\\
        \langle \theta_c, \env, \vset', \dir \rangle \Downarrow \hh_C' \\\\
        \vset'' = \filter((\vset - \vset') \cup \neigh(\vset', \dir), f_{c_2}, \context) \\\\
        % \vset''' = \filter(\vset'', f_{c_2}, \fstore, \dir) \\\\
        \langle \flow(\dir, f_{c_1}, f_{c_2}, \theta_c), \env', \vset'' \rangle \Downarrow \hh_C''
        }
        {
        \langle \flow(\dir, f_{c_1}, f_{c_2}, \theta_c), \env, \vset \rangle \Downarrow \hh_C''
        }
        \\\\
        \inferrule*[lab = \textsc{Op-flow}]
        {
        \dir = \forward \implies \vset = \filter(\{\ver \ | \ \ver.\inputt = \true\}, f_{c_2}, \context)\\\\
        \dir = \backward \implies \vset = \filter(\{\ver \ | \ \ver.\outputt = \true\}, f_{c_2}, \context) \\\\
        \langle \flow(\dir, f_{c_1}, f_{c_2}, \theta_c), \env, \vset \rangle \Downarrow \hh_C'
        }
        {
        \langle \flow(\dir, f_{c_1}, f_{c_2}, \theta_c), \env \rangle \Downarrow \env'
        }
        \\\\
        \inferrule*[lab = \textsc{Op-seq}]
        {
        \langle s_1; s_2, \env \rangle \Downarrow \fstore', \store', \hh_C' \\\\
        \langle s_1; s_2, \fstore', \store', \hh_C' \rangle \Downarrow \fstore'', \store'', \hh_C''
        }
        {
        \langle s_1; s_2, \env \rangle \Downarrow \fstore'', \store'', \hh_C''
        }
    \end{array}
    \]
\end{figure}
\clearpage
\section{Type Soundness}
\label{appendix:typesoundness}

\subsection{Type-checking for expressions}
\renewcommand{\thelemma}{\ref{lemma:optype-checking}}
\begin{lemma}
% \label{appendixlemma:tcexpressions}
If,
\begin{enumerate}
    \item $\Gamma, \tau_s \vdash \expr : \types$
    \item $\bot \sqsubset \types \sqsubset \top$
    \item $\fstore \sim \Gamma, \tau_s$
    \item $\store \sim \Gamma$
    \item $\hh_C \sim \tau_s$
    \item $\hh_C$ is finite
    % \item $V$ is the set of Neurons in $\range(\rho)$
\end{enumerate}
then,
\begin{enumerate}
    \item $\langle \expr, \context \rangle \Downarrow \val$
    \item $\vdash \val : \types'$
    \item $\types' \sqsubseteq \types$
\end{enumerate}
\end{lemma}
\renewcommand{\thelemma}{\Alph{section}.\arabic{lemma}}

\begin{proof}[Proof sketch]
This lemma states that if an expression type-checks, then it will terminate in the operational semantics and evaluate to a value. This value will either be the same type that the expression type checks to, or a subtype. We prove this lemma for all expressions, using induction on the structure of an expression. 
To prove the termination of every expression in \oldtool, the only non-trivial case is $\traverse$. For this case, we use a ranking function and show that it is bounded by 0 and decreases by at least 1 in each iteration.
\end{proof}

\begin{proof}Proof by induction on the structure of $\expr$\\
\textbf{Base Cases:}\\
$\mathbf{\expr \equiv \constant}$
\setcounter{number}{1}
\begin{align*}
    &\expr \equiv \constant_i \wedge \Gamma, \tau_s \models \expr : \intt \text{ or } \expr \equiv \constant_f \wedge \Gamma, \tau_s \models \expr : \float  && \text{From \textsc{T-const}} && \mycounter & \\
    &\text{From \textsc{Op-const}} \langle \expr, \fstore, \store, \hh_C \rangle \Downarrow \constant && \text{Consequent (1)} && \mycounter &\\
    &\text{From \textsc{V-int} and \textsc{V-real} } \vdash \constant_i : \intt \text{ and } \vdash \constant_f : \float && \text{Consequent 2} && \mycounter & \\
    &\intt \sqsubseteq \intt \text{ and } \float \sqsubseteq \float && \text{Consequent (3)} && &\\
\end{align*}
$\mathbf{\expr \equiv \var}$
\setcounter{number}{1}
\begin{align*}
    & \var \in \Gamma && \text{From (1) and \textsc{T-var}} && \mycounter &\\
    &\var \in \store(x) &&\text{From Antecedent (4)} && \mycounter &\\
    & \text{From (2) and \textsc{Op-var} } \langle \var,\fstore,\store,\hh_C \rangle \Downarrow \store(\var) && \text{Consequent (1)} && \mycounter &\\
    &\text{From Antecedent (4) } \Gamma, \tau_s \vdash \store(\var) : \Gamma(\var) \quad \Gamma(\var) \leq \Gamma(\var)&& \text{Consequents (2) and (3)}&& \mycounter &\\
\end{align*}

\textbf{Induction Cases:}\\
$\mathbf{\expr \equiv \var.\traverse(\dir, f_{c_1}, f_{c_2}, f_{c_3})}\{\expr_1\}$
\setcounter{number}{1}
\begin{align*}
& \inferrule*[lab = \textsc{T-traverse}]
    {
    \g \vdash \var : \polyexp \quad 
    \g \vdash \expr_1 : \ct \\\\ 
    \g \vdash f_{c_1} : \typeneuron \rightarrow \types' \quad 
    \g \vdash f_{c_2} : \typeneuron \rightarrow \bool \\\\ 
    \g \vdash f_{c_3} : \typeneuron \times \float \rightarrow \types \\\\
    \types' \leq \float \quad
    \types \leq \polyexp 
    % \quad \joinable(L, L_1, L_2, L_3)\\\\
    \types \in \{\typeneuron, \polyexp\} \implies \types'' = \polyexp \\\\
    \text{otherwise }\types'' = \types
    }
    {
    \g \vdash \var.\traverse(\dir, f_{c_1}, f_{c_2}, f_{c_3})\{\expr_1\} : \types''
    } && && \mycounter &\\
    &\text{From the induction hypothesis using (1) and Antecedents (3),(4) and (5)}\\
    & \langle \expr_1, \fstore, \store, \hh_C \rangle \Downarrow \val_1 && && \mycounter &\\
    &\vdash \val_1 :\types_1  \quad \types_1 \sqsubseteq \ct && && \mycounter &\\
    & \langle \var, \fstore, \store, \hh_C \rangle \Downarrow \val_2 && && \mycounter &\\
    &\vdash \val_2 :\types_2  \quad \types_2 \sqsubseteq \polyexp && && \mycounter &\\
    &\text{From \textsc{Op-traverse, Op-traverse-1, Op-traverse-2},}\\
    &\textsc{Neighbour, Filter, Priority, Vertices and Vertices-2 } \\
    &\text{The following are the steps to the traverse operational semantics:} && && \mycounter &\\
\end{align*}
\begin{enumerate}
    \item[a.] Start with an active set of neurons, $\vertices(\val_2)$ equals to the neurons in the polyhedral expression $\val_2$. 
    \item[b.] Create $\vset = \filter(\vertices(\val_2),f_{c_2}, \fstore, \store, \hh_C)$
    \item[c.] If $\vset = \emptyset$, output $\val_2$. Otherwise, do steps d through g below.
    \item[d.] Create $\vset' =  \priority(\vset, f_{c_1}, \context)$
    \item[e.] $\val_2 = \constant + \val_{\vset'} + \val_{\overline{\vset'}}$\\
    $\val'_2 = \constant + \val_{\overline{\vset'}} + \sum_{n} f_{c_3}(c_n,n)$ \\
    where $n$ is all the neurons in $\val_2$ that are also in $\vset'$ and $c_n$ is the coefficient of $n$ in $\val_2$.
    \item[f.] Create $\vset'' = \filter((\vset - \vset') \cup \neigh(\vset', \dir), f_{c_2}, \context)$
    \item[g.] $\vset = \vset'' \quad \val_2 = \val'_2$ Go back to step c. 
\end{enumerate}
To prove that this algorithm terminates we will define $L(V)$ as follows. Since the DNN has to be a DAG, there is a topological ordering of the neurons in the DNN, $n_0, n_1, \cdots, n_k$. Let's assume, without loss of generality, that we are traversing the DNN in the backward direction. We will define a function 
$$
\ell_{\vset}(n)=
\begin{cases}
1, n \in \vset \\
0,n \notin \vset \\
\end{cases}
$$
We will define $L(\vset)$ to be the binary number represented by  $\ell_{\vset}(n_k)\ell_{\vset}(n_{k-1})\cdots\ell_{\vset}(n_1)\ell_{\vset}(n_0)$\\
$L(\vset)$ is bounded by 0 by the definition of $\ell_{\vset}$.\\
 We will show that the algorithm above terminates by showing that $L(\vset)$ decreases by at least 1 in each iteration of the loop.\\
 \begin{enumerate}
     \item P returns the subset of $\vset$ containing all of the neurons with the highest priority. $\vset' \subseteq \vset$
     \item From (1) and the definition of $L$, $L(\vset) = L(\vset') + L(\vset \setminus \vset')$
     \item $\vset' = {n_{v_1}, \cdots, n_{v_j}}$. We know $\vset'$ is not empty because $\vset$ is not empty(otherwise the loop would be over), and since there are a finite number of neurons in the DNN, and therefore, a finite number in $\vset$ there has to be at least one neuron in $\vset$ with the highest priority. We can define $\vset'_0=\vset'$ and for $i \in [j]$, we can define $\vset'_i = \vset'_{i-1} \setminus n_{v_i} \cup \{n' | (n', n) \in E \}$  
     \item When going in the backward direction, $L(\vset'_{i-1} \setminus n_{v_i}) = L(\vset'_{i-1}) - 2^p$ where $p$ is the index of $n_{v_i}$ on the topological order. $L(\vset'_i = \vset'_{i-1} \setminus n_{v_i} \cup \{n' | (n', n) \in E \}) <= L(\vset'_{i-1}) - 2^p + \sum_{i=1}^{p-1} 2^i$ because when traversing backwards, every neighbor has to be before $n_{v_i}$ in the topological sort. Therefore, $L(\vset'_i) < L(\vset'_{i-1})$.
     \item $\neigh(\vset', \backward) = \vset'_j $ From (4), $L(\neigh(\vset', \backward)) = L(\vset'_j) < L(\vset')$
     \item $L((\vset - \vset') \cup \neigh(\vset', d)) <= L(\vset - \vset') + L(\neigh(\vset', d)) $
     \item From (5), $L(\vset - \vset') + L(\neigh(\vset', d))  < L(\vset - \vset') + L(\vset') = L(\vset)$
     \item $\filter$ returns a subset of the set of neurons passed to it. \\
     Therefore, $L(\filter((\vset - \vset') \cup \neigh(\vset', d)), f_{c_2}, \fstore, \store, \hh_C) <= L((\vset - \vset') \cup \neigh(\vset', d))$
     \item From (6), (7) and (8), $L(\filter((\vset - \vset') \cup \neigh(\vset', d)), f_{c_2}, \fstore, \store, \hh_C) < L(\vset) $
     \item From (9), $L(\vset'') < L(\vset)$ Therefore, $L(\vset)$ decreases in each iteration of the loop. Since it is a binary number that is bounded by 0, this means the loop must terminate.
 \end{enumerate}
Now that we have shown that this algorithm terminates and produces a value, $\val$, under operational semantics (Consequent (1)), we still have to prove that $\Gamma, \tau_s \vdash \val : \types$.
\setcounter{number}{1}
% \begin{figure}[H]
\begin{align*}
&\vdash \val_2 : \types_2 \quad \types_2 \sqsubseteq \polyexp && \text{From (5) above} && \mycounter &&\\
&\g \vdash f_{c_1} : \typeneuron \rightarrow \types'_1 \quad \types'_1 \sqsubseteq \float &&\text{From Antecedent (1) and \textsc{T-traverse}} && \mycounter &&\\
&\g \vdash f_{c_2} : \typeneuron \rightarrow \bool  &&\text{From Antecedent (1) and \textsc{T-traverse}} && \mycounter &&\\
& \g \vdash f_{c_3} : \typeneuron \times \float \rightarrow \types'_2 \quad \types'_2 \sqsubseteq \polyexp &&\text{From Antecedent (1) and \textsc{T-traverse}} && \mycounter &&\\
&\val'_2 = \constant + \val_{\overline{\vset'}} + \sum_{n} f_{c_3}(c_n,n) &&\text{From step e above }  && \mycounter &&\\
&\vdash \constant + \val_{\overline{\vset'}} : \types'_3 \quad \types'_3 \leq \polyexp && \text{From (1), (4) and (5)} && \mycounter &&\\
% \end{align*}
% \end{figure}
% \begin{figure}[H]
% \begin{align*}
&\types'_2 = \intt \implies \vdash \sum_{n} f_{c_3}(c_n,n) : \intt && \text{From \textsc{V-int}} && \mycounter &\\
&\types'_2 = \float \implies \vdash \sum_{n} f_{c_3}(c_n,n) : \float && \text{From \textsc{V-real}} && \mycounter &\\
&\types'_2 = \typeneuron \implies \vdash \sum_{n} f_{c_3}(c_n,n) : \polyexp && \text{From \textsc{V-polyexp}} && \mycounter &\\
&\types'_2 = \polyexp \implies \vdash \sum_{n} f_{c_3}(c_n,n) : \polyexp && \text{From \textsc{V-polyexp}} && \mycounter &\\
&\text{From (4),(7),(8),(9),(10) and \textsc{T-traverse}}\\
&\vdash \val : \types' \quad \types' \sqsubseteq \types && \text{Consequent (2) and (3)} && \mycounter &\\
\end{align*}
% \end{figure}

\noindent
$\mathbf{\expr \equiv \lp(\lpop,\expr_1,\expr_2)}$
\setcounter{number}{1}
\begin{align*}
    & \Gamma, \tau_s \vdash \expr_1 : \types_1 \quad \types_1 \sqsubseteq \float \quad \types = \types_1 && \text{From (1) and \textsc{T-solver}} && \mycounter &\\
    & \Gamma, \tau_s \vdash \expr_2 : \ct && \text{From (1) and \textsc{T-solver}} && \mycounter &\\
    &\text{From \textsc{Op-solver}, there are 3 possibilities:}\\
    &\text{From induction hypothesis using (1), (2)}\\
    &\text{and Antecedents (3),(4) and (5)}\\
    &\langle \expr_1, \fstore, \store, \hh_C \rangle \Downarrow \val_1 \quad \vdash \val_1 : \types_1 && && \mycounter &\\
    % &\text{From induction hypothesis using (2) and Antecedents (3),(4) and (5)}\\
    &\langle \expr_2, \fstore, \store, \hh_C \rangle \Downarrow \val_2 \quad \vdash \val_2 : \ct && && \mycounter &\\
    &\types_1= \intt \text{ or } \expr_2 \text{ is a conjunction of linear constraints containing}\\
    &\text { at least one expression, $s=n[x]$, such that $\tau_s(x) = \intt$:}\\
    &\text{An external call to a solver is made} && \text{From \textsc{Op-solver}} && \mycounter &\\
    &\langle \lp(\lpop,\expr_1,\expr_2), \fstore, \store, \hh_C \rangle \Downarrow \val  && \text{From (3),(4),(5)} && \mycounter &\\
    &\types_1 = \intt \implies \vdash \val : \intt \quad \types_1 = \float \implies \vdash \val : \float &&\text{From (5) and (6)} && \mycounter &\\
    &\expr_2 \text{ is a conjunction of linear constraints:}\\
    &\text{An external call to a solver is made} && \text{From \textsc{Op-solver}} && \mycounter &\\
    &\langle \lp(\lpop,\expr_1,\expr_2), \fstore, \store, \hh_C \rangle \Downarrow \val  && \text{From (3),(4),(8)} && \mycounter &\\
    &\vdash \val : \float &&\text{From (8) and (9)} && \mycounter &\\
    &\types_1 = \float && \text{From (1)} && \mycounter &\\
    &\expr_2 \text{ is not a conjunction of linear constraints:}\\
    &\text{An external call to an SMT solver is made} && \text{From \textsc{Op-solver}} && \mycounter &\\
    &\langle \lp(\lpop,\expr_1,\expr_2), \fstore, \store, \hh_C \rangle \Downarrow \val  && \text{From (3),(4),(12)} && \mycounter &\\
    &\types_1 = \intt \implies \vdash \val : \intt \quad \types_1 = \float \implies \vdash \val : \float  &&\text{From (12) and (13)} && \mycounter &\\
    &\text{From (6),(8) and (13) } \langle \lp(\lpop,\expr_1,\expr_2), \fstore, \store, \hh_C \rangle \Downarrow \val  && \text{Consequent (1)} && \mycounter &\\
    &\text{From (1),(7),(10),(11) and (14) } \vdash \val : \types \quad \types \sqsubseteq \types&& \text{Consequent (2) and (3)} && \mycounter &\\
\end{align*}
\newpage
$\mathbf{\expr \equiv \expr_1.\map(f)}$
\setcounter{number}{1}
\begin{align*}
&\text{Using \textsc{T-map-poly, T-map-sym}, } \\
& \Gamma, \tau_s \vdash \expr_1 : \polyexp \text{ or } \Gamma, \tau_s \vdash \expr_1 : \symexp &&  && \mycounter &\\
&\text{Case 1: } \Gamma, \tau_s \vdash \expr_1 : \polyexp&&&& \mycounter &\\
&\text{From the induction hypothesis using (2) and }\\
&\text{Antecedents (3),(4) and (5)}\\
&\langle \expr_1, \fstore, \store, \hh_C \rangle \Downarrow \val_1 && && \mycounter&\\
&\vdash \val_1 : \types'_1 \quad \types'_1 \sqsubseteq \polyexp && && \mycounter &\\
& \val_1 = \constant_0 + \sum_{i=0}^{i=l} \constant_i \cdot \ver_i  && \text{From (4)} && \mycounter &\\
&\Gamma, \tau_s \vdash f :  \typeneuron \times \float \rightarrow \types'' \quad \types'' \sqsubseteq \polyexp && \text{From \textsc{T-map-poly}}&& \mycounter&\\
&\forall i \in [l], \vdash c_i : \float \quad \vdash \ver_i : \typeneuron && \text{From (5)}&& \mycounter &\\
&\text{From (6),(7) and \textsc{T-func-call}} \\
&\forall i \in [l], \Gamma, \tau_s \vdash f(c_i,\ver_i) : \types'' && && \mycounter &\\
&\text{From induction hypothesis using (6),(8) and }\\
&\text{Antecedents (3),(4) and (5)}\\
&\langle f(c,\ver_i),\fstore, \store, \hh_C \rangle \Downarrow \val'_i && && \mycounter &\\
&\vdash \val'_i : \types'_i \quad \types'_i \sqsubseteq \polyexp && && \mycounter &\\
&\text{From (5), (10) and \textsc{Op-map}}\\
&\langle \expr_1.\map(f), \fstore, \store, \hh_C \rangle \Downarrow \constant_0 + \sum_{i=0}^{i=l} \val_i&& \text{Consequent (1)}&& \mycounter &\\
&\text{From (7),(10) and \textsc{V-polyexp} } \vdash \constant_0 + \sum_{i=0}^{i=l} \val_i : \types' \wedge \types' \sqsubseteq \types &&\text{Consequents (2) and (3)} && \mycounter&\\
&\text{Case 2: } \Gamma, \tau_s \vdash \expr_1 : \symexp &&&&  \mycounter&\\
&\text{From the induction hypothesis using (13) and }\\
&\text{Antecedents (3),(4) and (5)}\\
&\langle \expr_1, \fstore, \store, \hh_C \rangle \Downarrow \val_1 && && \mycounter&\\
&\vdash \val_1 : \types'_1 \quad \types'_1 \sqsubseteq \symexp && && \mycounter &\\
& \val_1 = \constant_0 + \sum_{i=0}^{i=l} \constant_i \cdot \ver_i  && \text{From (15)} && \mycounter &\\
&\Gamma, \tau_s \vdash f :  \typenoise \times \float \rightarrow \types'' \quad \types'' \sqsubseteq \symexp && \text{From \textsc{T-map-sym}}&& \mycounter&\\
&\forall i \in [l], \vdash c_i : \float \quad \vdash \ver_i : \typenoise && \text{From (16)}&& \mycounter &\\
&\text{From (17), (18) and \textsc{T-func-call}} \\
&\forall i \in [l], \Gamma, \tau_s \vdash f(c_i,\ver_i) : \types'' && && \mycounter & \\
&\text{From induction hypothesis using (17),(19) and }\\
&\text{Antecedents (3),(4) and (5)}\\
\end{align*}
\begin{align*}
&\langle f(c,\ver_i),\fstore, \store, \hh_C \rangle \Downarrow \val'_i && && \mycounter &\\
&\vdash \val'_i : \types'_i \quad \types'_i \sqsubseteq \symexp && && \mycounter &\\
&\text{From (16), (21) and \textsc{Op-map}}\\
&\langle \expr_1.\map(f), \fstore, \store, \hh_C \rangle \Downarrow \constant_0 + \sum_{i=0}^{i=l} \val_i&& \text{Consequent (1)}&& \mycounter &\\
&\text{From (18),(21) and \textsc{V-symexp} } \vdash \constant_0 + \sum_{i=0}^{i=l} \val_i : \types' \wedge \types' \sqsubseteq \types &&\text{Consequents (2) and (3)} && \mycounter&\\
\end{align*}

\noindent
$\mathbf{\expr \equiv f_c(\expr_1, \cdots, \expr_n)}$
\setcounter{number}{1}
\begin{align*}
&\text{From \textsc{T-func-call} and Antecedent (1)} \\
&\Gamma(f_c) = (\Pi_i^n \types_i)\rightarrow \types && && \mycounter &\\
&\text{From \textsc{T-func-call} and Antecedent (1)} \\
&\forall i, \g \vdash \expr_i : \types_i && && \mycounter &\\
&\bot \sqsubset \types \sqsubset \top \quad \bot \sqsubset \types_i \sqsubset \top&& \text{From \textsc{T-func}}&& \mycounter &\\
&\text{From the induction hypothesis using (2),(3) and}\\
&\text{Antecedents (3),(4) and (5)}\\
&\langle \expr_i, \context \rangle \Downarrow \val_i && && \mycounter &\\
&\vdash \val_i : \types_i && && \mycounter &\\
&f_c \in \dom(\fstore) \quad \fstore(f_c) = (\var_1, \cdots, \var_n),\expr' && \text{From Antecedent (4)} && \mycounter &\\
&\store' = \store[\var_1 \mapsto \val_1, \cdots, \var_n \mapsto \val_n] && && \mycounter &\\
&\text{From (2),(7) and Antecedents (3) and (4)} \\
&\store' \sim \Gamma && &&\mycounter &\\
&\Gamma, \tau_s \vdash \expr' : \types &&\text{From (1) and Antecedent (3)} && \mycounter &\\
&\text{From induction hypothesis using (3),(8),(9),}\\
&\text{and Antecedents (3) and (5)}\\
&\langle \expr', \fstore, \store', \hh_C \rangle \Downarrow \val  && && \mycounter &\\
&\vdash \val : \types'' \quad \types'' \sqsubseteq \types && \text{Consequents (2) and (3)}&& \mycounter &\\
&\text{From (4),(6),(7), (10) and \textsc{Op-func-call}}\\
&\langle f_c(\expr_1, \cdots, \expr_n), \fstore, \store', \hh_C \rangle \Downarrow \val &&\text{Consequent (1)} && &\\
\end{align*}
\newpage
\noindent
$\mathbf{\expr \equiv \expr_1 + \expr_2}$
\setcounter{number}{1}
\begin{align*}
&\text{From \textsc{T-binary-arith-1} }\Gamma, \tau_s \vdash \expr_1 : \types_1 \quad \Gamma, \tau_s \vdash \expr_2 : \types_2 &&  &&\mycounter &\\
&\text{From \textsc{T-binary-arith-1} }\types_1 \sqcup \types_2 \in \{ \polyexp, \symexp \} &&  && \mycounter &\\
&\text{From (2) }\bot \sqsubset \types_1 \sqsubset \top \quad \bot \sqsubset \types_2 \sqsubset \top && && \mycounter &\\
&\text{From the induction hypothesis using (1),(3) and Antecedents (3),(4) and (5)}\\
&\langle \expr_1, \fstore, \store, \hh_C \rangle \Downarrow \val_1 \quad \vdash \val_1 : \types'_1 \quad \types'_1 \sqsubset \types_1 && && \mycounter &\\
&\text{From the induction hypothesis using (1),(3) and Antecedents (3),(4) and (5)}\\
&\langle \expr_2, \fstore, \store, \hh_C \rangle \Downarrow \val_2 \quad \vdash \val_2 : \types'_2 \quad \types'_2 \sqsubset \types_2 && && \mycounter &\\
&\text{Here are all of the possible combinations of types, $\types_1,\types_2$ using (2)}\\
&\text{Using (4), (5), \textsc{Op-binary}}\\
&\types_1,\types_2 \in \{ \intt\} \text{:} \\
&\langle \expr, \fstore, \store, \hh_C \rangle \Downarrow \val_1 + \val_2 \quad  \vdash \val_1 + \val_2: \types_3  \quad \types_3 \sqsubseteq \intt = (\intt \sqcup \intt) &&   && \mycounter &\\
&\types_1,\types_2 \in \{ \intt, \float \} \text{:} \\
&\langle \expr, \fstore, \store, \hh_C \rangle \Downarrow \val_1 + \val_2 \quad  \vdash \val_1 + \val_2: \types_3  \quad \types_3 \sqsubseteq \float = (\intt \sqcup \float) &&   && \mycounter &\\
&\types_1,\types_2 \in \{ \float \} \text{:} \\
&\langle \expr, \fstore, \store, \hh_C \rangle \Downarrow \val_1 + \val_2 \quad  \vdash \val_1 + \val_2: \types_3  \quad \types_3 \sqsubseteq \float = (\float \sqcup \float) &&   && \mycounter &\\
&\types_1,\types_2 \in \{ \intt, \polyexp \} \text{:} \\
&\langle \expr, \fstore, \store, \hh_C \rangle \Downarrow \val_1 + \val_2 \quad  \vdash \val_1 + \val_2: \types_3  \quad \types_3 \sqsubseteq \polyexp = (\intt \sqcup \polyexp) &&   && \mycounter &\\
&\types_1,\types_2 \in \{ \float, \polyexp \} \text{:} \\
&\langle \expr, \fstore, \store, \hh_C \rangle \Downarrow \val_1 + \val_2 \quad  \vdash \val_1 + \val_2: \types_3  \quad \types_3 \sqsubseteq \polyexp = (\float \sqcup \polyexp) &&   && \mycounter &\\
&\types_1,\types_2 \in \{ \typeneuron, \polyexp \} \text{:} \\
&\langle \expr, \fstore, \store, \hh_C \rangle \Downarrow \val_1 + \val_2 \quad  \vdash \val_1 + \val_2: \types_3  \quad \types_3 \sqsubseteq \polyexp = (\typeneuron \sqcup \polyexp) &&   && \mycounter &\\
&\types_1,\types_2 \in \{ \polyexp, \polyexp \} \text{:} \\
&\langle \expr, \fstore, \store, \hh_C \rangle \Downarrow \val_1 + \val_2 \quad  \vdash \val_1 + \val_2: \types_3 \quad \types_3 \sqsubseteq \polyexp = (\polyexp\sqcup \polyexp) &&   && \mycounter &\\
&\types_1,\types_2 \in \{ \typeneuron, \typeneuron \} \text{:} \\
&\langle \expr, \fstore, \store, \hh_C \rangle \Downarrow \val_1 + \val_2 \quad  \vdash \val_1 + \val_2: \types_3  \quad \types_3 \sqsubseteq \polyexp = (\typeneuron \sqcup \typeneuron) &&   && \mycounter &\\
&\types_1,\types_2 \in \{ \intt, \symexp \} \text{:} \\
&\langle \expr, \fstore, \store, \hh_C \rangle \Downarrow \val_1 + \val_2 \quad  \vdash \val_1 + \val_2: \types_3 \quad \types_3 \sqsubseteq \symexp = (\intt \sqcup \symexp) &&   && \mycounter &\\
&\types_1,\types_2 \in \{ \float, \symexp \} \text{:} \\
&\langle \expr, \fstore, \store, \hh_C \rangle \Downarrow \val_1 + \val_2 \quad  \vdash \val_1 + \val_2: \types_3  \quad \types_3 \sqsubseteq \symexp = (\float \sqcup \symexp) &&   && \mycounter &\\
&\types_1,\types_2 \in \{ \typenoise, \symexp \} \text{:} \\
&\langle \expr, \fstore, \store, \hh_C \rangle \Downarrow \val_1 + \val_2 \quad  \vdash \val_1 + \val_2: \types_3 \quad \types_3 \sqsubseteq \symexp = (\typenoise \sqcup \symexp) &&   && \mycounter &\\
&\types_1,\types_2 \in \{ \symexp, \symexp \} \text{:} \\
&\langle \expr, \fstore, \store, \hh_C \rangle \Downarrow \val_1 + \val_2 \quad  \vdash \val_1 + \val_2: \types_3 \quad \types_3 \sqsubseteq \symexp = (\polyexp\sqcup \symexp) &&   && \mycounter &\\
&\types_1,\types_2 \in \{ \typenoise, \typenoise \} \text{:} \\
&\langle \expr, \fstore, \store, \hh_C \rangle \Downarrow \val_1 + \val_2 \quad  \vdash \val_1 + \val_2: \types_3 \quad \types_3 \sqsubseteq \symexp = (\typenoise \sqcup \typenoise &&   && \mycounter &
\end{align*}
\begin{align*}
&\text{In all of the cases above:}\\
&\langle \expr, \fstore, \store, \hh_C \rangle \Downarrow \val_1 + \val_2 && \text{Consequent (1)} && &\\
&\vdash \val_1 + \val_2 : \types_3 \quad \types_3 \sqsubseteq \types_1 \sqcup \types_2 && \text{Consequent (2)} && \mycounter &\\
\end{align*}

$\mathbf{\expr \equiv \expr_1 \leq \expr_2}$
\setcounter{number}{1}
\begin{align*}
&\text{From \textsc{T-comparison-1, T-comparison-2}} \\
&\Gamma, \tau_s \vdash \expr_1 : \types_1 \quad \Gamma, \tau_s \vdash \expr_2 : \types_2 &&  &&\mycounter &\\
&\text{From \textsc{T-comparison-1, T-comparison-2}} \\
&\types_1,\types_2 \in \{\intt, \float \} \text{ or } \types_1 \sqcup \types_2 \in \{ \polyexp \} &&  && \mycounter &\\
&\text{From Induction hypothesis using (1),(2) and}\\
&\text{Antecedents (3),(4) and (5)}\\
&\langle \expr_1, \fstore, \store, \hh_C \rangle \Downarrow \val_1  \quad \vdash \val_1 : \types'_1 \quad \types'_1 \sqsubseteq \types_1 && && \mycounter &\\
&\text{From Induction hypothesis using (1),(2) and}\\
&\text{Antecedents (3),(4) and (5)}\\
&\langle \expr_2, \fstore, \store, \hh_C \rangle \Downarrow \val_2  \quad \vdash \val_2 : \types'_2 \quad \types'_2 \sqsubseteq \types_2 && && \mycounter &\\
&\text{If }\types_1,\types_2 \in \{\intt, \float \} &&&& \mycounter &\\
&\types = \bool &&\text{From \textsc{T-comparison-1}} && \mycounter &\\
&\langle \expr_1 \leq \expr_2, \fstore, \store, \hh_C \rangle \Downarrow \val_1 \leq \val_2 && \text{From \textsc{Sym-binary}} &&\mycounter &\\
&\vdash \val_1 \leq \val_2 : \bool &&\text{From (3),(4), and (5)} &&\mycounter  &\\
&\text{If }\types_1 \sqcup \types_2 \in \{ \polyexp \} &&&& \mycounter &\\
&\types = \ct &&\text{From \textsc{T-comparison-2}} && \mycounter &\\
&\langle \expr_1 \leq \expr_2, \fstore, \store, \hh_C \rangle \Downarrow \val_1 \leq \val_2 && \text{From \textsc{Sym-binary}} &&\mycounter &\\
&\vdash \val_1 \leq \val_2 : \ct &&\text{From (3),(4), and (9)} &&\mycounter  &\\
&\text{From ((5),(6),(7),(8),(9),(10),(11) and (12)}\\
&\langle \expr_1 \leq \expr_2, \fstore, \store, \hh_C \rangle \Downarrow \val_1 \leq \val_2 \quad \vdash \val_1 \leq \val_2 : \types_3 \quad \types_3 \sqsubseteq \types && \text{Consequents (1) and (2)}&& &\\
\end{align*}
\end{proof}

\newpage

\subsection{Type-checking for statements}
\renewcommand{\thelemma}{\ref{lemma:statementtype-checking}}
\begin{lemma}
% \label{appendixlemma:tcstatements}
If,
\begin{enumerate}
    \item $\Gamma, \tau_s \vdash s : \Gamma'$
    \item $\fstore \sim \Gamma, \tau_s$
    \item $\store \sim \Gamma$
    \item $\hh_C \sim \tau_s$
    \item $\hh_C$ is finite
\end{enumerate}
then,
\begin{enumerate}
    \item $\langle s, \fstore, \store, \hh_C \rangle \Downarrow \fstore', \store', \hh'_C$
    \item $\fstore' \sim \Gamma', \tau_s$
    \item $\store' \sim \Gamma'$
    \item $\hh'_C \sim \tau_s$
\end{enumerate}
\end{lemma}
\renewcommand{\thelemma}{\Alph{section}.\arabic{lemma}}

\begin{proof}
    For function declaration statements and transformer definition statements, the expressions within them are successfully type-checked. So, according to Lemma~\ref{lemma:optype-checking}, the updated $\fstore$ and $\tstore$ remains consistent with $\Gamma', \tau_s$ after adding the new mapping.
    For \flow statements, an already defined abstract transformer is applied to the DNN $\hh_C$. Since $\tstore$ is consistent with $\Gamma, \tau_s$, after each application of the abstract transformer, the updated DNN $\hh'_C$ remains consistent with $\Gamma', \tau_s$.
\end{proof}

\subsection{Type soundness theorem}

\renewcommand{\thetheorem}{\ref{thm:welltyped}}
\begin{theorem}
    A well-typed program in \oldtool successfully terminates according to the operational semantics, i.e., $\ttt \models \ooo$. Formally, if $\cdot \vdash \Pi : \Gamma, \tau_s$ then $\langle\Pi, \hh_C\rangle \Downarrow \hh'_C$
\end{theorem}
\renewcommand{\thetheorem}{\Alph{section}.\arabic{theorem}}

The proof follows from Lemmas~\ref{lemma:optype-checking} and ~\ref{lemma:statementtype-checking}.

\clearpage
\section{Symbolic Values for Verification Procedure}
\label{appendix:sval}
\subsection{Definition}
\begin{grammar}
<Base-sym-val> $\bsval$ ::= $\sval_r$ | $\sval_i$ | $\sval_b$ | $\constant$
\alt $\neg \sval_1$ | $\sval_1 \oplus \sval_2$
\alt If$(\sval_1, \sval_2, \sval_3)$

<List-sym-val> $\lsval$ ::= $[\sval_{b_1}, \sval_{b_2}, \cdots ]$
\alt $If(\bsval, \sval_{l_1}, \sval_{l_2})$

<Sym-val> $\sval$ ::= $\bsval$ | $\lsval$
\end{grammar}

These are the possible values a symbolic expression can be evaluated to using symbolic semantics. A symbolic value can be a symbolic variable of type real, integer, or boolean, a unary or binary operation applied to a symbolic variable, and a conditional value of the form  $if(\sval_1, \sval_2, \sval_3)$. Integer, boolean, or real constants can also be symbolic values because sometimes the exact constant is known even during symbolic execution. Lists of symbolic values and conditional operations on lists of symbolic values are also considered symbolic values. 

\subsection{Operations on Symbolic Values}
% \begin{figure}[H]
    \fbox{$\height(\lsval)$}
    \[
    \begin{array}{c}
        \inferrule*[lab = \textsc{Sval-height-b}]
        {
        \lsval = [\sval_{b_1}, \sval_{b_2}, \cdots ] 
        }
        {
        \height(\lsval) = 0
        } 
        \hspace{1cm}
        \inferrule*[lab = \textsc{Sval-height-r}]
        {
        \lsval = If(\bsval, \sval_{l_1}, \sval_{l_2}) \\\\
        \constant_1 = \height(\sval_{l_1}) \quad 
        \constant_2 = \height(\sval_{l_2}) \\\\ 
        \constant = 1 + \max(\constant_1, \constant_2)
        }
        {
        \height(\lsval) = \constant
        } 
    \end{array}
    \]
% \end{figure}
$\height(\lsval)$ is used to determine the maximum number of nested conditional values in the branches of a conditional value. $\height(\lsval)$ is 0 when called on a symbolic value that is not conditional. This is used to prove lemmas that involve values of the form $[\sval_{b_1}, \cdots, \sval_{b_n}]$, which can be proved by induction on $\height(\lsval)$. A similar notion can be defined for basic values to find the number of nested conditional symbolic values before a basic symbolic value that is not conditional is reached.\\

\noindent
% \begin{figure}[H]
    \fbox{$\expand(\sval)$}
    \[
    \begin{array}{c}
        \inferrule*[lab = \textsc{Expanded-poly}]
        {
        \bsval = \sval_{b_0} + \sum \sval_{b_i} * n_i 
        }
        {
        \expand(\bsval)
        } 
        \hspace{1cm}
        \inferrule*[lab = \textsc{Expanded-sym}]
        {
        \bsval = \sval_{b_0} + \sum \sval_{b_i} * \epsilon_i 
        }
        {
        \expand(\bsval)
        } 
        \\\\
        \inferrule*[lab = \textsc{Expanded-list}]
        {
        \lsval = [\sval_{b_0}, \cdots, \sval_{b_n}] \\\\
        \forall i \in [n], \expand(\sval_{b_i})
        }
        {
        \expand(\lsval)
        }
        \hspace{1cm}
        \inferrule*[lab = \textsc{Expanded-if}]
        {
        \sval = If(\sval_1, \sval_2, \sval_3) \\\\
        \expand(\sval_2) \quad 
        \expand(\sval_3) 
        }
        {
        \expand(\sval)
        }
        \\\\
        \inferrule*[lab = \textsc{Expanded-binary}]
        {
        % \sval = \sval_1 \oplus \sval_2 \\\\
        \expand(\sval_1) \quad 
        \expand(\sval_2) 
        }
        {
        \expand(\sval_1 \oplus \sval_2)
        }
    \end{array}
    \]
% \end{figure}
These rules define expanded values that are necessary when symbolically executing functions like $\map$ and $\traverse$ because these operations apply a function to each pair of coefficient and neuron, or coefficient and symbolic variable, in a polyhedral or symbolic expression. Here, a polyhedral expression or symbolic expression has to have each of its constituent neurons or symbolic variables represented explicitly. as opposed to the whole polyhedral or symbolic variable being defined by one symbolic variable of type real, which is how these expressions are initially defined.\\

\noindent
% \begin{figure}[H]
    \fbox{$\add(\lsval)$}
    \[
    \begin{array}{c}
        \inferrule*[lab = \textsc{Sval-sum-b}]
        {
        \lsval = [\sval_{b_1}, \sval_{b_2}, \cdots \sval_{b_n}] \\\\
        \bsval' = \sum^n_{i=1}\sval_{b_i}
        }
        {
        \add(\lsval) = \bsval'
        } 
        \hspace{1cm}
        \inferrule*[lab = \textsc{Sval-sum-r}]
        {
        \sval = If(\bsval, \sval_{l_1}, \sval_{l_2}) \\\\
        \sval'_{b_1} = \add(\sval_{l_1}) \quad 
        \sval'_{b_2} = \add(\sval_{l_2}) \\\\ 
        \bsval' = If(\bsval, \sval'_{b_1}, \sval'_{b_2})
        }
        {
        \add(\lsval) = \bsval'
        } 
    \end{array}
    \]
% \end{figure}
These rules compute the sum of a list symbolic value. The sum of a list is the sum of each element in the list. The sum of a conditional symbolic value, $if(\sval_b, \sval_{l_1}, \sval_{l_2})$ is either the sum of $\sval_{l_1}$ or the sum of $\sval_{l_2}$ depending on the condition $\sval_b$.\\

\noindent
% \begin{figure}[H]
    \fbox{$\dott(\sval_{l_1}, \sval_{l_2})$}
    \[
    \begin{array}{c}
        \inferrule*[lab = \textsc{Sval-dot-b}]
        {
        \sval_{l_1} = [\sval_{b_1}, \sval_{b_2}, \cdots \sval_{b_n}] \\\\
        \sval_{l_2} = [\sval'_{b_1}, \sval'_{b_2}, \cdots \sval'_{b_n}] \\\\
        \bsval' = (\sum^{\min(n, m)}_{i=1}\sval_{b_i} * \sval'_{b_i})
        }
        {
        \dott(\sval_{l_1}, \sval_{l_2}) = \bsval'
        } 
        \hspace{1cm}
        \inferrule*[lab = \textsc{Sval-dot-r1}]
        {
        \sval_{l_1} = [\sval_{b_1}, \sval_{b_2}, \cdots \sval_{b_n}] \\\\
        \sval_{l_2} = If(\bsval, \sval'_{l_2}, \sval''_{l_2}) \\\\
        % lu_2 = If(bu, lu'_2, lu'_3) \\\\
        \bsval' = \dott(\sval_{l_1}, \sval'_{l_2}) \\\\
        \bsval'' = \dott(\sval_{l_1}, \sval''_{l_2}) \\\\
        % bu'_3 = \dott(lu_1, lu'_3) \\\\ 
        \bsval''' = If(\bsval, \bsval', \bsval'')
        }
        {
        \dott(\sval_{l_1}, \sval_{l_2}) = \bsval'''
        } 
        \hspace{1cm}
        \inferrule*[lab = \textsc{Sval-dot-r2}]
        {
        \sval_{l_1} = If(\bsval, \lsval', \lsval'') \\\\
        % lu_1 = If(bu, lu', lu'') \\\\
        % lu_2 = If(bu_1, lu', lu'') \\\\
        \bsval' = \dott(\lsval', \sval_{l_2}) \\\\
        \bsval'' = \dott(\lsval'', \sval_{l_2}) \\\\ 
        \bsval''' = If(\bsval, \bsval', \bsval'')
        }
        {
        \dott(\sval_{l_1}, \sval_{l_2}) = \bsval'''
        } 
    \end{array}
    \]
% \end{figure}
These rules compute the dot product of two lists. Similar to $\add$, this operation is defined recursively. However, since there are two lists, the recursion needs to be done on both.\\

\noindent
% \begin{figure}[H]
    \fbox{$\map(\bsval, f_c, \scontext)$}
    \[
    \begin{array}{c}
        \inferrule*[lab = \textsc{Sval-map-poly}]
        {
        \bsval = \sval_{b_0} + \sum_{i=0}^l \sval_{b_i} \cdot n_i \\\\
        \cc_0 = \cc \quad 
        \fstore(f_c) = (\var_1, \var_2), \expr \\\\
        \forall i \in [l], \ 
        \sstore_i = \sstore[\var_1 \mapsto \sval_{b_i}, \var_2 \mapsto n_i] \\\\
        \langle \expr, \fstore, \sstore_i, \hh_S, \cc_{i-1} \rangle \downarrow \sval'_{b_i}, \cc_i \quad 
        \bsval' = \sval_{b_0} + \sum_{i=0}^l \sval'_{b_i}
        }
        {
        \map(\bsval, f_c, \symc) = \bsval', \cc_l 
        } 
        \hspace{1cm}
        \inferrule*[lab = \textsc{Sval-map-sym}]
        {
        \bsval = \sval_{b_0} + \sum_{i=0}^l \sval_{b_i} \cdot \epsilon_i \\\\
        \cc_0 = \cc \quad 
        \fstore(f_c) = (\var_1, \var_2), \expr \\\\
        \forall i \in [l], \ 
        \sstore_i = \sstore[\var_1 \mapsto \sval_{b_i}, \var_2 \mapsto \epsilon_i] \\\\
        \langle \expr, \fstore, \sstore_i, \hh_S, \cc_{i-1} \rangle \downarrow \sval'_{b_i}, \cc_i \quad 
        \bsval' = \sval_{b_0} + \sum_{i=0}^l \sval'_{b_i}
        }
        {
        \map(\bsval, f_c, \scontext) = \bsval', \cc_l 
        } 
        \\\\
        \inferrule*[lab = \textsc{Sval-map-r}]
        {
        \bsval = If(\sval_{b_1}, \sval_{b_2}, \sval_{b_3}) \\\\
        % bu = If(bu_1, bu_2, bu_3) \\\\
        \map(\sval_{b_2}, f_c, \symc) = \sval'_{b_2}, \cc_2 \\\\
        \map(\sval_{b_3}, f_c, \fstore, \sstore, \hh_S, \cc_2) = \sval'_{b_3}, \cc_3 \\\\
        \bsval' = If(\sval_{b_1}, \sval'_{b_2}, \sval'_{b_3})
        }
        {
        \map(\bsval, f_c, \symc) = \bsval', \cc_3 
        } 
    \end{array}
    \]
% \end{figure}

The $\map$ function applies $f_c$ to each pair of neuron and coefficient in a polyhedral expression or symbolic variable and coefficient in a symbolic expression and then returns the sum of the outputs of each application of $f_c$. This has to be defined recursively because the top-level symbolic value might be a conditional symbolic value.\\

\noindent
% \begin{figure}[H]
    \fbox{$\mapl(\lsval, f_c, \symc)$}
    \[
    \begin{array}{c}
        \inferrule*[lab = \textsc{Sval-maplist-b}]
        {
        \lsval = [\sval_{b_1}, \cdots, \sval_{b_n}] \\\\
        \fstore(f_c) = (\var_1), \expr \\\\
        \cc_0 = \cc \quad 
        \forall i \in [l], \ 
        \sstore_i = \sstore[\var_1 \mapsto \sval_{b_i}] \\\\
        \langle \expr, \fstore, \sstore_i, \hh_S, \cc_{i-1} \rangle \downarrow \sval'_{b_i}, \cc_i \quad 
        \lsval' = [\sval'_{b_1}, \cdots, \sval'_{b_n}]
        }
        {
        \mapl(\lsval, f_c, \symc) = \lsval', \cc_n 
        } 
        \hspace{1cm}
        \inferrule*[lab = \textsc{Sval-maplist-r}]
        {
        \lsval = If(\bsval, \sval_{l_1}, \sval_{l_2}) \\\\
        \mapl(\sval_{l_1}, f_c, \symc) = \sval'_{l_1}, \cc_1 \\\\
        \mapl(\sval_{l_2}, f_c, \fstore, \sstore, \hh_S, \cc_1) = \sval'_{l_2}, \cc_2 \\\\ 
        \lsval' = If(\bsval, \sval'_{l_1}, \sval'_{l_2})
        }
        {
        \mapl(\lsval, f_c, \symc) = \lsval', \cc_2
        } 
    \end{array}
    \]
% \end{figure}

The $\mapl$ function applies $f_c$ to each element in a list, and outputs a list of the outputs of each application of $f_c$. Similar to $\map$ above, this function also has to be defined recursively.\\

\noindent
% \begin{figure}[H]
    \fbox{$\maxx(\bsval', \bsval'')$}
    \[
    \begin{array}{c}
        \inferrule*[lab = \textsc{Sval-max}]
        {
        \bsval = If(\bsval' \geq \bsval'', \bsval', \bsval'')
        }
        {
        \maxx(\bsval', \bsval'') = \bsval 
        } 
        \hspace{1cm}
        \inferrule*[lab = \textsc{Sval-min}]
        {
        \bsval = If(\sval_{b_1} \leq \sval_{b_2}, \sval_{b_1}, \sval_{b_2})
        }
        {
        \minn(\sval_{b_1}, \sval_{b_1}) = \bsval 
        } 
    \end{array}
    \]
% \end{figure}

\noindent
% \begin{figure}[H]
    \fbox{$\maxx(\lsval)$}
    \[
    \begin{array}{c}
        \inferrule*[lab = \textsc{Sval-max-emp}]
        {
        \lsval = []
        }
        {
        \maxx(\lsval) = 0 
        } 
        \hspace{1cm}
        \inferrule*[lab = \textsc{Sval-max-b-non-emp-1}]
        {
        \lsval = [\bsval]
        }
        {
        \maxx(\lsval) = \bsval
        } 
        \\\\
        \inferrule*[lab = \textsc{Sval-max-b-non-emp-r}]
        {
        \lsval = \bsval::\lsval' \\\\ 
        \bsval' = \maxx(\lsval')
        }
        {
        \maxx(\lsval) = If(\bsval \geq \bsval', \bsval, \bsval')
        } 
        \hspace{1cm} 
        \inferrule*[lab = \textsc{Sval-max-r}]
        {
        \lsval = If(\bsval, \lsval', \lsval'') \\\\ 
        \bsval' = \maxx(\lsval') \quad 
        \bsval'' = \maxx(\lsval'') 
        }
        {
        \maxx(\lsval) = If(\bsval, \bsval', \bsval'')
        } 
    \end{array}
    \]
% \end{figure}

The functions $\maxx$ and $\minn$ are overloaded since they are defined on two symbolic variables and defined on a list of symbolic variables. In the former case, they output a conditional value representing the maximum or minimum of the two input symbolic values. In the latter case, they output a conditional value representing the maximum or minimum element in the list of symbolic values. Again, for the list operation, the function must be defined recursively. Here, it is valid to compare elements with $\ge$ and $\le$ because these operations are defined symbolically for the types of basic symbolic values stated above.

\noindent
% \begin{figure}[H]
    \fbox{$\argmax(\lsval, f_c, \symc)$}
    \[
    \begin{array}{c}
        \inferrule*[lab = \textsc{Sval-compare}]
        {
        \lsval = [\sval_{b_1}, \cdots \sval_{b_n}] \quad 
        \fstore[f_c] = (\var_1, \var_2), \expr \quad 
        \cc_0 = \cc \\\\
        \forall i \in [n], \sstore_i = \sstore[\var_1 \mapsto \bsval, \var_2 \mapsto \sval_{b_i}] \\\\ 
        \langle \expr, \fstore, \sstore_i, \hh_C, \cc_{i-1} \rangle \downarrow \sval'_{b_i}, \cc_i \\\\
        \bsval' = \bigwedge^{n}_{i=1}\sval'_{b_i} \quad
        \lsval'' = If(\bsval', \bsval::\lsval', \lsval')}
        {
        \argmax(\bsval, \lsval, \lsval', f_c, \symc) = \lsval'', \cc_n
        } 
        \\\\
        \inferrule*[lab = \textsc{Sval-compare-emp}]
        {
        \sval_{l_1} = []
        }
        {
        \argmax(\sval_{l_1}, \sval_{l_2}, \sval_{l_3}, f_c, \symc) = \sval_{l_3}, \cc
        } 
        % \hspace{1cm}
        \inferrule*[lab = \textsc{Sval-compare-non-emp}]
        {
        \sval_{l_1} = \bsval::\lsval \\\\
        \
        \sval_{l_4}, \cc' = \argmax(\bsval, \val_{l_2}, \sval_{l_3}, \symc) \\\\
        \sval_{l_5}, \cc'' = \argmax(\lsval, \val_{l_2}, \sval_{l_4}, f_c, \fstore, \sstore, \hh_S, \cc')
        }
        {
        \argmax(\sval_{l_1}, \sval_{l_2}, \sval_{l_3}, f_c, \scontext) = \sval_{l_5}, \cc''
        } 
        \\\\
        \inferrule*[lab = \textsc{Sval-compare-emp}]
        {
        \lsval = []
        }
        {
        \argmax(\lsval, f_c, \symc) = \lsval, \cc 
        } 
        \hspace{0.5cm}
        \inferrule*[lab = \textsc{Sval-compare-non-emp}]
        {
        \lsval = [\bsval]
        }
        {
        \argmax(\lsval, f_c, \symc) = \lsval, \cc 
        }
        \\\\
        \inferrule*[lab = \textsc{Sval-compare-non-emp-r}]
        {
        \lsval = [\sval_{b_1}, \cdots \sval_{b_n}] \\\\ 
        \lsval', \cc' = \argmax(\lsval, \lsval, [], f_c, \symc)
        }
        {
        \argmax(\lsval, f_c, \symc) = \lsval', \cc'
        } 
        \hspace{1cm}
        \inferrule*[lab = \textsc{Sval-compare-if}]
        {
        \lsval = If(\bsval, \sval_{l_1}, \sval_{l_2}) \\\\ 
        \sval'_{l_1}, \cc' = \argmax(\sval_{l_1}, f_c, \symc) \\\\ 
        \sval'_{l_2}, \cc'' = \argmax(\sval_{l_2}, f_c, \fstore, \sstore, \hh_S, \cc') \\\\
        \lsval' = If(\bsval, \sval'_{l_1}, \sval'_{l_2})
        }
        {
        \argmax(\lsval, f_c, \symc) = \lsval', \cc''
        } 
        \\\\
    \end{array}
    \]
% \end{figure}

These rules are similar to the rules for $\argmax(\lval, f_c, \context)$. The difference is that these rules are defined on symbolic values using symbolic semantics instead of concrete values and concrete semantics. The main change is that instead of computing the maximum element according to $f_c$, the output is a nested conditional symbolic value where every combination of maximum elements appears. For example, if the list were $[\sval_1, \sval_2, \sval_3]$, the output of $\argmax(\lsval, f_c, \symc)$ is $if(f_c(\sval_1, \sval_2) \land f_c(\sval_1, \sval_3), \sval_1 :: (if(f_c(\sval_1, \sval_2) \land f_c(\sval_1, \sval_3), \sval_2 :: \cdots, \cdots ), \cdots)$.
\clearpage
\section{Symbolic Semantics of \oldtool}
\label{appendix:symsemantics}
\subsection{Symbolic DNN Expansion}
$\ex(\expr, \tau_s, \symc, \prop)$ outputs a new symbolic DNN, where the shape members/metadata that are of type $\polyexp$ or $\symexp$ that are represented by one variable are expanded. Expanding a $\polyexp$ variable entails adding neurons to $\hh_S$, adding the shape constraint applied to those neurons to $\cc$, and declaring real-valued symbolic variables to represent the constants in the $\polyexp$. Expanding a $\symexp$ variable entails declaring real-valued symbolic variables for the constants in the expanded symbolic expression and adding a constraint for each concrete symbolic variable.

\noindent
% \begin{figure}[H]
    \fbox{$\ex(\expr, \tau_s, \symc, \prop) = \hh_S', \cc'$}
    \[
    \begin{array}{c}
        \inferrule*[lab = \textsc{Add-metadata-elem}]
        {
        \vdash m : \types' \quad 
        \types = \reduced(\types') \\\\
        \sval = \sval^\types_{new} \quad
        \hh_S' = \hh_S[n[m] \mapsto \sval]
        }
        {
        \addg(n, m, \hh_S) = \hh_S'
        }
        \hspace{1cm}
        \inferrule*[lab = \textsc{Add-shape-elem}]
        {
        \tau_s(x) = \types' \quad
        \reduced(\types') = \types \\\\
        \sval = \sval^\types_{new} \quad
        \hh_S' = \hh_S[n[x] \mapsto \sval]
        }
        {
        \addg(n, x, \tau_s, \hh_S) = \hh_S'
        }
        \\\\
        \inferrule*[lab = \textsc{Add-neuron-b}]
        {
        n \in \hh_S
        }
        {
        \addg(n, \tau_s, \hh_S, \prop, \cc) = \hh_S, \cc 
        }
        \hspace{1cm}
        \inferrule*[lab = \textsc{Add-neuron-r}]
        {
        n \not\in \hh_S \\\\
        \textsf{Metadata}[m_1, \cdots, m_k] \quad 
        \hh'_{S_0} = \hh_S \\\\
        \forall i \in [k], \hh_{S_i}' = \addg(n, m, \hh'_{S_{i-1}}) \\\\
        \textsf{Shape}[\var_1, \cdots, \var_l] \quad 
        \hh''_{S_0} = \hh'_{S_k} \\\\
        \forall i \in [l], \hh_{S_i}'' = \addg(n, \var_i, \tau_s, \hh''_{S_{i-1}})
        }
        {
        \addg(n, \tau_s, \hh_S, \prop, \cc) = \hh_{S_l}'', \cc \wedge \prop((\hh''_{S_l}(n[\var_1]), \cdots), n)
        }
        \\\\
        \inferrule*[lab = \textsc{Expand-poly-b}]
        {
        n[\var] = \sval_{b_0} + \sum \sval_{b_i} * n_i 
        }
        {
        \exn(n, \var, \tau_s, \hh_S, \prop) = \hh_S, \true 
        }
        \hspace{1cm}
        \inferrule*[lab = \textsc{Expand-sym-b}]
        {
        n[\var] = \sval_{b_0} + \sum \sval_{b_i} * \epsilon_i 
        }
        {
        \exn(n, \var, \tau_s, \hh_S, \prop) = \hh_S, \true
        }
        \\\\
        \inferrule*[lab = \textsc{Expand-poly-r}]
        {
        \tau_s(\var) = \polyexp \quad
        \hh_S{n[\var]} = \sval_{b_r} \\\\
        \nn = [n'_1, \cdots n'_j] \quad
        \hh_{S_0} = \hh_S \\\\
        \forall i \in [j], \hh_{S_i}, \cc_i = \addg(n'_i, \tau_s, \hh_{S_{i-1}}, \prop, \cc_{i-1}) \\\\
        \bsval = \sval_{r_0} + \sum^j_{i=1} \sval_{r_i} * n'_i \\\\
        \hh_S' = \hh_{S_j}[n[\var] \mapsto \bsval]
        }
        {
        \exn(n, \var, \tau_s, \hh_S, \prop, \cc_0) = \hh'_S, \cc_j
        }
        \hspace{1cm}
        \inferrule*[lab = \textsc{Expand-sym-r}]
        {
        \tau_s(\var) = \symexp \quad
        \hh_S{n[\var]} = \sval_{b_r} \\\\
        \ee = [\mu'_1, \cdots \mu'_j] \\\\
        \forall i \in [j], \cc_i = (-1 \leq \mu'_i \leq 1) \wedge \cc_{i-1} \\\\
        \bsval = \sval_{0} + \sum^j_{i=1} \sval_{i} * n'_i \\\\
        \hh_S' = \hh_{S_j}[n[\var] \mapsto \bsval]
        }
        {
        \exn(n, \var, \tau_s, \hh_S, \prop, \cc_0) = \hh'_S, \cc_j
        }
        % \hspace{1cm}
        % \inferrule*[lab = EXPAND-SYM-R]
        % {
        % \tau_s(\var) = \typezono \quad
        % n[\var] = \sval_{b_r} \\\\
        % \ee = [\epsilon'_1, \cdots \epsilon'_j] \\\\
        % % \forall i \in [j], \hh_{S_i} = \add(n'_i, \tau_s, \hh_{S_{i-1}}) \\\\
        % \bsval = \sval_{r_0} + \sum^j_{i=1} \sval_{r_i} * \epsilon'_i \\\\
        % \hh_S' = \hh_S[n[\var] \mapsto \bsval] \\\\
        % \forall i \in [j], \cc_i = -1 \leq \epsilon'_i \leq 1
        % }
        % {
        % \exn(n, \var, \tau_s, \hh_S, \prop) = \hh'_S, \bigwedge^j_{i=1}\cc_i
        % }
        \\\\
        \inferrule*[lab = \textsc{E-const}]
        { }
        {\ex(\constant, \tau_s, \symc, \prop) = \hh_S, \cc}
        \hspace{1cm}
        \inferrule*[lab = \textsc{E-binary}]
        {
        \ex(\expr_1, \tau_s, \symc, \prop) = \hh_S', \cc' \\\\
        \ex(\expr_2, \tau_s, \fstore, \sstore, \hh_S', \cc', \prop) = \hh_S'', \cc''
        }
        {
        \ex(\expr_1 \oplus \expr_2, \tau_s, \scontext, \prop) = \hh_S'', \cc''
        }
        \\\\
    %     \end{array}
    % \]
    % \[
    %     \begin{array}{c}
        \inferrule*[lab = \textsc{E-shape-b}]
        {
        \ex(e,\fstore, \sstore, \hh_S,\cc) = \hh'_S, \cc'
        \\\\
        \langle \expr, \fstore, \sstore, \hh'_s, \cc' \rangle \downarrow n, \_ \\\\
        \exn(n, \var, \tau_s, \hh'_S, \prop) = \hh''_S, \cc''
        }
        {
        \ex(\expr[x], \tau_s, \symc, \prop) = \hh''_S, \cc' \wedge \cc'' 
        }
        % \\\\
        % \hspace{1cm}
        \inferrule*[lab = \textsc{E-shape-r}]
        {
        \langle \expr, \symc \rangle \downarrow [n_1, \cdots n_q] \quad
        \hh_{S_0} = \hh_S \\\\
        \forall i \in [q], \exn(n, \var, \tau_s, \hh_{S_{i-1}}, \prop) = \hh_{S_i}, \cc_i
        }
        {
        \ex(\expr[x], \tau_s, \scontext, \prop) = \hh_{S_q}, \bigwedge^q_{i=1} \cc_i
            }
        % \\\\
        % \inferrule*[lab = E-VAR]
        % % { \ex(\var, \tau_s, \fstore, \sstore, \hh_C, \cc, \prop) = \hh_S', \cc'}
        % { }
        % {\ex(\var, \tau_s, \fstore, \sstore, \hh_S, \cc, \prop) = \hh_S', \cc'}

    \end{array}
    \]
    % \caption{Caption}
    % \label{fig:my_label}
% \end{figure}

For the symbolic DNN expansion rules for $\expr$, the symbolic DNN has to be recursively expanded for each expression within $\expr$. This can be seen in G-SOLVER. The rest of the rules presented here are rules that require more than only recursively applying the DNN expansion procedure. For \map, the $\ex$ function has to be called on the input expression to \map. Below, the function $\foo$ is used because the symbolic semantics can output a conditional value. Since $\sstore$ only contains expanded values, for any function call expression, $\ex$ must be called on each argument expression. In \traverse, the DNN has to be expanded to add new neurons because the output of \traverse is a polyhedral expression with fresh variables. 

% \noindent
\begin{figure}[H]
\raggedright
    \fbox{$\tau_s, \symc, \prop \models \expr \leadsto \hh_S', \cc'$}
    \resizebox{\textwidth}{!}{
    $
        \begin{array}{c}
        \\\\
            \inferrule*[lab = \textsc{G-map-poly}]
            {
            \sval = \sval_{b_0} + \sum^j_{i=1} n'_i * \sval_{b_i} \\\\
            \fstore(f_c) = (\var_1, \var_2), \expr \\\\
            \forall i \in [j] \quad \sstore_i = \sstore[\var_1 \mapsto \sval_{b_i}, \var_2 \mapsto n'_i ]\\\\
            \tau_s, \fstore, \sstore_i, \hh_{S_{i-1}}, \cc_{i-1} \models \expr \leadsto \hh_{S_i}, \cc_i
            }
            {
            \foo(\tau_s, \fstore, \sstore, \hh_{S_0}, \cc_0, \prop, f_c, \sval) \leadsto \hh_{S_j}, \cc_j
            }
             \hspace{0.2cm}
            \inferrule*[lab = \textsc{G-map-sym}]
            {
            \sval = \sval_{b_0} + \sum^j_{i=1} \epsilon'_i * \sval_{b_i} \\\\
            \fstore(f_c) = (\var_1, \var_2), \expr \\\\
            \forall i \in [j] \quad \sstore_i = \sstore[\var_1 \mapsto \sval_{b_i}, \var_2 \mapsto \epsilon'_i ]\\\\
            \tau_s, \fstore, \sstore_i, \hh_{S_{i-1}}, \cc_{i-1} \models \expr \leadsto \hh_{S_i}, \cc_i
            }
            {
            \foo(\tau_s, \fstore, \sstore, \hh_{S_0}, \cc_0, \prop, f_c, \sval) \leadsto \hh_{S_j}, \cc_j
            }
            \\\\
            \inferrule*[lab = \textsc{G-func-call}]
            {
            \hh_{S_0} = \hh_S \quad \cc_0 = \cc \\\\
            \forall i \in [n], 
            \tau_s, \fstore, \sstore', \hh_{S_{i-1}}, \cc_{i-1}, \prop \models \expr_i \leadsto \overline{\hh_{S_i}}, \overline{\cc_{i}}  \\\\
            \ex(\expr_i, \tau_s, \fstore, \sstore, \overline{\hh_{S_i}}, \overline{\cc_{i}}) = \hh_{S_i}, \cc_i \\\\
            \hh'_S = \hh_{S_n} \quad \cc'_0 = \cc_n \\\\
            \forall i \in [n], \langle \expr_i, \fstore, \sstore, \hh'_S, \cc'_{i-1} \rangle \downarrow \sval_i, \cc'_i \\\\
            \fstore(f_c) = (\var_1, \cdots, \var_n), \expr \\\\
            \sstore' = \sstore[\var_1 \mapsto \sval_1, \cdots \var_n \mapsto \sval_n] \\\\
            \tau_s, \fstore, \sstore', \hh'_S, \cc_n, \prop \models \expr \leadsto \hh_S'', \cc'' 
            }
            {
            \tau_s, \scontext, \prop \models f_c(\expr_1, \cdots \expr_n) \leadsto \hh''_{S}, \cc''
            }
             \hspace{0.2cm}
             \inferrule*[lab = \textsc{G-map}]
            {
            \tau_s, \scontext, \prop \models \expr \leadsto \hh_{S_0}, \cc_0 \\\\
            \ex(\expr, \tau_s, \sstore, \hh_{S_0}, \cc_0, \prop) = \hh_{S}', \cc' \\\\
            \langle \expr, \fstore, \sstore, \hh'_S, \cc' \rangle \downarrow \sval, \_ \\\\
            \expand(\sval) = \true \\\\
            \foo(\tau_s, \fstore, \sstore, \hh'_S, \cc', \prop, f_c, \sval) = \hh''_S, \cc'''
            }
            {
            \tau_S, \symc, \prop \models \expr\cdot\map(f_c) \leadsto \hh_S'', \cc''' 
            }
            \\\\
            \inferrule*[lab = \textsc{G-solver}]
            {
            \tau_s, \symc, \prop \models \expr_1 \leadsto \hh_S', \cc' \\\\
            \tau_s, \fstore, \sstore, \hh_S', \cc', \prop \models \expr_2 \leadsto \hh_S'', \cc''
            }
            {
            \tau_s, \symc, \prop \models \lp(\lpop, \expr_1, \expr_2) \leadsto \hh_S'', \cc''
            }
             \hspace{0.2cm}
            \inferrule*[lab = \textsc{G-map-r}]
            {
            \sval = If(\sval', \sval_1, \sval_2) \\\\
            \foo(\tau_s, \fstore, \sstore, \hh_{S_0}, \cc_0, \prop, f_c, \sval_1) = \hh_{S_1}, \cc_1 \\\\
            \foo(\tau_s, \fstore, \sstore, \hh_{S_1}, \cc_1, \prop, f_c, \sval_2) = \hh_{S_2}, \cc_2
            }
            {
            \foo(\tau_s, \fstore, \sstore, \hh_{S_0}, \cc_0, \prop, f_c, \sval) \leadsto \hh_{S_2}, \cc_2
            }
            \\\\
            \inferrule*[lab = \textsc{G-traverse}]
            {
            \tau_s, \scontext, \prop \models \expr \leadsto \hh_{S_0}, \cc_0 \\\\ 
            % \cinput(\var \cdot \trav(\dir, f_{c_1}, f_{c_2}, f_{c_3})\{\expr\}, \fstore, \sstore, \hh_S', \m', \cc') = \true \\\\
            \nn = [n'_1, \cdots n'_j] \quad
            % \hh_{S_0}= \hh_S \quad 
            \forall i \in [j], \hh_{S_i}, \cc_i = \addg(n'_i, \tau_s, \hh_{S_{i-1}}, \prop, \cc_{i-1}) \\\\
            \hh''_{S_0} = \hh_{S_j} \quad 
            \cc''_0 = \cc_j \quad 
            \bsval = \sval_{b_0} + \sum^j_{i=1} \sval_{b_i} * n'_i \\\\
            \forall i \in [j], \tau_s, \fstore, \sstore, \hh''_{S_{i-1}}, \cc''_{i-1}, \prop \models f_{c_2}(n'_i, \sval_{b_i}) \leadsto \hh''_{S_i}, \cc''_i \\\\
            \hh'''_{S_0} = \hh''_{S_j} \quad 
            \cc'''_0 = \cc''_j \\\\
            \forall i \in [j], \tau_s, \fstore, \sstore, \hh'''_{S_{i-1}}, \cc'''_{i-1}, \prop \models f_{c_3}(n'_i, \sval_{b_i}) \leadsto \hh'''_{S_i}, \cc'''_i
            }
            {
            \tau_s, \scontext, \prop \models \var\cdot\trav(\dir, f_{c_1}, f_{c_2}, f_{c_3})\{\expr\} \leadsto \hh'''_{S_j}, \cc'''_j
            }
        \end{array}
    $
    }    
\end{figure}
\subsection{Symbolic Semantics for Expressions in \oldtool}
The general form for symbolic semantics is $\langle \expr, \fstore, \sstore, \hh_S, \cc \rangle \downarrow \sval, \cc'$. In the same way as operational semantics, $\fstore$ contains mappings of function names to their arguments and return expressions. $\sstore$ contains mappings from variables to expanded symbolic values. $\hh_S$ contains mappings from the shape members and metadata for each neuron in the symbolic DNN to symbolic values. $\cc$ is a conjunction of the constraints associated with each neuron's abstract shape and any constraints on fresh variables that have been generated by the symbolic semantics. The output of symbolic semantics is a symbolic value and a conjunction of $\cc$ and any new constraints introduced for the symbolic variables in $\sval$. The symbolic semantics for \traverse and \lp are where additional constraints are introduced. In \traverse, the symbolic semantics verifies the user-defined inductive invariant and then outputs a polyhedral expression with fresh variables and adds the constraint that this polyhedral expression satisfies the inductive invariant. For $\lp(\minimize, \expr_1, \expr_2)$, the symbolic semantics outputs a fresh real valued symbolic variable, $\sval$, and adds the constraint that when $\expr_2$ is true, $\sval \le \expr_1$. \\

% \begin{figure}[H]
    \fbox{$\langle \expr, \symc \rangle \downarrow \sval, \cc'$}\\
    \[
    \begin{array}{c}
        \inferrule*[lab = \textsc{Sym-const}]
        { }
        {\langle \constant, \symc \rangle \downarrow \constant, \cc}
        \hspace{1cm}
        \inferrule*[lab = \textsc{Sym-var}]
        { }
        {\langle \var, \symc \rangle \downarrow \sstore(\var), \cc}
        % \hspace{1cm}
        \\\\
        \inferrule*[lab = \textsc{Sym-noise}]
        { }
        {
        \langle \noise, \symc \rangle \downarrow \noise_{new}, \cc \wedge (-1 \leq \epsilon_{new} \leq 1)
        }
        % \hspace{1cm}
        % \inferrule*[lab = SYM-UNARY]
        % {\langle \expr, \symc \rangle \downarrow \sval} 
        % {\langle \sim\expr, \symc \rangle \downarrow \ \sim\sval}
        \\\\
        \inferrule*[lab = \textsc{Sym-binary}]
        {
        \langle \expr_1, \symc \rangle \downarrow \sval_1, \cc_1 \\\\
        \langle \expr_2, \fstore, \sstore, \hh_S, \cc_1 \rangle \downarrow \sval_2, \cc_2
        } 
        {\langle \expr_1 \oplus \expr_2, \symc \rangle \downarrow \sval_1 \binop \sval_2, \cc_2}
        \hspace{1cm}
        \inferrule*[lab = \textsc{Sym-ternary}]
        {
        \langle \expr_1, \symc \rangle \downarrow \sval_{b_1}, \cc_1 \\\\
        \langle \expr_2, \fstore, \sstore, \hh_S, \cc_1 \rangle \downarrow \sval_2, \cc_2 \\\\
        \langle \expr_3, \fstore, \sstore, \hh_S, \cc_2 \rangle \downarrow \sval_3, \cc_3 \\\\
        \sval = If(\sval_{b_1}, \sval_2, \sval_3)
        % \sval_1 = \true \implies \sval = \sval_2 \quad
        % \sval_1 = \false \implies \sval = \sval_3 
        } 
        {\langle (\expr_1 ? \expr_2 : \expr_3), \symc \rangle \downarrow \sval, \cc_3}
        \\\\
        \inferrule*[lab = \textsc{Sym-metadata}]
        {
        \langle \expr, \symc \rangle \downarrow n, \cc' \\\\
        \sval = \hh_S[n[m]] 
        } 
        {\langle \expr[m], \symc \rangle \downarrow \sval, \cc'}
        \hspace{1cm}
        % \\\\
        \inferrule*[lab = \textsc{Sym-shape}]
        {
        \langle \expr. \symc \rangle \downarrow n, \cc' \\\\
        \sval = \hh_S[n[x]] 
        } 
        {\langle \expr[x], \symc \rangle \downarrow \sval, \cc'}
        % \hspace{1cm}
        \\\\
        \inferrule*[lab = \textsc{Sym-max}]
        {
        \langle \expr, \symc \rangle \downarrow \sval, \cc'
        } 
        {\langle \maxx(\expr), \symc \rangle \downarrow \maxx(\sval), \cc'}
        \hspace{1cm}
        \inferrule*[lab = \textsc{Sym-compare}]
        {
        \langle \expr, \symc \rangle \downarrow \sval, \cc' \\\\
        \sval', \cc'' = \argmax(\sval, f_c, \fstore, \sstore, \hh_S, \cc')
        } 
        {\langle \argmax(\expr, f_c), \symc \rangle \downarrow \sval', \cc''}
        \\\\
        \inferrule*[lab = \textsc{Sym-func-call}]
        {
        \cc_0 = \cc \\\\ 
        \forall i \in [n], \ \langle \expr, \fstore, \sstore, \hh_S, \cc_{i-1} \rangle \downarrow \sval_i, \cc_i' \\\\
        \fstore(f_c) = (\var_1, \cdots, \var_n), \expr \\\\
        \sstore' = \sstore[\var_1 \mapsto \sval_1, \cdots \var_n \mapsto \sval_n] \\\\ 
        \langle \expr, \fstore, \sstore', \hh_S, \cc_n \rangle \downarrow \sval, \cc'
        }
        {
        \langle f_c(\expr_1, \cdots, \expr_n), \symc \rangle  \downarrow \sval, \cc'
        }
        \hspace{1cm}
        \inferrule*[lab = \textsc{Sym-map}]
        {
        \langle \expr, \symc \rangle \downarrow \bsval', \cc' \\\\
        \bsval, \cc'' = \map(\bsval', \symc')
        % \sval' = \constant_0 + \sum_{i=0}^{i=l} \constant_i \cdot \ver_i \\\\
        % \forall i \in [l], \ \langle f_c(\ver_i, \constant_i), \symc \rangle  \downarrow \sval_i \quad 
        % \sval = \constant_0 + \sum_{i=0}^{i=l} \sval_i
        }
        {
        \langle \expr.\map(f_c), \symc \rangle \downarrow \bsval, \cc''
        }
        % \\\\
        % \inferrule*[lab = INV-INPUT]
        % {
        % \langle \expr, \symc \rangle \downarrow \sval \\\\
        % \bsval = \unsat(\neg(\cc \implies \sval))
        % }
        % {
        % \cinput(\expr, \fstore, \sstore, \hh_S, \m, \cc) = \bsval
        % }
        % \hspace{1cm} 
        \end{array}
    \]
    \[
        \begin{array}{c}
        \inferrule*[lab = \textsc{Chech-induction}]
        {
        \nn = [n'_1, \cdots, n'_j] \quad
        \bsval = \sval_{r_0} + \sum^j_{i=1} \sval_{r_i} * n'_i \\\\
        \sstore' = \sstore[\var \mapsto \bsval] \quad 
        \langle \expr, \fstore, \sstore', \hh_S, \cc \rangle \downarrow \bsval', \cc_0 \\\\
        % \cc_0 = \cc \quad
        \forall i \in [j], \quad
        \langle f_{c_2}(n_i, \sval_{r_i}), \fstore, \sstore', \hh_S, \cc_{i-1} \rangle \downarrow \sval'_i, \cc_i \\\\
        \cc'_0 = \cc_j\\\\
        \forall i \in [j],         \langle f_{c_3}(n_i, \sval_{r_i}), \fstore, \sstore', \hh_{S}, \cc'_{i-1} \rangle \downarrow \sval''_i, \cc'_i \\\\
        \sval'' = \sval_{r_0} + \sum^j_{i=1} If(\sval'_i, \sval''_i, \sval_{r_i}*n_i) \quad 
        \sstore'' = \sstore[\var \mapsto \sval''] \\\\
        \langle \expr, \fstore, \sstore'', \hh_S, \cc'_j \rangle \downarrow \sval''', \cc'' \quad 
        \bsval'' = \unsat(\neg(\cc_0 \wedge \bsval' \implies \cc''_j \wedge \sval'''))
        }
        {
        \cinduction(\var \cdot \trav(\dir, f_{c_1}, f_{c_2}, f_{c_3})\{\expr\}, \fstore, \sstore, \hh_S, \cc) = \bsval''
        }
         \\\\
        \inferrule*[lab = \textsc{Chech-invariant}]
        {
        \langle \expr, \symc \rangle \downarrow \sval, \cc' \\\\
        \bsval = \unsat(\neg(\cc' \implies \sval))
        \\\\
        % \bsval' = \cinput(\expr, \fstore, \sstore, \hh_S, \m, \cc') \\\\
        \bsval' = \cinduction(\var \cdot \trav(\dir, f_{c_1}, f_{c_2}, f_{c_3})\{\expr\}, \fstore, \sstore, \hh_S, \cc)
        }
        {
        \cinvariant(\var \cdot \trav(\dir, f_{c_1}, f_{c_2}, f_{c_3})\{\expr\}, \fstore, \sstore, \hh_S, \cc) = \bsval \wedge \bsval', \cc'
        }
        \\\\
        \inferrule*[lab = \textsc{Sym-traverse}]
        {
        \cinvariant(\var\cdot\traverse(\dir, f_{c_1}, f_{c_2}, f_{c_3})\{\expr\}, \symc) = \true, \cc' \\\\
        \bsval = \sval_{0} + \sum^j_{i=1} \sval_{i} * \sval'_i \quad 
        \sstore' = \sstore[x \mapsto \bsval] \quad
        \langle \expr, \fstore, \sstore', \hh_S, \cc' \rangle \downarrow \sval, \cc''  
        }
        {
        \langle \var\cdot\traverse(\dir, f_{c_1}, f_{c_2}, f_{c_3})\{\expr\}, \symc \rangle \downarrow \bsval, \sval \wedge \cc''
        }
        % \hspace{1cm}
        \\\\
        \inferrule*[lab = \textsc{Sym-solver-min}]
        {
        \langle \expr_1, \symc \rangle \downarrow \sval_1, \cc_1 \\\\
        \langle \expr_2, \symc_1 \rangle \downarrow \sval_2, \cc_2 \\\\
        \sat(\cc_2 \wedge \sval_2) \\\\
        \bsval = \sval_r \quad 
        \cc' = \cc_2 \wedge (\sval_2 \implies \bsval \leq \sval_1)
        }
        {
        \langle \lp(\minimize, \expr_1, \expr_2), \symc \rangle \downarrow \bsval, \cc'
        }
        \\\\
        % \hspace{1cm}
        \inferrule*[lab = \textsc{Sym-solver-max}]
        {
        \langle \expr_1, \symc \rangle \downarrow \sval_1, \cc_1 \\\\
        \langle \expr_2, \symc_1 \rangle \downarrow \sval_2, \cc_2 \\\\
        \sat(\cc_2 \wedge \sval_2) \\\\
        \bsval = \sval_r \quad 
        \cc' = \cc_2 \wedge (\sval_2 \implies \bsval \geq \sval_1)
        }
        {
        \langle \lp(\maximize, \expr_1, \expr_2), \symc \rangle \downarrow \bsval, \cc'
        }
    \end{array}
    \]
% \end{figure}

\clearpage
\section{Definitions for Over-Approximation}
\label{appendix:definitions}
% \begin{figure}[H]
    \fbox{$\hh_C \prec_{\cc} \hh_S$}
    \[
    \begin{array}{c}
        \inferrule*[lab={\textsc{Over-approx-dnn}}]
        {
        \dom(\hh_S) \subseteq \dom(\hh_S) \\\\
        X = \cs(\hh_S, \cc) \quad 
        Y = \ns(\hh_S) \cup \es(\hh_S) \\\\
        Z = \ps(\hh_S) \cup \zs(\hh_S) \cup \cts(\hh_S) \\\\
        % Y = \vars(\hh_S) - \vars(\hh_C) \\\\
        \exists X \forall Y \exists Z (\cc \wedge \bigwedge_{t \in \dom(\hh_S)}\hh_S(t) = \hh_C(t))
        }
        {
        \hh_C \prec_{\cc} \hh_S
        }  
    \end{array}
    \]
% \end{figure}
\noindent
% \begin{figure}[H]
    \fbox{$\val, \hh_C \prec_{\cc} \sval, \hh_S$}
    \[
    \begin{array}{c}
        \inferrule*[lab={\textsc{Base-val-constraints}}]
        {
        \cc = (\bval = \bsval)
        }
        {
        CS(\bval, \bsval) = \cc
        } 
        \\\\
        \inferrule*[lab={\textsc{List-val-constraints-b}}]
        {
        \lval = [\val_{b_1}, \cdots, \val_{b_n}] \\\\
        \lsval = [\sval_{b_1}, \cdots, \sval_{b_n}]
        }
        {
        CS(\lval, \lsval) = \bigwedge^n_{i=1}(\sval_{b_i} = \val_{b_i})
        } 
        \hspace{1cm}
        \inferrule*[lab={\textsc{List-val-constraints-r}}]
        {
        \lval = [\val_{b_1}, \cdots, \val_{b_n}] \\\\
        \lsval = If(\sval_{b_1}, \sval_{l_1}, \sval_{l_2}) \\\\
        \cc_1 = CS(\lval, \sval_{l_1}) \quad 
        \cc_2 = CS(\lval, \sval_{l_2}) 
        }
        {
        CS(\lval, \lsval) = (\sval_{b_1} \implies \cc_1) \wedge (\neg(\sval_{b_1}) \implies \cc_2)
        } 
        \\\\
        \inferrule*[lab={\textsc{Over-approx-val}}]
        {
        \dom(\hh_S) \subseteq \dom(\hh_S) \\\\
        X = \cs(\hh_S, \cc) \quad 
        Y = \ns(\hh_S) \cup \es(\hh_S) \\\\
        Z = \ps(\hh_S) \cup \zs(\hh_S) \\\\
        \exists X \forall Y \exists Z (CS(\val, \sval) \wedge \cc \wedge \bigwedge_{t \in \dom(\hh_S)}\hh_S(t) = \hh_C(t))
        }
        {
        \val, \hh_C \prec_{\cc} \sval, \hh_S
        } 
    \end{array}
    \]
% \end{figure}
\noindent
% \begin{figure}[H]
    \fbox{$[\val], \hh_C \prec_{\cc} [\sval], \hh_S$}
    \[
    \begin{array}{c}
        \inferrule*[lab={\textsc{Over-approx-val-list}}]
        {
        \val \text{ and }\sval \text{ have the same length} \\\\
        \dom(\hh_S) \subseteq \dom(\hh_S) \\\\
        X = \cs(\hh_S, \cc) \quad 
        Y = \ns(\hh_S) \cup \es(\hh_S) \\\\
        Z = \ps(\hh_S) \cup \zs(\hh_S) \\\\
        \exists X \forall Y \exists Z (\bigwedge_{i}CS(\val_i, \sval_i) \wedge \cc \wedge \bigwedge_{t \in \dom(\hh_S)}\hh_S(t) = \hh_C(t))
        }
        {
        [\val], \hh_C \prec_{\cc} [\sval], \hh_S
        } 
    \end{array}
    \]
% \end{figure}
\noindent
% \begin{figure}[H]
    \fbox{$\store, \hh_C \prec_{\cc} \sstore, \hh_S$}
    \[
    \begin{array}{c}
        \inferrule*[lab={\textsc{Over-approx-store}}]
        {
        \dom(\hh_S) \subseteq \dom(\hh_S) \\\\
        \dom(\store) = \dom(\sstore) \\\\
        X = \cs(\hh_S, \cc) \quad 
        Y = \ns(\hh_S) \cup \es(\hh_S) \\\\
        Z = \ps(\hh_S) \cup \zs(\hh_S) \\\\
        % Y = \vars(\hh_S) - \vars(\hh_C) \\\\
        \exists X \forall Y \exists Z (\cc \wedge \bigwedge_{t \in \dom(\hh_S)}\hh_S(t) = \hh_C(t) \wedge \bigwedge_{t \in \dom(\rho)}CS(\rho(t),\rho'(t)))
        }
        {
        \rho, \hh_C \prec_{\cc} \rho', \hh_S
        } 
    \end{array}
    \]
    $\store$ represents the concrete store, which is a mapping from variables to concrete values. $\sstore$ represents the symbolic store, which is a mapping from variables to symbolic values. In our symbolic semantics, we will only add expanded values to $\sstore$. Formally, $\forall t \in \sstore, \expand(\sstore(t))$ \\
    
% \end{figure}
\noindent
% \begin{figure}[H]
    \fbox{$\fstore \sim \Gamma, \tau_s$}
    \fbox{$\tstore \sim \Gamma, \tau_s$}
    \[
    \begin{array}{c}
        \inferrule*[lab={\textsc{Sim-f}}]
        {
        \dom(\fstore) \subseteq \dom(\Gamma) \\\\
        \forall f \in \fstore, \Gamma, \tau_s \vdash \fstore(f) \ : \ \Gamma(f)
        }
        {
        \fstore \sim \Gamma, \tau_s
        } 
        \hspace{1cm}
        \inferrule*[lab={\textsc{Sim-$\tstore$}}]
        {
        \dom(\tstore) \subseteq \dom(\Gamma) \\\\
        \forall \theta \in \tstore, \Gamma, \tau_s \vdash \tstore(\theta) \ : \ \Gamma(\theta)
        }
        {
        \tstore \sim \Gamma, \tau_s
        } 
    \end{array}
    \]
    % \label{fig:my_label}
% \end{figure}
\noindent
% \begin{figure}[H]
    \fbox{$\store \sim \Gamma$}
    \[
    \begin{array}{c}
        \inferrule*[lab={\textsc{Sim-store}}]
        {
        \dom(\store) \subseteq \dom(\Gamma) \\\\
        \forall \var \in \store, \ \cdot \vdash \store(\var) \ : \ \Gamma(\var)
        }
        {
        \store \sim \Gamma
        } 
    \end{array}
    \]
    % \label{fig:my_label}
% \end{figure}
\noindent
% \begin{figure}[H]
    \fbox{$\sstore \sim \Gamma$}
    \[
    \begin{array}{c}
        \inferrule*[lab={\textsc{Sim-sym-store}}]
        {
        \dom(\sstore) \subseteq \dom(\Gamma) \\\\
        \forall \var \in \sstore, \  \vdash \sstore(\var) \ : \ \reduced(\Gamma(\var))
        }
        {
        \sstore \sim \Gamma
        } 
    \end{array}
    \]
    % \label{fig:my_label}
% \end{figure}
\noindent
% \begin{figure}[H]
    \fbox{$\hh_C \sim \tau_s$}
    \fbox{$\hh_S \sim \tau_s$}
    \[
    \begin{array}{c}
        \inferrule*[lab={\textsc{Sim-dnn}}]
        {
        \forall n \in \hh_C \\\\
        \forall \var \in \tau_s, \  \vdash \hh_C(n[\var]) \ : \ \tau_s(\var)
        }
        {
        \hh_C \sim \tau_s
        }
        \hspace{1cm}
%     \end{array}
%     \]
%     % \label{fig:my_label}
% % \end{figure}
% \noindent
% % \begin{figure}[H]
%     \fbox{$\hh_S \sim \tau_s$}
%     \[
%     \begin{array}{c}
        \inferrule*[lab={\textsc{Sim-sym-dnn}}]
        {
        \forall n \in \hh_S \\\\
        \forall \var \in \tau_s, \  \vdash \hh_S(n[\var]) \ : \ \reduced(\tau_s(\var))
        }
        {
        \hh_S \sim \tau_s
        } 
    \end{array}
    \]
    % \label{fig:my_label}
% \end{figure}

\begin{figure}[H]
    \fbox{$\tau_s,\fstore, \sstore, \hh_S, \cc, \prop \leadsto^* 
    \hh'_S, \cc'$}
    \[
    \begin{array}{c}
    \inferrule*[lab={\textsc{Multistep-expand}}]
        {
        \ex(\expr, \tau_s,\fstore, \sstore, \hh_S, \cc, \prop) =
    \hh'_S, \cc'
        }
        {
        \tau_s,\fstore, \sstore, \hh_S, \cc, \prop \leadsto^* 
    \hh'_S, \cc'
        } 
        \hspace{1cm}
        \inferrule*[lab={\textsc{Multistep-step}}]
        {
        \tau_s,\fstore, \sstore, \hh_S, \cc, \prop \models \expr \leadsto 
    \hh'_S, \cc'
        }
        {
        \tau_s,\fstore, \sstore, \hh_S, \cc, \prop \leadsto^* 
    \hh'_S, \cc'
        } 
        \\\\
        \inferrule*[lab={\textsc{Multistep-add}}]
        {
        \addg(n, \tau_s, \hh_S, \prop, \cc) =
    \hh'_S, \cc'
        }
        {
        \tau_s,\fstore, \sstore, \hh_S, \cc, \prop \leadsto^* 
    \hh'_S, \cc'
        } 
        \hspace{1cm}
        \inferrule*[lab={\textsc{Multistep-r}}]
        {
        \tau_s,\fstore, \sstore, \hh_S, \cc, \prop \leadsto^* 
    \hh''_S, \cc''\\\\
        \tau_s,\fstore, \sstore, \hh''_S, \cc'', \prop \leadsto^* 
    \hh'_S, \cc'
        }
        {
        \tau_s,\fstore, \sstore, \hh_S, \cc, \prop \leadsto^* 
    \hh'_S, \cc'
        } 
    \end{array}
    \]
    % \label{fig:my_label}
\end{figure}

\begin{figure}[H]
    \fbox{$\ll \expr, \fstore, \store, \sstore, \hh_C, \hh_S, \cc \gg \updownarrow \val, \sval, \cc', \m'$}
    \[
    \begin{array}{c}
        \inferrule*[lab={\textsc{Bisimulation}}]
        {
        \langle \expr, \fstore, \store, \hh_C \rangle \Downarrow \val \\\\
        \langle \expr, \fstore, \sstore, \hh_S, \cc \rangle \downarrow \sval, \cc' \\\\
        \exists! X \forall Y \exists!Z (CS(\sval,\val) \wedge \cc' \wedge \m \wedge \bigwedge_{t \in \dom(\hh_S)} \hh_S(t) = \hh_C(t))
        }
        {
        \ll \expr, \fstore, \store, \sstore, \hh_C, \hh_S, \cc \gg \updownarrow \val, \sval, \cc', \m
        } 
    \end{array}
    \]
    % \label{fig:my_label}
\end{figure}
\clearpage
\section{lemmas}
\label{appendix:lemmas}
\begin{lemma} 
\label{lemma:uniquez}
If 
\begin{enumerate}
    \item $\exists X \forall Y \exists Z (\bigwedge_{t \in \dom(\hh_S)}\hh_S(t) = \hh_C(t))$
\end{enumerate}
then,
\begin{enumerate}
    \item $\exists X \forall Y \exists! Z (\bigwedge_{t \in \dom(\hh_S)}\hh_S(t) = \hh_C(t))$
\end{enumerate}
\end{lemma}
\begin{proof}[Proof sketch]
This is a stronger claim than what is needed to prove over-approximation. We can make this claim because given a DNN during concrete execution, there will be concrete polyhedral and symbolic expressions representing the shape members. For example, $\hh_C(n[L])$ might be $4n_5 + 2n_3 + 1$. In this case, $\hh_S(n[L]) = \sval_1$, where $\sval_1 \in Z$. For each assignment of $n_5$ and $n_3$, $4n_5 + 2n_3 + 1$ equals a specific real number. Similarly, for all elements of $\hh_C$, $t$, given an assignment to $Y$, $\hh_C(t)$ is a real number. Since $\hh_S(t) = \hh_C(t)$, $\hh_S(t)$ has to equal the same real number. When we create $\hh_S$, we have a single symbolic variable representing each element $\range(\hh_S)$. During graph expansion, no variables are added to $Z$, so every variable of $Z$ is in $\range(\hh_S)$.
% Detailed proof: Proof~\ref{proof:uniquez}.
\end{proof}
% \clearpage
\begin{lemma} 
\label{lemma:uniquex}
If 
\begin{enumerate}
    \item $X \subseteq \cs(\hh_S)$
    \item $\exists X \forall Y \exists Z (\bigwedge_{t \in \dom(\hh_S)}\hh_S(t) = \hh_C(t))$
\end{enumerate}
then,
\begin{enumerate}
    \item $\exists! X \forall Y \exists Z (\bigwedge_{t \in \dom(\hh_S)}\hh_S(t) = \hh_C(t))$
\end{enumerate}
\end{lemma}
\begin{proof}[Proof sketch]
This is a stronger claim than what is needed to prove overapproximation. This lemma is only true when $X$ only contains variables from $\hh_S$. There are two possibilities for each variable in $X$, either it is in the range of $\hh_S$ or it is not. First, we could have a situation where $\hh_S(t) = x_1$ and $\hh_C(t) = 5$. In this case, there is  a unique assignment to $x_1$ so that $x_1 = 5$. Second, we could have a situation where $\hh_S(t) = x_1n_1 + x_2n_2$ and $\hh_C(t) = 2n_1 + 2n_2$. Since we have to satisfy the condition that $\exists x_1 \forall n_1, n_2 (x_1n_1+x_2n_2=2n_1+2n_2)$, there is a unique assignment to $x_1$ and $x_2$. If the $\exists X$ quantifier were to the right of the $\forall Y$ quantifier, this lemma would not be true.
\end{proof}
\begin{lemma} 
\label{lemma:uniquexz}
If 
\begin{enumerate}
    \item $X \subseteq \cs(\hh_S)$
    \item $\exists X \forall Y \exists Z (\bigwedge_{t \in \dom(\hh_S)}\hh_S(t) = \hh_C(t))$
\end{enumerate}
then,
\begin{enumerate}
    \item $\exists! X \forall Y \exists! Z (\bigwedge_{t \in \dom(\hh_S)}\hh_S(t) = \hh_C(t))$
\end{enumerate}
\end{lemma}
\begin{proof}[Proof sketch]
This lemma combines the previous two lemmas to form a stronger claim about how the symbolic graph over-approximates the concrete graph.
\end{proof}
% \begin{lemma}
% \label{appendixlemma:tcexpressions}
% If,
% \begin{enumerate}
%     \item $\Gamma, \tau_s \vdash \expr : \types$
%     \item $\bot \sqsubset \types \sqsubset \top$
%     \item $\fstore \sim \Gamma, \tau_s$
%     \item $\store \sim \Gamma$
%     \item $\hh_C \sim \tau_s$
%     % \item $V$ is the set of Neurons in $\range(\rho)$
% \end{enumerate}
% then,
% \begin{enumerate}
%     \item $\langle \expr, \context \rangle \Downarrow \val$
%     \item $\vdash \val : \types'$
%     \item $\types' \sqsubseteq \types$
% \end{enumerate}
% \end{lemma}
% \begin{proof}[Proof sketch]
% This lemma states that if an expression type-checks, then it will terminate in the operational semantics and evaluate to a value. This value will either be the same type that the expression type checks to, or a subtype. We prove this lemma for all expressions, using induction on the structure of an expression. Once we have proved this lemma for expressions, it is straightforward to prove that a whole program that type-checks will terminate and evaluate to something in the operational semantics. To prove the termination of every expression in \cf, the only non-trivial case is $\traverse$. For this case, we use a ranking function and show that it is bounded by 0 and decreases by at least 1 in each iteration.
% \end{proof}

\begin{lemma}
    \label{lemma:cexpand}
    If
    \begin{enumerate}
        \item $\langle \expr, \symc \rangle \downarrow \sval, \cc'$ 
    \end{enumerate}
    then,
    \begin{enumerate}
        \item $\cc' = \cc \wedge \cc''$
        % \item $\cc' \implies \cc$
    \end{enumerate}
\end{lemma}
\begin{proof}[Proof sketch]
This lemma states that the output conditions from symbolic execution are always a conjunction of the input conditions and another set of conditions. This structure comes from the fact that in the rules for symbolic execution semantics, there are two rules that create new conditions and return a conjunction of them with the input $\cc$, and the rest of the rules, just recursively accumulate the conditions outputted by symbolically executing the sub-expressions. The idea of these conditions is to add a condition whenever a fresh variable is generated for $\traverse$ or $\lp$ calls.
\end{proof}

\begin{lemma}
    \label{lemma:csequential}
    If
    \begin{enumerate} 
        \item $\forall i \in [n], \ll \expr_i, \fstore, \store_i, \sstore, \hh_C, \hh_S, \cc_0 \gg \updownarrow \val_i, \sval_i, \cc'_i, \m'_i$
        \item $\exists!X \forall Y \exists!Z (\cc_0 \wedge \bigwedge_{t \in \dom(\hh_S)} \hh_S(t) = \hh_C(t))$
    \end{enumerate}
    then,
    \begin{enumerate} 
        \item $\forall i \in [n], \langle \expr_i, \fstore, \sstore, \hh_S, \cc_{i-1} \rangle \downarrow \sval_i, \cc_i$
    \end{enumerate}
\end{lemma}

\begin{proof}[Proof sketch]
This lemma states that if expressions, $\expr_1, \cdots, \expr_n$ can all be symbolically evaluated under the conditions $\cc_0$, then they can be symbolically evaluated sequentially, where $\expr_i$ is evaluated under the conditions outputted by $\expr_{i-1}$. For most of the symbolic evaluation rules, the input conditions are simply propagated to the output conditions, so changing the input conditions will not affect whether or not the rules can be applied. In two rules, \textsc{Sym-traverse} and \textsc{Sym-solver}, the input conditions are used. For \textsc{Sym-traverse}, we have to prove that when symbolically evaluating $\expr_i$, if $\cinvariant$ was true when called with input condition $\cc_0$, then it will still be true when called with input condition $\cc_{i-1}$, which is the output condition of symbolically evaluating $\expr_{i-1}$. For \textsc{Sym-solver}, we also have to prove that if the conjunction of the input constraint and the condition outputted by the symbolic evaluation of the input constraint were satisfiable before, then that conjunction is still satisfiable when starting with the condition $\cc_{i-1}$ instead of $\cc_0$. These are proved using Lemma~\ref{lemma:cexpand}.
\end{proof}

\begin{lemma}
    \label{lemma:ctogether}
    If
    \begin{enumerate} 
        \item $\forall i \in [n], \ll \expr_i, \fstore, \store_i, \sstore, \hh_C, \hh_S, \cc_0 \gg \updownarrow \val_i, \sval_i, \cc'_i, \m'_i$
        \item $\exists!X \forall Y \exists!Z (\cc_0 \wedge \bigwedge_{t \in \dom(\hh_S)} \hh_S(t) = \hh_C(t))$
    \end{enumerate}
    then,
    \begin{enumerate}
        \item $\forall i \in [n], \langle \expr_i, \fstore, \sstore_i, \hh_S, \cc_{i-1} \rangle \downarrow \sval_i, \cc_i$
        \item $\exists! X' \forall Y' \exists!Z' (\bigwedge_i CS(\sval_i,\val_i) \wedge \cc_n \wedge \m \wedge \bigwedge_{t \in \dom(\hh_S)} \hh_S(t) = \hh_C(t))$
        \item $\cc_n \implies \cc_0$
    \end{enumerate}
\end{lemma}

\begin{proof}[Proof sketch]
This lemma states that if expressions, $\expr_1, \cdots, \expr_n$ can all be symbolically evaluated under the conditions $\cc_0$, and the bisumlation holds, then they can all be symbolically evaluated sequentially, where $\expr_i$ is evaluated under the conditions outputted by $\expr_{i-1}$, and the bisimulation holds on all the output symbolic values simultaneously, meaning there exists a unique assignment to $X$, for all assignments to $Y$, there exists a unique assignment to $Z$ such that all of the symbolic values equal the concrete values, each element of the range of $\hh_S$ equals the corresponding element in the range of $\hh_C$, the conditions outputted by symbolically evaluating $\expr_n$ holds, and some condition $M$ holds. In this proof, $M$ is set to a conjunction of each $M_i$ from the original bisumulation of each $\expr_i$. This proof uses Lemma~\ref{lemma:csequential} to show that the expressions can be symbolically evaluated sequentially. Then, it uses the fact that the condition outputted from the symbolic evaluation of each expression only adds constraints on fresh variables that are not present in the symbolic evaluation of other expressions. The second antecedent to this lemma implies that the assignment to the shared variables in the symbolic evaluation of $\expr_1, \cdots, \expr_n$ have the same assignments in their individual bisumulation. Because of this, the assignment to $X'$ in the third consequent can be created by combining the assignments to each $X_i$ used in each original bisimulation argument. 
\end{proof}

\begin{lemma}
\label{lemma:hexpand}
If
    \begin{enumerate}
        \item $\store, \hh_C \prec_{\cc} \sstore, \hh_S$
        \item $\ll \expr, \fstore, \store, \sstore, \hh_C, \hh_S, \cc \gg \updownarrow \val, \sval, \cc'_1, \m'_1$
        \item $\ex(\expr, \tau_s, \symc, \prop) = \hh_S', \cc'$
        \item $\exists! X \forall Y \exists! Z (\bigwedge_{t \in \dom(\hh_S)}\hh_S(t) = \hh_C(t))$
    \end{enumerate}
    then,
    \begin{enumerate}
        \item $\ll \expr, \fstore, \store, \sstore, \hh_C, \hh'_S, \cc' \gg \updownarrow \val, \sval', \cc'_2, \m'_2$
        \item $\expand(\sval')$
        \item $\exists! X' \forall Y' \exists! Z' (\bigwedge_{t \in \dom(\hh'_S)}\hh'_S(t) = \hh_C(t))$
        \item $\store, \hh_C \prec_{\cc'} \sstore, \hh_S'$
    \end{enumerate}
\end{lemma}

\begin{proof}[Proof sketch]
This lemma states that if bisimulation holds for an expression $\expr$ with a symbolic graph $\hh_S, \cc$ and then the symbolic graph is expanded to $\hh_S', \cc'$ using the $\ex$ function with input expression $\expr$, then, bisimulation holds for $\expr$ with symbolic graph $\hh_S', \cc'$ and the output of symbolically evaluating $\expr$ using $\hh_S'$ and $\cc'$ will be in expanded form. We use induction on the structure of $\expr$ to prove this. Showing that the new symbolic output will always be in expanded form is straightforward for most types of expressions because for most expressions, if the sub-expressions evaluate to expanded symbolic values, then the whole expression will also. Two exceptions to this are the expressions $\expr[\var]$, where $\var$ is a shape member, and $\expr[m]$, where $m$ is a type of metadata. Since $\ex$ expands all shape members and metadata used in $\expr$, the symbolic values of accessing shape members will always be expanded.
\end{proof}

\begin{lemma}
\label{lemma:multistephextendnew}
    If 
    \begin{enumerate}
        \item $\ll \expr, \fstore, \store, \sstore, \hh_C, \hh_S, \cc \gg \updownarrow \val, \sval, \hat{\cc_1}, \m_1$
        \item $\exists! X \forall Y \exists!Z (\bigwedge_{t \in \dom(\hh_S)} \hh_C(t) = \hh_S(t))$
        \item $\tau_s, \symc, \prop \leadsto^* \hh'_S, \cc'$\
        \item $\store, \hh_C \prec_{\cc} \sstore, \hh_S$
    \end{enumerate}
    then,
    \begin{enumerate}
        \item $\ll \expr, \fstore, \store, \sstore, \hh_C, \hh'_S, \cc' \gg \updownarrow \val, \sval', \cc'', \m''$
        \item $\exists! X' \forall Y' \exists!Z' (\bigwedge_{t \in \dom(\hh'_S)} \hh_C(t) = \hh'_S(t))$
        \item $\expand(\sval) \implies \expand(\sval')$
        \item $\store, \hh_C \prec_{\cc'} \sstore, \hh'_S$
    \end{enumerate}
\end{lemma}

\begin{proof}[Proof sketch]
This lemma states that if bisimulation holds for an expression $\expr$ with a symbolic graph $\hh_S, \cc$ and then the symbolic graph is expanded to $\hh_s', \cc'$ using the mutli-step expand rules (for $\leadsto^*$), then  bisimulation holds for $\expr$ with symbolic graph $\hh_s', \cc'$ and if the output of symbolically evaluating $\expr$ using $\hh_S$ and $\cc$ is in expanded form, then the output of symbolically evaluating $\expr$ using $\hh'_S$ and $\cc'$ is in expanded form. The proof for this lemma uses structural induction on the rules defining $\leadsto^*$, which is comprised of arbitrary combinations of the following functions, $\ex$, $\addg$, and $\leadsto$. It uses the Lemma~\ref{lemma:hexpand} to handle the case of $\ex$ and the rules for $\addg$ for the case of $\addg$. For the case of $\leadsto$, this proof uses induction on the structure of $\expr$. The basic argument in all of these cases is that the original bisimulation arguments holds only if the symbolic semantics rules outputted a value for $\expr$, which is only possible if the graph was expanded enough for all of the sub-expressions in $\expr$. If the graph expands for other expressions, the conditions in $\cc'$ should ensure that the symbolic value outputted for $\expr$ using $\hh'_S$ and $\cc'$ is equivalent to the symbolic value outputted for $\expr$ using $\hh_S$ and $\cc$.
\end{proof}

\begin{lemma} If
\label{lemma:overapprox}
    \begin{enumerate}
        \item $\store, \hh_C \prec_{\cc} \sstore, \hh_S$
        \item $\tau_s, \symc, \prop \models \expr \leadsto \hh_S', \cc'$
        \item $\langle \expr, \context \rangle \Downarrow \val$
        \item $\exists! X \forall Y \exists! Z (\bigwedge_{t \in \dom(\hh_S)}\hh_S(t) = \hh_C(t))$
        \item Inductive invariants are true and Solver constraints are feasible
    \end{enumerate}
    then,
    \begin{enumerate}
        \item $\store, \hh_C \prec_{\cc'} \sstore, \hh_S'$
        \item $\ll \expr, \fstore, \store, \sstore, \hh_C, \hh'_S, \cc' \gg \updownarrow \val, \sval, \cc'', \m$
        \item $\exists! X' \forall Y' \exists! Z' (\bigwedge_{t \in \dom(\hh'_S)}\hh'_S(t) = \hh_C(t))$
    \end{enumerate}
\end{lemma}

\begin{proof}[Proof sketch]
This lemma states that if a expression can be evaluated under concrete semantics and the symbolic graph has been created and expanded for $\expr$, then bisimulation is true for the concrete and symbolic values given by the operational and symbolic semantics. The antecedent of this lemma includes that all inductive invariants are true and solver constraints are feasible. The antecedent also includes the symbolic store over-approximating the concrete store and a stronger assumption on the symbolic graph over-approximating the concrete DNN because both of these statements are required to strengthen the induction hypothesis in order to prove this lemma. This lemma is proved by induction on the structure of $\expr$. It uses Lemmas~\ref{lemma:ctogether}~\ref{lemma:hexpand}~\ref{lemma:multistephextendnew}. We present the complicated cases in the induction on the structure of $\expr$. The other cases use similar arguments.
\end{proof}

\begin{lemma} If
\label{lemma:graphexpand}
    \begin{enumerate}
        \item $\Gamma, \tau_s \vdash \expr \ : \ \types$
        \item $\bot \sqsubset \types \sqsubset \top$
        \item $\fstore \sim \Gamma, \tau_s$
        \item $\store \sim \Gamma$
        \item $\hh_C \sim \tau_s$
        \item $\store, \hh_C \prec_{\cc} \sstore, \hh_S$
        \item $\exists!X \forall Y \exists! Z (\bigwedge_{t \in \dom(\hh_S)} \hh_S(t) = \hh_C(t))$
        \item Inductive invariants are true and Solver constraints are feasible
    \end{enumerate}
    then,
    \begin{enumerate}
        \item $\tau_s, \symc, \prop \models \expr \leadsto \hh_S', \cc'$
        \item $\store, \hh_C \prec_{\cc'} \sstore, \hh'_S$
    \end{enumerate}
\end{lemma}

\begin{proof}[Proof sketch]
This lemma states that if an expression, $\expr$, type-checks, the inductive invariants are true and the solver input constraints are feasible, then given a symbolic graph, $\hh_S$ and $\cc$, that over-approximates the concrete DNN, $\hh_S$ and $\cc$ can be expanded for $\expr$. This proof uses Lemma~\ref{lemma:overapprox} because some of the graph expansion rules, such as the ones for \map, require executing the symbolic semantics during graph expansion. We present the complicated cases in induction on the structure of $\expr$. The other cases use similar arguments.
\end{proof}

\renewcommand{\thetheorem}{\ref{soundnesstheorem}}
\begin{theorem}
% \label{soundnesstheorem}
    For a well-typed program $\Pi$, if \cf verification procedure proves it maintains the property $\prop$, then upon executing $\Pi$ on all concrete DNNs within the bound of verification, the property $\prop$ will be maintained at all neurons in the DNN. 
\end{theorem}
\renewcommand{\thetheorem}{\Alph{section}.\arabic{theorem}}
\renewcommand{\thetheorem}{\ref{completenesstheorem}}
\begin{theorem}
% \label{completenesstheorem}
    If executing a well-typed program $\Pi$ that does not use \traverse and \solver constructs on all concrete DNNs within the bounds of verification maintains the property $\prop$ for all neurons in the DNN, then it can be proved by the \cf verification procedure.
\end{theorem}
\renewcommand{\thetheorem}{\Alph{section}.\arabic{theorem}}

\clearpage
\section{Proofs}
\label{appendix:proofs}

\begin{proof}
\label{proof:uniquez}
Proof of lemma~\ref{lemma:uniquez}
\begin{enumerate}
    \item Let $m_X$ be a satisfying assignment to $X$ and $m_Y$ be an assignment to $Y$.
    \item For a given assignment to $Y$, for all $t \in \dom(\hh_C)$, $\hh_C(t)$ is a fixed value.
    \item $m_1, m_2$ are two assignments of $Z$. 
    \begin{enumerate}
        \item $m_X \cup m_Y \cup m_1 \models (\bigwedge_{t \in \dom(\hh_S)}\hh_S(t) = \hh_C(t))$
        \item $m_X \cup m_Y \cup m_2 \models (\bigwedge_{t \in \dom(\hh_S)}\hh_S(t) = \hh_C(t))$
    \end{enumerate}
    % \item There are two types of variables in Z
    \item $Z \subseteq \range(\hh_S)$
    \item For each $z$ in $Z$, there is a $t$ in $\dom(\hh_S)$ s.t. $\hh_S(t) = z$
    % \item $\forall z \in Z, \exists t \in \dom(\hh_c)$
    \item From (3a) and (5), for each $z$ in $Z$, $m_1(z) = \hh_C(t)$
    \item From (3b) and (5), for each $z$ in $Z$, $m_2(z) = \hh_C(t)$
    \item From (6) and (7), for each $z$ in $Z$, $m_1(z) = m_2(z)$
    \item So, for given assignments $m_X$ and $m_Y$, there is a unique assignment to $Z$.
    \item Therefore, $\exists X \forall Y \exists! Z (\bigwedge_{t \in \dom(\hh_S)}\hh_S(t) = \hh_C(t))$
\end{enumerate}
\end{proof}
\noindent
\begin{proof}
Proof of lemma~\ref{lemma:uniquex}
\label{proof:uniquex}
    \begin{enumerate}
    \item $m_1$ and $m_2$ are two satisfying assignments of $X$
    \item There are two types of variables in $X$:
    \begin{enumerate}
        \item For some of the $t$ in $\dom(\hh_S)$, $\hh_C(t)$ are polyhedral expressions or symbolic expressions. For such values, there are corresponding variables in $X$ which represent the coefficients and the constants of the polyhedral expressions or the symbolic expressions.  
        \item For the remaining $t$ in $\dom(\hh_S)$, $\hh_C(t)$ is some constant $c$, an $x$ in $X$, represents this constant, s.t., $\hh_S(t) = x$.
    \end{enumerate}
    \item In case (2a),
    \begin{enumerate}
        \item Let $\hh_C(t) = c_0 + c_1 * n_1 + c_2 * n_2 + \cdots$. \\ 
        In this case, $\hh_S(t) = x_0 + x_1 * n_1 + x_2 * n_2 + \cdots$ where $n_1, n_2, \cdots \in Y$. \\ 
        Since the equality $c_0 + c_1 * n_1 + c_2 * n_2 + \cdots = x_0 + x_1 * n_1 + x_2 * n_2 + \cdots$ must hold for all assignments to $Y$, there is a unique assignment to $x_0, x_1, \cdots$.
        \item Let $\hh_C(t) = c_0 + c_1 * \epsilon_1 + c_2 * \epsilon_2 + \cdots$. \\ 
        In this case, $\hh_S(t) = x_0 + x_1 * \epsilon_1 + x_2 * \epsilon_2 + \cdots$ where $\epsilon_1, \epsilon_2, \cdots \in Y$ \\
        Since the equality $c_0 + c_1 * \epsilon_1 + c_2 * \epsilon_2 + \cdots = x_0 + x_1 * \epsilon_1 + x_2 * \epsilon_2 + \cdots$ must hold for all assignments to $Y$, there is a unique assignment to $x_0, x_1, \cdots$.
    \end{enumerate}
    \item In case (2b), since $\hh_C(t)$ is a constant, there is a unique assignment to $x$ such that $\hh_C(t) = \hh_S(t)$, i.e., $x = \hh_C(t)$
    Hence, there is a unique assignment to all the variables in $X$.
    \item Therefore, $\exists! X \forall Y \exists Z (\bigwedge_{t \in \dom(\hh_S)}\hh_S(t) = \hh_C(t))$ 
\end{enumerate}

\end{proof}

% \clearpage

\begin{proof}
Proof of lemma~\ref{lemma:uniquexz} \\
\label{proof:uniquexz}
\begin{enumerate}
    \item $\exists X \forall Y \exists Z (C \wedge \bigwedge_{t \in \dom(\hh_S)}\hh_S(t) = \hh_C(t)) \ \implies \ \exists X \forall Y \exists Z (\bigwedge_{t \in \dom(\hh_S)}\hh_S(t) = \hh_C(t))$
    \item $\exists X \forall Y \exists Z (\bigwedge_{t \in \dom(\hh_S)}\hh_S(t) = \hh_C(t)) \ \implies \ \exists! X \forall Y \exists! Z (\bigwedge_{t \in \dom(\hh_S)}\hh_S(t) = \hh_C(t))$
    \item This follows from lemma~\ref{lemma:uniquex} and lemma~\ref{lemma:uniquez}
\end{enumerate}
\end{proof}
\clearpage
% \input{sections/appendix/typecheckinglemma}
% \begin{lemma}
%     \label{lemma:cexpand}
%     If
%     \begin{enumerate}
%         \item $\langle \expr, \symc \rangle \downarrow \sval, \cc'$ 
%     \end{enumerate}
%     then,
%     \begin{enumerate}
%         \item $\cc' = \cc \wedge \cc''$
%         % \item $\cc' \implies \cc$
%     \end{enumerate}
% \end{lemma}
\begin{proof}[Proof of lemma ~\ref{lemma:cexpand}]
    This can be seen by looking at the rules in the section symbolic semantics for expressions.  We can do induction on the structure of $\expr$ to conclude that in each rule, $\cc'$ is either equal to $\cc$, in which case $\cc'' = \true$ or $\cc'$ is created by taking conjunctions of $\cc$ with other conditions.
\end{proof}

% \begin{lemma}
%     If
%     \begin{enumerate} 
%         \item $\forall i \in [n], \ll \expr_i, \fstore, \store_i, \sstore, \hh_C, \hh_S, \cc_0 \gg \updownarrow \val_i, \sval_i, \cc'_i, \m'_i$
%         \item $\exists!X \forall Y \exists!Z (\cc_0 \wedge \bigwedge_{t \in \dom(\hh_S)} \hh_S(t) = \hh_C(t))$
%     \end{enumerate}
%     then,
%     \begin{enumerate} 
%         \item $\forall i \in [n], \langle \expr_i, \fstore, \sstore, \hh_S, \cc_{i-1} \rangle \downarrow \sval_i, \cc_i$
%     \end{enumerate}
% \end{lemma}
\begin{proof} [Proof of lemma ~\ref{lemma:csequential}]
This can be proven by induction on the structure of $\expr_i$ and induction on $i$. Using Lemma~\ref{lemma:cexpand}, $\forall i \in [m], \cc'_i = \cc_0 \wedge \cc''_i $. There are only three rules in which symbolic evaluation uses $\cc$: \textsc{Inv, Inv-invariant} and \textsc{Sym-solver}. When these rules are not involved, by looking at the rules, it is easy to see that $\langle \expr_i, \fstore, \sstore, \hh_S, \cc_{i-1} \rangle \downarrow \sval_i, \cc'_i \wedge \bigwedge_{k = 1}^{i-1} \cc''_k$ \\
\textbf{\textsc{Inv-invariant}}
\begin{enumerate}
    \item $\cinduction(\var \cdot \trav(\dir, f_{c_1}, f_{c_2}, f_{c_3})\{\expr\}, \fstore, \sstore, \hh_S, \m, \cc) = \bsval''$
    \item From Antecedent (1) and \textsc{Inv-invariant}, $\sval''_b = \cc_0 \wedge \sval'_{b} \implies \cc''_j \wedge \sval''' $
    \item From (1), $\sval''_b = (\bigwedge_{k = 1}^{i-1} \cc''_k \wedge \cc_0 \wedge \sval'_{b}) \implies (\bigwedge_{k = 1}^{i-1} \cc''_k \wedge \cc''_j \wedge \sval''') $
    \item $\cinduction(\var \cdot \trav(\dir, f_{c_1}, f_{c_2}, f_{c_3})\{\expr\}, \fstore, \sstore, \hh_S, \m, \cc \wedge \bigwedge_{k = 1}^{i-1} \cc''_k) = \bsval''$
\end{enumerate}
\textbf{\textsc{Inv}}
\begin{enumerate}
    \item $\cinvariant(\var \cdot \trav(\dir, f_{c_1}, f_{c_2}, f_{c_3})\{\expr\}, \fstore, \sstore, \hh_S, \m, \cc) = \sval_b \wedge \sval'_b, \cc'$
    \item From Antecedent (1) and \textsc{Inv}, $\sval_b = \cc' \implies \sval $
    \item From (1), $\sval_b = \bigwedge_{k = 1}^{i-1} \cc''_k \wedge \cc' \implies \sval $
    \item $\cinvariant(\var \cdot \trav(\dir, f_{c_1}, f_{c_2}, f_{c_3})\{\expr\}, \fstore, \sstore, \hh_S, \m, \cc \wedge \bigwedge_{k = 1}^{i-1} \cc''_k) = \sval_b \wedge \sval'_b, \bigwedge_{k = 1}^{i-1} \cc''_k \wedge \cc'$
\end{enumerate}
\textbf{\textsc{Sym-solver}}
\begin{enumerate}
    \item $\langle \lp(\minimize, \expr_1, \expr_2), \symc \rangle \downarrow \bsval, \cc'$
    \item From \textsc{Sym-solver}, $\sat(\cc_2 \wedge \sval_2)$
    \item From Antecedents (1) and (2), because the symbolic evaluation rules only add variables that are fresh and will not be shared between $\cc''_i$'s, $\exists! X \forall Y \exists! Z(\bigwedge_{k = 1}^{i-1} \cc''_k \wedge \bigwedge_{k=1}^{i-1} \m'_k \wedge \bigwedge_{t \in \dom(\hh_S)} \hh_S(t) = \hh_C(t))$
    \item The satisfying assignment to $\cc_2 \wedge \sval_2$ can be extended to include the unique assignment to the the variables in $\bigwedge_{k = 1}^{i-1} \cc''_k$, since these variables are only in $X$. We know the fresh variables added to $\cc$ are in $X$ by looking at \textsc{Sym-traverse} and \textsc{Sym-solver}.
    \item From (4), $\sat(\bigwedge_{k = 1}^{i-1} \cc''_k \wedge \cc_2 \wedge \sval_2)$
    \item $\langle \lp(\minimize, \expr_1, \expr_2), \fstore, \sstore, \hh_C, \cc \wedge \bigwedge_{k = 1}^{i-1} \cc''_k \rangle \downarrow \bsval, \cc' \wedge \bigwedge_{k = 1}^{i-1} \cc''_k $
\end{enumerate}
\end{proof}
\newpage

% \begin{lemma}
%     If
%     \begin{enumerate} 
%         \item $\forall i \in [n], \ll \expr_i, \fstore, \store_i, \sstore, \hh_C, \hh_S, \cc_0 \gg \updownarrow \val_i, \sval_i, \cc'_i, \m'_i$
%         \item $\exists!X \forall Y \exists!Z (\cc_0 \wedge \bigwedge_{t \in \dom(\hh_S)} \hh_S(t) = \hh_C(t))$
%     \end{enumerate}
%     then,
%     \begin{enumerate}
%         \item $\forall i \in [n], \langle \expr_i, \fstore, \sstore_i, \hh_S, \cc_{i-1} \rangle \downarrow \sval_i, \cc_i$
%         \item $\exists! X' \forall Y' \exists!Z' (\bigwedge_i \sval_i = \val_i \wedge \cc_n \wedge \m \wedge \bigwedge_{t \in \dom(\hh_S)} \hh_S(t) = \hh_C(t))$
%         \item $\cc_n \implies \cc_0$
%     \end{enumerate}
% \end{lemma}
\begin{proof}[Proof of lemma ~\ref{lemma:ctogether}]
\setcounter{number}{1}
\begin{align*}
    &\text{From Antecedent (1) and (2)}  \\
    & \forall i \in [n], \langle \expr_i, \fstore, \sstore_i, \hh_S, \cc_0 \rangle \downarrow \sval_i, \cc'_i && && \mycounter & \\ 
    &\text{From Lemma~\ref{lemma:csequential} using Antecedent (1)} \\
    & \forall i \in [n], \langle \expr_i, \fstore, \sstore_i, \hh_S, \cc_{i-1} \rangle \downarrow \sval_i, \cc_i && \text{Consequent 1} && \mycounter & \\ 
    & \text{We will proceed using induction on $n$}\\
    &\textbf{Base Case: } n=1\\
    & \text{Consequent 2 follow directly from Antecedent (1)}\\
    &\textbf{Base Case: } n=2\\
    &\text{From (1) and Lemma~\ref{lemma:cexpand}}\\
    & \cc_1 = \cc'_1 = \cc_0 \wedge \cc''_1 &&  && \mycounter &\\
    &\text{From (3) and Antecedent (1)} \\
    & \exists! X_1 \forall Y_1 \exists! Z_1 (CS(\sval_1,\val_1) \wedge \cc_1 \wedge \m'_1 \wedge \bigwedge_{t \in \dom(\hh_S)} \hh_S(t) = \hh_C(t)) &&  && \mycounter &\\
    &\text{From Antecedent (1)} \\
    & \exists! X_2 \forall Y_2 \exists! Z_2 (CS(\sval_2,\val_2) \wedge \cc'_2 \wedge \m'_2 \wedge \bigwedge_{t \in \dom(\hh_S)} \hh_S(t) = \hh_C(t)) &&  && \mycounter &\\
&\text{From (1) and Lemma~\ref{lemma:cexpand}} \\
    & \cc'_2 = \cc_0 \wedge \cc''_2 && && \mycounter &\\
    &\text{From (2) and Lemma~\ref{lemma:cexpand}} \\
    & \cc_2 = \cc_1 \wedge \cc''_2 && && \mycounter &\\
    &\text{From (3) and (4)}\\
    & \exists! X_1 \forall Y_1 \exists! Z_1 (CS(\sval_1,\val_1) \wedge \cc_0 \wedge \cc''_1 \wedge \m'_1 \wedge \bigwedge_{t \in \dom(\hh_S)} \hh_S(t) = \hh_C(t)) &&  && \mycounter &\\
    &\text{From (5) and (6)}\\
    & \exists! X_2 \forall Y_2 \exists! Z_2 (CS(\sval_2,\val_2) \wedge \cc_0 \wedge \cc''_2 \wedge \m'_2 \wedge \bigwedge_{t \in \dom(\hh_S)} \hh_S(t) = \hh_C(t)) &&  && \mycounter &\\
    & \text{Given an arbitrary assignment to $Y_1 \cup Y_2$, $m_{Y}$, } \\
    & \text{Call the unique assignment to $X_1$, $m_{X_1}$.}\\
    & \text{There exists a unique assignment to $Z_1$, $m_{Z_1}$ such that }\\
    & m_{X_1} \cup m_{Y} \cup m_{Z_1} \models (CS(\sval_1,\val_1) \wedge \cc_0 \wedge \cc''_1 \wedge \m'_1 \wedge \bigwedge_{t \in \dom(\hh_S)} \hh_S(t) = \hh_C(t)) && \text{From (8)} && \mycounter &\\
    & \text{Call the unique assignment to $X_2$, $m_{X_2}$.}\\
    & \text{There exists a unique assignment to $Z_2$, $m_{Z_2}$ such that }\\
    & m_{X_2} \cup m_{Y} \cup m_{Z_2} \models (CS(\sval_2,\val_2) \wedge \cc_0 \wedge \cc''_2 \wedge \m'_2 \wedge \bigwedge_{t \in \dom(\hh_S)} \hh_S(t) = \hh_C(t)) && \text{From (9)} && \mycounter &\\
    & m_{X_1} \cup m_{Y} \cup m_{Z_1} \models (\cc_0 \wedge \bigwedge_{t \in \dom(\hh_S)} \hh_S(t) = \hh_C(t)) && \text{From (10)} && \mycounter &\\
    \end{align*}
\begin{align*}
    & m_{X_2} \cup m_{Y} \cup m_{Z_2} \models (\cc_0 \wedge \bigwedge_{t \in \dom(\hh_S)} \hh_S(t) = \hh_C(t)) && \text{From (11)} && \mycounter &\\
    & \text{For $a \in \vars(\cc_0) \cup \vars(\hh_S)$, }\\
    &\text{From (12),(13) and Antecedent (2)}  \\
    & (m_{X_1} \cup m_{Y} \cup m_{Z_1})(a) = (m_{X_2} \cup m_{Y} \cup m_{Z_2})(a) && && \mycounter &\\
    & X_{1,2} = X_1 \cup X_2 \quad Y_{1,2} = Y_1 \cup Y_2 \quad Z_{1,2} = Z_1 \cup Z_2 && && \mycounter &\\
    &\vars(\sval_1) \cup \vars(\cc''_1) \subseteq \vars(\cc_0 \wedge \m'_1 \wedge \bigwedge_{t \in \dom(\hh_S)} \hh_S(t) = \hh_C(t)) && && \mycounter&\\
    &\vars(\m'_1)\setminus \vars(\cc_0 \wedge \bigwedge_{t \in \dom(\hh_S)} \hh_S(t) = \hh_C(t)) \text{ are fresh variables}&& && \mycounter&\\
    &\vars(\sval_2) \cup \vars(\cc''_2) \subseteq \vars(\cc_0 \wedge \m'_2 \wedge \bigwedge_{t \in \dom(\hh_S)} \hh_S(t) = \hh_C(t)) && && \mycounter&\\
    &\vars(\m'_2)\setminus \vars(\cc_0 \wedge \bigwedge_{t \in \dom(\hh_S)} \hh_S(t) = \hh_C(t)) \text{ are fresh variables}&& && \mycounter&\\
    & \text{Since the fresh variables in $\m'_1$ are different from those in $\m'_2$ }\\
    &\text{From (14), (17) and (19)}\\
    & \text{$m_{X_1} \cup m_{X_2}$ is a well-defined assignment to the variables in $X_{1,2}$} &&  && \mycounter &\\
    &\text{From (14), (17) and (19)}\\
    & \text{$m_{Z_1} \cup m_{Z_2}$ is a well-defined assignment to the variables in $Z_{1,2}$} &&  && \mycounter &\\
&\text{From (3),(7),(10), (11),(20) and (21)} \\
    & \exists X_{1,2} \forall Y_{1,2} \exists Z_{1,2} (\bigwedge_i CS(\sval_i,\val_i) \wedge \cc_2 \wedge \m'_1 \wedge \m'_2 \wedge \bigwedge_{t \in \dom(\hh_S)} \hh_S(t) = \hh_C(t) ) && && \mycounter &\\
    & \text{For any assignment $m_{1,2}$ to $X_{1,2}$, } \\
    &\text{From (3,6,7,22)}\\
    & m_{1,2} \models \forall Y_{1,2} \exists Z_{1,2} (CS(\sval_1,\val_1) \wedge \cc_0 \wedge \cc''_1 \wedge \m'_1 \wedge \bigwedge_{t \in \dom(\hh_S)} \hh_S(t) = \hh_C(t)) && && \mycounter&\\
    &\text{From (8) and (23)}\\
    & \text{The assignment to $X_{1,2} \cap X_1$ = $m_{x_1}$} &&  && \mycounter &\\
    & \text{For any assignment $m_{1,2}$ to $X_{1,2}$, } \\
    &\text{From (3,6,7,22)}\\
    & m_{1,2} \models \forall Y_{1,2} \exists Z_{1,2} (CS(\sval_2,\val_2) \wedge \cc_0 \wedge \cc''_2 \wedge \m'_2 \wedge \bigwedge_{t \in \dom(\hh_S)} \hh_S(t) = \hh_C(t)) && && \mycounter&\\
    &\text{From (9) and (25)}\\
    & \text{The assignment to $X_{1,2} \cap X_2$ = $m_{x_2}$} &&  && \mycounter &\\
    &\text{(24) and (26)}\\
    & \exists! X_{1,2} \forall Y_{1,2} \exists Z_{1,2} (\bigwedge_i CS(\sval_i,\val_i) \wedge \cc_2 \wedge \m'_1 \wedge \m'_2 \wedge \bigwedge_{t \in \dom(\hh_S)} \hh_S(t) = \hh_C(t) ) &&  && \mycounter &\\
    \end{align*}
\begin{align*}
    & \text{Similar logic can be used to show that the assignment to $Z_{1,2}$ is unique}\\
    & \exists! X_{1,2} \forall Y_{1,2} \exists Z_{1,2} (\bigwedge_i CS(\sval_i,\val_i) \wedge \cc_2 \wedge \m \wedge \bigwedge_{t \in \dom(\hh_S)} \hh_S(t) = \hh_C(t) ) &&&& \text{Consequent 2} &\\
    &\text{From (3) and (7), } \cc_2 \implies \cc_0 &&&&\text{Consequent 3}  &\\
    &\textbf{Induction Case: } n >2 \\
    &\text{From induction hypothesis using Antecedents (1) and (2)} \\
    & \exists! X_1 \forall Y_1 \exists ! Z_1 (\bigwedge_{i=1}^{n-1} CS(\sval_i,\val_i) \wedge \cc_{n-1} \wedge \m_{n-1} \wedge \bigwedge_{t \in \dom(\hh_S)} \hh_S(t) = \hh_C(t)) && && \mycounter &\\
    &\cc_{n-1} \implies \cc_0 && && \mycounter &\\
    & \text{For an assignment $m_1$ to $X_1 \cup Y_1 \cup Z_1$ such that }\\
    & m_1 \models (\bigwedge_{i=1}^{n-1} CS(\sval_i,\val_i) \wedge \cc_{n-1} \wedge \m_{n-1} \wedge \bigwedge_{t \in \dom(\hh_S)} \hh_S(t) = \hh_C(t)) &&&& \mycounter&\\
    &\text{From (28) and (29)}\\
    & m_1 \models (\cc_0 \wedge \bigwedge_{t \in \dom(\hh_S)} \hh_S(t) = \hh_C(t)) && && \mycounter & \\
    &\text{From (1) and Lemma~\ref{lemma:cexpand}}\\
    & \cc'_n = \cc_0 \wedge \cc''_n &&  && \mycounter &\\
    & \text{For an assignment $m_2$ to $X_2 \cup Y_2 \cup Z_2$ such that }\\
    & m_2 \models (CS(\sval_n,\val_n) \wedge \cc'_n \wedge \m'_n \wedge \bigwedge_{t \in \dom(\hh_S)} \hh_S(t) = \hh_C(t)) &&&& \mycounter&\\
% \end{align*}
% \begin{align*}
    &\text{From (32) and Antecedent (1)}\\
    & m_2 \models (\cc_0 \wedge \bigwedge_{t \in \dom(\hh_S)} \hh_S(t) = \hh_C(t)) && && \mycounter & \\
    &\text{From (31), (33) and Antecedent (2)}\\
    & \text{$m_1$ and $m_2$ map $\vars(\cc_0 \wedge \bigwedge_{t \in \dom(\hh_S))} \hh_S(t) = \hh_C(t)$ to the same constants} &&  && \mycounter &\\
    & (X_2 \cup Y_2 \cup Z_2) \setminus \vars(\cc_0 \wedge \bigwedge_{t \in \dom(\hh_S))} \hh_S(t) = \hh_C(t) \text{ only has fresh variables} && && \mycounter&\\
    &\text{From (35) and (36)}\\
    & m_1 \cup m_2 \text{ is a well-defined assignment} && && \mycounter &\\
    &\text{From (28), (30), (33), (37) and Antecedent (1)} \\
    &\exists!X_3 \forall Y_3 \exists! Z_3((\bigwedge_{i=1}^{n} CS(\sval_i,\val_i) \wedge \cc_{n} \wedge \m_{n} \wedge \bigwedge_{t \in \dom(\hh_S)} \hh_S(t) = \hh_C(t))) &&&& \text{Consequent 2} &\\
    &\text{From (2) and (29) } \cc_n \implies \cc_0 &&&& \text{Consequent 3} &
\end{align*}
\end{proof}

\newpage
\begin{proof}[Proof of Lemma~\ref{lemma:hexpand}] Proof by induction on the structure of $\expr$ \\ 
\textbf{Base cases:} \\ 
$\mathbf{\expr \equiv \constant}$
\setcounter{number}{1}
\begin{align*}
    &\ex(\constant, \tau_s, \symc, \prop) = \hh_S, \cc && \text{\textsc{E-const}} && \mycounter & \\ 
    &\text{From (1), }\hh_S' = \hh_S, \cc' = \cc &&&& \mycounter &\\
    &\text{From (2) and Antecedent (1), } \store, \hh_C \prec_{\cc'} \sstore, \hh_S' &&\text{Consequent (4)} && &\\
    &\text{From (2) and Antecedent (2), } \ll \expr, \fstore, \store, \sstore, \hh_C, \hh'_S, \cc' \gg \updownarrow \val, \sval, \cc'_1, \m'_1 &&\text{Consequent (1)} && &\\
    &\text{From \textsc{Expanded-const} } \expand(\constant) && \text{Consequent (2)} && & \\ 
    &\text{From (2) and Antecedent (4), } \exists! X \forall Y \exists! Z (\bigwedge_{t \in \dom(\hh'_S)} \hh'_S(t) = \hh_C(t))&& \text{Consequent (3)} && & \\
\end{align*}
$\mathbf{\expr \equiv \var}$
\setcounter{number}{1}
\begin{align*}
    &\ex(\constant, \tau_s, \symc, \prop) = \hh_S, \cc && \text{\textsc{E-var}} && \mycounter & \\ 
    &\text{From (1), }\hh_S' = \hh_S, \cc' = \cc &&&& \mycounter &\\
    &\text{From (2) and Antecedent (1), } \store, \hh_C \prec_{\cc'} \sstore, \hh_S' &&\text{Consequent (4)} && &\\
    &\text{From (2) and Antecedent (2), } \ll \expr, \fstore, \store, \sstore, \hh_C, \hh'_S, \cc' \gg \updownarrow \val, \sval, \cc'_1, \m'_1 &&\text{Consequent (1)} && &\\
    &\text{From (2) and Antecedent (4), } \exists! X \forall Y \exists! Z (\bigwedge_{t \in \dom(\hh'_S)} \hh'_S(t) = \hh_C(t))&& \text{Consequent (3)} && & \\
    &\var \in \sstore &&\text{From Antecedent (2)} && \mycounter & \\
    &\text{Since $\sstore$ maps variables to expanded values}\\
    &\expand(\sstore(t)) && \text{From (3)} && \mycounter&\\
    &\text{From (4), Antecedent (2) and \textsc{Sym-var}, } \expand(\sval) &&  \text{Consequent (2)} && &\\
\end{align*}

\newpage
\textbf{Induction Cases:}\\
$\mathbf{\expr \equiv \expr_1[\var]}$
\setcounter{number}{1}
\begin{align*}
    &\inferrule*[]
        {
        \ex(\expr_1,\tau_s, \fstore, \sstore, \hh_S,\cc, \prop) = \hh'_S, \cc_1
        \\\\
        \langle \expr_1, \fstore, \sstore, \hh'_s, \cc_1 \rangle \downarrow n, \_ \\\\
        \exn(n, \var, \tau_s, \hh'_S, \prop) = \hh''_S, \cc_3
        }
        {
        \ex(\expr_1[x], \tau_s, \symc, \prop) = \hh''_S, \cc_1 \wedge \cc_3 
        } && \text{\textsc{E-shape-b}} && \mycounter & \\
    & \inferrule*[]
        {
        \langle \expr_1, \context \rangle \Downarrow n' \\\\
        \val = \hh_C[n'[x]] 
        } 
        {\langle \expr_1[x], \context \rangle \Downarrow \val} && \text{\textsc{Op-shape}} && \mycounter &\\
    & \text{From the induction hypothesis using (1),(2), }\\
    & \text{Antecedents (1), (2) and (4) }\\
    & \ll \expr_1, \fstore, \store, \sstore, \hh_C, \hh'_S, \cc_1 \gg \updownarrow n', n, \cc_2, \m_2 && && \mycounter &\\
    & \exists! X \forall Y \exists! Z (\bigwedge_{t \in \dom(\hh'_S)} \hh'_S(t) = \hh_C(t)) && && \mycounter &\\
    & \store, \hh_C \prec_{\cc_2} \sstore, \hh'_S && && \mycounter &\\
    &\ex(n) && && \mycounter &\\
    & \text{Since $n'$ is a neuron, $n = n'$} && \text{From (3)} && \mycounter &\\
    & \val = \hh_C(n[x])&& \text{From (3)} && \mycounter &\\
    &\text{From (3) and \textsc{Expanded-poly}}\\
    & \exists X \forall Y \exists 
    Z(\hh_C(n[x]) = \hh'_S(n[x]))&&  && \mycounter &\\
    &\text{From (1),(8) and \textsc{Sym-shape} } \langle \expr_1[\var], \fstore, \sstore, \hh'_S, \cc \rangle \downarrow \hh'_S[n[x]], \cc'_2 &&  && \mycounter &\\
    &\text{From \textsc{Sym-shape}}\\
    & \sval = \hh'_S[n[x]] &&  && \mycounter&\\
    &\text{From (1) and \textsc{Sym-shape}}\\
    & \sval' = \hh''_S[n[\var]] &&  && \mycounter&\\
    &\text{From (1), \textsc{Expand-poly-r}, \textsc{Expand-sym-r}, }\\
    &\text{\textsc{Expand-poly-b and Expand-sym-b} } \expand(\hh''_S[n[\var]]) && \text{Consequent (2)} &&\mycounter &\\
    & \text{\textsc{Expand-poly-r, Expand-sym-r} adds }\\
    & \text{fresh variables to $\hh'_S$ to get $\hh''_S$,}\\
    &\text{\textsc{Expand-poly-b} and \textsc{Expand-sym-b} returns $\hh''_S = \hh'_S$} && && \mycounter&\\
    &\text{From (14), } \exists! X'' \forall Y'' \exists! Z'' (\bigwedge_{t \in \dom(\hh''_S)} \hh''_S(t) = \hh_C(t)) && \text{Consequent 3} && \mycounter &\\
    & \text{From (14),(15) and Antecedent (1), } \store, \hh_C \prec_{\cc_1 \wedge \cc_3} \sstore, \hh''_S && \text{Consequent (4)} && \mycounter &\\
    & \text{From (2),(12) and (15) } \ll \expr, \fstore, \store, \sstore, \hh_C, \hh'_S, \cc_1 \gg \updownarrow \val, \sval', \cc_2, \m_2 && \text{Consequent (1)} && \mycounter &\\
\end{align*}
\newpage
\fbox{$\mathbf{\expr \equiv \expr_1 \oplus \expr_2}$}
\setcounter{number}{1}
\begin{align*}
&\text{From Antecedent (2) and \textsc{E-binary}}\\
& \inferrule*[]
            {
            \ex(\expr_1, \tau_s, \fstore, \sstore, \hh_{S}, \cc, \prop) = \hh_{S_1}, \cc_1 \\\\
            \ex(\expr_2, \tau_s, \fstore, \sstore, \hh_{S_1}, \cc_1, \prop) = \hh_{S_2}, \cc_2
            }
            {
            \ex(\expr_1 \oplus \expr_2, \tau_s, \symc, \prop) = \hh_{S_2}, \cc_2
            } &&  && \mycounter &\\
&\text{From Antecedent (3) and \textsc{Op-binary}}\\
& \inferrule*[]
            {
        \langle \expr_1, \context \rangle \Downarrow \val_1 \\\\
        \langle \expr_2, \context \rangle \Downarrow \val_2
        } 
        {\langle \expr_1 \oplus \expr_2, \context \rangle \Downarrow \val_1 \binop \val_2} &&  && \mycounter &\\
% & \inferrule*[]
%             {
%         \langle \expr_1, \fstore, \sstore, \hh_{S_2}, \cc_2 \rangle \downarrow \sval_1, \cc'_2 \\\\
%         \langle \expr_2, \fstore, \sstore, \hh_{S_2}, \cc'_2 \rangle \downarrow \sval_2, \cc''_2
%         } 
%         {\langle \expr_1 \oplus \expr_2, \fstore, \sstore, \hh_{S_2}, \cc_2 \rangle \downarrow \sval_1 \binop \sval_2. \cc''_2} && \text{From Antecedent (3) and OP-BINARY} && \mycounter &\\
& \text{From induction hypothesis using (1) (2), antecedents (1,4)} \\
& \ll \expr_1, \fstore, \store, \sstore, \hh_C, \hh_{S_1}, \cc_1 \gg \updownarrow \val_1, \sval''_1, \cc_3, \m_3&& && \mycounter&\\
&\expand(\sval'_1) && && \mycounter &\\
&\exists! X' \forall Y' \exists! Z' (\bigwedge_{t \in \dom(\hh_{S_1})}\hh_{S_1}(t) = \hh_C(t)) &&&& \mycounter &\\
&\store, \hh_C \prec_{\cc_1} \sstore, \hh_{S_1} &&&& \mycounter &\\
& \text{From induction hypothesis using (1,2,5) and antecedent (1)}\\
& \ll \expr_2, \fstore, \store, \sstore, \hh_C, \hh_{S_2}, \cc_2 \gg \updownarrow \val_2, \sval'_2, \cc_4, \m_4 && && \mycounter&\\
&\expand(\sval'_2) && && \mycounter &\\
&\exists! X'' \forall Y'' \exists! Z'' (\bigwedge_{t \in \dom(\hh_{S_2})}\hh_{S_2}(t) = \hh_C(t)) &&\text{Consequent (3)} && \mycounter &\\
&\store, \hh_C \prec_{\cc_2} \sstore, \hh_{S_2} && \text{Consequent (4)} && \mycounter &\\
& \hh_{S_1}, \cc_1 \leadsto^* \hh_{S_2}, \cc_2&& \text{From (1)} && \mycounter &\\
& \text{From Lemma~\ref{lemma:multistephextendnew} using (3), (5) and (11)}\\
& \ll \expr_1, \fstore, \store, \sstore, \hh_C, \hh_{S_2}, \cc_2 \gg \updownarrow \val_1, \sval'_1, \cc_5, \m_5&& && \mycounter&\\
& \text{From Lemma~\ref{lemma:ctogether} using (7),(9) and (12)}\\
&\langle \expr_2, \fstore, \sstore, \hh_{S_2}, \cc_5 \rangle \downarrow \sval'_2, \cc_6 &&&& \mycounter&\\
& \exists! X''' \forall Y''' \exists!Z''' (CS(\sval'_1,\val_1) \wedge CS(\sval'_2,\val_2) \wedge \cc_6 \wedge \m \wedge \\
&\bigwedge_{t \in \dom(\hh_{S_2})} \hh_{S_2}(t) = \hh_C(t)) &&&& \mycounter &\\
&\langle \expr_1 \oplus \expr_2, \fstore, \sstore, \hh_{S_2}, \cc_2 \rangle \downarrow \sval'_1 \oplus \sval'_2, \cc_6 &&\text{From (12) and (13)} && \mycounter &\\
& \exists! X''' \forall Y''' \exists!Z''' (CS(\sval'_1 \oplus \sval'_2,\val_1 \oplus \val_2) \wedge \cc_6 \wedge \m \wedge \\
&\text{From (14) and \textsc{Val-constraints} rules}\\
&\bigwedge_{t \in \dom(\hh_{S_2})} \hh_{S_2}(t) = \hh_C(t)) && && \mycounter &\\
&\text{From (2), (15) and (16) }\\
& \ll \expr, \fstore, \store, \sstore, \hh_C, \hh_{S_2}, \cc_{2} \gg \updownarrow \val_1 \oplus \val_2, \sval'_1 \oplus \sval'_2, \cc_6, \m && \text{Consequent (1)} && \mycounter &\\
&\text{From (4), (8) and \textsc{Expanded-binary} }\expand(\sval'_1 \oplus \sval'_2) && \text{Consequent (2)} && \mycounter &\\
\end{align*}
\end{proof}

% \begin{lemma}
%     If 
%     \begin{enumerate}
%         \item $\ll \expr, \fstore, \store, \sstore, \hh_C, \hh_S, \cc \gg \updownarrow \val, \sval, \hat{\cc_1}, \m_1$
%         \item $\exists! X \forall Y \exists!Z (\bigwedge_{t \in \dom(\hh_S)} \hh_C(t) = \hh_S(t))$
%         \item $\tau_s, \symc, \prop \leadsto^* \hh'_S, \cc'$\
%         \item $\store, \hh_C \prec_{\cc} \sstore, \hh_S$
%     \end{enumerate}
%     then,
%     \begin{enumerate}
%         \item $\ll \expr, \fstore, \store, \sstore, \hh_C, \hh'_S, \cc' \gg \updownarrow \val, \sval', \cc'', \m''$
%         \item $\exists! X' \forall Y' \exists!Z' (\bigwedge_{t \in \dom(\hh'_S)} \hh_C(t) = \hh'_S(t))$
%         \item $\expand(\sval) \implies \expand(\sval')$
%         \item $\store, \hh_C \prec_{\cc'} \sstore, \hh'_S$
%     \end{enumerate}
% \end{lemma}
\begin{proof}[Proof of Lemma~\ref{lemma:multistephextendnew}]
    \item $\leadsto^*$ is recursively defined as a sequence of one of the following three operations: 
    \begin{enumerate}
        \item $\tau_s, \symc, \prop \models \expr \leadsto \hh'_S, \cc'$
        \item $\ex(\expr, \tau_s, \symc, \prop) = \hh'_S, \cc'$
        \item $\addg(n, \tau_s, \hh_S, \prop) = \hh'_S, \cc'$
    \end{enumerate}
    We can prove this lemma using structural induction on $\leadsto^*$.\\
    \textbf{Base Case: }$\addg(n, \tau_s, \hh_S, \prop) = \hh'_S, \cc'$\\
    There are two rules for $\addg(n, \tau_s, \hh_S, \prop) = \hh'_S, \cc'$:
    \begin{enumerate}
        \item ADD-NEURON-B 
        \begin{enumerate}
            \item $\hh'_S = \hh_S \quad \cc' = \cc$
            \item Consequents 1, 2 and 4 follow from (a) and Antecedents 1, 2 and 4.
            \item $\expand(\sval) \implies \expand(\sval)$
            \item Consequent 3 follows from (a) and (c)
        \end{enumerate}
        \item \textsc{Add-neuron-r}
        \begin{enumerate}
            \item $\inferrule*[]
        {
        n \not\in \hh_S \\\\
        \textsf{Metadata}[m_1, \cdots, m_k] \quad 
        \hh'_{S_0} = \hh_S \\\\
        \forall i \in [k], \hh_{S_i}' = \addg(n, m, \hh'_{S_{i-1}}) \\\\
        \textsf{Shape}[\var_1, \cdots, \var_l] \quad 
        \hh''_{S_0} = \hh'_{S_k} \\\\
        \forall i \in [l], \hh_{S_i}'' = \addg(n, \var_i, \tau_s, \hh''_{S_{i-1}})
        }
        {
        \addg(n, \tau_s, \hh_S, \prop, \cc) = \hh_{S_l}'', \cc \wedge \prop(n, \hh''_{S_l})
        }$
            \item Since the same symbolic variables are used for neurons in concrete operational semantics and symbolic operation semantics, $n \in \hh_C$.
            \item $X'$ can be defined as the constants in $\hh_S$ and the constants in the metadata and shape elements of $n$ that are added to form $\hh'_S = \hh''_{S_l}$
            \item $X' \setminus X$ consists of elements in the range of $\hh'_S$. We can create an assignment $m_{x'}$ to the set $X' \setminus X$. From (b), for each $x \in X' \setminus X$, there exists $t$ such that $\hh'_S(t) = x$ and $\hh'_C(t) = c$ for some constant $c$. We can ass $[x \mapsto c]$ to $m_{x'}$. 
            \item We can call the unique assignment to $X$ satisfying antecedent (2), $m_x$ and  create an assignment $m'_x = m_x \cup m_{x'}$. Using $m'_x$, given an arbitrary assignment to $Y$, every expression in the range of $\hh_C$ will have a fixed, constant value. 
            \item $Y' = Y \cup Y''$, where $Y''$ includes $n$ and new neurons and symbolic variables used to define the shape and metadata of $n$.
            \item Given an arbitrary assignment to $Y'$, we can call the unique assignment to $Z$ that satisfies Antecedent (2), $m_z$. We can construct $m_{z'}$. From (b) and (e), for each $z \in Z' \setminus Z$, there exists $t \in \hh'_S$, such that $\hh'_S(t) = z$.  We can add $[z \mapsto \hh_C(t)]$. Call $m'_z = m_z \cup m_{z'}$
            \item From (d) and (g), $\exists X' \forall Y' \exists Z' (\bigwedge_{t \in \dom(\hh'_S)} \hh'_S(t) = \hh_C(t))$
            \item From lemma~\ref{lemma:uniquexz} using (h), $\exists X'! \forall Y' \exists Z'! (\bigwedge_{t \in \dom(\hh'_S)} \hh'_S(t) = \hh_C(t))$. This is Consequent (2).
            \item For all $t \in \dom(\hh_S), \hh'_S(t) = 
            hh_S(t)$ because ADD-NEURON-R does not change any mappings within $\hh_S$. 
            \item $\prop(n, \hh'_S)$ is only defined on the shape and metadata of $n$.
            \item From (j) and (k), $\langle \expr, \fstore, \sstore, \hh'_S, \cc' \rangle \downarrow \sval, \cc'' \quad \sval' = \sval$
            \item Using (i), (l), Lemma~\ref{lemma:cexpand} and Antecedent (1), $\exists X'! \forall Y' \exists Z'! (CS(\sval,\val) \wedge \cc'' \wedge \m' \wedge \bigwedge_{t \in \dom(\hh'_S)} \hh'_S(t) = \hh_C(t))$
            \item Consequent (1) follows from Antecedent (1), (l), and (m).
            \item Consequent (2) follows from (l)\\
            \item Consequent (4) follows from (i) and Antecedents (2) and (4).
        \end{enumerate}
    \end{enumerate}
    \textbf{Base Case: } $\ex(\expr, \tau_s, \symc, \prop) = \hh'_S, \cc'$\\
    Consequents (1),(2),(3) and (4) directly follow from Lemma~\ref{lemma:hexpand} using Antecedents (1),(2),(3) and (4).\\
    \textbf{Base Case: } $\tau_s, \symc, \prop \models \expr \leadsto \hh'_S, \cc'$\\
    We will prove this by structural induction on $\expr$.\\
    \begin{enumerate}
        \item \textbf{Base Cases: }$\expr \equiv \constant $ or $ \expr \equiv \var$ \\
        \begin{enumerate}
            \item $\hh'_S = \hh_S \quad \cc' = \cc$
            \item Consequents 1, 2 and 4 follow from (a) and Antecedents 1, 2 and 4.
            \item $\expand(\sval) \implies \expand(\sval)$
            \item Consequent 3 follows from (a) and (c)
        \end{enumerate}

        \item \textbf{Induction Case: }$\expr \equiv \expr_1 \oplus \expr_2$ \\
        \begin{enumerate}
            \item $ \inferrule*[]
            {
            \tau_s, \symc, \prop \models \expr_1 \leadsto \hh_{S_1}, \cc_{1} \\\\
            \tau_s, \fstore, \sstore, \hh_{S_1}, \cc_1, \prop \models \expr_2 \leadsto \hh_{S_2}, \cc_2
            }
            {
            \tau_s, \symc, \prop \models \expr_1 \oplus \expr_2 \leadsto \hh_{S_2}, \cc_2
            } $ From Antecedent (3) and \textsc{G-binary}
            \item $ \inferrule*[]
            {
        \langle \expr_1, \context \rangle \Downarrow \val_1 \\\\
        \langle \expr_2, \context \rangle \Downarrow \val_2
        } 
        {\langle \expr_1 \oplus \expr_2, \context \rangle \Downarrow \val_1 \binop \val_2}$ from Antecedent (1) and \textsc{Op-binary}
            \item From (a), $\tau_s, \symc, \prop \models \expr_1 \leadsto \hh_{S_1}, \cc_{1}$
            \item From (a), $\tau_s, \fstore, \sstore, \hh_{S_1}, \cc_1 \prop \models \expr_2 \leadsto \hh_{S_2}, \cc_{2}$
            \item $ \inferrule*[]
            {
        \langle \expr_1, \symc \rangle \downarrow \sval_1 \\\\
        \langle \expr_2, \symc \rangle \downarrow \sval_2
        } 
        {\langle \expr_1 \oplus \expr_2, \symc \rangle \downarrow \sval_1 \binop \sval_2}$  from Antecedent (1) and \textsc{Sym-binary}
            \item From the induction hypothesis on $\expr$ using (b), (c) and Antecedents (1),(2) and (4):
            \begin{enumerate}
                \item $\ll \expr_1, \fstore, \store, \sstore, \hh_C, \hh_{S_1}, \cc_1 \gg \updownarrow \val_1, \sval'_1, \cc''_1, \m''_1$
                \item $\exists! X' \forall Y' \exists!Z' (\bigwedge_{t \in \dom(\hh_{S_1})} \hh_C(t) = \hh_{S_1}(t))$
                \item $\expand(\sval_1) \implies \expand(\sval'_1)$
                \item $\store, \hh_C \prec_{\cc_1} \sstore, \hh_{S_1}$
            \end{enumerate}
            \item From the induction hypothesis on $\expr$ using (b), (d), (i), (ii) and (iv):
            \begin{enumerate}
                \item $\ll \expr_2, \fstore, \store, \sstore, \hh_C, \hh_{S_2}, \cc_2 \gg \updownarrow \val_2, \sval'_2, \cc''_2, \m''_2$
                \item $\exists! X'' \forall Y'' \exists!Z'' (\bigwedge_{t \in \dom(\hh_{S_2})} \hh_C(t) = \hh_{S_2}(t))$ This is Consequent (2).
                \item $\expand(\sval_2) \implies \expand(\sval'_2)$
                \item $\store, \hh_C \prec_{\cc_2} \sstore, \hh_{S_2}$. 
                 This is Consequent (4).
            \end{enumerate}
            \item From the induction hypothesis on $\expr$ using (d), (f i), (f ii) and (f iv):
            \begin{enumerate}
                \item $\ll \expr_1, \fstore, \store, \sstore, \hh_C, \hh_{S_2}, \cc_2 \gg \updownarrow \val_1, \sval''_1, \cc'_1, \m'_1$
                \item $\expand(\sval'_1) \implies \expand(\sval''_1)$
            \end{enumerate}
            \item From Lemma~\ref{lemma:ctogether} using (g i) and (h i):
            \begin{enumerate}
                \item $\langle \expr_1, \fstore, \sstore, \hh_{S_2}, \cc'_1 \rangle \downarrow \sval'_2, \cc'_2 $ 
                \item $\exists! X''' \forall Y''' \exists! Z''' (CS(\val_1,\sval''_1) \wedge 
                CS(\val_2,\sval'_2) \wedge \cc'_2 \wedge \m \wedge \bigwedge_{t \in \dom(\hh_{S_2})} \hh_{S_2}(t) = \hh_C(t))$
            \end{enumerate}
            \item From (b), (h i), (i i), (i ii), and \textsc{Sym-binary},  $\ll \expr, \fstore, \store, \sstore, \hh_C, \hh_{S_2}, \cc_2 \gg \updownarrow \val_1 \oplus \val_2, \sval''_1 \oplus \sval'_2, \cc'_2, \m$. This is Consequent (1).
            \item From (f iii), (g iii), (h ii) and \textsc{Expanded-binary}, \\
            $\expand(\sval_1 \oplus \sval_2) \implies \expand(\sval''_1 \oplus \sval'_2)$. This is Consequent (3).
        \end{enumerate}

        \item The other induction cases are similar. \textsc{G-traverse} involves $\addg$ so the induction case for $\expr \equiv \var \cdot \trav(\dir, f_{c_1}, f_{c_2}, f_{c_3})\{\expr_1\} $ will use the proof of the first base case above ($\addg(n, \tau_s, \hh_S, \prop) = \hh'_S, \cc'$). \textsc{G-map} and \textsc{G-func-call} involve $\ex$, so the proof of the induction cases $\expr \equiv \expr\cdot\map(f_c) $ and $\expr \equiv f_c(\expr_1, \cdots \expr_n) $ will use the proof of the second base case above ($\ex(\expr, \tau_s, \symc, \prop) = \hh'_S, \cc'$).\\
    \end{enumerate}
\textbf{Induction Case: }  $ \inferrule*[]
            {
            \tau_s, \symc, \prop \leadsto^* \hh_{S_1}, \cc_{1} \\\\
            \tau_s, \fstore, \sstore, \hh_{S_1}, \cc_1, \prop  \leadsto^* \hh_{S_2}, \cc_{2}
            }
            {
            \tau_s, \symc, \prop \leadsto^* \hh_{S_2}, \cc_{2}
            } $
\begin{enumerate}
    \item $\tau_s, \symc, \prop \leadsto^* \hh_{S_1}, \cc_{1}$
    \item $\tau_s, \fstore, \sstore, \hh_{S_1}, \cc_1, \prop  \leadsto^* \hh_{S_2}, \cc_{2}$
    \item From the induction hypothesis. using (1) and Antecedents (1),(2) and (4):
    \begin{enumerate}
        \item $\ll \expr, \fstore, \store, \sstore, \hh_C, \hh_{S_1}, \cc_1 \gg \updownarrow \val, \sval', \cc'_1, \m'_1$
        \item $\exists! X' \forall Y' \exists!Z' (\bigwedge_{t \in \dom(\hh_{S_1})} \hh_C(t) = \hh_{S_1}(t))$
        \item $\expand(\sval) \implies \expand(\sval')$
        \item $\store, \hh_C \prec_{\cc_1} \sstore, \hh_{S_1}$
    \end{enumerate}
    \item From the induction hypothesis using (2), (3 a), (3 b) and (3 d):
    \begin{enumerate}
        \item $\ll \expr, \fstore, \store, \sstore, \hh_C, \hh_{S_2}, \cc_2 \gg \updownarrow \val, \sval'', \cc'_2, \m'_2$ This is Consequent (1)
        \item $\exists! X'' \forall Y'' \exists!Z'' (\bigwedge_{t \in \dom(\hh_{S_2})} \hh_C(t) = \hh_{S_2}(t))$ This is Consequent (2). 
        \item $\expand(\sval') \implies \expand(\sval'')$
        \item $\store, \hh_C \prec_{\cc_2} \sstore, \hh_{S_2}$ This is Consequent (4).
    \end{enumerate}
    \item Consequent (3) follows from (3 c) and (4 c).
\end{enumerate}

\end{proof}

\clearpage
\begin{proof}[Proof of lemma~\ref{lemma:overapprox}] We prove this by induction on the structure of $\expr$ \\ 
\\
\textbf{Base Cases}: \\ 
\fbox{$\mathbf{\expr \equiv \constant}$}
\setcounter{number}{1}
\begin{align*}
    &\text{From Antecedent (2)}\\
    &\text{From \textsc{G-const}, }\tau_s, \symc, \prop \models \expr \leadsto \hh_S, \cc && &&\mycounter & \\
    &\text{From \textsc{Sym-const}, }\langle \constant, \symc \rangle \downarrow \constant, \cc &&  && \mycounter & \\ 
    &\text{From Antecedent (3)}\\
    &\text{From \textsc{Op-const}, }\langle \constant, \context \rangle \Downarrow \constant, \cc &&  && \mycounter & \\ 
    &\text{From (1), (2), }\hh'_S = \hh_S \quad \cc' = \cc &&&& \mycounter\\
    % &\text{From (4), Antecedent (4), }\exists! X' \forall Y' \exists! Z' (\cc'' \wedge \bigwedge_{t \in \dom(\hh_S')}\hh'_S(t) = \hh_C(t)) && \text{Consequent (3)} && \\
    &\text{From (4), Antecedents (1,4) } \\
    &\exists!X'\forall Y' \exists! Z' (\constant = \constant \wedge \cc \wedge \bigwedge_{t \in \dom(\hh_S')}\hh'_S(t) = \hh_C(t) )  && && \mycounter & \\
    &\text{From (2), (5) and Antecedent (1) } \ll \expr, \fstore, \store, \sstore, \hh_C, \hh_S, \cc \gg \updownarrow \val, \sval, \cc', [] && \text{Consequent (2)} && &\\
    &\text{From (3), Antecedent (1), }\store, \hh_C \prec_{\cc'} \sstore, \hh'_S  && \text{Consequent (1)} && \\
    & \text{From (4), Antecedent (4) } \exists! X' \forall Y' \exists! Z' (\bigwedge_{t \in \dom(\hh_S)}\hh'_S(t) = \hh_C(t)) && \text{Consequent (3)} && &\\
\end{align*}
\fbox{$\mathbf{\expr \equiv \var}$}
\setcounter{number}{1}
\begin{align*}
    &\text{From \textsc{G-var}, }\tau_s, \symc, \prop \models \expr \leadsto \hh_S, \cc && \text{Antecedent (2)}&&\mycounter & \\
    &\text{From Antecedents (1) and (3)}\\
    &\text{From \textsc{Sym-var}, }\langle \constant, \symc \rangle \downarrow \sstore(\var), \cc &&  && \mycounter & \\ 
    &\text{From \textsc{Op-var}, }\langle \constant, \context \rangle \downarrow \store(\var), \cc && \text{Antecedent (3)} && \mycounter & \\ 
    &\text{From (1), (2), }\hh'_S = \hh_S \quad \cc' = \cc'' = \cc &&&& \mycounter\\
    % &\text{From (4), Antecedent (4), }\exists! X' \forall Y' \exists! Z' (\cc'' \wedge \bigwedge_{t \in \dom(\hh_S')}\hh'_S(t) = \hh_C(t)) && \text{Consequent (3)} && \\
    &\text{From (4), Antecedent (1, 4), }\store(\var), \hh_C \prec_{\cc''} \sstore(\var), \hh_S' && \text{Consequent (1)} && \mycounter & \\
    &\text{From (5), }\\
    &\exists!X''\forall Y'' \exists! Z'' (CS(\sstore(\var),\store(\var)) \wedge \cc \wedge \bigwedge_{t \in \dom(\hh_S')}\hh'_S(t) = \hh_C(t) )  && && \mycounter &\\
    &\text{From (2), (6) and Antecedent (1), }\\
    &\ll \expr, \fstore, \store, \sstore, \hh_C, \hh_S, \cc \gg \updownarrow \val, \sval, \cc', [] && \text{Consequent (2)} && &\\
    & \text{From (4), Antecedent (4) } \exists! X' \forall Y' \exists! Z' (\bigwedge_{t \in \dom(\hh_S)}\hh'_S(t) = \hh_C(t)) && \text{Consequent (3)} && &\\
\end{align*}

\newpage
\textbf{Induction cases:} \\
\fbox{$\mathbf{\expr \equiv \var \cdot \trav(\dir, f_{c_1}, f_{c_2}, f_{c_3})\{\expr_1\}}$}
\setcounter{number}{1}
\begin{align*}
    &\text{From Antecedent (2) and  \textsc{G-traverse}}\\
    & \tau_s, \symc, \prop \models \expr_1 \leadsto \hh_{S_0}, \cc_{0} &&  && \mycounter & \\
    &\text{From Antecedent (3), \textsc{Op-traverse}}\\
    & \langle \expr_1, \context \rangle \Downarrow \val_{0} &&  && \mycounter &\\
    &\text{From induction hypothesis using (1), (2),}\\
    &\text{Antecedents (1), (4),(5)} \\
    &\langle \expr_1, \fstore, \sstore, \hh_{S_0}, \cc_0 \rangle \downarrow \hat{\sval_0}, \bar{\cc}_0 &&&& \mycounter \\
    &\exists ! X' \forall Y' \exists ! Z' (\hat{\sval_0} = \val_0 \wedge \bar{\cc}_0 \wedge \bar{\m}_0 \wedge \\
    &\bigwedge_{t \in \dom(\hh_{S_0})}\hh_{S_0}(t) = \hh_C(t)) &&&& \mycounter &\\
    &\store, \hh_C \prec_{\cc_0} \sstore, \hh_{S_0} &&&& \mycounter & \\
    &\exists ! X' \forall Y' \exists ! Z' ( \bigwedge_{t \in \dom(\hh_{S_0})}\hh_{S_0}(t) = \hh_C(t)) &&&& \mycounter &\\
    &\forall i \in [j], \hh_{S_i}, \cc_i = \addg(n'_i, \tau_s, \hh_{S_{i-1}}, \prop, \cc_{i-1}) && \text{From \textsc{G-traverse}} && \mycounter & \\
    &\hh''_{S_0} = \hh_{S_j} \quad \cc''_0 = \cc_j &&\text{From \textsc{G-traverse}}&& \mycounter\\
    & \exists X!'' \forall Y'' \exists! Z'' ( \bigwedge_{t \in \dom(\hh''_{S_0})}\hh''_{S_0}(t) = \hh_C(t)) && \text{From (6,7,8)} && \mycounter & \\
    &\text{Since the concrete values satisfy the properties}\\
    & \exists X!'' \forall Y'' \exists! Z'' (\bigwedge_{t \in \dom(\hh''_{S_0})}\hh''_{S_0}(t) = \hh_C(t)) && \text{From (7,9)} && \mycounter & \\
    &\store, \hh_C \prec_{\cc''_0} \sstore, \hh''_{S_0} &&\text{From (5),(8),(10)}&& \mycounter \\
    & \tau_s, \fstore, \sstore, \hh''_{i-1}, C''_{i-1}, \prop \models f_{c_2}(n'_i, \sval_{b_i}) \leadsto \hh''_{S_i}, \cc''_{i} && \text{From \textsc{G-traverse}} && \mycounter & \\
    &\text{From Antecedent (3), \textsc{Op-traverse}}\\
    & \langle \var , \fstore, \store, \hh_C \rangle \Downarrow \val_{b_0} + \sum_{i=1}^j \val_{b_i} * n'_i &&  && \mycounter & \\
    &\text{From Antecedent (3), \textsc{Op-traverse}}\\
    & \langle f_{c_2}(n'_i, \val_{b_i}), \fstore, \store, \hh_C \rangle \Downarrow \val_i &&  && \mycounter & \\
    &\text{From induction hypothesis, using (9), (11), (12), (14)} && \text{By induction on $i$}\\
    &\text{and Antecedent (5)}\\
    &\langle f_{c_2}(n'_i, \sval_{b_i}), \fstore, \sstore, \hh''_{S_{i}}, \cc''_{i} \rangle \downarrow \hat{\sval_i}, \bar{\cc}''_i &&&& \mycounter &\\
    &\exists ! X_i \forall Y_i \exists ! Z_i (\hat{\sval_i} = \val_i \wedge \cc''_i \wedge \m''_i \wedge \\
    &\bigwedge_{t \in \dom(\hh''_{S_i})}\hh''_{S_i}(t) = \hh_C(t)) &&&& \mycounter &\\
    &\store, \hh_C \prec_{\cc''_i} \sstore, \hh''_{S_i} &&&& \mycounter & \\
    & \exists! X_i \forall Y_i \exists! Z_i (\bigwedge_{t \in \dom(\hh''_{S_i})}\hh''_{S_i}(t) = \hh_C(t)) && && \mycounter & \\
\end{align*}
\begin{align*}
    &\hh'''_{S_0} = \hh''_{S_j} \quad \cc'''_0 = \cc''_j &&&& \mycounter \\
    &\text{From \textsc{G-traverse}}\\
    &\tau_s, \fstore, \sstore, \hh'''_{S_{i-1}}, \cc'''_{i-1}, \prop \models f_{c_3}(n'_i, \sval_{b_i}) \leadsto \hh'''_{S_i}, \cc'''_i &&  && \mycounter & \\
    &\text{From Antecedent (5)}\\
    & \langle f_{c_3}(n'_i, \val_{b_i}), \fstore, \store, \hh_C \rangle \Downarrow \val'_i &&  && \mycounter & \\
    &\text{By induction on $i$}\\
    &\text{From induction hypothesis, using (17), (18), (19), (20), (21)} && \\
    &\text{and Antecedent (5)}\\
    &\langle f_{c_3}(n'_i, \sval_{b_i}), \fstore, \sstore, \hh'''_{S_{i}}, \cc'''_{i} \rangle \downarrow \hat{\sval'_i}, \overline{\cc_i'''} &&&& \mycounter &\\
    & \exists! X'_i \forall Y'_i \exists! Z'_i (\val'_i = \hat{\sval'_i} \wedge \overline{\cc'''_i} \wedge \overline{\m'''_i} \wedge \bigwedge_{t \in \dom(\hh'''_{S_i})}\hh'''_{S_i}(t) = \hh_C(t)) && && \mycounter & \\
    &\store, \hh_C \prec_{\cc'''_i} \sstore, \hh'''_{S_i} && && \mycounter & \\
    & \exists! X'_i \forall Y'_i \exists! Z'_i (\bigwedge_{t \in \dom(\hh'''_{S_i})}\hh'''_{S_i}(t) = \hh_C(t)) && && \mycounter & \\
    &\hh_S' = \hh'''_{S_j} \quad \cc' = \cc'''_j &&&& \mycounter & \\
    &\text{From (24), (26), }\store, \hh_C \prec_{\cc'} \sstore, \hh'_{S} && \text{Consequent 1}&& \mycounter & \\
    & \text{From (25), (26), }\exists X' \forall Y' \exists! Z' (\bigwedge_{t \in \dom(\hh'_{S})}\hh'_{S}(t) = \hh_C(t)) && \text{Consequent 3}&& \mycounter & \\
    &\text{From (1), (7), (12), (20), (26), } \\
    &\hh_{S_0}, \cc_0  \leadsto^* \hh'_S, \cc' &&&& \mycounter & \\
    &\forall i \in [j], \hh''_{S_i}, \cc''_i \leadsto^* \hh'_S, \cc' &&&& \mycounter & \\
    &\forall i \in [j], \hh'''_{S_i}, \cc'''_i \leadsto^* \hh'_S, \cc' &&&& \mycounter & \\
    &\text{From lemma~\ref{lemma:multistephextendnew} using (4), (6) and (29)}\\
    &\langle \expr_1, \fstore, \sstore, \hh'_S, \cc' \rangle \downarrow \sval_0, \hat{\cc''_0}  &&&& \mycounter & \\
    & \ll \expr_1, \fstore, \store, \sstore, \hh_C, \hh'_S, \cc' \gg \updownarrow \val_0, \sval_0, \hat{\cc''_0}, \hat{\m''_0}   &&&& \mycounter & \\
    &\text{From lemma~\ref{lemma:multistephextendnew} using (16), (18) and (30)}\\
    &\forall i \in [j], \langle f_{c_2}(n'_i,\sval_{b_i}), \fstore, \sstore, \hh'_S, \cc' \rangle \downarrow \sval_i, \ddot{\cc''_i}  &&&& \mycounter & \\
    &\ll f_{c_2}(n'_i,\sval_{b_i}), \fstore, \store, \sstore, \hh_C, \hh'_S, \cc' \rangle \downarrow \val_i, \sval_i, \ddot{\cc''_i}, \ddot{\m''_i}  &&&& \mycounter & \\
    &\text{From lemma~\ref{lemma:multistephextendnew} using (23), (25) and (31)}\\
    &\forall i \in [j], \langle f_{c_3}(n'_i,\sval_{b_i}), \fstore, \sstore, \hh'_S, \cc' \rangle \downarrow \sval'_i, \ddot{\cc'''_i}  &&&& \mycounter & \\
    &\ll f_{c_3}(n'_i,\sval_{b_i}), \fstore, \store, \sstore, \hh_C, \hh'_S, \cc' \rangle \downarrow \val'_i, \sval'_i, \ddot{\cc'''_i}, \ddot{\m'''_i}  &&&& \mycounter & \\
    &\text{From lemmas~\ref{lemma:ctogether} using (27),(28),(35) and (36) }\\
    &\langle f_{c_2}(n'_i, \sval_{b_i}), \fstore, \sstore, \hh'_S, \hat{\cc''_{i-1}} \rangle \downarrow \sval_i, \hat{\cc''_{i}} &&&& \mycounter &\\
    & \hat{\cc'''_0} = \hat{\cc''_j} &&&& \mycounter &\\
    &\langle f_{c_3}(n'_i, \sval_{b_i}), \fstore, \sstore, \hh'_S, \hat{\cc'''_{i-1}} \rangle \downarrow \sval'_i, \hat{\cc'''_{i}} &&&& \mycounter &\\
\end{align*}
\begin{align*}
    &\exists! \ddot{X} \forall \ddot{Y} \exists! \ddot{Z} (\bigwedge_{i=0}^j CS(\sval_i,\val_i) \wedge \bigwedge_{i=1}^j CS(\sval'_i,\val'_i) \wedge \hat{\cc'''_{j}} \wedge \m'''_j \wedge \bigwedge_{t \in \dom(\hh'_{S})}\hh'_{S}(t) = \hh_C(t)) &&&& \mycounter &\\
    &\val'' = \val_{b_0} + \sum^j_{i=1} \val_i \quad \sval'' = \sval_{b_0} + \sum^j_{i=1} If(\sval_i, \sval'_i, n'_i * \sval_{b_i})  &&&& \mycounter &\\
    &\text{From (41) and (42), }\\
    &\exists! \ddot{X} \forall \ddot{Y} \exists! \ddot{Z} (CS(\val'',\sval'') \wedge \hat{\cc'''_{j}} \wedge \m'''_j \wedge \bigwedge_{t \in \dom(\hh'_{S})}\hh'_{S}(t) = \hh_C(t)) &&&& \mycounter & \\
    &\store'' = \store[\var \mapsto \val''] \quad \sstore'' = \sstore[\var \mapsto \sval''] &&&& \mycounter & \\
    &\text{From (27), (43), (44), }\store'', \hh_C \prec_{{\cc'''_j}}\sstore'', \hh'_S &&&& \mycounter & \\
    &\text{From (45) and  lemma~\ref{lemma:cexpand} using (38, 39,40) } \store'', \hh_C \prec_{{\cc}}\sstore'', \hh_S &&  && \mycounter & \\
    &\text{From induction hypothesis using (46) and Antecedents (2), (3),(4) and (5)}\\
    &\langle \expr_1, \fstore, \sstore'', \hh'_S, \cc' \rangle \downarrow \sval''', \cc''' &&&& \mycounter & \\
    &\text{From lemma~\ref{lemma:csequential}, using (38), (40) and (47)} \langle \expr, \fstore, \sstore'', \hh'_S, \cc'''_j \rangle \downarrow \overline{\sval'''}, \hat{\cc'''} &&&& \mycounter & \\
    &\text{Antecedent (5) and (33)}\\
    & \cinvariant(\var.\trav(\dir,f_{c_1}, f_{c_2}, f_{c_3})\{\expr_1\}, \fstore, \sstore, \hh'_S, \cc') = \true, \hat{\cc''_2} && && \mycounter & \\
    & \langle \var.\trav(\dir, f_{c_1}, f_{c_2}, f_{c_3}), \fstore, \sstore, \hh'_S, \cc' \rangle \downarrow \sval_b, \ddot{\cc} &&  && \mycounter &\\
    &\text{From \textsc{Sym-traverse}}\\
    &\sval_b = \sval_{r_0} + \sum_{i = 1}^k \sval_{r_i} * n'_i \text{ where $\sval_{r_i}$ are fresh variables representing constants} &&  && \mycounter& \\
    &\text{From (33), }\exists! X_2 \forall Y_2 \exists! Z_2 (Cs(\sval_0,\val_0) \wedge \hat{\cc''_0} \wedge \hat{\m''_0} \wedge \bigwedge_{t \in \dom{\hh'_S}} \hh_C(t) = \hh'_S(t) ) &&  && \mycounter & \\
    &\text{From \textsc{Sym-traverse} }\sstore' = \sstore[\var \mapsto \sval_b] \quad \store' = \store[\var \mapsto \val] && && \mycounter & \\
    &\ll \expr_1, \fstore, \store', \sstore', \hh_C, \hh'_S, \hat{\cc''_0} \gg \updownarrow \val_{inv}, \sval_{inv}, \ddot{\cc}', \ddot{\m'} &&&& \mycounter &\\
    &\text{From \textsc{Op-traverse}, $\val$ is in the form $c_0 + \sum_{i=0}^k c_i * n_i$} &&  && \mycounter & \\
    & \text{Since we use the same symbolic variables to represent neurons }\\
    &\text{in symbolic and concrete expressions, }\\
    & \exists! X_3 \forall Y_2 (\sval_b = \val)&& && \mycounter &\\
    &\text{From (51) } X_3 \cap X_2 = \emptyset \quad X_4 = X_3 \cap X_2 &&  && \mycounter &\\
    &\text{From (51), (52) and (56)}\\
    & \exists! X_4 \forall Y_2 \exists! Z_2 ((\sval_b = \val) \wedge \sval_0 = \val_0 \wedge \hat{\cc''_0} \wedge \hat{\m''_0} \wedge \bigwedge_{t \in \dom{\hh'_S}} \hh_C(t) = \hh'_S(t)) &&  && \mycounter &\\
    &\text{From (54) and (58)}\\
    & \exists! X'_4 \forall Y'_2 \exists! Z'_2 ((\sval_b = \val) \wedge \ddot{\cc'} \wedge \m \wedge \bigwedge_{t \in \dom{\hh'_S}} \hh_C(t) = \hh'_S(t)) &&  && \mycounter &\\
\end{align*}
\begin{align*}
    & \text{The invariant is true for the symbolic output of traverse}\\
    &\text{and the symbolic output overapproximates the concrete output:} \\
    &\text{From (49) (54), } \langle \expr_1, \fstore, \store', \hh_C \rangle \Downarrow \true &&  && \mycounter &\\
    &\text{From (54) and (60)}\\
    & \exists! X'_4 \forall Y'_2 \exists! Z'_2 (\sval_{inv} = \true \wedge (\sval_b = \val) \wedge \ddot{\cc'} \wedge \m \wedge \bigwedge_{t \in \dom{\hh'_S}} \hh_C(t) = \hh'_S(t)) &&  && \mycounter &\\
    &\text{From (50), (54) and \textsc{Sym-traverse}}\\
    & \ddot{\cc} = \ddot{\cc} \wedge \sval_{inv} &&  && \mycounter &\\
    &\text{From (61) and (62)}\\
    &\exists! \hat{X'} \forall \hat{Y'} \exists! \hat{Z'} (\sval_b = \val \wedge \ddot{\cc} \wedge \m \wedge \bigwedge_{t \in \dom{\hh'_S}} \hh_C(t) = \hh'_S(t) ) &&  && \mycounter & \\
    & \text{From (53), (60) and Antecedent (3)}\\
    & \ll \expr, \fstore, \store, \sstore, \hh_C, \hh'_S, \cc' \gg \updownarrow \val, \sval_b, \ddot{\cc}, \m && &&\text{Consequent (2)} &\\
\end{align*}
\fbox{$\mathbf{\expr \equiv \expr_1 \oplus \expr_2}$}
\setcounter{number}{1}
\begin{align*}
    &\text{From Antecedent (2) and \textsc{G-binary}}\\
    & \inferrule*[]
    {
        \tau_s, \symc, \prop \models \expr_1 \leadsto \hh_{S_1}, \cc_{1} \\\\
        \tau_s, \fstore, \sstore, \hh_{S_1}, \cc_1, \prop \models \expr_2 \leadsto \hh_{S_2}, \cc_2
    }
    {
        \tau_s, \symc, \prop \models \expr_1 \oplus \expr_2 \leadsto \hh_{S_2}, \cc_2
    } &&  && \mycounter &\\
    &\text{From Antecedent (3) and \textsc{Op-binary}}\\
    & \inferrule*[]
    {
        \langle \expr_1, \context \rangle \Downarrow \val_1 \\\\
        \langle \expr_2, \context \rangle \Downarrow \val_2
    } 
    {
        \langle \expr_1 \oplus \expr_2, \context \rangle \Downarrow \val_1 \binop \val_2
    } &&  && \mycounter &\\
    & \text{From induction hypothesis using (1) (2),}\\
    &\text{antecedents (1,4,5)} \\
    & \store, \hh_C \prec_{\cc_1} \sstore, \hh_{S_1} && && \mycounter &\\
    & \ll \expr_1, \fstore, \store, \sstore, \hh_C, \hh_{S_1}, \cc_1 \gg \val_1, \hat{\sval_1}, \cc'_1, \m'_1 && && \mycounter &\\
    & \exists! X_1 \forall Y_1 \exists!Z_1 (\bigwedge_{t \in \dom(\hh_{S_1})} \hh_C(t) = \hh_{S_1}(t)) && && \mycounter &\\
    & \text{From induction hypothesis using (1,2,5,6)}\\
    &\text{ and antecedent (5)}\\
    & \store, \hh_C \prec_{\cc_2} \sstore, \hh_{S_2} && \text{Consequent 1} && \mycounter &\\
    & \ll \expr_2, \fstore, \store, \sstore, \hh_C, \hh_{S_2}, \cc_2 \gg \val_2, \sval_2, \cc'_2, \m'_2 && && \mycounter &\\
    & \exists! X_2 \forall Y_2 \exists!Z_2 (\bigwedge_{t \in \dom(\hh_{S_2})} \hh_C(t) = \hh_{S_2}(t)) && \text{Consequent 3} && \mycounter &\\
    & \hh_{S_1}, \cc_1 \leadsto^* \hh_{S_2}, \cc_2&& && \mycounter &\\
\end{align*}
\begin{align*}
    & \text{From lemma~\ref{lemma:multistephextendnew} using (4,5,9)}\\
    & \langle \expr_1, \fstore, \sstore, \hh_{S_2}, \cc_2 \rangle \downarrow \sval_1, \cc_3 && && \mycounter &\\
    & \ll \expr_1, \fstore, \store, \sstore, \hh_{S_2}, \cc_2 \gg \updownarrow \val_1, \sval_1, \cc_3, \m_3 && && \mycounter &\\
    & \exists! X''_1 \forall Y''_1 \exists! Z''_1 (\bigwedge_{t \in \dom(\hh_{S_2})} \hh_C(t) = \hh_{S_2}(t)) && && \mycounter &\\
    & \cc'_2 = \cc_2 \wedge \overline{\cc'_2} && \text{From lemma~\ref{lemma:cexpand} using (7)} && \mycounter &\\
    & \exists! X \forall Y \exists! Z (\cc_2 \wedge \bigwedge_{t \in \dom(\hh_{S_2})} \hh_C(t) = \hh_{S_2}(t)) && \text{From (7), (8) and (13)}&& \mycounter&\\
    & \text{From lemma~\ref{lemma:ctogether} using (7), (11) and (14)}\\
    & \ll \expr_1, \fstore, \store, \sstore, \hh_C, \hh_{S_2}, \cc_2 \gg \downarrow \sval_1, \cc_3, \m_3 && && \mycounter &\\
    & \ll \expr_2, \fstore, \store, \sstore, \hh_C, \hh_{S_2}, \cc_3 \gg \downarrow \sval_2, \cc_4, \m_4 && && \mycounter &\\
    & \exists! X''' \forall Y''' \exists!Z''' (CS(\sval_1,\val_1) \wedge CS(\sval_2,\val_2) \wedge \cc_4 \wedge \m \wedge \\
    &\bigwedge_{t \in \dom(\hh_{S_2})} \hh_S(t) = \hh_C(t)) && && \mycounter &\\
    & \exists! X''' \forall Y''' \exists!Z''' (CS(\sval_1 \oplus \sval_2,\val_1 \oplus \val_2) \wedge \cc_4 \wedge \m \wedge \\
    &\bigwedge_{t \in \dom(\hh_{S_2})} \hh_S(t) = \hh_C(t)) && \text{From (18)} && \mycounter &\\
    &\text{From \textsc{Sym-binary}, (15) and (16)}\\
    & \langle \expr_1 \oplus \expr_2, \fstore, \sstore, \hh_C, \hh_{S_2}, \cc_2 \rangle \downarrow \sval_1 \oplus \sval_2, \cc_4 &&   && \mycounter & \\
    &\text{From (2), (18) and (19)}\\
    &\ll \expr_1 \oplus \expr_2, \fstore, \store, \sstore, \hh_C, \hh_{S_2}, \cc_2 \gg \updownarrow \val_1 \oplus \val_2, \sval_1 \oplus \sval_2, \cc_4, \m && \text{Consequent 2} && &\\ 
\end{align*}
\fbox{$\mathbf{\expr \equiv \lp(\minimize, \expr_1, \expr_2)}$}
\setcounter{number}{1}
\begin{align*}
&\text{From Antecedent (2) and \textsc{G-solver}}\\
    & \inferrule*[]
    {
        \tau_s, \symc, \prop \models \expr_1 \leadsto \hh_{S_1}, \cc_{1} \\\\
        \tau_s, \fstore, \sstore, \hh_{S_1}, \cc_1, \prop \models \expr_2 \leadsto \hh_{S_2}, \cc_2
    }
    {
        \tau_s, \symc, \prop \models 
        \lp(\minimize, \expr_1, \expr_2) \leadsto \hh_{S_2}, \cc_2
    } &&  && \mycounter &\\
&\text{From Antecedent (3) and \textsc{Op-solver}}\\
    & \inferrule*[]
    {
        \langle \expr_1, \context \rangle \Downarrow \val_1 \\\\
        \langle \expr_2, \context \rangle \Downarrow \val_2
    } 
    {
        \langle \lp(\minimize, \expr_1, \expr_2), \context \rangle \Downarrow \minimize(\val_1, \val_2)
    } &&  && \mycounter &\\
    & \text{From induction hypothesis using (1) (2), antecedents (1,4,5)} \\
    & \store, \hh_C \prec_{\cc_1} \sstore, \hh_{S_1} && && \mycounter &\\
    & \ll \expr_1, \fstore, \store, \sstore, \hh_C, \hh_{S_1}, \cc_1 \gg \val_1, \hat{\sval_1}, \cc'_1, \m'_1 && && \mycounter &\\
    & \exists! X_1 \forall Y_1 \exists!Z_1 (\bigwedge_{t \in \dom(\hh_{S_1})} \hh_C(t) = \hh_{S_1}(t)) && && \mycounter &\\
\end{align*}
\begin{align*}
    & \text{From induction hypothesis using (1,2,5,6) and antecedent (5)}\\
    & \store, \hh_C \prec_{\cc_2} \sstore, \hh_{S_2} && \text{Consequent 1} && \mycounter &\\
    & \ll \expr_2, \fstore, \store, \sstore, \hh_C, \hh_{S_2}, \cc_2 \gg \val_2, \sval_2, \cc'_2, \m'_2 && && \mycounter &\\
    & \exists! X_2 \forall Y_2 \exists!Z_2 (\bigwedge_{t \in \dom(\hh_{S_2})} \hh_C(t) = \hh_{S_2}(t)) && \text{Consequent 3} && \mycounter &\\
    & \hh_{S_1}, \cc_1 \leadsto^* \hh_{S_2}, \cc_2&& && \mycounter &\\
    & \text{From lemma~\ref{lemma:multistephextendnew} using (4,5,9)}\\
    & \langle \expr_1, \fstore, \sstore, \hh_{S_2}, \cc_2 \rangle \downarrow \sval_1, \cc_3 && && \mycounter &\\
    & \ll \expr_1, \fstore, \store, \sstore, \hh_{S_2}, \cc_2 \gg \updownarrow \val_1, \sval_1, \cc_3, \m_3 && && \mycounter &\\
    & \exists! X''_1 \forall Y''_1 \exists! Z''_1 (\bigwedge_{t \in \dom(\hh_{S_2})} \hh_C(t) = \hh_{S_2}(t)) && && \mycounter &\\
    &\text{From lemma~\ref{lemma:cexpand} using (7)}\\
    & \cc'_2 = \cc_2 \wedge \overline{\cc'_2} && && \mycounter &\\
    &\text{From (7), (8) and (13)}\\
    & \exists! X''' \forall Y''' \exists! Z''' (\cc_2 \wedge \bigwedge_{t \in \dom(\hh_{S_2})} \hh_C(t) = \hh_{S_2}(t)) && && \mycounter&\\
    & \text{From lemma~\ref{lemma:ctogether} using (7), (11) and (14)}\\
    & \ll \expr_1, \fstore, \store, \sstore, \hh_C, \hh_{S_2}, \cc_2 \gg \downarrow \sval_1, \cc_3, \m_3 && && \mycounter &\\
    & \ll \expr_2, \fstore, \store, \sstore, \hh_C, \hh_{S_2}, \cc_3 \gg \downarrow \sval_2, \cc_4, \m_4 && && \mycounter &\\
    & \exists! X \forall Y \exists!Z (\sval_1 = \val_1 \wedge \sval_2 = \val_2 \wedge \cc_4 \wedge \m \wedge \bigwedge_{t \in \dom(\hh_{S_2})} \hh_S(t) = \hh_C(t)) && && \mycounter &\\
    & \cc_2 \implies \sat(\sval_2) && \text{Antecedent (5)} && \mycounter&\\
    & \text{From (15), (16) and (18) using \textsc{Sym-solver}}\\
    & \langle \lp(\minimize, \expr_1, \expr_2), \fstore, \sstore, \hh_{S_2}, \cc_2 \rangle \downarrow \sval, \cc_4 \wedge (\sval_2 \implies \sval \le \sval_1) && && \mycounter &\\
    &\text{From (2)}\\
    & \val = \minimize(\val_1, \val_2) &&  && \mycounter & \\
    &\text{From \textsc{Sym-solver}, $\sval$ is a fresh variable } &&  && \mycounter &\\
    &\text{From (17) and (21)}\\
    & \exists! X \forall Y \exists!Z (\sval = \val \wedge \bigwedge_{i=1}^2 \sval_i = \val_i \wedge \cc_4 \wedge \m \wedge \bigwedge_{t \in \dom(\hh_{S_2})} \hh_S(t) = \hh_C(t)) && && \mycounter &\\
    &\text{From (20)}\\
    & \text{$\val$ is the minimum value of $\val_1$ under $\val_2$: } \val_2 \implies \val \leq \val_1 &&  && \mycounter &\\
    &\text{From (22),(23)}\\
    & \exists X! \forall Y \exists! Z (\val_2 \implies \val \leq \val_1 \wedge \sval = \val \wedge \sval_1 = \val_1 \wedge \sval_2 = \val_2 &&  &&  &\\
    & \wedge \cc_4 \wedge \m \wedge \bigwedge_{t \in \dom(\hh_{S_2})} \hh_C(t) = \hh_{S_2}(t)) &&  && \mycounter &\\
\end{align*}
\begin{align*}
    &\text{From (24)}\\
    & \exists X! \forall Y \exists! Z (\sval_2 \implies \sval \leq \sval_1 \wedge \sval = \val \wedge \sval_1 = \val_1 \wedge \sval_2 = \val_2  &&&&  &\\
    & \wedge \cc_4 \wedge \m \wedge \bigwedge_{t \in \dom(\hh_{S_2})} \hh_C(t) = \hh_{S_2}(t)) &&  && \mycounter &\\
    &\text{From (25)}\\
    & \exists! X \forall Y \exists! Z (\sval = \val \wedge \sval_2 \implies \sval \leq \sval_1 \wedge \cc_4 \wedge \m \wedge \bigwedge_{t \in \dom(\hh_{S_2})} \hh_C(t) = \hh_{S_2}(t)) &&  && \mycounter &\\
    & \text{From (2), (19) and (26)}\\
    &\ll \tau_s, \fstore, \store, \sstore, \hh_C, \hh_{S_2}, \cc_2 \gg \updownarrow \val, \sval, \cc_4 \wedge \sval_2 \implies \sval \leq \sval_1, \m &&&& \text{Consequent 2}  &\\
\end{align*}
\fbox{$\mathbf{\expr \equiv \expr\cdot\map(f_c)}$}
\setcounter{number}{1}
\begin{align*}
&\text{From Antecedent (2) and \textsc{G-map}}\\
    & \inferrule*[]
    {
            \tau_s, \scontext, \prop \models \expr \leadsto \hh_{S_0}, \cc_0 \\\\
            \ex(\expr, \tau_s, \sstore, \hh_{S_0}, \cc_0, \prop) = \hh_{S}', \cc' \\\\
            \langle \expr, \fstore, \sstore, \hh'_S, \cc' \rangle \downarrow \sval, \_ \\\\
            \expand(\sval) = \true \\\\
            \foo(\tau_s, \fstore, \sstore, \hh'_S, \cc', \prop, f_c, \sval) = \hh''_S, \cc'''
            }
            {
            \tau_S, \symc, \prop \models \expr\cdot\map(f_c) \leadsto \hh_S'', \cc''' 
            } &&  && \mycounter &\\
&\text{From Antecedent (3) and \textsc{Op-map}}\\
    & \inferrule*[]
    {
        \langle \expr, \context \rangle \Downarrow \bval' \quad 
        \bval' = \constant_0 + \sum_{i=0}^{i=l} \constant_i \cdot \ver_i \\\\
        \forall i \in [l], \ \langle f_c(\ver_i, \constant_i), \context \rangle  \Downarrow \val_i \quad 
        \bval = \constant_0 + \sum_{i=0}^{i=l} \val_i
    }
    {
        \langle \expr.\map(f_c), \context \rangle \Downarrow \bval
    } &&  && \mycounter &\\
    &\text{From (1), }\ex(\expr, \tau_s, \sstore, \hh_{S_0}, \cc_0 \prop) = \bar{\hh_{S}'}, \bar{\cc_1} &&&& \mycounter & \\
    &\text{From (2), }\langle \expr, \context \rangle \Downarrow \bval' &&&& \mycounter & \\
    & \text{From induction hypothesis using (1),(4) and Antecedents (1),(4) and (5)}\\
    & \ll \expr, \fstore, \store, \sstore, \hh_C, \hh_{S_0}, \cc_0 \gg \updownarrow \val'_b,\bar{\sval^\#}, \bar{\cc_2^\#}, \bar{\m_2^\#}  && && \mycounter &\\
    & \text{From lemma~\ref{lemma:hexpand} using (3) (5), antecedents (1,4)} \\
    &\expand(\bar{\sval'}) &&&& \mycounter & \\
    & \ll \expr, \fstore, \store, \sstore, \hh_C, \bar{\hh'_{S}}, \bar{\cc_1} \gg \updownarrow \val'_b,\bar{\sval'}, \bar{\cc_2}, \bar{\m_2}  && && \mycounter &\\
    & \exists! X_1 \forall Y_1 \exists!Z_1 (\bigwedge_{t \in \dom(\bar{\hh_{S'}})} \hh_C(t) = \bar{\hh'_{S}}(t)) && && \mycounter &\\
    &\store, \hh_C \prec_{\bar{\cc_1}} \sstore, \bar{\hh'_S} &&&& \mycounter & \\
    &\foo(\tau_s, \fstore, \sstore, \bar{\hh'_S}, \bar{\cc_2}, \prop, f_c, \bar{\sval'}) = \bar{\hh''_S}, \bar{\cc_3} && \text{From (1)} && \mycounter & \\
    &\forall i \in [l], \ \langle f_c(\ver_i, \constant_i), \context \rangle  \Downarrow \val_i && \text{From (2)} && \mycounter & \\
\end{align*}
\begin{align*}
    & \text{We can do induction on }\height(\bar{\sval'}) \\
    & \textbf{Base Case: } \height(\bar{\sval'}) = 0\\
    %&\text{By induction on }\height(\bar{\sval'})\text{ using (6), (8), (9), (10), (11)} &&&& & \\
    & \bar{\sval'} = \bar{\sval'_{b_0}} + \sum_{i=1}^j \alpha'_i * \bar{\sval'_{b_i}} \text{ where $\alpha'_i = n'_i$ or $\alpha'_i = \epsilon'_i$} && \text{From (10)} && \mycounter &\\
    &\text{Since the variables $n'_i$ and $\epsilon'_i$ are stored in $Y$} \\
    &\text{and shared between the concrete and symbolic expressions}\\
    & j = l \quad \ver_i = \alpha'_i && \text{From (7) and (11)} && \mycounter &\\
    & \fstore[f_c] = (\var_1, \var_2), \expr_c && \text{From (10)} && \mycounter &\\
    & \forall i \in [j], \sstore_i = [\var_1 \mapsto \bar{\sval'_{b_i}}, \var_2 \mapsto m_i] \quad m_i = n'_i \text{ or } m_i = \epsilon'_i && && \mycounter & \\
    & \forall i \in [j], \store_i = [\var_1 \mapsto c_i, \var_2 \mapsto \ver_i] && \text{From (13)} && \mycounter & \\
    &\forall i \in [j], \store_i, \hh_C \prec_{\bar{\cc_1}} \sstore_i, \bar{\hh'_S} && \text{From (7) and (9)} && \mycounter &\\
    & \hh_{S_0} = \bar{\hh'_S} \quad \cc_0 = \bar{\cc_1} && && \mycounter &\\
    & \tau_s, \fstore, \sigma_i, \hh_{S_{i-1}}, \cc_{i-1} \vdash \expr_c \leadsto \hh_{S_i}, \cc_i  &&\text{From (10) and (15)}&& \mycounter & \\
    & \text{From induction hypothesis using (8), (17), (19) and }\\
    &\text{\textsc{G-map-poly}, \textsc{G-map-sym} and Antecedent (5)} \\
    & \ll \expr_c, \fstore, \store_i, \sstore_i, \hh_C, \hh_{S_{i-1}}, \cc_{i-1} \gg \updownarrow \val_i, \overline{\sval''_{b_i}}, \cc'_i, \m'_i && && \mycounter &\\
    &\exists!X'_i \forall Y'_i \exists! Z'_i (\bigwedge_{t \in \dom(\hh_{S_i})} \hh_{S_i}(t) = \hh_C(t)) && && \mycounter &\\
    &\store_i, \hh_C \prec_{\cc_{i}} \sstore_i, \hh_{S_i}&& && \mycounter &\\
    &\bar{\hh'_S}, \bar{\cc_1} \leadsto^* \hh_{S_j}, \cc_j && \text{From (18) and (19)} && \mycounter &\\
    &\bar{\hh_{S_{i-i}}}, \cc_{i-1} \leadsto^* \hh_{S_j}, \cc_j && \text{From (18) and (19)} && \mycounter &\\
    & \text{From (6) and lemma~\ref{lemma:multistephextendnew} using (7), (8) and (23)} \\
    &\ll \expr, \fstore, \store, \sstore, \hh_C, \hh_{S_j}, \cc_j \gg \updownarrow \val'_b, \bar{\sval^\#}, \hat{\cc_0}, \hat{\m_0} \quad \expand(\bar{\sval^\#}) && && \mycounter &\\
    &\text{From lemma~\ref{lemma:multistephextendnew} using (20), (21) and (24)} \\
    & \ll \expr_c, \fstore, \store_i, \sstore_i, \hh_C, \hh_{S_{j}}, \cc_{j} \gg \updownarrow \val_i, \sval''_{b_i}, \hat{\cc_i}, \hat{\m_i} &&&& \mycounter &\\
    &\text{From lemma~\ref{lemma:ctogether} using (25) and (26)} \\
    &\cc''_0 = \hat{\cc_0} &&&& \mycounter &\\
    &\forall i \in [j], \langle \expr_c, \fstore, \sstore_i, \hh_{S_{j}}, \cc''_{i-1} \rangle \downarrow \sval''_{b_i}, \cc''_i &&&& \mycounter &\\
    &\exists! X_2 \forall Y_2 \exists! Z_2 (CS(\bar{\sval^\#},\ver'_b) \wedge \bigwedge_{i=1}^j CS(\sval''_{b_i},\ver_i)  &&&&  &\\
    &\wedge \cc''_j \wedge \m \wedge \bigwedge_{t \in \dom(\hh_{S_j})} \hh_{S_j}(t) = \hh_C(t)) &&&& \mycounter &\\
    & \exists! X_2 \forall Y_2 \exists! Z_2 (CS(\sval'_{b_0},c_0) \wedge \bigwedge_{i=1}^j CS(\sval''_{b_i},\ver_i) &&  &&  &\\
    & \wedge \cc''_j \wedge \m \wedge \bigwedge_{t \in \dom(\hh_{S_j})} \hh_{S_j}(t) = \hh_C(t)) &&\text{From (7) and (29) } && \mycounter &\\
\end{align*}
\begin{align*}
    &\text{From (30) } \exists! X_2 \forall Y_2 \exists! Z_2 (CS(\sval'_{b_0} + \sum_{i=0}^j \sval''_{b_i},c_0 + \sum_{i=0}^j \val_i)  &&  &&  &\\
    &\wedge \cc''_j \wedge \m \wedge \bigwedge_{t \in \dom(\hh_{S_j})} \hh_{S_j}(t) = \hh_C(t)) &&  && \mycounter &\\
    &\text{From (7), (12) and (28) }\langle \expr.\map(f_c), \fstore, \sstore, \bar{\hh'_S}, \bar{\cc_1} \rangle \downarrow \sval'_{b_0} + \sum_{i=0}^j \sval''_{b_i}, \cc''_j &&  && \mycounter &\\
    &\text{From (2), (31) and (32)}\\
    &\ll \expr.\map(f_c), \fstore, \store, \sstore, \hh_C, \bar{\hh'_S}, \bar{\cc_1} \gg \updownarrow c_0 + \sum_{i=0}^j \val_i, \sval'_{b_0} + \sum_{i=0}^j \sval''_{b_i}, \cc''_j, \m &&&& \text{Consequent (2)} &\\
    & \exists! X_j \forall Y_j \exists ! Z_j (\bigwedge_{t \in \dom(\hh_{S_j})} \hh_{S_j}(t) = \hh_C(t)) &&&& \text{Consequent (3)}  &\\
    &\text{From (8) and (9)}\\
    & \exists! X \forall Y \exists! Z(\bar{\cc'_1} \wedge \bigwedge_{t \in \dom(\bar{\hh'_S})}\bar{\hh'_S}(t) = \hh_C(t) \wedge \bigwedge_{t \in \dom(\sstore)}\sstore(t) = \store(t)) &&  && \mycounter &\\
    &\text{From (22)}\\
    & \exists! X'' \forall Y'' \exists! Z''(\cc_j \wedge \bigwedge_{t \in \dom(\hh_{S_j})}\hh_{S_j}(t) = \hh_C(t) \wedge \bigwedge_{t \in \dom(\sstore)}\sstore_j(t) = \store_j(t)) &&  && \mycounter &\\\\
    &\text{From (3) and (24)} \\
    & \cc_j = \cc \wedge \cc''' \quad \hh_S \subseteq \hh_{S_j} && && \mycounter &\\
    &\text{From (33), (34) and (35) } \store, \hh_C \prec_{\cc_j} \sstore, \hh_{S_j} &&&& \text{Consequent (1)}  &\\
    & \textbf{Induction Case: } \height(\bar{\sval'}) > 0\\
    &\text{From (10)}\\
    & \bar{\sval'} = If(\bar{\sval'_1}, \bar{\sval'_2}, \bar{\sval'_3}) &&  && \mycounter &\\
    & \exists!X_2 \forall Y_2 \exists!Z_2 (\bar{\sval'_1} \implies (CS(\bar{\sval'_2},\val_b)) \wedge \neg \bar{\sval'_1} \implies (CS(\bar{\sval'_3}, \val_b)) \wedge \\
    &\text{From (7) and (36)}\\
    &\bar{\cc_2} \wedge \bar{\m_2} \wedge \bigwedge_{t \in 
    dom(\bar{\hh'_S})}\bar{\hh'_S}(t) = \hh_C(t)) &&  && \mycounter &\\
    & \text{Given any assignment, $m$, to $X_2 \cup Y_2 \cup Z_2$, $\bar{\sval'_1}$ is either true or false} && && \mycounter &\\
    &\text{If $\bar{\sval'_1}$ is true under $m$} \\
    &\text{from the induction hypothesis on $\height(\bar{\sval'})$ using (1), (2), (7), (8) and (10)} \\
    & \map(\bar{\sval'_2}, f_c, \fstore, \sstore, \bar{\hh'_S}, \bar{\cc_2}) = \sval'_2, \cc_2 &&&& \mycounter&\\
    & \exists!X_{C_2} \forall Y_{C_2} \exists!Z_{C_2}( m \models (CS(\sval'_2,\val_b) \wedge \cc_2 \wedge \m_2)) && && \mycounter &\\
    &\exists! X_2 \forall Y_2 \exists! Z_2 (\bigwedge_{t \in \dom(\hh'_{S_2})} \hh'_{S_2}(t) = \hh_C(t)) \quad \store, \hh_C \prec_{\bar{\cc_2}} \sstore, \hh'_{S_2} && && \mycounter & \\
\end{align*}
\begin{align*}
    &\text{If $\bar{\sval'_1}$ is false under $m$} \\
    &\text{From the induction hypothesis on $\height(\bar{\sval'})$ using (1), (2), (7), (10), (40)} \\
    & \map(\bar{\sval'_3}, f_c, \fstore, \sstore, \bar{\hh'_S}, \cc_2) = \sval'_3, \cc_3 \quad \cc_3 = \bar{\cc_3}&&&& \mycounter&\\
    &\exists!X_{C_3} \forall Y_{C_3} \exists!Z_{C_3}( m \models (CS(\sval'_3,\val_b) \wedge \cc_3 \wedge \m_3)) && && \mycounter &\\
    &\exists! X_3 \forall Y_3 \exists! Z_3 (\bigwedge_{t \in \dom(\hh''_{S_2})} \hh''_{S_2}(t) = \hh_C(t)) \quad \store, \hh_C \prec_{\bar{\cc_3}} \sstore, \hh''_{S_2} && && \mycounter & \\
    &\text{From lemma~\ref{lemma:cexpand} using (31), (43),}\\
    &\text{\textsc{Sval-map-r, Sval-map-poly and Sval-map-sym}}\\
    &\cc_3 = \cc_2 \wedge \cc'_3 \quad \cc_2 = \bar{\cc'_1} \wedge \cc'_1 && && \mycounter &\\
    & \text{Given an assignment to $X \cup Y \cup Z$, $\m_2$ and $\m_3$ contain constraints where}\\
    &\text{the left hand side is a fresh variable and the right hand side is a constant}&& && \mycounter &\\
    &\text{From (39), (40), (42), (43), (45) and (46)}\\
    &\exists! X_3 \forall Y_3 \exists! Z_3 (\bar{\sval'} \implies (CS(\sval'_2,\val_b)) \wedge \neg \bar{\sval'} \implies \\
    &(\sval'_3 = \val_b \wedge \cc'_3 \wedge (\m_3 \setminus \m_2)) \wedge \cc_2 \wedge \m_2 \bigwedge_{t \in \dom(\bar{\hh''_S})} \bar{\hh''_S}(t) = \hh_C(t)) && && \mycounter&\\
    &\text{From (44), (46) and (47)}\\
    &\exists! X_3 \forall Y_3 \exists! Z_3 (\bar{\sval'} \implies (CS(\sval'_2,\val_b)) \wedge \neg \bar{\sval'} \implies \\
    &(\sval'_3 = \val_b) \wedge \cc_3 \wedge \m_3 \bigwedge_{t \in \dom(\bar{\hh''_S})} \bar{\hh''_S}(t) = \hh_C(t)) &&  && \mycounter &\\
    &\text{From (39), (42) and \textsc{Sym-map}}\\
    &\langle \expr.\map(f_c), \fstore, \sstore, \hh''_S, \bar{\cc_3} \rangle \downarrow If(\bar{\sval_1}, \sval'_2, \sval'_3), \cc_3 &&  && \mycounter &\\
    & \text{From (2), (48)}\\
    &\ll \expr.\map(f_c), \fstore, \store, \sstore, \hh_C, \hh''_S, \bar{\cc_3} \gg \updownarrow \val, If(\bar{\sval_1}, \sval'_2, \sval'_3), \cc_3, \m_3 &&&&\text{Consequent 2}  &\\
\end{align*}
\fbox{$\mathbf{\expr \equiv f_c(\expr_1, \cdots \expr_n)}$}
\setcounter{number}{1}
\begin{align*}
&\text{From Antecedent (2), \textsc{G-func-call}}\\
    & \inferrule*[]
    {
            \hh_{S_0} = \hh_S \quad \cc_0 = \cc \\\\
            \forall i \in [n], 
            \tau_s, \fstore, \sstore', \hh_{S_{i-1}}, \cc_{i-1}, \prop \models \expr_i \leadsto \overline{\hh_{S_i}}, \overline{\cc_{i}}  \\\\
            \ex(\expr_i, \tau_s, \fstore, \sstore, \overline{\hh_{S_i}}, \overline{\cc_{i}}) = \hh_{S_i}, \cc_i \\\\
            \hh'_S = \hh_{S_n} \quad \cc'_0 = \cc_n \\\\
            \forall i \in [n], \langle \expr_i, \fstore, \sstore, \hh'_S, \cc'_{i-1} \rangle \downarrow \sval_i, \cc'_i \\\\
            \fstore(f_c) = (\var_1, \cdots, \var_n), \expr \\\\
            \sstore' = \sstore[\var_1 \mapsto \sval_1, \cdots \var_n \mapsto \sval_n] \\\\
            \tau_s, \fstore, \sstore', \hh'_S, \cc_n, \prop \models \expr \leadsto \hh_S'', \cc'' 
            }
            {
            \tau_s, \scontext, \prop \models f_c(\expr_1, \cdots \expr_n) \leadsto \hh''_{S}, \cc''
            } &&  && \mycounter &\\
\end{align*}
\begin{align*}
&\text{From Antecedent (3), \textsc{Op-func-call}}\\
    & \inferrule*[]
    {
        \forall i \in [n], \ \langle \expr, \context \rangle \Downarrow \val_i \\\\
        \fstore(f_c) = (\var_1, \cdots, \var_n), \expr \\\\
        \store' = \store[\var_1 \mapsto \val_1, \cdots \var_n \mapsto \val_n] \\\\ 
        \langle \expr, \fstore, \store', \hh_C \rangle \Downarrow \val   
    }
    {
        \langle f_c(\expr_1, \cdots, \expr_n), \context \rangle  \Downarrow \val
    } &&  && \mycounter &\\
    &\text{From (1), }\ex(\expr_i, \tau_s, \fstore, \sstore, \overline{\hh_{S_i}}, \overline{\cc_{i}}) = \hh_{S_i}, \cc_i&&&& \mycounter & \\
    &\text{From (2), }\langle \expr_i, \context \rangle \Downarrow \val_i &&&& \mycounter & \\
    &\text{From induction hypothesis using (1), (4)}\\
    &\text{and Antecedents (1),(4) and (5)}\\
    &  \ll \expr_i, \fstore, \store, \sstore, \hh_C, \overline{\hh_{S_i}}, \overline{\cc_{i}} \gg \updownarrow \val_i, \bar{\sval_i^\#}, \bar{\cc_i^\#}, \bar{\m_i^\#} && && \mycounter &\\
    & \text{From lemma~\ref{lemma:hexpand} using (3) (5), antecedents (1,4)} \\
    & \ll \expr_i, \fstore, \store, \sstore, \hh_C, \hh_{S_i}, \cc_i \gg \updownarrow \val_i, \bar{\sval'_i}, \bar{\cc_i}, \bar{\m_i}&& && \mycounter &\\
    & \exists X_i! \forall Y_i \exists!Z_i (\bigwedge_{t \in \dom(\hh_{S_i})} \hh_C(t) = \hh_{S_i}(t)) && && \mycounter &\\
    &\store, \hh_C \prec_{\cc_i} \sstore, \hh_{S_i} &&&& \mycounter & \\
    &\hh'_S = \hh_{S_n}, \quad  \cc'_0 = \cc_n &&\text{From (1)} && \mycounter & \\
    &\text{From (7), (9), } \exists X' \forall Y' \exists!Z' (\bigwedge_{t \in \dom(\hh'_{S})} \hh_C(t) = \hh'_{S}(t)) && && \mycounter &\\
    &\text{From (3), (9), }\hh_{S_i}, \cc_i \leadsto^* \hh'_S, \cc'_0 &&&& \mycounter & \\
    &\text{From lemma~\ref{lemma:multistephextendnew} using  (6), (7) and (11)}\\
    &\ll \expr_i, \fstore, \store, \sstore, \hh_C, \hh'_{S}, \cc' \gg \updownarrow \val_i, \sval_i, \bar{\cc'_i}, \bar{\m'_i} &&&& \mycounter & \\
    &\text{From lemmas~~\ref{lemma:ctogether} using (8), (9), (10) and (12) }\\
    &\cc''_0 = \cc'_0 && && \mycounter &\\
    &\langle \expr_i, \fstore, \sstore, \hh'_S, \cc''_{i-1} \rangle \downarrow \sval_i, \cc''_i &&&& \mycounter & \\
    &\exists! X' \forall Y \exists! Z' (\bigwedge_{i=0}^j CS(\sval_i,\val_i) \wedge \cc''_j \wedge \m \wedge \bigwedge_{t \in \dom(\hh'_S)} \hh'_S(t) = \hh_C(t)) &&&& \mycounter & \\
    &\sstore' = \sstore[\var_1 \mapsto \sval_1, \cdots \var_n \mapsto \sval_i] \quad \store' = \store[\var_1 \mapsto \val_1, \cdots \var_n \mapsto \val_i]&& \text{From (1) and (2)}&& \mycounter & \\
    &\store, \hh_C \prec_{\cc'_0} \sstore, \hh'_S && \text{From (8) and (15)} && \mycounter &\\
    &\text{From induction hypothesis using (1), (2), (10) and (17)}\\
    &\text{and Antecedent (5)}\\
    & \langle \expr, \fstore, \sstore', \hh'_S, \cc'_0 \rangle \downarrow  \sval^\#, \bar{\cc''} &&&& \mycounter&\\
    & \store', \hh_c \prec_{\cc''} \sstore', \hh''_S &&  && \mycounter &\\
    & \exists! X'' \forall Y'' \exists! Z'' (\bigwedge_{t \in \dom(\hh''_S)} \hh''_S(t) = \hh_C(t)) && \text{Consequent (3)} && \mycounter &\\
\end{align*}
\begin{align*}
    &\ll \expr, \fstore, \store', \sstore', \hh_C, \hh'_S, \cc'  \gg \updownarrow \val, \sval^\#, \bar{\cc''}, \bar{\m''} && && \mycounter &\\
    & \hh'_S \subseteq \hh''_S  \quad \cc'' = \cc'_0 \wedge \cc'''&& \text{From (11) }&& \mycounter &\\
    & \text{From (7), (8), (19) (20) and (22) } \store, \hh_c \prec_{\cc''} \sstore, \hh''_S && \text{Consequent (1)} && &\\
    &\text{From lemma~\ref{lemma:multistephextendnew} using (1), (2), (10) and (21)}\\
    & \ll \expr, \fstore, \store', \sstore', \hh_C, \hh''_S, \cc''  \gg \updownarrow \val, \sval', \cc''_2, \m''_2  && && \mycounter &\\
    &\text{From lemma~\ref{lemma:multistephextendnew} using (1), (10) and (12)}\\
    &\ll \expr_i, \fstore, \store, \sstore, \hh_C, \hh''_S, \cc'' \gg \updownarrow \val_i, \bar{\sval_i}, \bar{\cc''_i}, \bar{\m''_i} && && \mycounter &\\
    &\text{From lemma~\ref{lemma:ctogether} using (20), (23) and (24) }\\
    &\cc''_{0} = \cc'' && && &\\
    &\langle \expr_i, \fstore, \sstore, \hh''_{S}, \cc''_{i-1} \rangle \downarrow \bar{\sval_i}, \cc''_i && && \mycounter &\\
    &  \langle \expr, \fstore, \sstore', \hh''_S, \cc''_j  \rangle \downarrow \sval', \cc'''_2 && && \mycounter &\\
    &\exists! X''' \forall Y''' \exists!Z'''(CS(\sval',\val) \wedge \cc''_j \wedge \m \wedge \bigwedge_{t \in \dom(\hh''S)} \hh''_S(t) = \hh_C(t)) && && \mycounter &\\
    &\text{From (24),(25), (26) and SYM-FUNC-CALL}\\
    & \langle f_c(\expr_1, \dots, \expr_n), \fstore, \sstore, \hh''_S, \cc'' \downarrow \sval', \cc''_j &&  && \mycounter&\\
    &\text{From (2), (27) and (28)}\\
    &\ll f_c(\expr_1, \dots, \expr_n), \fstore, \store, \sstore, \hh_C, \hh''_S, \cc'' \gg \updownarrow \val, \sval', \cc''_j, \m &&\text{Consequent (2)} && &\\
\end{align*}
\end{proof}

% \clearpage
\begin{proof}
Proof of lemma~\ref{lemma:graphexpand}. We prove this by induction on the structure of $\expr$.
    
\textbf{Base cases:} \\ 
$\mathbf{\expr \equiv \constant}$
\setcounter{number}{1}
\begin{align*}
    &\text{From \textsc{G-const} } \tau_s, \symc, \prop \models \constant \leadsto \hh_S, \cc && \text{Consequent 1} &&  &\\
    &\text{From Antecedent (6) } \store, \hh_C \prec_{\cc} \sstore, \hh_S&& \text{Consequent 2} &&  &\\
\end{align*}

\noindent
$\mathbf{\expr \equiv \var}$
\setcounter{number}{1}
\begin{align*}
    &\text{From \textsc{G-var} } \tau_s, \symc, \prop \models \var \leadsto \hh_S, \cc && \text{Consequent 1} &&  &\\
    &\text{From Antecedent (6) } \store, \hh_C \prec_{\cc} \sstore, \hh_S&& \text{Consequent 2} &&  &\\
\end{align*}

\newpage
\textbf{Induction Cases:} \\
$\mathbf{\expr_1 \oplus \expr_2}$
\setcounter{number}{1}
\begin{align*}
    &\text{From \textsc{T-binary-bool, T-binary-arith-1}, }\\
    &\text{ \textsc{T-binary-arith-2, T-binary-arith-3, T-comparison-1}, }\\
    &\text{\textsc{T-comparison-2} and \textsc{T-comparison}-3}\\
    &\Gamma, \tau_s \vdash \expr_1 : \types_1 \quad \bot \sqsubset \types_1 \sqsubset \top \quad \Gamma, \tau_s \vdash \expr_2 : \types_2 \quad \bot \sqsubset \types_2 \sqsubset \top  && && \mycounter &\\
    &\text{From induction hypothesis using (1) and}\\
    &\text{Antecedents (3-8)}\\
    &\tau_s, \fstore, \sstore, \hh_S, \cc, \prop \vdash \expr_1 \leadsto \hh'_S, \cc' && && \mycounter &\\
    &\text{From lemma~\ref{lemma:optype-checking} using Antecedents (1),(2),(3),(4) and (5)}\\
    & \langle \expr_1 \oplus \expr_2, \fstore, \store, \hh_C \rangle \Downarrow \val && && \mycounter &\\
    &\text{From (2) and \textsc{Op-binary}, } \langle \expr_1 \fstore, \store, \hh_C \rangle \Downarrow \val_1 &&  && \mycounter &\\
    &\text{From lemma~\ref{lemma:overapprox} using (2), (4), Antecedents (6-8)}\\
    & \store, \hh_C, \prec_{\cc'} \sstore, \hh'_S && && \mycounter &\\
    & \exists!X' \forall Y' \exists! Z' (\bigwedge_{t \in \dom(\hh'_S)} \hh'_S(t) = \hh_C(t)) && && \mycounter &\\
    &\text{From induction hypothesis using (1), (5), (6) and}\\
    &\text{Antecedents (3), (4),(5), and (8)}\\
    &\tau_s, \fstore, \sstore, \hh'_S, \cc', \prop \vdash \expr_2 \leadsto \hh''_S, \cc'' && && \mycounter &\\
    &\store, \hh_C \prec_{\cc''} \sstore, \hh''_S && \text{Consequent (2)} &&   &\\
    &\text{From (2), (7) and \textsc{G-binary} } \\ 
    &\tau_s, \fstore, \sstore, \hh_S, \cc, \prop \vdash \expr_1 \oplus \expr_2 \leadsto \hh''_S, \cc'' && \text{Consequent (1)} && &\\
\end{align*}

$\mathbf{f_c(\expr_1, \cdots, \expr_n)}$
\setcounter{number}{1}
\begin{align*}
    &\text{From lemma~\ref{lemma:optype-checking} using Antecedents (1),(2),(3),(4) and (5)}\\
    &\langle f_c(\expr_1, \cdots, \expr_n), \fstore, \store, \hh_C \rangle \Downarrow \val && && \mycounter &\\
    &\forall i \in [n], \langle \expr_i, \fstore, \store, \hh_C \rangle \Downarrow \val_i && \text{From \textsc{Op-func-call}}&& \mycounter &\\
    &\Gamma, \tau_s \vdash \expr_i : \types_i \quad \bot \sqsubset \types_i \sqsubset \top && \text{From \textsc{T-func-call}} && \mycounter &\\
    &\hh_{S_0} = \hh_S \quad \cc_0 = \cc && && \mycounter &\\
    &\text{We can do induction on $n$}\\
    &\textbf{Base Case: } n=0\\
    &\forall i \in [n], \tau_s, \fstore, \sstore, \hh_{S_{i-1}}, \cc_{1-i}, \prop \vdash \expr_i \leadsto \overline{\hh_{S_i}}, \overline{\cc_i} && && \mycounter &\\
    &\forall i \in [n], \ex(\expr_1, \tau_s, \fstore, \sstore, \overline{\hh_{S_i}}, \overline{\cc_{i}}) = \hh_{S_{i}}, \cc_i && && \mycounter &\\
    & \store, \hh_C \prec_{\cc} \sstore, \hh_S && \text{From Antecedent (6)} && \mycounter &\\
    & \exists!X \forall Y \exists! Z (\bigwedge_{t \in \dom(\hh_S)} \hh_S(t) = \hh_C(t)) && \text{From Antecedent (7)} && \mycounter &\\
\end{align*}
\begin{align*}
    &\textbf{Base Case: } n=1\\
    &\text{From the induction hypothesis on $\expr$ using (3)}\\
    &\text{and Antecedents (3),(4),(5),(6),(7), and (8)} \\
    &\tau_s, \fstore, \sstore, \hh_{S_0}, \cc_0, \prop \vdash \expr_1 \leadsto \overline{\hh_{S_1}}, \overline{\cc_1} && && \mycounter &\\
    &\ex(\expr_1, \tau_s, \fstore, \sstore, \overline{\hh_{S_1}}, \overline{\cc_{1}}) = \hh_{S_1}, \cc_1 && && \mycounter &\\
    &\text{From lemma~\ref{lemma:overapprox} using (2), (7), (8), (9),}\\
    &\text{and Antecedent (8)}\\
    & \store, \hh_C \prec_{\overline{\cc_1}} \sstore, \overline{\hh_{S_1}} &&  && \mycounter &\\
    & \exists!\overline{X_1} \forall \overline{Y_1} \exists! \overline{Z_1} (\bigwedge_{t \in \dom(\overline{\hh_{S_1}})} \overline{\hh_{S_1}}(t) = \hh_C(t)) &&  && \mycounter &\\
    & \ll \expr_1, \fstore, \store, \sstore, \hh_C, \hh_{S_1}, \cc_1 \gg \updownarrow \val_1, \sval_1, \hat{\cc_1}, \hat{\m_1} && && \mycounter &\\
    &\text{From lemma~\ref{lemma:hexpand} using (10), (11), (12) and (13)}\\
    & \store, \hh_C \prec_{\cc_1} \sstore, \hh_{S_1} && && \mycounter &\\
    & \exists!X_1 \forall Y_1 \exists! Z_1 (\bigwedge_{t \in \dom(\hh_{S_1})} \hh_{S_1}(t) = \hh_C(t)) &&  && \mycounter &\\
    &\textbf{Induction Case: } n > 1\\
    &\text{From the induction hypothesis on $n$:}\\
    &\forall i \in [n-1], \tau_s, \fstore, \sstore, \hh_{S_{i-1}}, \cc_{1-i}, \prop \vdash \expr_i \leadsto \overline{\hh_{S_i}}, \overline{\cc_i} && && \mycounter &\\
    &\forall i \in [n-1], \ex(\expr_1, \tau_s, \fstore, \sstore, \overline{\hh_{S_i}}, \overline{\cc_{i}}) = \hh_{S_{i}}, \cc_i && && \mycounter &\\
    & \store, \hh_C \prec_{\cc_{n-1}} \sstore, \hh_{S_{n-1}} && && \mycounter &\\
    & \exists!X_{n-1} \forall Y_{n-1} \exists! Z_{n-1} (\bigwedge_{t \in \dom(\hh_{S_{n-1}})} \hh_{S_{n-1}}(t) = \hh_C(t)) &&  && \mycounter &\\
    &\text{From the induction hypothesis on $\expr$ using (3), (18), (19),}\\
    &\text{and Antecedents (3), (4),(5), and (8)} \\
    &\tau_s, \fstore, \sstore, \hh_{S_{n-1}}, \cc_{n-1}, \prop \vdash \expr_n \leadsto \overline{\hh_{S_n}}, \overline{\cc_n} && && \mycounter &\\
    &\ex(\expr_n, \tau_s, \fstore, \sstore, \overline{\hh_{S_n}}, \overline{\cc_{n}}) = \hh_{S_n}, \cc_n && && \mycounter &\\
    &\text{From lemma~\ref{lemma:overapprox} using (2), (18), (19) and (20),}\\
    &\text{and Antecedent (8)}\\
    & \store, \hh_C \prec_{\overline{\cc_n}} \sstore, \overline{\hh_{S_n}} &&  && \mycounter &\\
    & \exists!\overline{X_n} \forall \overline{Y_n} \exists! \overline{Z_n} (\bigwedge_{t \in \dom(\overline{\hh_{S_n}})} \overline{\hh_{S_n}}(t) = \hh_C(t)) &&  && \mycounter &\\
    & \ll \expr_n, \fstore, \store, \sstore, \hh_C, \hh_{S_n}, \cc_1 \gg \updownarrow \val_n, \sval_n, \hat{\cc_n}, \hat{\m_n} && && \mycounter &\\
    &\text{From lemma~\ref{lemma:hexpand} using (21), (22), (23), and (24)}\\
    & \store, \hh_C \prec_{\cc_n} \sstore, \hh_{S_n} && && \mycounter &\\
    & \exists!X_n \forall Y_n \exists! Z_n (\bigwedge_{t \in \dom(\hh_{S_n})} \hh_{S_n}(t) = \hh_C(t)) &&  && \mycounter &\\
\end{align*}
\begin{align*}
    &\text{From induction on } n \textbf{:}\\
    &\forall i \in [n], \tau_s, \fstore, \sstore, \hh_{S_{i-1}}, \cc_{1-i}, \prop \vdash \expr_i \leadsto \overline{\hh_{S_i}}, \overline{\cc_i} && && \mycounter &\\
    &\forall i \in [n], \ex(\expr_1, \tau_s, \fstore, \sstore, \overline{\hh_{S_i}}, \overline{\cc_{i}}) = \hh_{S_{i}}, \cc_i && && \mycounter &\\
    & \store, \hh_C \prec_{\cc_{n}} \sstore, \hh_{S_{n}} && && \mycounter &\\
    & \exists!X_{n} \forall Y_{n} \exists! Z_{n} (\bigwedge_{t \in \dom(\hh_{S_{n}})} \hh_{S_{n}}(t) = \hh_C(t)) &&  && \mycounter &\\
    &\text{From \textsc{G-func-call}, }\hh_S, \cc \leadsto^* \hh_{S_n}, \cc_n &&  && \mycounter &\\
    &\text{From (7), (14), (25), }\forall i \in [n], \store, \hh_C \prec_{\cc_i} \sstore, \hh_{S_i} && && \mycounter &\\
    &\text{From (8), (15), (26)}\\
    &\forall i \in [n],\exists!X_i \forall Y_i \exists! Z_i (\bigwedge_{t \in \dom(\hh_{S_i})} \hh_{S_i}(t) = \hh_C(t)) && && \mycounter &\\
    &\text{From lemmas~\ref{lemma:overapprox},~\ref{lemma:multistephextendnew} using (2),(31),(32),(33),}\\
    &\text{and \textsc{G-func-call} and Antecedent (8)}\\
    &\forall i \in [n], \ll \expr_i, \fstore, \store, \sstore, \hh_C, \hh_{S_n}, \cc_n \gg \updownarrow \val_i, \sval_i, \cc''_i, \m''_i && && \mycounter&\\
    &\text{Since $\vars(\cc_n) \subseteq \vars(\hh_{S_n})$}\\
    &\text{From (29) and (30), }\exists!X_{n} \forall Y_{n} \exists! Z_{n} (\cc_n \wedge \bigwedge_{t \in \dom(\hh_{S_{n}})} \hh_{S_{n}}(t) = \hh_C(t))&&  && \mycounter &\\
    &\text{From lemma~\ref{lemma:ctogether} using (34) and (35)}\\
    &\cc'_0 = \cc_n && && \mycounter &\\
    &\forall i \in [n], \langle \expr_i, \fstore, \sstore, \hh_{S_n}, \cc'_{i-1} \rangle \downarrow \sval_i, \cc'_i && && \mycounter &\\
    &\exists! X' \forall Y' \exists! Z' (\bigwedge_{i = 1}^n CS(\val_i,\sval_i) \wedge \cc'_n \wedge \m \bigwedge_{t \in \dom(\hh_{S_n})} (\hh_{S_n}(t) = \hh_C(t))) && && \mycounter&\\
    &\cc'_n \implies \cc_n && && \mycounter &\\
    & \store' = \store[\var_1 \mapsto \val_1, \cdots \var_n \mapsto \val_n] \text{ and}\\
    &\text{From \textsc{Op-func-call} and \textsc{G-func-call}}\\
    &\sstore' = \sstore[\var_1 \mapsto \sval_1, \cdots \var_n \mapsto \sval_n]  &&  && \mycounter &\\
    &\text{From (29),(38),(39) and (40) }\store', \hh_C \prec_{\cc_n} \sstore', \hh_{S_n} &&  && \mycounter &\\
    &\text{From \textsc{T-func-call} } \Gamma(f_c) = (\Pi_i^n \types_i)\rightarrow \types'   && && \mycounter &\\
    &\text{From Antecedent (3) } \fstore(f_c) = (\var_1, \cdots, \var_n),\expr \quad \types, \tau_s \vdash \expr : \types' &&  && \mycounter &\\
    &\text{From the induction hypothesis on $\expr$ using }\\
    &\text{ (26),(41),(43) and Antecedents (3),(4),(5),(8)}\\
    &\tau_s, \fstore, \sstore', \hh_{S_n}, \cc_n, \prop \models \expr \leadsto \hh'_S, \cc' && && \mycounter &\\
    &\store', \hh_C \prec_{\cc'} \sstore', \hh'_S && && \mycounter&\\
    &\text{From (27),(28),(40), (44) and \textsc{G-func-call}}\\
    &\tau_s, \symc, \prop \models \expr \leadsto \hh_S', \cc' && && \mycounter&\\
\end{align*}
\begin{align*}
    &\text{From lemma~\ref{lemma:overapprox} using (1), (46) and}\\
    &\text{Antecedents (6),(7),(8)}\\
    &\store, \hh_C \prec_{\cc'} \sstore, \hh'_S && \text{Consequent (2)} && &\\
\end{align*}
\end{proof}

\noindent
\textbf{Induction Case:} \\
$\mathbf{\expr.\map(f)}$
\setcounter{number}{1}
\begin{align*}
    &\text{From \textsc{T-map-poly} and \textsc{T-map-sym}}\\
    &\Gamma, \tau_s \vdash \expr : \types_1 \quad \bot \sqsubset \types_1 \sqsubset \top && && \mycounter &\\
    &\text{From the induction hypothesis using (1) and }\\
    &\text{antecedents (3),(4),(5),(6),(7), and (8)}\\
    &\tau_s, \fstore, \sstore, \hh_S, \cc, \prop \models \expr \leadsto \hh_{S_0}, \cc_0 && && \mycounter &\\
    &\ex(\expr, \tau_s, \sstore, \hh_{S_0}, \cc_0, \prop) = \hh''_S, \cc'' && && \mycounter &\\
    &\text{From lemma~\ref{lemma:optype-checking} using Antecedents (1-5)}\\
    &\langle \expr.\map(f), \fstore, \store, \hh_C \rangle \Downarrow \val && && \mycounter &\\
    &\langle \expr, \fstore, \store, \hh_C \rangle \Downarrow \val_1 && \text{From \textsc{Op-map}} && \mycounter &\\
    &\text{From lemma~\ref{lemma:overapprox} using (2),(4)}\\
    &\text{and Antecedents (6),(7),(8)}\\
    &\store, \hh_C \prec_{\cc_0} \sstore, \hh_{S_0} && && \mycounter &\\
    &\exists!X_1 \forall Y_1 \exists! Z_1 (\bigwedge_{t \in \dom(\hh_{S_0})} \hh_{S_0}(t) = \hh_C(t)) && && \mycounter &\\
    &\ll \expr, \fstore, \store, \sstore, \hh_C, \hh_{S_0}, \cc_0 \gg \updownarrow \val_1, \sval_1, \cc_1, \m_1 && && \mycounter &\\
    &\text{From lemma~\ref{lemma:hexpand} using (3),(6),(7) and (8)}\\
    &\ll \expr, \fstore, \store, \sstore, \hh_C, \hh''_S, \cc'' \gg \updownarrow \val_1, \sval'_1, \cc'_1, \m'_1 && && \mycounter &\\
    &\exists! X_2 \forall Y_2 \exists!Z_2 (\bigwedge_{t \in \dom(\hh''_S)} \hh_C(t) = \hh''_S(t)) &&&& \mycounter &\\
    &\expand(\sval'_1) &&&& \mycounter &\\
    &\store, \hh_C \prec_{\cc''} \sstore, \hh''_S &&&& \mycounter &\\
    &\langle \expr, \fstore, \sstore, \hh''_S, \cc'' \rangle \downarrow \sval'_1,\cc'_1 &&\text{From (9)} && \mycounter &\\
    &\text{We can do induction on $\height(\sval'_1)$:}\\
    &\textbf{\text{Base Case: }} \height(\sval'_1) = 0 \\
    &\text{From (11) and \textsc{T-map-sym} and \textsc{T-map-poly}:}\\
    &\sval'_1 = \sval_{b_0} + \sum^j_{i=1} n_i * \sval_{b_i} \quad n=n'_i\text{ or }n=\epsilon'_i  && && \mycounter &\\
    &\text{From \textsc{Op-map, T-map-sym, T-map-poly}}\\
    &\fstore(f) = (\var_1,\var_2),\expr' &&  && \mycounter &\\
\end{align*}
\begin{align*}
    &\text{From (14), \textsc{G-map-sym, G-map-poly}}\\
    &\sstore_i = [\var_1 \mapsto \sval_{b_i}, \var_2 \mapsto n_i]  &&  && \mycounter &\\
    &\sval_1 = \constant_0 + \sum_{i=0}^{i=l} \constant_i \cdot \ver_i && \text{From (4) and OP-MAP} && \mycounter &\\
    &\text{From (15), \textsc{Op-map} and \textsc{Op-func-call}}\\
    &\store_i = [\var_1 \mapsto \constant_i, \var_2 \mapsto \ver_i] &&  && \mycounter &\\
    &\text{Since we use the same symbolic variables to}\\
    &\text{represent neurons and other symbolic variables in }\\
    &\text{concrete and symbolic operational semantics}\\
    &\store_i, \hh_C \prec{\cc''} \sstore_i, \hh''_S \quad j=l && \text{From (9) and (11)} && \mycounter &\\
    &\text{From \textsc{T-map-sym, T-map-poly, T-func}}\\
    &\Gamma, \tau_s \vdash \expr' : \types_2 \quad \bot \sqsubset \types_2 \sqsubset \top &&  &&\mycounter &\\
    &\text{From \textsc{T-func-call} and antecedent (4), }\store_i \sim \Gamma_i &&  && \mycounter&\\
    &\text{We can do induction on $j$:}\\
    &\text{Base Case: } j=1 \\
    &\text{From the induction hypothesis on $\expr$ using}\\
    &\text{(10),(19),(20),(21) and antecedents (3),(5) and (8)}\\
    &\tau_s, \fstore, \sstore_1, \hh''_{S}, \cc'' \vdash \expr' \leadsto \hh_{S_1}, \cc_1 && && \mycounter &\\
    &\text{From lemma~\ref{lemma:overapprox} using (4),(10),(19) and (22),}\\
    &\text{and Antecedent (8)}\\
    &\exists! X'_1 \forall Y'_1 \exists! Z'_1 (\bigwedge_{t \in \dom(\hh_{S_1})} \hh_{S_1}(t) = \hh_C(t)) && && \mycounter&\\
    &\store_1, \hh_C \prec_{\cc_1} \sstore_1, \hh_{S_1} && && \mycounter &\\
    &\store, \hh_C \prec_{\cc_1} \sstore, \hh_{S_1} && \text{From (10),(12) and (23)} && \mycounter &\\
    &\text{Base Case: } j=1 \\
    &\text{From the induction hypothesis on $j$}\\
    &\store, \hh_C \prec_{\cc_{j-1}} \sstore, \hh_{S_{j-1}} && && \mycounter &\\
    &\exists! X'_{j-1} \forall Y'_{j-1} \exists! Z'_{j-1} (\bigwedge_{t \in \dom(\hh_{S_{j-1}})} \hh_{S_{j-1}}(t) = \hh_C(t)) && && \mycounter&\\
    &\text{From the induction hypothesis on $\expr$ using}\\
    &\text{(20),(21),(26),(27) and antecedents (3), (5) and (8)}\\ 
    &\tau_s, \fstore, \sstore_j, \hh_{S_{j-1}}, \cc_{j-1} \vdash \expr' \leadsto \hh_{S_j}, \cc_j && && \mycounter &\\
    &\text{From lemma~\ref{lemma:overapprox} using (4),(26),(27) and (28),}\\
    &\text{and Antecedent (8)}\\
    &\exists! X'_j \forall Y'_j \exists! Z'_j (\bigwedge_{t \in \dom(\hh_{S_j})} \hh_{S_j}(t) = \hh_C(t)) && && \mycounter&\\
\end{align*}
\begin{align*}
    &\store_1, \hh_C \prec_{\cc_j} \sstore_j, \hh_{S_j} && && \mycounter &\\
    &\store, \hh_C \prec_{\cc_j} \sstore, \hh_{S_j} && \text{From (10),(12) and (29)} && \mycounter &\\
    &\text{End of induction on $j$. Form induction on $j$:}\\
    &\tau_s, \fstore, \sstore_j, \hh_{S_{j-}}, \cc_{j-1} \vdash \expr' \leadsto \hh_{S_j}, \cc_j && && \mycounter &\\
    &\text{From (2),(3),(11),(13),(32) and G-MAP} \\
    & \tau_s, \fstore, \sstore_j, \hh_{S}, \cc \vdash \expr.\map(f) \leadsto \hh_{S_j}, \cc_j  && \text{Consequent (1)}&& \mycounter &\\\
    &\text{From lemma~\ref{lemma:overapprox} using (4),(33) and}\\
    &\text{Antecedents (6),(7) and (8), } \store, \hh_C \prec_{\cc_j} \sstore, \hh_{S_j} &&&& \text{Consequent (2)}  &\\
    &\textbf{\text{Induction Case: }} \height(\sval'_1) > 0 \\
    &\sval'_1 = if(\sval', \sval''_1, \sval''_2) && && \mycounter &\\
    &\text{From induction hypothesis on $\height(\sval'_1)$}\\
    & \foo(\tau_s, \fstore, \sstore, \hh''_{S}, \cc'', \prop, f_c, \sval_1) = \hh_{S_1}, \cc_1 && && \mycounter &\\
    &\text{From induction hypothesis on $\height(\sval'_1)$}\\
    &\foo(\tau_s, \fstore, \sstore, \hh_{S_1}, \cc_1, \prop, f_c, \sval_2) = \hh_{S_2}, \cc_2 &&&& \mycounter &\\
    &\text{From (35),(36) and \textsc{G-map-r}}\\
    & \tau_s, \fstore, \sstore_j, \hh_{S}, \cc \vdash \expr.\map(f) \leadsto \hh_{S_2}, \cc_2&& \text{Consequent (1)} && \mycounter &\\
    &\text{From lemma~\ref{lemma:overapprox} using (4),(37) and}\\
    &\text{Antecedents (6),(7) and (8)}\\
    &\store, \hh_C \prec_{\cc_2} \sstore, \hh_{S_2} && \text{Consequent (2)} && &\\
\end{align*}

\begin{proof}[Proof for theorem ~\ref{soundnesstheorem}]
% Using the definitions and the lemmas stated above, we can finally prove Theorem~\ref{soundnesstheorem}. 
    From Lemma~\ref{lemma:graphexpand}, we can conclude that for each abstract transformer $\theta$ specified in a \oldtool program $\Pi$, we can create an expanded symbolic DNN that can over-approximate a given concrete DNN that is within the bounds of the verification procedure. From Lemma~\ref{lemma:overapprox}, the verification procedure and the operational semantics execute the abstract transformer on the symbolic DNN and the concrete DNN respectively, to output a tuple of symbolic and concrete values respectively, representing the new abstract shape, with the symbolic values over-approximating the concrete values. Hence if the soundness property holds on the over-approximated symbolic output, i.e., the abstract transformer is sound on the symbolic DNN, then the transformer is also sound on the concrete DNN.
    % We have to prove the soundness of the verification procedure for each abstract transformer in the program. For each expression, $\expr$, in the output of an abstract transformer, from Lemmas ~\ref{gexpand} and ~\ref{symbolicoverapprox}, $\langle \expr, \fstore, \store, \hh_C \rangle \Downarrow \val$ and a symbolic graph over-approximating the concrete DNN can be created and expanded. From Lemma ~\ref{symbolicoverapprox}, the symbolic execution terminates and produces a symbolic value and conditions, $\sval, \cc''$, where $\sval$ over-approximates $\val$ under the conditions, $\cc''$. If the verification procedure proves the soundness of the abstract transformer, then $\cc'' \wedge (\curr = op(\prev)) \wedge (\curr = \curr') \implies \prop(\curr')$. Since the expanded symbolic graph over-approximates the concrete DNN and each shape member's newly computed symbolic value over-approximates each shape member's newly computed concrete value, the abstract transformer is sound for all concrete DNNs that can be represented by the expanded symbolic graph.
\end{proof}

\begin{proof}[Proof for theorem ~\ref{completenesstheorem}]
   For all of the constructs except for \solver and \traverse, the symbolic semantics are exact w.r.t. the operational semantics. So, we can prove, using structural induction, that given the output of symbolic evaluation of an expression, $e$, $\sval, \cc$, defined on $\vars(\hh_S)$, for every assignment to $\vars(\hh_S)$, that satisfies $\cc$, there exists a concrete DNN, $\hh_C$, s.t. concretely evaluating $e$ using the operational semantics will output $v$ equal to the value of $\sval$ under the given assignment to $\vars(\hh_S)$. If the \cf verification procedure fails, it will return a satisfying assignment to $\vars(\hh_S)$, which can be mapped to a concrete DNN that does not satisfy the over-approximation based soundness property. 
\end{proof}
\clearpage
\section{DNN Certifiers used for Evaluation}
\label{appendix:casestudies}
In this section, we provide the details and the \oldtool codes for the DNN certifiers and the transformers used in the evaluation in \S~\ref{sec:evaluation}.

\subsection{Dataset for Evaluation of Unsound Transformers}
\label{appendix:unsoundness}
To create the dataset for evaluating the detection of unsound behavior in the DNN certifiers, we randomly introduced bugs in the existing as well as the new certifiers defined in the evaluation. The following bugs were injected to create unsound behavior:
\begin{enumerate}
    \item Changing the operations to other operations with similar types of operands. For instance, changing + to -, max to min, etc.
    \item Changing the shape member to another shape member with the same type. For instance, changing \curr[l] to \curr[u].
    \item Changing function calls to other functions with the same signature.
    \item Changing the neurons, for example, \prev to \curr, when \prev represents a single neuron. 
\end{enumerate}

\subsection{\oldtool codes for Sigmoid and Tanh}
\label{appendix:sigmoid}
In the following code, we give the specifications for Sigmoid and Tanh transformers. Similar specifications can be given for transformers corresponding to Exponential and Reciprocal DNN operations. An attention layer is a composition of these primitive operations. 
The abstract transformers corresponding to all of these primitive operations can be specified in \oldtool. However, only the transformers for which the verification queries can fit into a decidable SMT theory can be verified by the \cf verification procedure.   In the following, we present the transformers for Sigmoid and Tanh operations for the DeepPoly certifier. These transformers for other DNN certifiers are quite similar and are thus, omitted here.

\begin{lstlisting}[escapeinside={(*@}{@*)}]
Def Shape as (Real l, Real u, PolyExp L, PolyExp U){[(curr[l]<=curr),(curr[u]>=curr),(curr[L]<=curr),(curr[U]>=curr)]};

Func Sigmoid_deriv(Real x) = 1 - Sigmoid(x);
Func Tanh_deriv(Real x) = 1 - Tanh(x)*Tanh(x);

Func lambda_s(Real l, Real u) = (Sigmoid(u) - Sigmoid(l)) / (u-l);
Func lambda_t(Real l, Real u) = (Tanh(u) - Tanh(l)) / (u-l);

Func lambda_p_s(Real l, Real u) = min(Sigmoid_deriv(l), Sigmoid_deriv(u));
Func lambda_p_t(Real l, Real u) = min(Tanh_deriv(l), Tanh_deriv(u));

Func f(Neuron n1, Neuron n2) = n1[l] >= n2[u];

Transformer DeepPoly{
    Sigmoid -> (Sigmoid(prev[l]), 
        Sigmoid(prev[u]), 
        (prev[l]>0 ? Sigmoid(prev[l]) + lambda_s(prev[l], prev[u])*(prev-prev[l]) : Sigmoid(prev[l]) + lambda_p_s(prev[l], prev[u])*(prev-prev[l])), 
        (prev[u]<=0 ? Sigmoid(prev[u]) + lambda_s(prev[l], prev[u])*(prev-prev[u]) : Sigmoid(prev[u]) + lambda_p_s(prev[l], prev[u])*(prev-prev[u])));

    Tanh -> (Tanh(prev[l]), 
        Tanh(prev[u]), 
        (prev[l]>0 ? Tanh(prev[l]) + lambda_s(prev[l], prev[u])*(prev-prev[l]) : Tanh(prev[l]) + lambda_p_s(prev[l], prev[u])*(prev-prev[l])), 
        (prev[u]<=0 ? Tanh(prev[u]) + lambda_t(prev[l], prev[u])*(prev-prev[u]) : Tanh(prev[u]) + lambda_p_t(prev[l], prev[u])*(prev-prev[u])));
}

Flow(forward, priority, true, DeepPoly);\end{lstlisting}

\subsection{\oldtool codes for State-of-the-art DNN Certifiers}
\label{appendix:existingcert}
In the following case studies, we show the ConstraintFlow code for the implementations of different DNN certifiers. We show the transformers for the DNN operations which can be verified by the ConstraintFlow verification procedure. To avoid clutter, we show only the operations - \affine, \maxpool, \relu, \abs, and \hswish.

\subsubsection{\textbf{DeepPoly}}
Following is the code for the DeepPoly certifier
\begin{lstlisting}[escapeinside={(*@}{@*)}]
Def Shape as (Float l, Float u, PolyExp L, PolyExp U){[(curr[l]<=curr),(curr[u]>=curr),(curr[L]<=curr),(curr[U]>=curr)]};

Func simplify_lower(Neuron n, Float coeff) = (coeff >= 0) ? (coeff * n[l]) : (coeff * n[u]);
Func simplify_upper(Neuron n, Float coeff) = (coeff >= 0) ? (coeff * n[u]) : (coeff * n[l]);

Func replace_lower(Neuron n, Float coeff) = (coeff >= 0) ? (coeff * n[L]) : (coeff * n[U]);
Func replace_upper(Neuron n, Float coeff) = (coeff >= 0) ? (coeff * n[U]) : (coeff * n[L]);

Func priority(Neuron n) = n[layer];
Func priority2(Neuron n) = -n[layer];

Func stop(Int x, Neuron n, Float coeff) = true;

Func backsubs_lower(PolyExp e, Neuron n, Int x) = (e.traverse(backward, priority2, stop(x), replace_lower){e <= n}).map(simplify_lower);
Func backsubs_upper(PolyExp e, Neuron n, Int x) = (e.traverse(backward, priority2, stop(x), replace_upper){e >= n}).map(simplify_upper);

Func f(Neuron n1, Neuron n2) = n1[l] >= n2[u];

Func slope(Float x1, Float x2) = ((x1 * (x1 + 3))-(x2 * (x2 + 3))) / (6 * (x1-x2));
Func intercept(Float x1, Float x2) = x1 * ((x1 + 3) / 6) - (slope(x1, x2) * x1);

Func f1(Float x) = x < 3 ? x * ((x + 3) / 6) : x;
Func f2(Float x) = x * ((x + 3) / 6);
Func f3(Neuron n) = max(f2(n[l]), f2(n[u]));

Func compute_l(Neuron n1, Neuron n2) = min([n1[l]*n2[l], n1[l]*n2[u], n1[u]*n2[l], n1[u]*n2[u]]);
Func compute_u(Neuron n1, Neuron n2) = max([n1[l]*n2[l], n1[l]*n2[u], n1[u]*n2[l], n1[u]*n2[u]]);

Transformer deeppoly{
    Affine -> (backsubs_lower(prev.dot(curr[weight]) + curr[bias], curr, curr[layer]), backsubs_upper(prev.dot(curr[weight]) + curr[bias], curr, curr[layer]), prev.dot(curr[weight]) + curr[bias], prev.dot(curr[weight]) + curr[bias]);

    Maxpool -> len(argmax(prev, f)) > 0 ? (max(prev[l]), max(prev[u]),  avg(argmax(prev, f)), avg(argmax(prev, f))) : (max(prev[l]), max(prev[u]), max(prev[l]), max(prev[u]));

    Relu -> ((prev[l]) >= 0) ? ((prev[l]), (prev[u]), (prev), (prev)) : (((prev[u]) <= 0) ? (0, 0, 0, 0) : (0, (prev[u]), 0, (((prev[u]) / ((prev[u]) - (prev[l]))) * (prev)) - (((prev[u]) * (prev[l])) / ((prev[u]) - (prev[l]))) ));

    Abs -> ((prev[l]) >= 0) ? ((prev[l]), (prev[u]), (prev), (prev)) : (((prev[u]) <= 0) ? (0-(prev[u]), 0-(prev[l]), 0-(prev), 0-(prev)) : (0, max(prev[u], 0-prev[l]), prev, prev*(prev[u]+prev[l])/(prev[u]-prev[l]) - (((2*prev[u])*prev[l])/(prev[u]-prev[l]))) );
    
    HardSwish -> (prev[l] < -3) ? 
                    (prev[u] < -3 ? 
                        (0, 0, 0, 0) : 
                        (prev[u] < 0 ? 
                            (-3/8, 0, -3/8, 0) : 
                            (-3/8, f1(prev[u]), -3/8, f1(prev[u]) * (prev - prev[l])))) : 
                    ((prev[l] < 3) ? 
                        ((prev[u] < 3) ? 
                            (-3/8, f3(prev), -3/8, slope(prev[u], prev[l]) * prev + intercept(prev[u], prev[l])) : 
                            (-3/8, prev[u], -3/8, prev[u] * ((prev + 3) / (prev[u] + 3)))) :
                        (prev[l], prev[u], prev, prev)); 

}

flow(forward, priority, true, deeppoly);\end{lstlisting}

\clearpage
\subsubsection{\textbf{Vegas}}
Following is the code for the Vegas certifier. Vegas uses a forward analysis followed by a backward analysis, both using different transformers. The metadata $\equations$ is used to refer to a list of equations relating the neurons in the current layer to the neurons in \prev. For the $\revaffine$ operation, the code uses the \solver construct to find the minimum and maximum value of the neuron given the bounds of the neurons in \prev.
\begin{lstlisting}
Def Shape as (Float l, Float u, PolyExp L, PolyExp U){[(curr[l]<=curr),(curr[u]>=curr),(curr[L]<=curr),(curr[U]>=curr)]};

Func simplify_lower(Neuron n, Float coeff) = (coeff >= 0) ? (coeff * n[l]) : (coeff * n[u]);
Func simplify_upper(Neuron n, Float coeff) = (coeff >= 0) ? (coeff * n[u]) : (coeff * n[l]);

Func replace_lower(Neuron n, Float coeff) = (coeff >= 0) ? (coeff * n[L]) : (coeff * n[U]);
Func replace_upper(Neuron n, Float coeff) = (coeff >= 0) ? (coeff * n[U]) : (coeff * n[L]);

Func priority(Neuron n) = n[layer];
Func priority2(Neuron n) = -n[layer];

Func stop(Int x, Neuron n, Float coeff) = true;

Func backsubs_lower(PolyExp e, Neuron n, Int x) = (e.traverse(backward, priority2, stop(x), replace_lower){e <= n}).map(simplify_lower);
Func backsubs_upper(PolyExp e, Neuron n, Int x) = (e.traverse(backward, priority2, stop(x), replace_upper){e >= n}).map(simplify_upper);

Func f(Neuron n1, Neuron n2) = n1[l] >= n2[u];

Func slope(Float x1, Float x2) = ((x1 * (x1 + 3))-(x2 * (x2 + 3))) / (6 * (x1-x2));
Func intercept(Float x1, Float x2) = x1 * ((x1 + 3) / 6) - (slope(x1, x2) * x1);

Func f1(Float x) = x < 3 ? x * ((x + 3) / 6) : x;
Func f2(Float x) = x * ((x + 3) / 6);
Func f3(Neuron n) = max(f2(n[l]), f2(n[u]));

Func compute_l(Neuron n1, Neuron n2) = min([n1[l]*n2[l], n1[l]*n2[u], n1[u]*n2[l], n1[u]*n2[u]]);
Func compute_u(Neuron n1, Neuron n2) = max([n1[l]*n2[l], n1[l]*n2[u], n1[u]*n2[l], n1[u]*n2[u]]);

Func create_c(Neuron n, PolyExp e) = n == e;

Transformer vegas_forward{
    Affine -> (backsubs_lower(prev.dot(curr[weight]) + curr[bias], curr, curr[layer]), backsubs_upper(prev.dot(curr[weight]) + curr[bias], curr, curr[layer]), prev.dot(curr[weight]) + curr[bias], prev.dot(curr[weight]) + curr[bias]);

    Maxpool -> len(argmax(prev, f)) > 0 ? (max(prev[l]), max(prev[u]),  avg(argmax(prev, f)), avg(argmax(prev, f))) : (max(prev[l]), max(prev[u]), max(prev[l]), max(prev[u]));

    Relu -> ((prev[l]) >= 0) ? ((prev[l]), (prev[u]), (prev), (prev)) : (((prev[u]) <= 0) ? (0, 0, 0, 0) : (0, (prev[u]), 0, (((prev[u]) / ((prev[u]) - (prev[l]))) * (prev)) - (((prev[u]) * (prev[l])) / ((prev[u]) - (prev[l]))) ));

    Abs -> ((prev[l]) >= 0) ? ((prev[l]), (prev[u]), (prev), (prev)) : (((prev[u]) <= 0) ? (0-(prev[u]), 0-(prev[l]), 0-(prev), 0-(prev)) : (0, max(prev[u], 0-prev[l]), prev, prev*(prev[u]+prev[l])/(prev[u]-prev[l]) - (((2*prev[u])*prev[l])/(prev[u]-prev[l]))) );
    
    HardSwish -> (prev[l] < -3) ? 
                    (prev[u] < -3 ? 
                        (0, 0, 0, 0) : 
                        (prev[u] < 0 ? 
                            (-3/8, 0, -3/8, 0) : 
                            (-3/8, f1(prev[u]), -3/8, f1(prev[u]) * (prev - prev[l])))) : 
                    ((prev[l] < 3) ? 
                        ((prev[u] < 3) ? 
                            (-3/8, f3(prev), -3/8, slope(prev[u], prev[l]) * prev + intercept(prev[u], prev[l])) : 
                            (-3/8, prev[u], -3/8, prev[u] * ((prev + 3) / (prev[u] + 3)))) :
                        (prev[l], prev[u], prev, prev)); 

}

Transformer vegas_backward{
    rev_Affine -> (lp(minimize, curr, (curr[equations].map_list(create_c curr))), lp(maximize, curr, (curr[equations].map_list(create_c curr))), curr[L], curr[U]);
    rev_Maxpool -> (curr[l], min(curr[u], min(prev[u])), curr[L], curr[U]);
    rev_Relu -> 
    (prev[l]) > 0 ? 
    (
        (prev[u]) >= 0 ? 
        (max((prev[l]), curr[l]), min((prev[u]), curr[u]), curr[L], curr[U]) : 
        (max((prev[l]), curr[l]),curr[u], curr[L], curr[U])
    ) : 
    (
        (prev[u]) >= 0 ? 
        (curr[l], min((prev[u]), curr[u]), curr[L], curr[U]) : 
        (curr[l], curr[u], curr[L], curr[U])
    );
    rev_Abs -> (max(-prev[u], curr[l]), min(prev[u], curr[u]), curr[L], curr[U]);
    rev_HardSwish -> prev[l] >= 3 ? 
        (max(prev[l], curr[l]), min(prev[u], curr[u]), curr[L], curr[U]) :
        (prev[l] > 0 ? 
            (prev[u] >= 3 ?
                (max(prev[l] / 2, curr[l]), curr[u], curr[L], curr[U]) :
                (max(prev[l] / 2, curr[l]), min(2 * prev[u], curr[u]), curr[L], curr[U])
            ) : 
            (prev[u] >= 0 ? 
                (curr[l], curr[u], curr[L], curr[U]):
                (min(0, max(-3, curr[l])), min(0, max(-3, curr[u])), curr[L], curr[U])
            )

        );
}

flow(forward, priority, true, vegas_forward);
flow(backward, -priority, true, vegas_backward);    \end{lstlisting}

\clearpage
\subsubsection{\textbf{DeepZ}}
Following is the correct code for DeepZ certifier. The abstract shape has two concrete values which serve as the lower and upper concrete bounds. Additionally, it has a zonotope expression, which can be represented as a \symexp in \oldtool. The \cf verification procedure is able to prove the soundness of this Transformer.
\begin{lstlisting}
Def Shape as (Float l, Float u, ZonoExp z){[(curr[u]>=curr),(curr In curr[z]),(curr[l]<=curr)]};

Func simplify_lower(Neuron n, Float coeff) = (coeff >= 0) ? (coeff * n[l]) : (coeff * n[u]);
Func simplify_upper(Neuron n, Float coeff) = (coeff >= 0) ? (coeff * n[u]) : (coeff * n[l]);

Func priority(Neuron n) = n[layer];
Func f(Neuron n1, Neuron n2) = n1[l] >= n2[u];

Func s1(Float x1, Float x2) = ((x1 * (x1 + 3))-(x2 * (x2 + 3))) / (6 * (x1-x2));
Func i1(Float x1, Float x2) = x1 * ((x1 + 3) / 6) - (s1(x1, x2) * x1);

Func f1(Float x) = x < 3 ? x * ((x + 3) / 6) : x;
Func f2(Float x) = x * ((x + 3) / 6);

Transformer DeepZ{
    Affine -> ((prev.dot(curr[weight]) + curr[bias]).map(simplify_lower), (prev.dot(curr[weight]) + curr[bias]).map(simplify_upper), prev[z].dot(curr[weight]) + (curr[bias]));
    Maxpool -> len(argmax(prev, f)) > 0 ? (max(prev[l]), max(prev[u]),  avg(argmax(prev, f)[z])) :
     (max(prev[l]), max(prev[u]), ((max(prev[u]) + max(prev[l])) / 2) + (((max(prev[u]) - max(prev[l])) / 2) * eps));
    Relu -> ((prev[l]) >= 0) ? ((prev[l]), (prev[u]), (prev[z])) : (((prev[u]) <= 0) ? (0, 0, 0) : (0, (prev[u]), ((prev[u]) / 2) + (((prev[u]) / 2) * eps)));
    Abs -> ((prev[l]) >= 0) ? 
                ((prev[l]), (prev[u]), (prev[z])) : 
                (((prev[u]) <= 0) ? 
                    (-(prev[u]), -(prev[l]), -(prev[z])) : 
                    (0, max(-prev[l], prev[u]), ((max(-prev[l], prev[u])) / 2) + (((max(-prev[l], prev[u])) / 2) * eps)));
    HardSwish -> (prev[l] < -3) ? 
                    (prev[u] < -3 ? 
                        (0, 0, 0) : 
                        (prev[u] < 0 ? 
                            (-3/8, 0, (-3/16) * (1 - eps)) : 
                            (-3/8, f1(prev[u]), (f1(prev[u])/2 - (3/16)) + ((f1(prev[u])/2 + (3/16)) * eps) ))) : 
                    ((prev[l] < 3) ? 
                        ((prev[u] < 3) ? 
                            (-3/8, max(f2(prev[l]), f2(prev[u])), ((max(f2(prev[l]), f2(prev[u]))/2 )- (3/16)) + (eps * (max(f2(prev[l]), f2(prev[u]))/2 + (3/16)))) : 
                            (-3/8, prev[u], (prev[u]/2 - (3/16)) + (eps * (prev[u]/2 + (3/16))) )) :
                        (prev[l], prev[u], prev[z])); 
}

flow(forward, priority, true, DeepZ);\end{lstlisting}

\clearpage
\subsubsection{\textbf{RefineZono}}
Following is the correct code for RefineZono certifier. The abstract shape has two concrete values which serve as the lower and upper concrete bounds, a zonotope expression, which is represented as a \symexp in \oldtool, and a constraint of the type $\ct$. The \cf verification procedure is able to prove the soundness of this Transformer.
% \begin{figure}[H]
    \begin{lstlisting}
Def Shape as (Float l, Float u, ZonoExp z, Ct c){[(curr[l]<=curr),(curr[u]>=curr),(curr In curr[z]),curr[c]]};

Func simplify_lower(Neuron n, Float coeff) = (coeff >= 0) ? (coeff * n[l]) : (coeff * n[u]);
Func simplify_upper(Neuron n, Float coeff) = (coeff >= 0) ? (coeff * n[u]) : (coeff * n[l]);

Func priority(Neuron n) = n[layer];
Func foo(Neuron n) = n[c];
Func f(Neuron n1, Neuron n2) = n1[l] >= n2[u];

Func s1(Float x1, Float x2) = ((x1 * (x1 + 3))-(x2 * (x2 + 3))) / (6 * (x1-x2));
Func i1(Float x1, Float x2) = x1 * ((x1 + 3) / 6) - (s1(x1, x2) * x1);

Func f1(Float x) = x < 3 ? x * ((x + 3) / 6) : x;
Func f2(Float x) = x * ((x + 3) / 6);

Transformer RefineZono{
    Affine -> (lp(minimize, prev.dot(curr[weight]) + curr[bias], prev.map_list(foo)), lp(maximize, prev.dot(curr[weight]) + curr[bias], prev.map_list(foo)), prev[z].dot(curr[weight]) + (curr[bias]), (prev.dot(curr[weight]) + curr[bias]) == curr);
    Maxpool -> len(argmax(prev, f)) > 0 ? (max(prev[l]), max(prev[u]),  avg(argmax(prev, f)[z]), (curr <= max(prev[u])) and (curr >= max(prev[l]))) :
        (max(prev[l]), max(prev[u]), ((max(prev[u]) + max(prev[l])) / 2) + (((max(prev[u]) - max(prev[l])) / 2) * eps), (curr <= max(prev[u])) and (curr >= max(prev[l])));
    Relu -> (prev[l] >= 0) ? 
        (prev[l], prev[u], prev[z], (prev[l] <= curr) and (prev[u] >= curr)) : 
        (
            (prev[u] <= 0) ? 
            (0, 0, 0, curr == 0) : 
            (0, prev[u], (prev[u] / 2) + ((prev[u] / 2) * eps), 
            (prev[l] <= prev) and (prev[u] >= prev) and 
            (((prev <= 0) and (curr == 0) ) or ((prev > 0) and (curr == prev)) )
            )
        );
    Abs -> (prev[l] >= 0) ? 
        (prev[l], prev[u], prev[z], (prev == curr)) : 
            (prev[u] <= 0) ? 
                (-prev[u], -prev[l], -prev[z], (curr == -prev)) : 
                (0, max(-prev[l], prev[u]), (max(-prev[l], prev[u]) / 2) + ((max(-prev[l], prev[u]) / 2) * eps), 
                (((prev <= 0) and (curr == -prev) ) or ((prev > 0) and (curr == prev)) )
                );
    HardSwish -> (prev[l] < -3) ? 
                (prev[u] < -3 ? 
                    (0, 0, 0, curr==0) : 
                    (prev[u] < 0 ? 
                        (-3/8, 0, (-3/16) * (1 - eps), ((curr >= (-3/8)) and (curr <= 0))) : 
                        (-3/8, f1(prev[u]), (f1(prev[u])/2 - (3/16)) + ((f1(prev[u])/2 + (3/16)) * eps), ((curr >= (-3/8)) and (curr <= f1(prev[u]))) ))) : 
                ((prev[l] < 3) ? 
                    ((prev[u] < 3) ? 
                        (-3/8, max(f2(prev[l]), f2(prev[u])), ((max(f2(prev[l]), f2(prev[u]))/2 )- (3/16)) + (eps * (max(f2(prev[l]), f2(prev[u]))/2 + (3/16))), ((curr >= (-3/8)) and (curr <= max(f2(prev[l]), f2(prev[u]))))) : 
                        (-3/8, prev[u], (prev[u]/2 - (3/16)) + (eps * (prev[u]/2 + (3/16))), ((curr >= (-3/8)) and (curr <= prev[u])) )) :
                    (prev[l], prev[u], prev[z], curr==prev)); 
    
}

flow(forward, priority, true, RefineZono);\end{lstlisting}

\clearpage
\subsubsection{\textbf{IBP}}
Following is the correct code for IBP certifier. The abstract shape has two concrete values which serve as the lower and upper concrete bounds. The \cf verification procedure can prove the soundness of this Transformer.
% \begin{figure}[H]
    \begin{lstlisting}
Def Shape as (Float l, Float u){[(curr[l]<=curr),(curr[u]>=curr)]};

Func simplify_lower(Neuron n, Float coeff) = (coeff >= 0) ? (coeff * n[l]) : (coeff * n[u]);
Func simplify_upper(Neuron n, Float coeff) = (coeff >= 0) ? (coeff * n[u]) : (coeff * n[l]);

Func priority(Neuron n) = n[layer];

Func hswish(Float x) = x <= -3 ? 0 : (x >= 3 ? x : (x * ((x + 3) / 6))); 

Transformer Ibp{
    Affine -> ((prev.dot(curr[weight]) + curr[bias]).map(simplify_lower), (prev.dot(curr[weight]) + curr[bias]).map(simplify_upper));
    Maxpool -> (max(prev[l]), max(prev[u]));
    Relu -> ((prev[l]) >= 0) ? ((prev[l]), (prev[u])) : (((prev[u]) <= 0) ? (0, 0) : (0, (prev[u])));
    Abs -> (((prev[l]) >= 0) ? ((prev[l]), (prev[u])) : (((prev[u]) <= 0) ? (-prev[u], -prev[l]) : (0, max(-prev[l], prev[u]))));
    HardSwish -> prev[u] <= (-3/2) ? (hswish(prev[u]), hswish(prev[l])) : (prev[l] > (-3/2) ? (hswish(prev[l]), hswish(prev[u])) : (-3/8, max(hswish(prev[u]), hswish(prev[l]))));
}

flow(forward, priority, true, Ibp);\end{lstlisting}

\subsubsection{\textbf{Hybrid Zonotope}}
Following is the \oldtool code for Hybrid Zonotope certifier.
    \begin{lstlisting}
Def Shape as (Float l, Float u, Float b, ZonoExp z)
{[curr[b] >= 0, curr[l] <= curr, curr[u] >= curr, curr In (curr[z] + (curr[b]*eps))]};

Func simplify_lower(Neuron n, Float coeff) = (coeff >= 0) ? (coeff * n[l]) : (coeff * n[u]);
Func simplify_upper(Neuron n, Float coeff) = (coeff >= 0) ? (coeff * n[u]) : (coeff * n[l]);

Func replace_abs(Neuron n, Float coeff) = (coeff >= 0) ? (coeff * n[b]) : (-coeff * n[b]);

Func priority(Neuron n) = n[layer];

Func relu(Float r) = r >= 0 ? r : 0;
Func f(Neuron n1, Neuron n2) = n1[l] >= n2[u];
Func abs(Float x) = x > 0 ? x : -x;

Func s1(Float x1, Float x2) = ((x1 * (x1 + 3))-(x2 * (x2 + 3))) / (6 * (x1-x2));
Func i1(Float x1, Float x2) = x1 * ((x1 + 3) / 6) - (s1(x1, x2) * x1);

Func f1(Float x) = x < 3 ? x * ((x + 3) / 6) : x;
Func f2(Float x) = x * ((x + 3) / 6);

Transformer HybridZonotope{
    Neuron_add -> ((prev_0[l] + prev_1[l]), (prev_0[u] + prev_1[u]), (prev_0[b] + prev_1[b]), (prev_0[z] + prev_1[z]));
    Maxpool -> (max(prev[l]), max(prev[u]), max(abs(max(prev[l])), abs(max(prev[u]))),0);
    Relu -> abs(prev[l]) > abs(prev[u]) ? 
                ((prev[l]) >= 0) ? 
                (prev[l], prev[u], (prev[b]), (prev[z])) : 
                (-prev[b], prev[b] + relu(prev[u]), prev[b], ((1 + eps) * (relu(prev[u]/2)))) :
                (((prev[l]) < 0) and ((prev[u]) > 0)) ? 
                (0, prev[u], (prev[b]), (prev[z] - (((1 + eps) * (prev[l])) / 2))) : 
                (((prev[l]) >= 0) ? 
                    (prev[l], prev[u], (prev[b]), (prev[z])) : 
                    (0, 0, prev[b], ((1 + eps) * (relu(prev[u]/2)))));
    Abs -> (prev[l] > 0) ?
                (prev[l], prev[u], prev[b], prev[z]) :
                (prev[u] < 0) ?
                    (-prev[u], -prev[l], prev[b], -prev[z]) : 
                    (0, max(-prev[l], prev[u]), max(-prev[l], prev[u]), 0);
    HardSwish -> (prev[l] < -3) ? 
                    (prev[u] < -3 ? 
                        (0, 0, 0, 0) : 
                        (prev[u] < 0 ? 
                            (-3/8, 0, 3/8, 0) : 
                            (-3/8, f1(prev[u]), max(f1(prev[u]), 3/8), 0))) : 
                    ((prev[l] < 3) ? 
                        ((prev[u] < 3) ? 
                            (-3/8, max(f2(prev[l]), f2(prev[u])),max(3/8, max(f2(prev[l]), f2(prev[u]))), 0) : 
                            (-3/8, prev[u], max(3/8, f1(prev[u])), 0 )) :
                        (prev[l], prev[u], prev[b], prev[z])); 
}

flow(forward, priority, true, HybridZonotope);\end{lstlisting}

\clearpage
\subsection{\oldtool codes for New DNN Certifiers}
\label{appendix:newcert}

\subsubsection{\textbf{\deeppolyH}}
Following is the \oldtool code for \deeppolyH certifier.
\begin{lstlisting}
Def Shape as (Float l, Float u, PolyExp L, PolyExp U){[(curr[l]<=curr),(curr[u]>=curr),(curr[L]<=curr),(curr[U]>=curr)]};

Func simplify_lower(Neuron n, Float coeff) = (coeff >= 0) ? (coeff * n[l]) : (coeff * n[u]);
Func simplify_upper(Neuron n, Float coeff) = (coeff >= 0) ? (coeff * n[u]) : (coeff * n[l]);

Func replace_lower(Neuron n, Float coeff) = (coeff >= 0) ? (coeff * n[L]) : (coeff * n[U]);
Func replace_upper(Neuron n, Float coeff) = (coeff >= 0) ? (coeff * n[U]) : (coeff * n[L]);

Func priority(Neuron n) = n[layer];
Func priority2(Neuron n) = -n[layer];

Func stop(Int x, Neuron n, Float coeff) = n[layer] >= (x - 2);

Func backsubs_lower(PolyExp e, Neuron n, Int x) = (e.traverse(backward, priority2, stop(x), replace_lower){e <= n}).map(simplify_lower);
Func backsubs_upper(PolyExp e, Neuron n, Int x) = (e.traverse(backward, priority2, stop(x), replace_upper){e >= n}).map(simplify_upper);

Func f(Neuron n1, Neuron n2) = n1[l] >= n2[u];
Func abs(Float x) = x > 0 ? x : (-x);

Func s1(Float x1, Float x2) = ((x1 * (x1 + 3))-(x2 * (x2 + 3))) / (6 * (x1-x2));
Func i1(Float x1, Float x2) = x1 * ((x1 + 3) / 6) - (s1(x1, x2) * x1);

Func f1(Float x) = x < 3 ? x * ((x + 3) / 6) : x;
Func f2(Float x) = x * ((x + 3) / 6);

Transformer BalanceCert{
    Affine -> (backsubs_lower(prev.dot(curr[weight]) + curr[bias], curr, curr[layer]), backsubs_upper(prev.dot(curr[weight]) + curr[bias], curr, curr[layer]), prev.dot(curr[weight]) + curr[bias], prev.dot(curr[weight]) + curr[bias]);
    Maxpool -> len(argmax(prev, f)) > 0 ? (max(prev[l]), max(prev[u]),  avg(argmax(prev, f)), avg(argmax(prev, f))) : (max(prev[l]), max(prev[u]), max(prev[l]), max(prev[u]));
    Relu -> ((prev[l]) >= 0) ? ((prev[l]), (prev[u]), ((prev)), ((prev))) : (((prev[u]) <= 0) ? (0.0, 0.0, 0.0, 0.0) : (((abs(prev[l]) < abs(prev[u])) ? (prev[l]) : 0), (prev[u]), ((abs(prev[l]) < abs(prev[u])) ? (prev) : 0), (((prev[u]) / ((prev[u]) - (prev[l]))) * ((prev))) + ((((prev[u] * (-1))) * (prev[l])) / ((prev[u]) - (prev[l]))) ));
    Abs -> ((prev[l]) >= 0) ? 
                ((prev[l]), (prev[u]), (prev), (prev)) : 
                (((prev[u]) <= 0) ? 
                    (-(prev[u]), -(prev[l]), -(prev), -(prev)) : 
                    (0, max(prev[u], -prev[l]), ((-prev[l])>prev[u]) ? -prev : prev, prev*(prev[u]+prev[l])/(prev[u]-prev[l]) - (((2*prev[u])*prev[l])/(prev[u]-prev[l]))) );
    HardSwish -> (prev[l] < -3) ? 
                    (prev[u] < -3 ? 
                        (0, 0, 0, 0) : 
                        (prev[u] < 0 ? 
                            (-3/8, 0, -3/8, 0) : 
                            (-3/8, f1(prev[u]), -3/8, f1(prev[u]) * (prev - prev[l])))) : 
                    ((prev[l] < 3) ? 
                        ((prev[u] < 3) ? 
                            (-3/8, max(f2(prev[l]), f2(prev[u])), -3/8, s1(prev[u], prev[l]) * prev + i1(prev[u], prev[l])) : 
                            (-3/8, prev[u], -3/8, prev[u] * ((prev + 3) / (prev[u] + 3)))) :
                        (prev[l], prev[u], prev, prev)); 
}

flow(forward, priority, true, BalanceCert);\end{lstlisting}

\clearpage
\subsubsection{\textbf{\deeppolyNew}}
Following is the \oldtool code for \deeppolyNew certifier.
\begin{lstlisting}
def Shape as (Float l, Float u, PolyExp L, PolyExp U, PolyExp Lc, PolyExp Uc)
{[(curr[l]<=curr),(curr[u]>=curr),(curr[L]<=curr),(curr[U]>=curr),(curr[Lc]<=curr),(curr[Uc]>=curr)]};

func simplify_lower(Neuron n, Float coeff) = (coeff >= 0) ? (coeff * n[l]) : (coeff * n[u]);
func simplify_upper(Neuron n, Float coeff) = (coeff >= 0) ? (coeff * n[u]) : (coeff * n[l]);

func replace_lower(Neuron n, Float coeff) = (coeff >= 0) ? (coeff * n[L]) : (coeff * n[U]);
func replace_upper(Neuron n, Float coeff) = (coeff >= 0) ? (coeff * n[U]) : (coeff * n[L]);

func replace_lower2(Neuron n, Float coeff) = (coeff >= 0) ? (coeff * n[Lc]) : (coeff * n[Uc]);
func replace_upper2(Neuron n, Float coeff) = (coeff >= 0) ? (coeff * n[Uc]) : (coeff * n[Lc]);

func priority(Neuron n) = n[layer];
func priority2(Neuron n) = -n[layer];

func backsubs_lower(PolyExp e, Neuron n) = (e.traverse(backward, priority2, true, replace_lower){e <= n}).map(simplify_lower);
func backsubs_upper(PolyExp e, Neuron n) = (e.traverse(backward, priority2, true, replace_upper){e >= n}).map(simplify_upper);

func f(Neuron n1, Neuron n2) = n1[l] >= n2[u];
func s1(Float x1, Float x2) = ((x1 * (x1 + 3))-(x2 * (x2 + 3))) / (6 * (x1-x2));
func i1(Float x1, Float x2) = x1 * ((x1 + 3) / 6) - (s1(x1, x2) * x1);

func f1(Float x) = x < 3 ? x * ((x + 3) / 6) : x;
func f2(Float x) = x * ((x + 3) / 6);

transformer ReuseCert{
    Affine -> (backsubs_lower(prev.dot(curr[weight]) + curr[bias], curr).map(simplify_lower), backsubs_upper(prev.dot(curr[weight]) + curr[bias], curr).map(simplify_upper), prev.dot(curr[weight]) + curr[bias], prev.dot(curr[weight]) + curr[bias], (prev.dot(curr[weight]) + curr[bias]).map(replace_lower2), (prev.dot(curr[weight]) + curr[bias]).map(replace_upper2));
    Maxpool -> len(argmax(prev, f)) > 0 ? 
                (max(prev[l]), max(prev[u]), avg(argmax(prev, f)), avg(argmax(prev, f)), avg(argmax(prev, f)[Lc]), avg(argmax(prev, f)[Uc])) : 
                (max(prev[l]), max(prev[u]), max(prev[l]), max(prev[u]), max(prev[l]), max(prev[u]));
    Relu -> ((prev[l]) >= 0) ? 
                ((prev[l]), (prev[u]), (prev), (prev), (prev[Lc]), (prev[Uc])) : 
                (((prev[u]) <= 0) ? 
                    (0, 0, 0, 0, 0, 0) : 
                    (0, (prev[u]), 0, (((prev[u]) / ((prev[u]) - (prev[l]))) * (prev)) - (((prev[u]) * (prev[l])) / ((prev[u]) - (prev[l]))), 0, (((prev[u]) / ((prev[u]) - (prev[l]))) * (prev[Uc])) - (((prev[u]) * (prev[l])) / ((prev[u]) - (prev[l]))) ));

    Abs -> ((prev[l]) >= 0) ? 
                ((prev[l]), (prev[u]), (prev), (prev), (prev[Lc]), (prev[Uc])) : 
                (((prev[u]) <= 0) ? 
                    (-prev[u], -prev[l], -prev, -prev, -prev[Uc], -prev[Lc]) : 
                    (0, max(-prev[l], prev[u]), 0, prev*(prev[u]+prev[l])/(prev[u]-prev[l]) - (((2*prev[u])*prev[l])/(prev[u]-prev[l])), 0, ((-prev[l])>prev[u] ? prev[Lc] : prev[Uc])*(prev[u]+prev[l])/(prev[u]-prev[l]) - (((2*prev[u])*prev[l])/(prev[u]-prev[l]))));

    HardSwish -> (prev[l] <= -3) ? 
                    (prev[u] <= -3 ? 
                        (0, 0, 0, 0, 0, 0) : 
                        (prev[u] <= 0 ? 
                            (-3/8, 0, -3/8, 0, -3/8, 0) : 
                            (-3/8, f1(prev[u]), -3/8, f1(prev[u]) * (prev - prev[l]), -3/8, f1(prev[u]) * (prev[Uc] - prev[l])))) : 
                    ((prev[l] <= 3) ? 
                        ((prev[u] <= 3) ? 
                            (-3/8, max(f2(prev[l]), f2(prev[u])), -3/8, s1(prev[u], prev[l]) * prev + i1(prev[u], prev[l]), -3/8, max(f2(prev[l]), f2(prev[u]))) : 
                            (-3/8, prev[u], -3/8, prev[u] * ((prev + 3) / (prev[u] + 3)), -3/8, prev[u] * ((prev[Uc] + 3) / (prev[u] + 3)))) :
                        (prev[l], prev[u], prev, prev, prev, prev)); 
}

flow(forward, priority, true, ReuseCert);\end{lstlisting}

\clearpage
\subsubsection{\textbf{\sympoly}}
Following is the \oldtool code for \sympoly certifier.
\begin{lstlisting}
Def Shape as (Float l, Float u, PolyExp L, PolyExp U){[(curr[l]<=curr),(curr[u]>=curr),(curr[L]<=curr),(curr[U]>=curr)]};

Func simplify_lower(Neuron n, Float coeff) = (coeff >= 0) ? (coeff * n[l]) : (coeff * n[u]);
Func simplify_upper(Neuron n, Float coeff) = (coeff >= 0) ? (coeff * n[u]) : (coeff * n[l]);

Func replace_lower(Neuron n, Float coeff) = (coeff >= 0) ? (coeff * n[L]) : (coeff * n[U]);
Func replace_upper(Neuron n, Float coeff) = (coeff >= 0) ? (coeff * n[U]) : (coeff * n[L]);

Func priority(Neuron n) = n[layer];
Func priority2(Neuron n) = -n[layer];

Func stop(Int x, Neuron n, Float coeff) = true;

Func backsubs_lower(PolyExp e, Neuron n, Int x) = (e.traverse(backward, priority2, stop(x), replace_lower){e <= n}).map(simplify_lower);
Func backsubs_upper(PolyExp e, Neuron n, Int x) = (e.traverse(backward, priority2, stop(x), replace_upper){e >= n}).map(simplify_upper);

Func f(Neuron n1, Neuron n2) = n1[l] >= n2[u];

Func s1(Float x1, Float x2) = ((x1 * (x1 + 3))-(x2 * (x2 + 3))) / (6 * (x1-x2));
Func i1(Float x1, Float x2) = x1 * ((x1 + 3) / 6) - (s1(x1, x2) * x1);

Func f1(Float x) = x < 3 ? x * ((x + 3) / 6) : x;
Func f2(Float x) = x * ((x + 3) / 6);

Transformer SymPoly{
    Affine -> (backsubs_lower(prev.dot(curr[weight]) + curr[bias], curr, curr[layer]), backsubs_upper(prev.dot(curr[weight]) + curr[bias], curr, curr[layer]), prev.dot(curr[weight]) + curr[bias], prev.dot(curr[weight]) + curr[bias]);
    Maxpool -> len(argmax(prev, f)) > 0 ? (max(prev[l]), max(prev[u]),  avg(argmax(prev, f)), avg(argmax(prev, f))) : (max(prev[l]), max(prev[u]), max(prev[l]), max(prev[u]));
    Relu -> ((prev[l]) >= 0) ? ((prev[l]), (prev[u]), (prev), (prev)) : (((prev[u]) <= 0) ? (0, 0, 0, 0) : (0, (prev[u]), ((((1 + eps) / 2)) * (prev)), (((prev[u]) / ((prev[u]) - (prev[l]))) * (prev)) - (((prev[u]) * (prev[l])) / ((prev[u]) - (prev[l]))) ));
    Abs -> ((prev[l]) >= 0) ? ((prev[l]), (prev[u]), (prev), (prev)) : (((prev[u]) <= 0) ? (0-(prev[u]), 0-(prev[l]), 0-(prev), 0-(prev)) : (0, max(prev[u], 0-prev[l]), (eps * prev), prev*(prev[u]+prev[l])/(prev[u]-prev[l]) - (((2*prev[u])*prev[l])/(prev[u]-prev[l]))) );
    HardSwish -> (prev[l] < -3) ? 
                    (prev[u] < -3 ? 
                        (0, 0, 0, 0) : 
                        (prev[u] < 0 ? 
                            (-3/8, 0, -3/8, 0) : 
                            (-3/8, f1(prev[u]), -3/8, f1(prev[u]) * (prev - prev[l])))) : 
                    ((prev[l] < 3) ? 
                        ((prev[u] < 3) ? 
                            (-3/8, max(f2(prev[l]), f2(prev[u])), -3/8, s1(prev[u], prev[l]) * prev + i1(prev[u], prev[l])) : 
                            (-3/8, prev[u], -3/8, prev[u] * ((prev + 3) / (prev[u] + 3)))) :
                        (prev[l], prev[u], prev, prev)); 
}

flow(forward, priority, true, SymPoly);\end{lstlisting}

\clearpage
\subsubsection{\textbf{PolyZ}}
Following is the \oldtool code for PolyZ certifier
    \begin{lstlisting}
Def Shape as (Float l, Float u, PolyExp L, PolyExp U, ZonoExp z){[curr[l]<=curr,curr[u]>=curr,curr[L]<=curr,curr[U]>=curr,curr In curr[z]]};

Func simplify_lower(Neuron n, Float coeff) = (coeff >= 0) ? (coeff * n[l]) : (coeff * n[u]);
Func simplify_upper(Neuron n, Float coeff) = (coeff >= 0) ? (coeff * n[u]) : (coeff * n[l]);

Func replace_lower(Neuron n, Float coeff) = (coeff >= 0) ? (coeff * n[L]) : (coeff * n[U]);
Func replace_upper(Neuron n, Float coeff) = (coeff >= 0) ? (coeff * n[U]) : (coeff * n[L]);

Func priority(Neuron n) = n[layer];
Func priority2(Neuron n) = -n[layer];

Func backsubs_lower(PolyExp e, Neuron n) = (e.traverse(backward, priority2, true, replace_lower){e <= n}).map(simplify_lower);
Func backsubs_upper(PolyExp e, Neuron n) = (e.traverse(backward, priority2, true, replace_upper){e >= n}).map(simplify_upper);

Func f(Neuron n1, Neuron n2) = n1[l] >= n2[u];
Func f1(Float x) = x < 3 ? x * ((x + 3) / 6) : x;
Func f2(Float x) = x * ((x + 3) / 6);

Func s1(Float x1, Float x2) = ((x1 * (x1 + 3))-(x2 * (x2 + 3))) / (6 * (x1-x2));
Func i1(Float x1, Float x2) = x1 * ((x1 + 3) / 6) - (s1(x1, x2) * x1);

Transformer polyz{
    Affine -> (max((prev.dot(curr[weight]) + curr[bias]).map(simplify_lower),backsubs_lower(prev.dot(curr[weight]) + curr[bias], curr)), min((prev.dot(curr[weight]) + curr[bias]).map(simplify_upper),backsubs_upper(prev.dot(curr[weight]) + curr[bias], curr)), prev.dot(curr[weight]) + curr[bias], prev.dot(curr[weight]) + curr[bias], prev[z].dot(curr[weight]) + curr[bias]);
    Maxpool -> len(argmax(prev, f)) > 0 ? (max(prev[l]), max(prev[u]), avg(argmax(prev, f)), avg(argmax(prev, f)), avg(argmax(prev, f)[z])) : (max(prev[l]), max(prev[u]), max(prev[l]), max(prev[u]), ((max(prev[u]) + max(prev[l])) / 2) + (((max(prev[u]) - max(prev[l])) / 2) * eps));
    Relu -> ((prev[l]) >= 0) ? 
    ((prev[l]), (prev[u]), (prev), (prev), (prev[z])) : 
    (
        ((prev[u]) <= 0) ? 
        (0, 0, 0, 0, 0) : 
        (0, (prev[u]), 0, (((prev[u]) / ((prev[u]) - (prev[l]))) * (prev)) - (((prev[u]) * (prev[l])) / ((prev[u]) - (prev[l]))), ((prev[u] + prev[l]) / 2) + (((prev[u] - prev[l]) / 2) * eps))
    );
    Abs -> ((prev[l]) >= 0) ? 
    ((prev[l]), (prev[u]), (prev), (prev), (prev[z])) : 
    (
        ((prev[u]) <= 0) ? 
        (-prev[u], -prev[l], -prev, -prev, -prev[z]) : 
        (0, max(-prev[l], prev[u]), 0, prev*(prev[u]+prev[l])/(prev[u]-prev[l]) - (((2*prev[u])*prev[l])/(prev[u]-prev[l])), ((max(-prev[l], prev[u])) / 2) + (((max(-prev[l], prev[u])) / 2) * eps))
    );
    HardSwish -> (prev[l] < -3) ? 
                    (prev[u] < -3 ? 
                        (0, 0, 0, 0, 0) : 
                        (prev[u] < 0 ? 
                            (-3/8, 0, -3/8, 0, (-3/16) * (1 - eps)) : 
                            (-3/8, f1(prev[u]), -3/8, f1(prev[u]) * (prev - prev[l]), (f1(prev[u])/2 - (3/16)) + ((f1(prev[u])/2 + (3/16)) * eps)))) : 
                    ((prev[l] < 3) ? 
                        ((prev[u] < 3) ? 
                            (-3/8, max(f2(prev[l]), f2(prev[u])), -3/8, s1(prev[u], prev[l]) * prev + i1(prev[u], prev[l]), ((max(f2(prev[l]), f2(prev[u]))/2 )- (3/16)) + (eps * (max(f2(prev[l]), f2(prev[u]))/2 + (3/16)))) : 
                            (-3/8, prev[u], -3/8, prev[u] * ((prev + 3) / (prev[u] + 3)), (prev[u]/2 - (3/16)) + (eps * (prev[u]/2 + (3/16))))) :
                        (prev[l], prev[u], prev, prev, prev[z])); 
}

flow(forward, priority, true, polyz);\end{lstlisting}

\end{document}